%% file: main.tex
\documentclass[
    reprint,
    aps,
    prx,
    twocolumn,
    nofootinbib
]{revtex4-2}

\usepackage{tikz}
\usetikzlibrary{backgrounds}
\usetikzlibrary{arrows}
\usetikzlibrary{shapes,shapes.geometric,shapes.misc}
\usetikzlibrary{decorations.pathmorphing}
\usetikzlibrary{fadings}
\usetikzlibrary{decorations.markings}
\usetikzlibrary{decorations.pathreplacing,calligraphy}
\tikzstyle{tikzfig}=[baseline=-0.25em,scale=0.5]
\pgfkeys{/tikz/tikzit fill/.initial=0}
\pgfkeys{/tikz/tikzit draw/.initial=0}
\pgfkeys{/tikz/tikzit shape/.initial=0}
\pgfkeys{/tikz/tikzit category/.initial=0}
\pgfdeclarelayer{edgelayer}
\pgfdeclarelayer{nodelayer}
\pgfsetlayers{background,edgelayer,nodelayer,main}
\tikzstyle{none}=[inner sep=0mm]
\tikzstyle{every picture}=[tikzfig]
\tikzstyle{every loop}=[]
\tikzstyle{new style 0}=[fill=black, draw=black, shape=circle, scale=0.4]
\tikzstyle{new edge style 0}=[-, draw={rgb,255: red,32; green,172; blue,215}]
\input{quantum.tikzstyles}

\usepackage{graphicx}
\usepackage{bm}
\usepackage{physics}
\usepackage{amsthm}
\usepackage{cancel}
\newtheorem{theorem}{Theorem}
\newtheorem*{problem}{Problem}
\usepackage{amssymb}
\newcommand{\im}{\mathrm{i}}
\newcommand{\e}{\mathrm{e}}
\newcommand{\poly}{\mathrm{poly}}
\usepackage{booktabs}
\usepackage{svg}
\usepackage{multirow}
\usepackage[ruled,vlined,linesnumbered]{algorithm2e}
\usepackage{hyperref}
\hypersetup{colorlinks=true, citecolor=blue}

\usepackage{soul}

\begin{document}

\title{End-to-End Quantum Algorithms for the Jones Polynomial}

\author{Tuomas Laakkonen$^1$}
\email[Correspondence to: ]{tsrl@mit.edu}
\author{Enrico Rinaldi$^1$}
\author{Chris N. Self$^1$}
\author{Eli Chertkov$^2$}
\author{Matthew DeCross$^2$}
\author{David Hayes$^2$}
\author{Brian Neyenhuis$^2$}
\author{Marcello Benedetti$^1$}
\author{Konstantinos Meichanetzidis$^1$}
\email[Correspondence to: ]{k.mei@quantinuum.com}

\affiliation{$^1$Quantinuum, Partnership House, Carlisle Place, London SW1P 1BX, United Kingdom\\
$^2$Quantinuum, 303 S. Technology Ct., Broomfield, Colorado 80021, USA}

\date{\today}

\begin{abstract}
\vspace{1em} We present an end-to-end algorithmic pipeline where a noisy digital quantum computer is used to approximate the value of the Jones polynomial at the fifth root of unity for any input link, i.e. a closed braid. This problem is DQC1-complete for Markov-closed braids and BQP-complete for Plat-closed braids, and we accommodate both versions of the problem. Even though it is widely believed that DQC1 is strictly contained in BQP, and so is `less quantum', the resource requirements of classical algorithms for the DQC1 version are at least as high as for the BQP version, and so we potentially gain `more advantage' by focusing on Markov-closed braids in our exposition. We demonstrate our quantum algorithm on Quantinuum's H2-2 quantum computer and show the effect of problem-tailored error-mitigation techniques. Further, leveraging that the Jones polynomial is a link invariant, we construct an efficiently verifiable benchmark to characterise the effect of noise present in a given quantum processor. In parallel, we develop and benchmark the state-of-the-art tensor-network-based classical algorithms for computing the Jones polynomial. The reconfigurable tools provided in this work allow for precise resource estimation to identify minimum link sizes for near-term quantum advantage in practice for a meaningful quantum-native problem in knot theory, if a candidate set of links are provided.
\end{abstract}

\maketitle

\section{Introduction}

Knot theory offers a wide range of computational problems that are notably difficult to solve.
A paradigmatic example is the evaluation of certain \emph{link invariants}, quantities that do not change for knots that are equivalent in the topological sense.
One such quantity is the Jones polynomial, a Laurent polynomial with coefficients and powers determined solely by the link topology.
This means that if two links have different Jones polynomials, then the links are topologically inequivalent~\cite{Kauffman2001}. The exact evaluation of the Jones polynomial is in general \#P-hard~\cite{jaeger1990computational}, meaning that both classical and quantum algorithms are believed to be  inefficient at this task.

We instead consider the problem of
additively approximating the Jones polynomial, up to an exponentially large scale-factor, at non-lattice roots of unity for \emph{a given link}.
Specifically, the problem we study is defined in terms of braids, as any link can be formed by `closing' a braid.
Remarkably, braids have some unitary representations,
and there are two ways to close a braid.
For the `Markov closure', the problem reduces to estimating the weighted trace of a unitary matrix, and is DQC1-complete~\cite{shor2008estimating}.
For the `Plat closure', the problem reduces to estimating a quantum amplitude, i.e. an element of a unitary matrix, and is BQP-complete~\cite{aharonov2006polynomial}.
We illustrate these in the top and bottom panels of Fig.~\ref{fig:markov-plat}(a).
It is widely believed that problems in these complexity classes cannot be efficiently solved by a classical computer in the worst case. Thus, they are excellent candidates for demonstrating a separation between classical and quantum models of computation.

The DQC1-complete formulation of the problem can be solved in the restricted model of quantum computation called `one clean qubit'. There, a single qubit is initialized in a pure state, while the remaining ones are in the maximally mixed state, and measurements are performed only at the end of the computation~\cite{Passante_2009}.
In this work, we instead use a universal \emph{gate-based digital quantum computer}, i.e. a BQP machine, and provide an end-to-end quantum algorithm that can treat \emph{both} DQC1 and BQP versions of the problem.
We choose the fifth root of unity as the evaluation point of the Jones polynomial. This is a well-studied case for which we can use the Fibonacci unitary representation of braids~\cite{KAUFFMAN_2008, shor2008estimating}.
Notably, it is approximately \emph{universal} for quantum computation~\cite{Pachos2012}.
Given this unitary representation of braids, we focus on estimating the desired weighted trace or amplitude, employing the Hadamard test as the key algorithmic primitive.
Using its control-free variant,
and making use of echo verification~\cite{Polla_2023}
and efficient compilation of braid unitaries,
we minimise the circuit depth and number of shots, and therefore, the time-to-solution.
We also employ methods for error detection~\cite{Self_2024} and error mitigation~\cite{cai2023quantum} to increase the accuracy and precision of the algorithm's estimate at the cost of minimal overhead in gates and shots.

\begin{figure*}[t]
    \centering
    \begin{tabular}{@{}p{.5\textwidth}@{}p{.5\textwidth}@{}}
    \begin{center}(a)\end{center} 
    & 
    \begin{center}(b)\end{center} \\[0.25em]
    \input{figs/markov-plat.tikz}
    &
    \input{figs/markov-plat-slide.tikz}
    \end{tabular}
    \caption{(a) Markov closure $M(B)$ and Plat closure $P(B)$ of braid $B$, which we enclose in the blue box. Evaluating $V_{M(B)}$ to additive error is DQC1-complete and evaluating $V_{P(B)}$ to additive error is BQP-complete. In general $M(B)$ and $P(B)$ are topologically inequivalent. (b) Using slide moves, the Markov closure $M(B)$ can be reduced to the Plat closure $P(B')$ of a braid $B'$ on twice the number of strands and more crossings in general.}
    \label{fig:markov-plat}
\end{figure*}
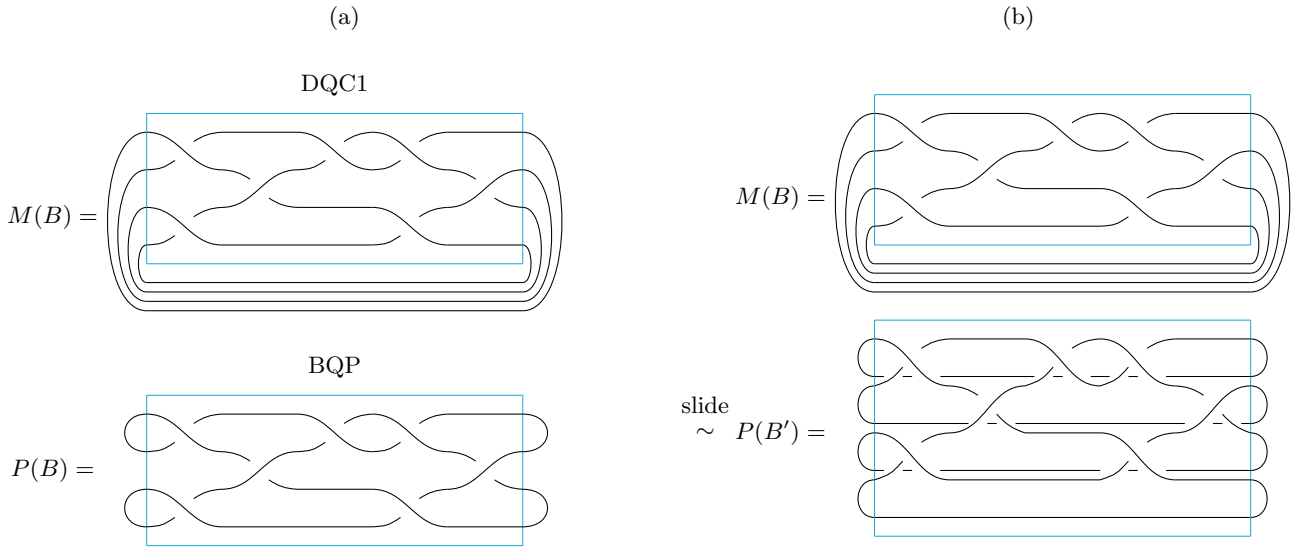

Our highly optimised quantum algorithm can be used to define a problem-tailored benchmark for noisy digital quantum computers. The Jones polynomial, being a link invariant, remains unchanged when the link changes in a way that preserves its topology.
Previous work leverages this property to design benchmarks and to determine the circuit sizes that a quantum processor can faithfully execute according to a given error budget and noise model. 
Ref.~\cite{self2022estimating} employs the BQP version of the problem and chooses a lattice root of unity as the evaluation point for which the problem is actually in P, witnessed by the resulting circuits being Clifford which are classically efficiently simulable~\cite{gottesman1998heisenberg,townsendteague2022simplification}.
Ref.~\cite{goktas2019benchmarking} uses the DQC1 version with the fifth root of unity as the evaluation point, albeit with a suboptimal compilation of the vanilla Hadamard test, to define a benchmark that cannot be efficiently verified were it to be scaled up.
Our benchmark goes beyond previous proposals as it works for \emph{both} the DQC1 and BQP versions and, importantly, is \emph{efficiently verifiable}.

We begin by generating random braids whose Jones polynomial upon closure can be evaluated classically within a negligible amount of time. 
We then randomly augment them to generate a set of topologically equivalent links, each of which has the same Jones polynomial as the original link.
The corresponding circuits span various widths and depths and are guaranteed to ideally return the pre-computed answer up to efficient classical pre- and post-processing.
The circuits can finally be executed on a quantum computer of given specifications to characterise the effect of noise and identify the largest instances within our family of braids for which a desired accuracy can be achieved.

Having established an efficiently verifiable benchmark, we demonstrate how our pipeline can be used to estimate resource requirements for \emph{quantum advantage}, in terms of \emph{time to solution} and \emph{energy consumption}. After fixing the output precision and assuming a \emph{realistic noise model} for a \emph{near-term} quantum computer of given specifications, we identify braid sizes that can be successfully attacked with our algorithm. We also quantify the classical computational resources necessary to compute the same quantity within the same error. Our family of random braids might neither be representative of the average case nor capture the most difficult instances. Yet, if we assume that the classically hard instances are at least as hard as our family of braids, then the benchmark provides a conservative estimate for where quantum advantage is possible.

Because the approximations of the Jones polynomial we consider are complete for DQC1 and BQP, which are classes widely believed to be classically hard, instances likely exist which are classically hard to solve. The tools provided in this work enable one to narrow down the focus of the search for instances exhibiting a near-term quantum computational advantage, specifically by lower bounding their \emph{size}.
We envision that this practical approach inspires concrete quantification and identification of quantum advantage for other quantum-native problems. We hope our work stimulates new ideas towards quantum utility in the field of computational topology.

\section{Evaluating the Jones polynomial on a noisy digital quantum computer}
\label{sec:problem}

Any braid $B$ can be composed from two generators, $\sigma_i^{\pm 1}$, shown in Fig.~\ref{fig:generators}, by starting from the trivial braid on $l$ strands, which we denote as $1_l$.
Any sequence of generators is called a \emph{braid word}.
The \emph{Markov closure} $M(B)$ of a braid $B$ is obtained by connecting the endpoints on one side with those on the other side and without adding any crossings.
The \emph{Plat closure} $P(B)$ is obtained by connecting neighbouring pairs of strands with ``cups and caps''.
The two closures are shown in Fig.~\ref{fig:markov-plat}(a). Note that any Markov-closed braid can be deformed into a Plat-closed braid on twice the strands, as shown in Fig.~\ref{fig:markov-plat}(b).
In Fig.~\ref{fig:example-markov-braids} we show some examples of links represented as Markov-closed braids.

\begin{figure}[t]
    \centering
    \input{figs/braid-gens.tikz}
    \caption{Braid generators $\sigma_i^{\pm 1}$, out of which any braid $B$ is composed.
    }
    \label{fig:generators}
\end{figure}

\begin{figure}[t]
    \centering
    \input{figs/example-braids.tikz}
    \caption{Examples of Markov-closed braids $M(B)$ corresponding to the unknot $0_1$, the trefoil knot $3_1$, the $6_3$ knot, and the $6^3_2$ link. Each example braid is enclosed in a blue box and we also show its braid word.
    }
    \label{fig:example-markov-braids}
\end{figure}
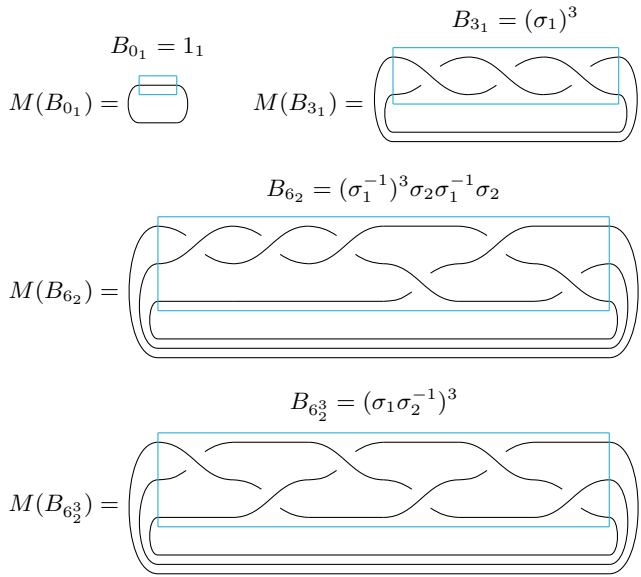

\subsection{Problem definition}

The DQC1-complete and BQP-complete problems of approximating the Jones polynomials $V_{M(B)}(t)$ and $V_{P(B)}(t)$ at $t=\e^{\im \frac{2\pi}{5}}$ of Markov-closed and Plat-closed braids $B$ respectively, are stated as:

\begin{problem}
    APPROX-JONES-MARKOV: Given a braid $B$ on $n-1$ strands and $c=O(\poly(n))$ crossings, compute $R\in\mathbb{C}$ such that $|V_{M(B)}(\e^{\im \frac{2\pi}{5}}) - R| \leq \epsilon \phi^{n-2}$ with high probability, where $\epsilon^{-1}=O(\poly(n))$.
\end{problem}

\begin{problem}
    APPROX-JONES-PLAT: Given a braid $B$ on $n-1$ strands and $c=O(\poly(n))$ crossings, compute $R\in\mathbb{C}$ such that $|V_{P(B)}(\e^{\im \frac{2\pi}{5}}) - R| \leq \epsilon \phi^{\frac{n}{2} - \frac{3}{2}}$ with high probability, where $\epsilon^{-1}=O(\poly(n))$.
\end{problem}

Here $\phi = \frac{1+\sqrt{5}}{2}$ is the golden ratio. Note that the estimates are given up to an additive error with an exponentially large scale-factor of $\phi^{O(n)}$. This is factor is necessary - in fact, \cite[Theorem 1.2]{kuperberg2015how} shows that an additive approximation without this scaling factor is \#P-hard to compute, and hence unlikely to be efficient for a quantum computer. Such additive bounds with an exponential scale factor are also common for similar quantum algorithms for approximate counting \cite[Definition 1]{bordewich2005approximate}. However, the completeness of these problems for DQC1 and BQP respectively indicates that even this weak approximation of the Jones polynomial is likely to be hard for classical computers.

For the chosen evaluation point, the Fibonacci unitary representation of the braid generators, $U_{\sigma_i}$ and $U_{\sigma_i^{-1}}=U_{\sigma_i}^\dagger = U_{\sigma_i}^*$, in the computational basis, is~\cite{KAUFFMAN_2008}:
\begin{equation}
\label{eq:unitary}
\hspace{-2mm}\begin{aligned}
    &\langle 101 | U_{\sigma_i} | 101 \rangle = \phi^{-1}\e^{\im \frac{4\pi}{5}} & &\langle 111 | U_{\sigma_i} | 101 \rangle = \phi^{-\frac{1}{2}}\e^{\im\frac{7\pi}{5}} \\
    &\langle 101 | U_{\sigma_i} | 111 \rangle = \phi^{-\frac{1}{2}}\e^{\im\frac{7\pi}{5}} & &\langle 111 | U_{\sigma_i} | 111 \rangle = -\phi^{-1}\\
    &\langle 110 | U_{\sigma_i} | 110 \rangle = \e^{ \im \frac{3\pi}{5}} & &\langle 011 | U_{\sigma_i} | 011 \rangle = \e^{ \im \frac{3\pi}{5}} \\
    &\langle 010 | U_{\sigma_i} | 010 \rangle = \e^{-\im\frac{4\pi}{5}}
\end{aligned}
\end{equation}
The submatrix defined by Eq.~\eqref{eq:unitary} is fully specified by unitarity and acts on the subspace of $3$-bit strings that do not contain two consecutive zeros.
Then, the braid-unitary $U_B$ is obtained by composition of the $3$-qubit gates $U_{\sigma_i^{\pm 1}}$ according to the braid word,
such that a braid on $(n-1)$ strands is represented as a quantum circuit on $n$ qubits.
As shown by the examples in Fig.~\ref{fig:braid-unitary-and-compiled-circuit}, strands lie in between qubits  such that $U_{\sigma_i}$ acts on qubits $i-1$, $i$, and $i+1$.
Therefore, $U_B$ acts on a $f_{n}$-dimensional subspace of the $n$-qubit Hilbert space, where $f_n$  is the $n$-th Fibonacci number defined by the recurrence relation $f_{n+2}=f_{n+1}+f_n$, $f_0=0$, $f_1=1$,
spanned by the Fibonacci basis states 
\begin{align}\label{eq:fibonacci-basis}
    \mathcal{F}_n = \{ |s\rangle = \otimes_{i=0}^{n-1} |s_i\rangle \; | \; s_0=0 \; \wedge \; s_i + s_{i+1} > 0 \; \forall i \}.
\end{align} 
In general, we call bit strings that contain no consecutive zeros (of which $\mathcal{F}_n$ is a subset) \emph{Fibonacci strings}.

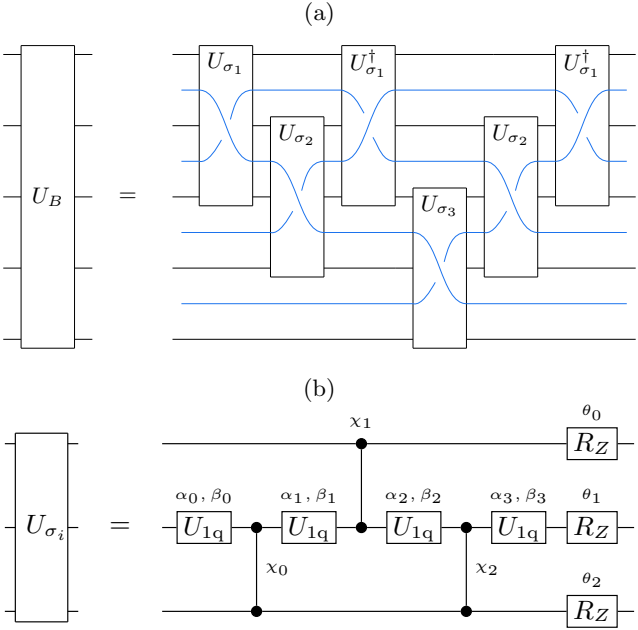
\begin{figure}[t]
    (a)\vspace{6pt}\\
    \resizebox{\linewidth}{!}{\input{figs/example-braid-unitary.tikz}}\vspace{10pt}\\
    (b)\\
    \resizebox{\linewidth}{!}{\input{figs/braid-generator-circuit1.tikz}}
    \caption{
    Constructing quantum circuits for a braid using optimized circuits for the braid generators.
    (a)~Example of the $5$-qubit unitary for the $4$-strand braid $B=\sigma_1\sigma_2\sigma_1^{-1}\sigma_3\sigma_2\sigma_1^{-1}$. The $3$-qubit gates $U_{\sigma_i}$, from which any such braid-unitary is composed, are defined in Eq.~\eqref{eq:unitary}.
    (b)~Optimized circuit, composed of Z-phase gates $R_Z(\theta)$, native 1-qubit gates $U_{1q}(\alpha, \beta)$ and native 2-qubit (Ising) gates $R_{ZZ}(\chi)$ \cite{Datasheets}, implementing the 3-qubit gate $U_{\sigma_i}^k$ of Eq.~\eqref{eq:unitary}. See App.~\ref{app:compilation-params} for the parameter values.
    }
    \label{fig:braid-unitary-and-compiled-circuit}
\end{figure}

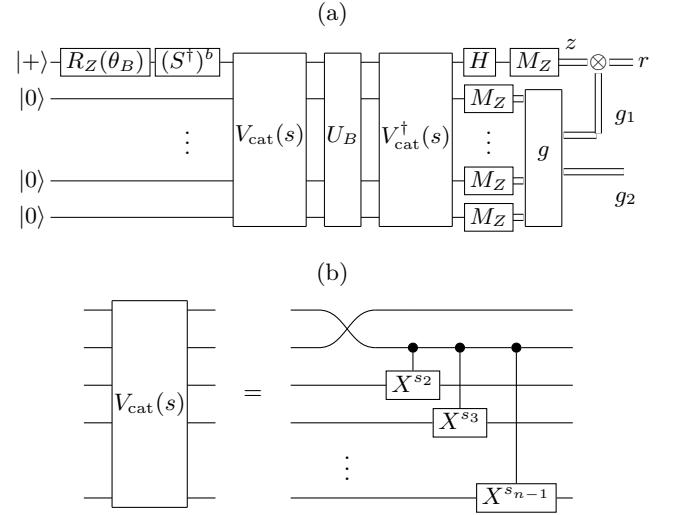
\begin{figure}[t]
    \centering
    (a)\vspace{4pt}\\
    \resizebox{\linewidth}{!}{\input{figs/H-test-CFEV.tikz}} \vspace{4pt}\\
    (b)\vspace{4pt}\\
    \input{figs/catstate-unitary.tikz}
    \caption{The control-free echo-verification Hadamard test. 
    (a)~We employ cat states prepared by $V_\mathrm{cat}$. We measure all qubits and postprocess with a function $g$ that returns two outputs $g_1$ and $g_2$.  Sampling this circuit multiple times allows us to estimate both $\mathbb{E}[\mathrm{Re}\langle s|U_{B}|s \rangle]$ for $b=0$ and $\mathbb{E}[\mathrm{Im}\langle s|U_{B}|s \rangle]$ for $b=1$ using the $g_1$ output, up to additive error (see Fig.~\ref{fig:braid-unitary-and-compiled-circuit}(a) for an example of a braid-unitary $U_B$). The $g_2$ output can be used to detect when hardware errors have occurred.
    (b)~Circuit preparing the $n$-qubit cat state $\frac{1}{\sqrt{2}}(|0\rangle^{\otimes n} + |s\rangle) = V_\mathrm{cat}(s)(|+\rangle\otimes|0\rangle^{\otimes n-1})$.}
    \label{fig:H-test-CFEV-with-catstate-unitary}
\end{figure}

\subsection{Quantum algorithms}

Now we present the quantum algorithms that estimate $V_{M(B)}(\e^{\im \frac{2\pi}{5}})$ and $V_{P(B)}(\e^{\im \frac{2\pi}{5}})$, respectively, for a given braid $B$. It is possible to estimate $V_{M(B)}(\e^{\im \frac{2\pi}{5}})$ from the \emph{weighted trace} of $U_B$ as done by Shor and Jordan~\cite{shor2008estimating} in the one clean qubit model.
In our work, equipped with a gate-based digital quantum computer, we simulate the one clean qubit model using the Hadamard test, as suggested by Aharonov, Jones, and Landau~\cite{aharonov2006polynomial}.
We express the desired quantity as
\begin{align}
\label{eq:jones}
    V_{M(B)}(\e^{\im \frac{2\pi}{5}}) = (- \e^{- \im \frac{3\pi}{5}})^{3 w_{B}} \phi^{n-2} \mathbb{E}_{s\sim p(s)}[ \langle s | U_{B} | s \rangle ] ,
\end{align}
where $w_{B}=\sum_{\sigma \in B} \mathrm{sign}(\sigma)$ is the writhe of the braid defined as the difference between positive and negative crossings, and the basis states are sampled from the following probability distribution (see App.~\ref{app:zeckendorf-algo}):
\begin{align}
\label{eq:prob}
    p(s) = \begin{cases}
    \frac{\phi^{s_{n-1}}}{\phi^{n-1}} & \text{if } |s\rangle \in \mathcal{F}_n\\
    0 & \text{otherwise.}
    \end{cases}
\end{align}

To estimate the quantity in Eq.~\eqref{eq:jones}, we adapt the control-free echo-verification protocol \texttt{cfev}~\cite{Polla_2023}.
This protocol can be understood as a control-free version of the Hadamard test~\cite{Mitarai2019}, which also combines elements of the SWAP test~\cite{barenco1996stabilisation,Buhrman_2001} to reduce the variance of the vanilla Hadamard test~\cite{aharonov2006polynomial} (see App.~\ref{app:qualgorithm} for further details).  Given $U$, $\ket{\psi}$, and an eigenstate $\ket{\phi}$ of $U$, \texttt{cfev} allows one to estimate $\bra{\psi}U\ket{\psi}$ by preparing the cat state $\ket{\psi_\mathrm{cat}} = \frac{1}{\sqrt{2}}(\ket{\psi} + \ket{\phi})$. We adapt this to estimate $\mathbb{E}_{s \sim p(s)}[\bra{s}U_B\ket{s}]$ by sampling different $\ket{s}$ states for each measurement shot, rather than keeping $\ket{s}$ fixed as in the original \texttt{cfev} protocol.

To determine an eigenstate of $U_B$, since Shor and Jordan's method~\cite{shor2008estimating} does not prescribe the value of $U_{\sigma_i}\ket{x}$ for non-Fibonacci bitstrings $x$, we define that $U_{\sigma_i}\ket{000} = \e^{\im\alpha}\ket{000}$ for some angle $\alpha$ that will be fixed when compiling $U_{\sigma_i}$ to a circuit. Therefore, we will have that $U_B\ket{0}^{\otimes n} = \e^{\im \theta_B}\ket{0}^{\otimes n}$ with $\theta_B = w_B \alpha$, and the cat state is given by $\frac{1}{\sqrt{2}}(|0\rangle^{\otimes n} + |s\rangle) = V_\mathrm{cat}(s)(|+\rangle\otimes|0\rangle^{\otimes n-1})$, where $V_\mathrm{cat}(s)$ is the circuit illustrated in Fig.~\ref{fig:H-test-CFEV-with-catstate-unitary}(b). 

The \emph{full quantum algorithm} runs as follows.
First, we sample a computational basis state $| s \rangle$ with probability $p(s)$ (Eq.~\eqref{eq:prob}). This can be done classically in linear time as shown in App.~\ref{app:zeckendorf-algo}.
Second, we execute the $n$-qubit circuit shown in Fig.~\ref{fig:H-test-CFEV-with-catstate-unitary}(a). Third, we feed the measurement outcome $x_1, x_2, \dots, x_n$ to a function $g_1$ defined as: $g_1 = 1$ if $x_2 = x_3 = \cdots = x_n = 0$, and $g_1 = 0$ otherwise (function $g_2$ in Fig.~\ref{fig:H-test-CFEV-with-catstate-unitary}(a) is discussed later).
Finally, we compute the random variable $r = g_1 \cdot z \in\{-1,0,+1\}$ where $z = (-1)^{x_1}$. As shown in App.~\ref{app:qualgorithm}, $r$ is a crude estimator of the real or imaginary part of $\mathbb{E}_{s \sim p(s)}[\bra{s}U_B\ket{s}]$. By repeating these steps $N$ times and taking the average over samples we obtain a Monte Carlo estimate of the desired quantity
\begin{align}
\label{eq:jones-approx}
     &V_{M(B)}(\e^{\im \frac{2\pi}{5}}) = (- \e^{- \im \frac{3\pi}{5}})^{3 w_{B}} \phi^{n-2} \left(R+\epsilon\right),\\\nonumber
&\textrm{where} ~ R = N^{-1} \sum_{i=1}^N r_i \\\nonumber
&\mathrm{and} ~ \epsilon=\epsilon_\mathrm{shot}+\epsilon_\mathrm{noise}.
\end{align}

The real and imaginary parts, $\mathrm{Re}(R) $ and $\mathrm{Im}(R)$, are computed separately by running the above algorithm without ($b=0$) and with ($b=1$) the phase gate $S^\dagger$ in the circuit of Fig.~\ref{fig:H-test-CFEV-with-catstate-unitary}(a), respectively.

The above protocol can now be minimally tweaked to compute $V_{P(B)}(\e^{\im \frac{2\pi}{5}})$. 
In this case, instead of sampling $s\sim p(s)$, we always prepare the $n$-qubit state $\ket{s} = \ket{0101\cdots 10}$ and estimate (see App.~\ref{app:jones-plat}):
\begin{align}
\label{eq:jones-approx-plat}
    V_{P(B)}(\e^{\im\frac{2\pi}{5}}) &= (-\e^{-\im\frac{3\pi}{5}})^{3w_B}\phi^{\frac{n}{2} - \frac{3}{2}} \bra{s}U_B\ket{s}\\\nonumber
    &=  (-\e^{-\im\frac{3\pi}{5}})^{3w_B}\phi^{\frac{n}{2} - \frac{3}{2}}\left(R+\epsilon\right).
\end{align}

The approximation is characterised by two types of errors:
the \emph{Monte Carlo shot error} $\epsilon_\mathrm{shot}$, which induces variance due to a finite number of shots $N$,
and the \emph{gate noise error}, $\epsilon_\mathrm{noise}$, which induces a bias, which given a quantum computer with certain specifications can be determined.

To implement the circuit shown in Fig.~\ref{fig:H-test-CFEV-with-catstate-unitary} in practice, we optimise the compilation of the $3$-qubit matrix $U_{\sigma_i^{\pm 1}}$ of Eq.~\eqref{eq:unitary} in terms of the native gate set of the underlying hardware. We specifically aim to minimise the number of 2-qubit gates, as they incur the most errors in most quantum processors. With Quantinuum systems in mind~\cite{Datasheets}, as they currently exhibit the highest fidelities, we choose an ansatz composed of Z-phase gates $R_Z(\theta)$, native 1-qubit gates $U_{1q}(a,b)$ and entangling 2-qubit gates $R_{ZZ}(\chi) = \e^{-\im \frac{\chi}{2} Z \otimes Z }$. We determined the parameters of the ansatz, only requiring the desired action in the Fibonacci subspace, and that $\ket{000}$ is an eigenstate. The resulting circuit is shown in Fig.~\ref{fig:braid-unitary-and-compiled-circuit}(b), and contains only three 2-qubit gates. We repeated this process for $U_{\sigma_i}^k$ with $2 \leq k < 10$, and replace any consecutive runs of identical braid generators with these optimized circuits.
Given another quantum computing platform, this compilation is to be carried out accordingly.

\begin{figure*}
    \centering
    \begin{tabular}{@{}c@{}}
        \resizebox{.45\linewidth}{!}{\input{figs/braid.tikz}}
            \vspace{8mm}\\
        \includegraphics[width=.45\linewidth]{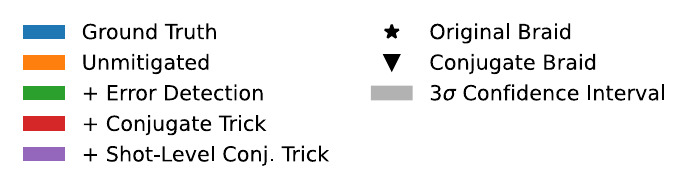}\\
    \end{tabular}
    \qquad
    \begin{tabular}{@{}c@{}}
        \includegraphics[width=.45\linewidth]{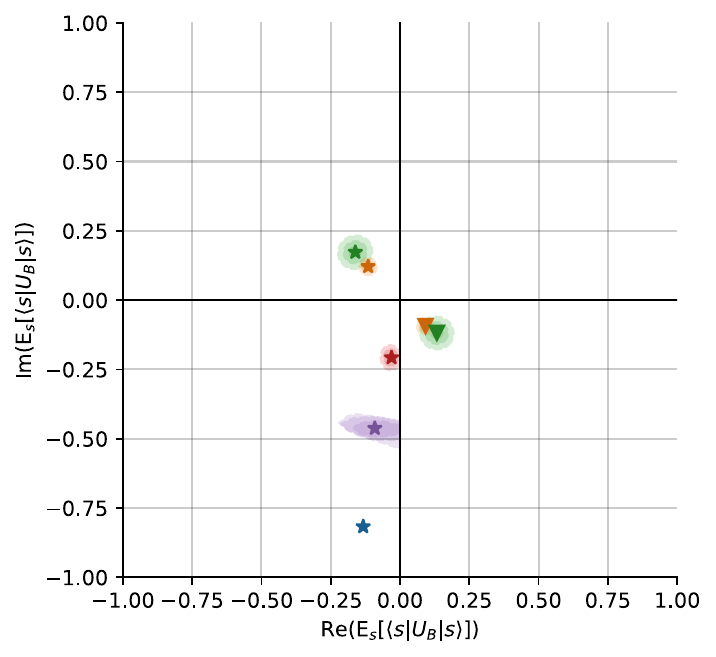}
    \end{tabular}
    \caption{A demonstration of our algorithm using Quantinuum's H2-2 quantum computer. We generated a braid $B$ with 104 crossings using the benchmarking procedure described in Sec.~\ref{sec:benchmark} whose Markov closure is topologically equivalent to the tensor product of one left-handed trefoil knot and three right-handed trefoil knots. The resulting quantum circuits acted on $16$ qubits with $340$ 2-qubit gates. Each value was estimated with 4000 shots. Non-Fibonacci error detection improved the estimate by $\sim 25\%$. There is a substantial phase difference between the unmitigated estimate and the ground truth due to coherent memory errors - this can be mitigated by the conjugate trick, but the shot-level conjugate trick described in App.~\ref{app:mitigation} performs better. Our error mitigation and detection techniques are described in Sec.~\ref{ssec:error_mitigation}.}
    \label{fig:accuracy}
\end{figure*}

\subsection{Error mitigation}
\label{ssec:error_mitigation}

To mitigate the impact of $\epsilon_\mathrm{noise}$, we can exploit the fact that $U_B$ only acts on a small fraction of the total Hilbert space (that is, the Fibonacci subspace, and the state $\ket{0}^{\otimes n}$). Recall that the algorithm consists of sampling a  Fibonacci string $s = s_1, s_2, \dots, s_n$, running the circuit in Fig.~\ref{fig:H-test-CFEV-with-catstate-unitary}(a), and then measuring all the qubits to obtain $x = x_1, x_2, \dots, x_n$. We introduce a function $g_2$ defined as: $g_2 = 1$ if $x_2 = 0$ and $(x_i \oplus s_i) + (x_{i + 1} \oplus s_{i+1}) > 0$ for all $i \geq 3$, and $g_2 = 0$ otherwise. From the definition of $U_{\sigma_i}$, we have that $\bra{x}U_B\ket{y} = 0$ whenever $x \in \mathcal{F}_n$ and $y \notin \mathcal{F}_n$ or vice-versa. Therefore, if $g_2 = 0$, a hardware error \emph{must} have occurred and that sample can be \emph{discarded} and not included in the estimate $R$. This is similar to the method of `symmetry constraint' error mitigation discussed in~\cite{cai2023quantum}. Under reasonable noise assumptions our technique decreases the bias and increases the variance (see App.~\ref{app:mitigation}), potentially reducing the overall error. We call this technique the \emph{non-Fibonacci error detection}. 

Unfortunately, our algorithm is particularly sensitive to coherent phase errors of the form $\e^{-\im\theta Z/2}$ on the second qubit of $U_B$. This is because $U_{\sigma}$ is diagonal in the first qubit, and is also diagonal in its middle qubit if the first qubit is in the $\ket{0}$ state. Since $\ket{\psi_\mathrm{cat}} = \frac{1}{\sqrt{2}}(\ket{s} + \ket{0}^{\otimes n})$ and $s_0 = 0$, then $\ket{\psi_\mathrm{cat}}$ has the first qubit always in the $\ket{0}$ state. This implies that $U_B$ is always diagonal on the second qubit, and hence any coherent phase errors on this qubit will commute with $U_B$ and be combined into one large phase error. Since this qubit acts as the control qubit of the Hadamard test (via $V_\mathrm{cat}$), a phase error of $\theta$ here corresponds to a rotation of the expectation value, so that we estimate $R\e^{\im\theta} \approx \mathbb{E}_s[\e^{\im\theta}\bra{s}U_B\ket{s}]$ instead. 

The relevant source of such errors on Quantinuum's trapped ion platforms is \emph{memory errors}, namely the cumulative effect of errors that happen to qubits while idling. To mitigate this, we can split the algorithm into two phases - first, we estimate $\e^{\im\theta}R$ for braid $B$ using \texttt{cfev}. Then, we estimate $\e^{\im\theta}R^*$ by running the algorithm for braid $B^{*}$ by replacing each $U_{\sigma}$ with $U_{\sigma^{-1}}$. From this, we can calculate
\begin{equation}
    R = \pm\frac{|\e^{\im\theta}R + \e^{\im\theta}R^*|}{2} \pm \im\frac{|\e^{\im\theta}R - \e^{\im\theta}R^*|}{2} ,
\end{equation}
to eliminate the phase error. We call this the \emph{conjugate trick}. Note that the variance of this estimate is half the average variance of $\e^{\im\theta}R$ and $\e^{\im\theta}R^*$, so no additional shots are required. Here, $R$ is determined only up to the sign of each component, but this can be disambiguated under the assumption that $|\theta| < \frac{\pi}{2}$. See App.~\ref{app:mitigation} for more details.
We have assumed that the coherent error $\e^{\im\theta}$ is the same in both cases - we believe this assumption is justified since the structure of the conjugate circuits is identical except for the sign flip of the parameters of every gate. In App.~\ref{app:mitigation}, we discuss how the conjugate trick behaves when this assumption is invalid, and we present a method that can still mitigate the phase error under much weaker assumptions about the form of the coherent errors, but cannot determine the signs of the components and may be biased. We call this the \emph{shot-level conjugate trick}.

Finally, a comment regarding choice of protocol and its efficient compilation is in order. The \texttt{cfev} protocol has two advantages over the vanilla Hadamard test (see App.~\ref{app:qualgorithm} for more details on the variants of the Hadamard test). Firstly, by using echo verification, we have a smaller shot error $\epsilon_\mathrm{shot} \propto N^{-\frac{1}{2}}$, as compared to the vanilla Hadamard test where the shot error is $\epsilon_\mathrm{shot} \propto {(2-R^2)}^{\frac{1}{2}}{N}^{-\frac{1}{2}}$. Secondly, the vanilla Hadamard test requires implementing many \mbox{controlled-$U_{\sigma_i^{\pm 1}}$} operations, and using a similar ansatz to Fig.~\ref{fig:braid-unitary-and-compiled-circuit}(b) these require eight 2-qubit gates each. By comparison, \texttt{cfev} uses only $U_{\sigma_i}$ operations, which require only three 2-qubit gates each, thus reducing the impact of $\epsilon_\mathrm{noise}$ substantially. 
Overall, by switching to the \texttt{cfev} protocol and optimizing the implementation of $U_{\sigma_i}$ for trapped-ion hardware, we reduce the number of 2-qubit gates required by $\approx 15$x and the number of shots required by $\approx 25$\% as compared to the estimates in Ref.~\cite{goktas2019benchmarking}.
This, combined with the error mitigation method mentioned above, may allow much larger instances of the problem to be run on NISQ hardware than was previously anticipated. 

Figure~\ref{fig:accuracy} shows an example of executing our algorithm on Quantinuum's H2-2 quantum computer for a braid with 15 strands and 104 crossings that was randomly generated using the method in Sec.~\ref{sec:benchmark}. Applying both the conjugate trick and non-Fibonacci error detection, and sampling 4000 shots per circuit, we observe a relative error of 75\% (corresponding to $\epsilon_\mathrm{noise} = 0.62$), with a relative standard deviation of 9\% (corresponding to $\epsilon_\mathrm{shot} = 0.02$). The shot-level conjugate trick improves the relative error to 43\% (corresponding to $\epsilon_\mathrm{noise} = 0.36$), but at the cost of increasing the relative standard deviation to 14\% (corresponding to $\epsilon_\mathrm{shot} = 0.06$). In the next section, we will see how this error scales using noisy simulations of a hypothetical NISQ computer.

\subsection{Less quantum, more advantage}

It is believed that P $\subset$ DQC1 $\subset$ BQP, and so DQC1 is `less quantum' than BQP, at least in the worst case.
On the other hand, for a given braid, \emph{classically} estimating the trace of the corresponding exponentially large matrix, which solves the DQC1 version of the problem, is generally more challenging than estimating an amplitude of the same matrix, which is required to solve the BQP version. Thus, the experimental results in the next sections focus on the DQC1 formulation of the problem, where our quantum algorithm can achieve `more advantage', even though our pipeline allows one to search for quantum advantage in the BQP version, as well.

Chen \emph{et al.}~\cite{chen2023complexity} define the NISQ complexity class and prove that BPP$^O$ $\subset$ NISQ$^O$ and NISQ$^O$ $\subset$ BQP$^O$, relative to error-free oracles $O$. Our hardware assumptions are compatible with the NISQ complexity class, but our DQC1 problem does not rely on an error-free oracle. Further work is needed to understand the relation between NISQ and DQC1. A relevant result by Morimae \emph{et al.}~\cite{Morimae2017} is that the one clean qubit model remains hard to simulate classically for suitably small depolarizing noise applied to the clean qubit.

\section{Efficiently verifiable benchmark leveraging topological invariance}
\label{sec:benchmark}

Benchmarks measure the computing performance, such as error-rate and time-to-solution scaling, on certain tasks. An ideal quantum computing benchmark would be based on a task that requires non-trivial quantum resources, resembles tasks of practical relevance, and can be efficiently verified by a classical computer while not being efficiently classically computable. An obvious example would be factoring \cite{Shor_1997}, where classical verification simply amounts to multiplying the factors.
However, it is difficult to implement such an ideal benchmark on near-term hardware. Benchmarks that rely on random circuits from a universal gateset, often do not resemble circuits used in concrete computational problems.
For example, quantum volume~\cite{Cross_2019} is not efficiently classically verifiable and a weakness of mirror benchmarking~\cite{mayer2023theory,Hines_2023} is that classical methods can `cheat' by claiming to have measured the initial state. 
Application-oriented benchmarks~\cite{Lubinski_2023,mills2021application}, are appealing from the practitioner's point of view, but require the exponentially hard classical simulation of circuits for verification. When the benchmark is efficiently verifiable~\cite{magesan2011scalable}, it often relies on non-universal families of quantum circuits, such as Clifford and matchgates, that can be efficiently simulated classically~\cite{Aaronson_2004,Wick1950,Jozsa_2008}. Furthermore, the size of the task used in the benchmark should be under fine control~\cite{blumekohout2020volumetric}. This is needed if we want to study the \emph{scaling} of performance as quantum circuit depth and width are independently varied. Application-oriented benchmarks based on variational algorithms~\cite{Benedetti_2019,mccaskey2019quantum,siddiqui2024stressing} easily decouple circuit width and depth, but, in general, they are affected by an intractable sampling overhead during variational optimisation.

Here, we introduce a benchmark based on APPROX-JONES-MARKOV with the appealing feature that it uses a universal gate set relevant to the knot theory problem while being \emph{efficiently verifiable} by classical means. We aim at characterising the effect of noise on the estimate of $V_{M(B)}(\e^{\im \frac{2\pi}{5}})$, i.e. the scaling of the $\epsilon_\mathrm{noise}$ error with the braid size. Our purpose is to identify the problem sizes that can be solved to a given accuracy and with a given fidelity specification. The links $M(B)$ and $M(B')$ are topologically equivalent iff their braid representatives $B$ and $B'$ can be rewritten into each other via a finite sequence of `Markov' moves (see App.~\ref{app:braid-simplification}, Fig.~\ref{fig:markov-moves}).

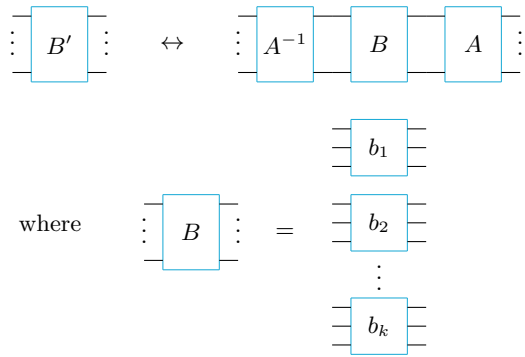
\begin{figure}
    \centering
    \input{figs/conjugation.tikz}
    \caption{Augmenting a $3k$-strand braid $B=\otimes_{i=1}^k b_i$ by conjugation with a random braid $A$, and obtaining braid $B'$. Braid $B$ is composed of random $3$-strand braids $b_1, \dots, b_k$, so its Jones polynomial can be evaluated classically exactly in negligible time as $V_{M(B)}(\e^{\im \frac{2\pi}{5}}) = \phi^{k-1}\prod_{i=1}^k V_{M(b_i)}(\e^{\im \frac{2\pi}{5}})$. See App.~\ref{app:braid-augmentation} for details. $B$ is generated so that the magnitude of $V_{M(B)}$ is approximately uniformly distributed. The random braid $A$ and its inverse $A^{-1}$ are generated to be inverse to each other but not mirror images.}
    \label{fig:augmentation}
\end{figure}

We start with a braid $B$ constructed by tensoring $k$ $3$-strand braids $b_i$, so we can classically compute $V_{M(B)}(\e^{\im \frac{2\pi}{5}}) = \phi^{k-1} \prod_{i=1}^k V_{M(b_i)}(\e^{\im \frac{2\pi}{5}})$ at negligible cost.
We then augment $B$ by conjugating it with a random braid $A$, making the braid effectively wider and deeper, respectively, as shown in Fig.~\ref{fig:augmentation}. 
This results in a braid $B'=A^{-1} B A$ that by definition shares the same Jones polynomial upon closure, $V_{M(B)}=V_{M(B')}$. Therefore, by measuring the value of $V_{M(B')}$ using our quantum algorithm, and comparing it to $V_{M(B)}$, which is computed efficiently classically, we can measure the error incurred by the quantum computer under test. 

If such a $B'$ is constructed naively, then this is very similar to mirror benchmarking \cite{proctor2022measuring}. However, in App.~\ref{app:braid-augmentation}, we detail a method which generates $A$ and $A^{-1}$ that are inverses to each other, but do not have similar structures. In particular, $A$ is generated randomly and is much longer than $A^{-1}$. It appears close to a random braid in structure, except that the signs of the crossings are picked such that $A^{-1}$ can be constructed with a shallow braid. This results in overall braids $B'$ that appear similar to random brick-wall braids. 

Furthermore, since we wish to characterize $\epsilon_\mathrm{noise}$ rather than $\epsilon_\mathrm{shot}$, which has a known distribution, we pick the braids $b_i$ from a probability distribution that is set up so that the distribution of $|\mathbb{E}_s[\langle s|U_B|s\rangle]|$ is as uniform as possible over the range $[0, 1]$. If the $b_i$ were picked uniformly randomly, then the distribution of $|\mathbb{E}_s[\langle s | U_B | s\rangle]|$ would concentrate exponentially around $0$ as the number of strands increases (since the average value of $|\mathbb{E}_{s_i}[\langle s_i | U_{b_i} | s_i\rangle]|$ is less than one), which would require an exponentially increasing number of shots to ensure that $\epsilon_\mathrm{shot}$ is not the dominant error. 

Therefore, via augmentation, we can obtain benchmarking sets of braids with varying numbers of strands, crossings, and depths. To characterise the scaling of the error with increasing number of crossings for potential NISQ computers, we constructed a dataset of 24k braids of between 10 and 15 strands and 50 to 600 crossings. We classically simulated \texttt{cfev} with non-Fibonacci error detection on these braids and calculated the relative error in the output. We used a large number (10k) of shots per braid, so that this error characterises $\epsilon_\mathrm{noise}$ accurately.

The simulation modelled the native gate-set used in the ansatz in Fig.~\ref{fig:braid-unitary-and-compiled-circuit}(b) with errors inserted as depolarizing channels on 1- and 2-qubit gates with probabilities $\epsilon_\mathrm{1q}$ and $\epsilon_\mathrm{2q}$. We tested both $\epsilon_\mathrm{2q} = 5\cdot 10^{-4}$ and $\epsilon_\mathrm{2q} = 1 \cdot 10^{-4}$ as the dominant error, and set $\epsilon_\mathrm{1q} \approx \epsilon_\mathrm{2q} / 10$. SPAM bit-flip errors were also included with $\epsilon_\mathrm{SPAM} \approx \epsilon_\mathrm{2q}$ per qubit. We did not attempt to model coherent effects such as memory errors since this varies widely by platform and technology. See App.~\ref{app:simulations} for more details.

\begin{figure}[t]   
    \centering
    \includegraphics[width=\linewidth]{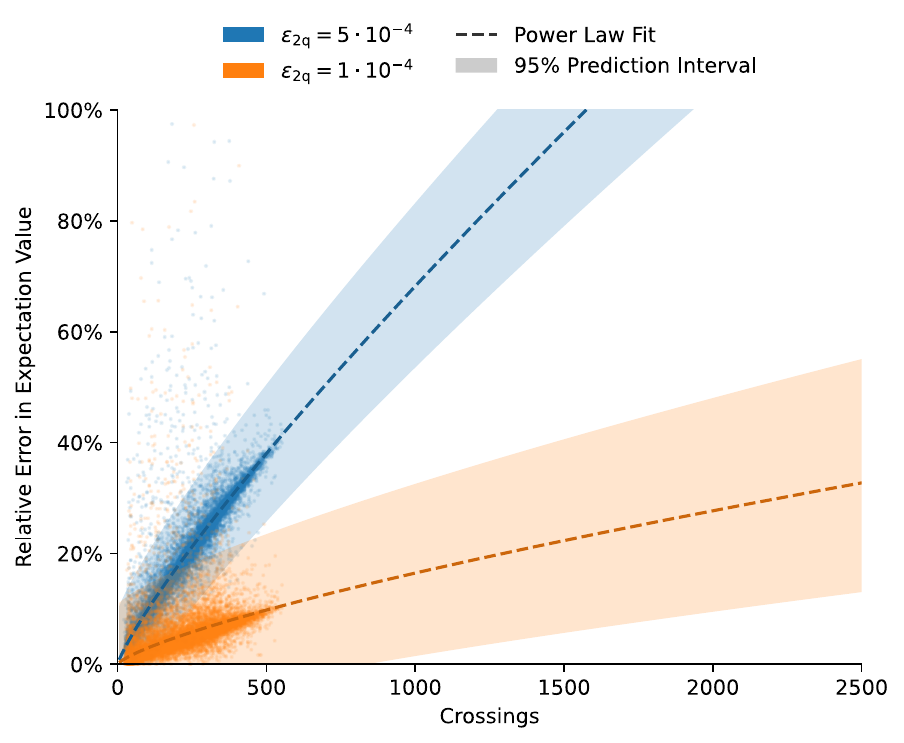}
    \caption{For a benchmark set of 24k randomly generated braids and two different rates of depolarizing 2-qubit errors, this shows how the relative error $|R - \mathbb{E}_s[\bra{s}U_B\ket{s}]| / |\mathbb{E}_s[\bra{s}U_B\ket{s}]|$ between the estimated value (over 10k shots) and the exactly known value depends on the number of crossings (and hence 2-qubit gates).}
    \label{fig:noise-characterisation}
\end{figure}

In Fig.~\ref{fig:noise-characterisation}, we show how the relative error for this benchmark set scaled with the number of crossings, along with a power-law extrapolation to larger numbers of crossings. We observed only weak scaling in terms of number of strands and depth when fixing the number of crossings (Fig.~\ref{fig:unexplained-error}). This aligns with the expectation that errors from a depolarizing noise model ought to depend mainly on the number of gates. These results inform us about the maximum crossing number, or equivalently gate count (specifically 2-qubit gates since they are the dominant noise factor~\cite{Datasheets}), we can reach while remaining within a desired error budget.

These simulations should be broadly representative of NISQ devices with similar error rates, at least to the extent that the depolarizing channel is accurate, or to the extent that the circuit structure sufficiently mixes any structure in the actual error channel. This picture can be ensured by employing randomized compiling techniques, and we note that this should be approximately realized in our experiments on random braids~\cite{mann2017complexity}. There is one notable exception, which is coherent phase effects, for example due to memory errors (which are especially apparent on quantum computers with a relatively slower gate speed, such as trapped-ion devices) or cross-talk between qubits (which are more common on superconducting devices).

For the results of these simulations to be applicable, we assume that the coherent portion of this error is negligible, either by applying the conjugate trick or by mitigating it at a hardware level, and that any remaining coherent errors can be incorporated into an effective 1- and 2-qubit depolarizing error rates. For example, for trapped-ion technology, this can be achieved via dynamical decoupling~\cite{viola1999dynamical}. This is a reasonable assumption as we expect rates of coherent errors to be improved at a hardware level before the 2-qubit gate-error rates considered here are achieved.

Since these results are obtained via simulation, we must \emph{extrapolate} to larger sizes. If one uses an actual quantum backend which is challenging to simulate classically, one would perform an \emph{interpolation} on similarly obtained data. Further, note that this benchmark will take significantly longer to run on an actual device as clock speeds of quantum processors are orders of magnitude lower than those of classical processors, but one may compensate with fewer braids in the benchmark set or fewer shots, in which case it may be beneficial to apply a technique such as deconvolution~\cite{diggle2018deconvolution} to separate the effects of $\epsilon_\mathrm{shot}$ and $\epsilon_\mathrm{noise}$. While we have performed a relatively naive extrapolation by gathering as much data as is feasible given the constraints of classical simulation, our benchmarking pipeline can be used as-is with other extrapolation methods, or alternative data (for instance, combining large-scale classical simulation with a few data points from quantum hardware).

\section{Resource estimates for near-term quantum advantage}
\label{sec:resources}

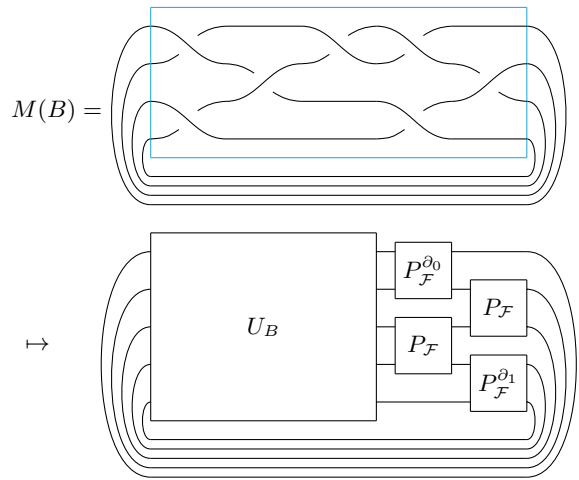
\begin{figure}[t]
    \centering
    \input{figs/tn-proj.tikz}
    \caption{The Markov closure $M(B)$ of a braid $B$, and the corresponding tensor network \texttt{tn-proj} that computes the weighted trace $\sum_{s}p(s)\bra{s}U_B\ket{s}$, with the 2-qubit projectors $P_\mathcal{F}$ inserting the appropriate weights.}
    \label{fig:tn-proj}
\end{figure}

\begin{figure*}[t]
    \centering
    \qquad (a) \qquad\qquad\qquad\qquad\qquad\qquad\qquad\qquad\qquad\qquad\qquad\qquad\qquad (b) \qquad\qquad \\
    \includegraphics[width=0.49\textwidth]{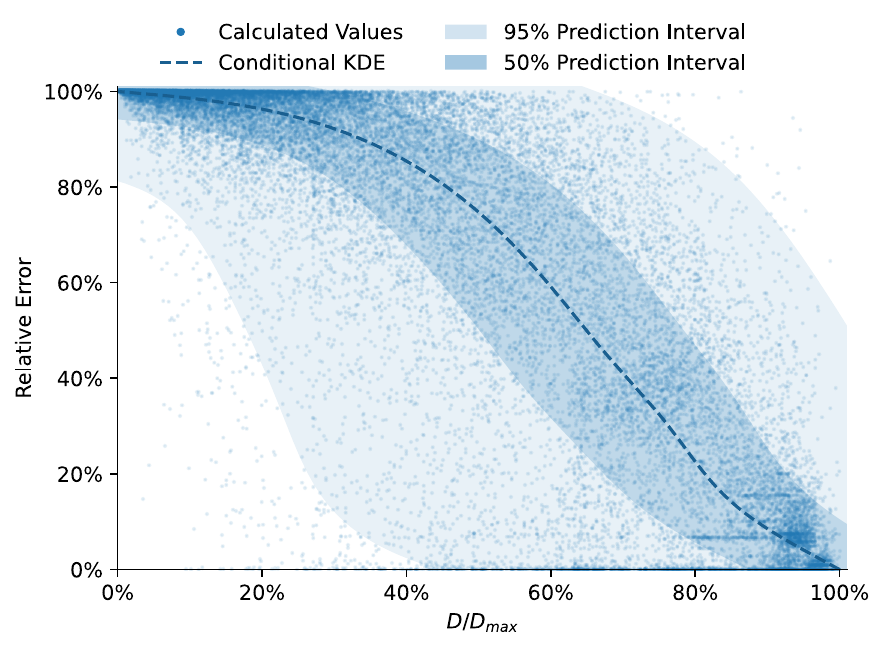}
    \includegraphics[width=0.49\textwidth]{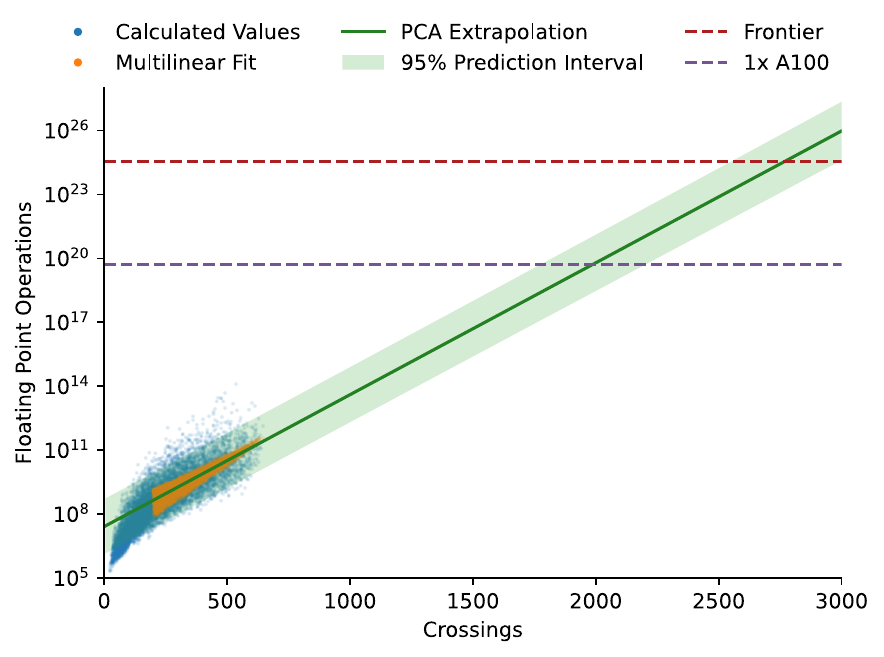}
    \caption{
    Studying the error introduced by compression of the MPO computation through limiting the bond dimension, as well as the scaling of run-time with number of crossings when the bond dimension is large enough that the computation is error-free. The calculated values are from a dataset constructed using the benchmark method in Sec.~\ref{sec:benchmark} and comprising 10k braids of up to 600 crossings and 30 strands.
    (a) A comparison of the relative error in $V_{M(B)}(\e^{\im\frac{2\pi}{5}})$ achieved by the \texttt{mpo-proj} algorithm as it depends on the maximum required memory $D$, normalized to the case $D_\mathrm{max}$ where $\chi$ is set very large (so that the relative error is zero). We can see that as the available memory decreases, the relative error increases. Note that above a threshold of $\sim 80\%$ relative error, $D = 0$ is sufficient with at least $5\%$ confidence, so we consider the cost of executing \texttt{mpo-proj} as effectively zero. The prediction intervals and regression line were obtained via Nadaraya-Watson conditional kernel density estimation~\cite{nadaraya1964estimating}. See Fig.~\ref{fig:flops-derating} for how the amount of computation required depends on the required memory and number of crossings. 
    (b) The number of floating-point operations (FLOPS) required to evaluate $V_{M(B)}(\e^{\im\frac{2\pi}{5}})$ using the \texttt{mpo-proj} method for the same benchmark set. The bond dimension $\chi = \chi_\mathrm{max}$ was allowed to grow as large as needed to avoid any lossy compression. The extrapolation is a multi-linear fit of $\log_{10}(\mathrm{FLOPS})$ against the number of strands $n - 1$, crossings $c$, and depth of the braid for $c \geq 200$. It is extrapolated along the principal component of the dataset to represent possible performance for larger datasets generated in the same way. The lines labelled `Frontier' and `1x A100' represent one month of computation on the respective hardware, for scale. See Fig.~\ref{fig:mem-mpo} for the equivalent plot in terms of memory consumption.}
    \label{fig:chi-and-mpo-crossings}
\end{figure*}

\begin{figure*}[t]
    \centering
    \qquad (a) \qquad\qquad\qquad\qquad\qquad\qquad\qquad\qquad\qquad\qquad\qquad\qquad\qquad (b) \qquad\qquad \\
    \includegraphics[width=0.49\textwidth]{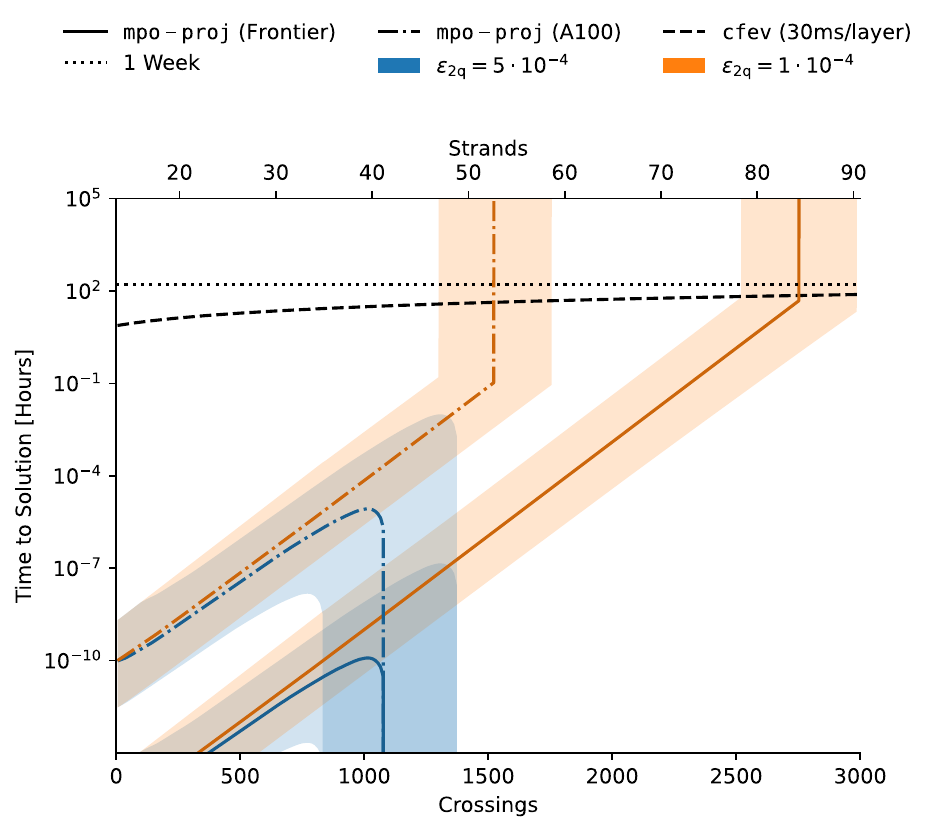}
    \includegraphics[width=0.49\textwidth]{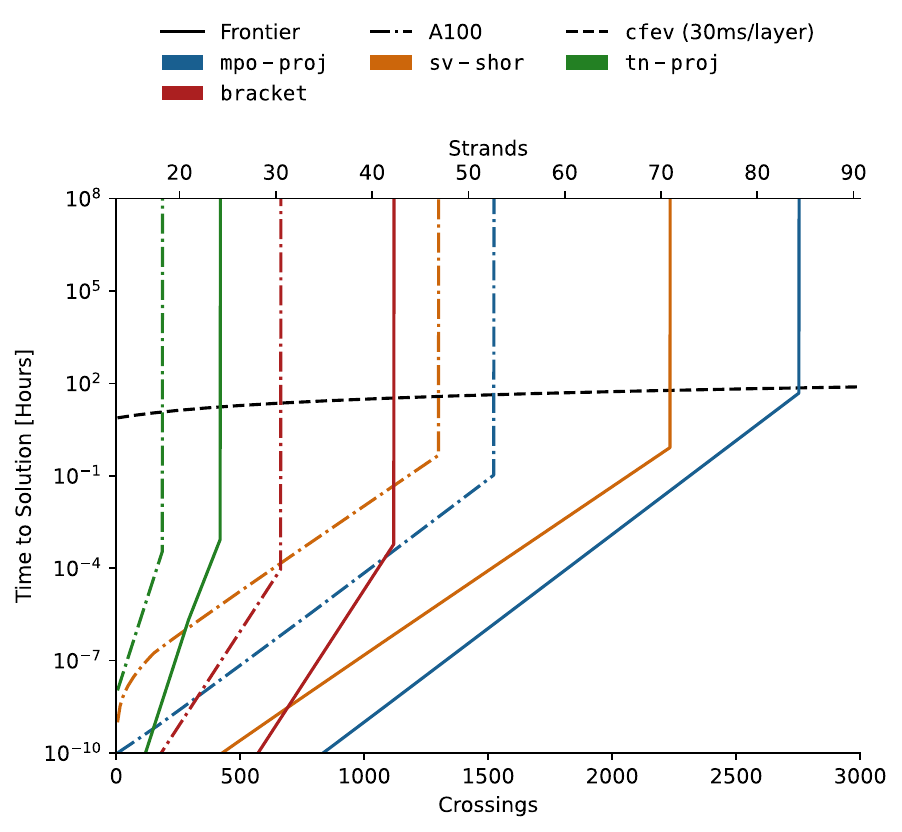}
    \caption{
    Time comparisons for calculating $V_{M(B)}(\e^{\im\frac{2\pi}{5}})$ on the same benchmark set as in Fig.~\ref{fig:chi-and-mpo-crossings}, to the same relative error achieved by \texttt{cfev}, as it depends on the number of crossings $c$ and the 2-qubit error rate $\epsilon_\mathrm{2q}$. The solid lines represent computation on the Frontier cluster, while the dot-dashed lines represent the same computation on a single A100 GPU. The lines vertically downward represent the points where the relative error of \texttt{cfev} exceeds the 80\% threshold given by Fig.~\ref{fig:chi-and-mpo-crossings}, whereas the lines vertically upward represent the points where required memory exceeds the available memory of the classical hardware.
    (a) The time needed for the \texttt{mpo-proj} method, as it depends on the number of crossings $c$ and the 2-qubit error rate $\epsilon_\mathrm{2q}$. We can see that for $\epsilon_\mathrm{2q} = 5 \cdot 10^{-4}$, the relative error of \texttt{cfev} grows large enough that \texttt{mpo-proj} is always faster than \texttt{cfev}. Conversely, for $\epsilon_\mathrm{2q} = 1 \cdot 10^{-4}$, the memory limit is reached before this point. Even if this were not the case, we would expect \texttt{cfev} to be faster than \texttt{mpo-proj} on Frontier when $c \geq 2800$.
    (b) A comparison of the time needed for each classical method when $\epsilon_\mathrm{2q} = 1 \cdot 10^{-4}$, as it depends on the number of crossings $c$. We can see that \texttt{mpo-proj} is the most competitive classical method, although the memory limit of Frontier is reached before it becomes slower than \texttt{cfev}. See Figs.~\ref{fig:time-sv}-\ref{fig:mem-bracket} for plots of the FLOPS, memory consumption, and time-to-solution of each classical method individually.
    }
    \label{fig:mpo-comparison-and-time-to-soln}
\end{figure*}

\begin{table}
    \centering
    \begin{tabular}{p{0.19\linewidth}@{\hskip 4mm}p{0.15\linewidth}p{0.18\linewidth}p{0.15\linewidth}p{0.20\linewidth}}
    \toprule
    Algorithm  & Space & Time & $\epsilon_\mathrm{shot}$ & $1 - \epsilon_\mathrm{noise}$ \\
    \midrule
    \texttt{cfev} & $n$  & $(c + n)N$   & $N^{-1/2}$ & $(1 - \epsilon_\mathrm{2q})^{3 c}$  \\
    \midrule
    \texttt{tn-proj} & $\alpha^c$ & $\beta^c$ & - & -\\
    \texttt{mpo-proj} & $n\chi^2$ & $cn\chi^3$ & - & $n\chi^2 \gamma^{-c}$ \\
    \texttt{sv-shor} & $\phi^n$ & $c \phi^n N$ & $N^{-1}$ & - \\
    \midrule
    \texttt{bracket} & $\mu^c $ & $\nu^c$ & - & - \\
    \bottomrule
    \end{tabular}
    \caption{Approximate scaling in runtime and space of the algorithms we consider as a function of $n$ qubits or equivalently strands, $c$ crossings, $N$ shots, maximum bond dimension $\chi$, and $\epsilon_\mathrm{2q}$ 2-qubit gate error. The parameters $\alpha \approx 1.04$, $\beta \approx 1.05$, $\gamma \approx 1.01$, and $\mu, \nu \approx 1.03$ are representative numbers taken from the extrapolations shown in Fig.~\ref{fig:chi-and-mpo-crossings}(a) and~\ref{fig:mpo-comparison-and-time-to-soln}(b). Note that since $c \sim 10^3$, the difference between $\gamma^c$ and $\beta^c$ is very large (e.g $10^{16}$).}
    \label{tab:complexities-markov}
\end{table}

We now demonstrate how one may use the algorithmic tools introduced above to perform resource estimates for APPROX-JONES-MARKOV, for our quantum algorithm and for classical algorithms that either simulate a quantum circuit or approach the problem in a non-quantum-related way.
Regarding our quantum algorithm, given a braid of certain size, i.e. number of strands $n-1$ and crossing number $c$, with a required additive error $\epsilon$, we can infer the required number of shots $N=\epsilon_\mathrm{shot}^{-2}$ from $\epsilon_\mathrm{shot}=\epsilon-\epsilon_\mathrm{noise}$.
Then, $N$ determines the time-to-solution of the quantum algorithm (see Tab.~\ref{tab:complexities-markov}),
noting that $\epsilon_\mathrm{noise}$ depends on the braid size. 

From this, we estimate the braid \emph{size} for which we might obtain quantum advantage, by analysing the runtime of the most competitive classical algorithm returning the solution to the same error. We note that this method has limitations, and we must make some assumptions -- these are explained in Section \ref{sec:assumptions}. Only if these assumptions are satisfied can we draw concrete conclusions about the existence of a quantum advantage for generic braids. If this is not the case, we can still use the method of Section \ref{sec:benchmark} as a flexible benchmark of noisy quantum hardware.

\subsection{Classical simulations of the quantum algorithm}\label{sec:simulatingquantum}

We can construct a tensor network to exactly evaluate $V_{M(B)}(\e^{\im\frac{2\pi}{5}})$ by representing the weighted trace of Ref.~\cite{shor2008estimating} as the usual trace $\mathrm{tr}(U_B)$ and insert the weights using 2-qubit projectors onto the Fibonacci subspace, $P_\mathcal{F}^{\partial_0}=\ket{01}\bra{01}$,
$P_\mathcal{F}=\ket{01}\bra{01}+\ket{10}\bra{10}+\ket{11}\bra{11}$,
$P_\mathcal{F}^{\partial_1}=\ket{10}\bra{10}+\phi(\ket{01}\bra{01}+\ket{11}\bra{11})$,
as shown in Fig.~\ref{fig:tn-proj} for a five qubit example. From this, we construct three classical algorithms. Firstly, the tensor network can be contracted directly and exactly, we call this method \texttt{tn-proj}. Since it cannot exploit the sparsity of $U_B$ (which is only defined on a small subspace of the total Hilbert space), but only its structure, we would expect this to be very inefficient in general. 

Alternatively, we consider the approximating $\mathrm{tr}(U_B P_\mathcal{F})$ using the Hutch++ stochastic trace estimator~\cite{meyer2021hutch}, which estimates such traces to relative error $\epsilon$ as a linear combination of $O(\epsilon^{-1})$ amplitudes $\bra{\psi}U_BP_\mathcal{F}\ket{\psi}$. Each amplitude is computed using a statevector simulation, which can effectively exploit the sparsity of $U_B$ by storing only the coefficients corresponding to bitstrings in $\mathcal{F}$, thus lowering the complexity of computation from $O(2^n)$ to $O(\phi^n)$. This is effectively a statevector simulation of Shor and Jordan's algorithm, so we call it \texttt{sv-shor}. 

Finally, we consider approximating the contraction of \texttt{tn-proj} by constructing a matrix product operator (MPO) from $U_B P_\mathcal{F}$, from which the trace can be computed in polynomial time with respect to the bond dimension $\chi$ of the MPO by exact contraction. We call this method \texttt{mpo-proj}. By fixing a maximum bond dimension and compressing bonds using the usual singular-value method, we can control the approximation cost. 

The scaling of the $\chi$ required to meet the desired error budget can be characterised empirically using the same benchmark method used to characterise the quantum algorithm (Sec.~\ref{sec:benchmark}). As in \texttt{sv-shor}, this method can take advantage of the sparsity of $U_B$ by discarding singular values which are exactly zero. In practice, we find that the maximum bond dimensions are often Fibonacci numbers, indicating that this is indeed occurring (see Fig.~\ref{fig:maxchi-fib}).

\subsection{Classical algorithms for the Jones polynomial}

Various classical methods to compute the Jones polynomial have been proposed, however, there are comparatively few results on the complexity of these algorithms at scale, either empirically or theoretically. As discussed further in App.~\ref{app:classicalalgos}, we believe the most competitive method is due to Kauffman~\cite[p. 125]{Kauffman2001} which computes the Jones polynomial, via the bracket polynomial, by contracting a tensor network defined over a polynomial ring (this network is distinct from the one considered by \texttt{tn-proj}). We call this method \texttt{bracket}.

\begin{figure}[t]
    \includegraphics[width=0.5\textwidth]{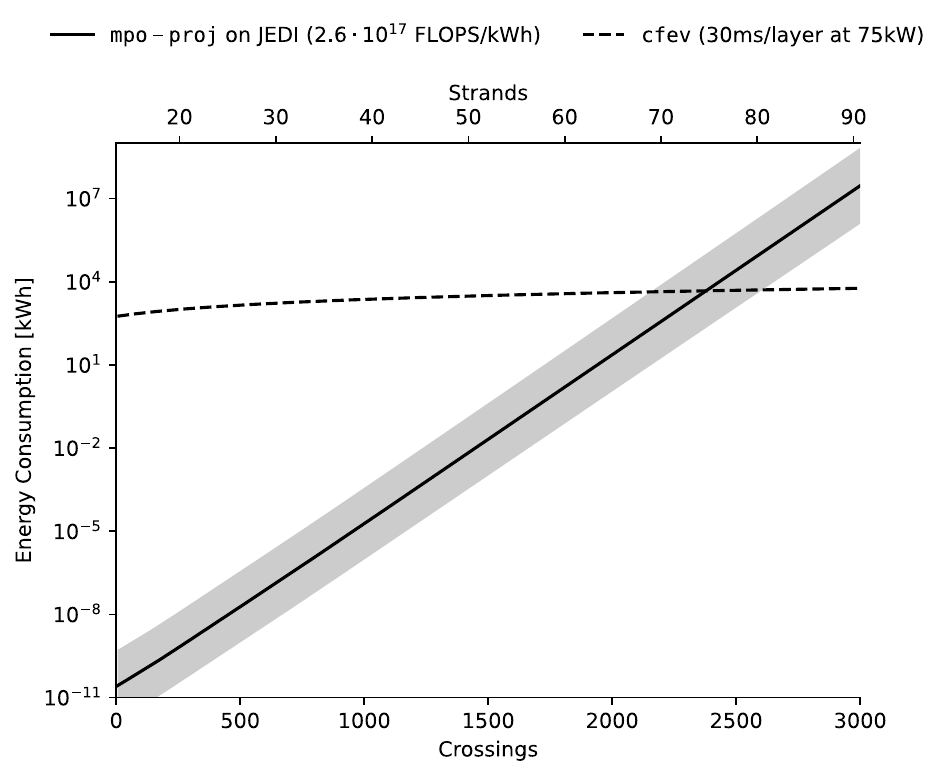}
    \caption{A comparison of the energy consumption between the most competitive classical algorithm \texttt{mpo-proj}, assuming it is running on J\"ulich's JEDI supercomputer \cite{2024jedi} (without considering memory limitations) with an efficiency of $2.6\cdot 10^{17}$ FLOPS/kWh, and our quantum algorithm \texttt{cfev}, assuming the same estimated 75kW power usage as Quantinuum's H2 quantum computer~\cite{decross2024computational}.}
    \label{fig:energy}
\end{figure}

\subsection{Empirical characterization}

To determine the exact scaling of \texttt{tn-proj}, \texttt{sv-shor}, \texttt{mpo-proj}, and \texttt{bracket}, we benchmark them empirically on a dataset of braids generated according to the method in Sec.~\ref{sec:benchmark}. We aim to compare the time required for each classical method to approximate $V_{M(B)}(\e^{\im \frac{2\pi}{5}})$ to the same relative error as \texttt{cfev}. 

We assume two different classical hardware targets: an A100 80GB GPU~\cite{choquette2021a100} to represent typical maximum performance of a single compute node, and Frontier~\cite{atchley2023frontier} representing typical maximum performance available in a large high-performance supercomputer. In both cases, we assume perfect parallelism and utilization, and calculate the time required as the number of 32-bit floating point operations (FLOPS) divided by theoretical maximum FLOPS per second. For the quantum algorithms, we assume the time to solution scales with the depth of the quantum circuit as 30ms per layer (this is similar to existing trapped-ion systems~\cite[Table I]{moses2023h2}), and that the relative error scales as in Fig.~\ref{fig:noise-characterisation}.

The dataset contained 10k braids of up to 600 crossings and between 10 and 30 strands. When constructing large instances, it is not sufficient to scale up the number of crossings or strands individually while keeping the other fixed, since this will lead to braids that are too shallow or too deep. Therefore, in order to extrapolate to larger sizes than we could benchmark in practice, we computed the principal component analysis of the dataset in terms of crossings, depth, and strands, and extrapolated along the principal component, so that crossings and strands increase together.

For \texttt{tn-proj} and \texttt{bracket}, we used Cotengra~\cite{gray2021hyper} to calculate optimized contraction paths and estimate the contraction cost. For \texttt{sv-shor}, we utilize the estimates of distributed statevector computations from~\cite{imamura2022mpiqulacs} to calculate the cost per crossing. We use an optimistic QBF factor of $0.25$ (a parameter related to the FLOPS per statevector entry per gate) so as to handicap \texttt{cfev} as much as possible. The number of amplitudes $N$ to calculate was taken as $\lceil 4\epsilon_\mathrm{noise}^{-1}\rceil$ as in~\cite[Theorem 10]{meyer2021hutch}. 

Finally, for \texttt{mpo-proj}, we used Quimb~\cite{gray2018quimb} to implement MPO evolution. For each braid, we first ran the algorithm without compressing any non-zero singular values, so that we could observe the run-time scaling for an error-free computation. These results are shown in Fig.~\ref{fig:chi-and-mpo-crossings}(b). Then, we ran the same computation with the bond dimension limited to $\chi \leq 2^k$ for $k \in [3, 10]$, to estimate the error resulting from a given level of compression. In Fig.~\ref{fig:chi-and-mpo-crossings}(a), using the known value of the Jones polynomial for each braid, we computed the relative error as it depends on the maximum available memory. Taking the $97.5\%$ lowest quantile of this relationship (that is, the lower bound of a 95\% confidence interval) to handicap \texttt{cfev} as much as possible, we computed the maximum value of $\chi$ needed to achieve the same relative error as \texttt{cfev}, and using the theoretical FLOP counts of each operation~\cite[Table 3.13]{blackford1999lapack} we computed the time to solution across the extrapolated dataset for both $\epsilon_\mathrm{2q} = 1 \cdot 10^{-4}$ and $\epsilon_\mathrm{2q} = 5 \cdot 10^{-4}$.

The results are shown in Fig.~\ref{fig:mpo-comparison-and-time-to-soln}(a). We can see that for $c \geq 1100$, the relative error of the quantum computer for $\epsilon_\mathrm{2q} = 5 \cdot 10^{-4}$ is so high so that the bond dimension required to approximate it (according to Fig.~\ref{fig:chi-and-mpo-crossings}) drops to zero. This happens while the time to solution for \texttt{mpo-proj} is still much smaller than that of \texttt{cfev}, and hence it is likely not possible to achieve quantum advantage with $\epsilon_\mathrm{2q} = 5\cdot 10^{-4}$.

In Fig.~\ref{fig:mpo-comparison-and-time-to-soln}(b), we compare the time-to-solution scaling for all methods when $\epsilon_\mathrm{2q} = 1\cdot 10^{-4}$. We observe that \texttt{mpo-proj} is the most competitive classical method, but that all methods hit the memory limits of the corresponding hardware target before they become slower than the quantum computer. Above these limits, the computation would likely not be possible on that classical hardware. Even if this were not the case, our results suggest that \texttt{mpo-proj} executed on Frontier will likely be slower than \texttt{cfev} when $c \geq 2800$, at which point the relative error of \texttt{cfev} would be $\sim 40\%$.

Furthermore, at large scales, \texttt{cfev} may be more energy efficient. For example, \cite{decross2024computational} pessimistically estimates the energy consumption of Quantinuum's H2 quantum computer at 75kW. By comparison, the most efficient classical computing cluster as of November 2024 is the JEDI system at the J\"ulich Supercomputer Center~\cite{2024jedi}, with an efficiency of $2.6 \cdot 10^{17}$ FLOPS/kWh. Assuming these values, we expect that \texttt{cfev} becomes more energy-efficient than \texttt{mpo-proj} around $c \geq 2400$, as shown in Fig.~\ref{fig:energy}.

\subsection{Limitations}\label{sec:assumptions}

In the previous section, we determined the size of braid for which the \texttt{cfev} method may be faster or more energy-efficient than the competing classical methods, when evaluated on the set of benchmark braids. A priori, there is no reason to conclude that the advantage extends to braids outside of this set. Indeed, the benchmark braid instances are not generic or typical, and they have a special structure by virtue of their construction. 

However, it is reasonable to assume that, 
for generic classical algorithms, generic braids without the same special structure will be at least as hard to evaluate. Under this assumption, the results of the previous section would provide a lower bound on the classical hardness of a generic family of braids. 
In Appendix~\ref{app:limitations}, we show that the \texttt{mpo-proj} method may be able to exploit some of the structure of the benchmark braids. This is encouraging, as it suggests that generic braids may be even harder to solve with the classical methods tested here. In contrast, the runtime of the \texttt{cfev} quantum algorithm does not depend on the family of braids, but only on their size. Thus, given the aforementioned lower bound on the classical hardness of generic braids, the estimated crossover point where we expect a quantum advantage for the benchmark braids is an estimate of the upper bound for where we might see a quantum advantage for generic braids.

It is also important to stress that our conclusions here are based on extrapolations from necessarily limited data (up to $c = 600$). The main obstacle to benchmarking the classical methods (especially \texttt{mpo-proj}) on larger braids was the amount of memory available on a single GPU and the computation time required. This is promising, as it indicates these problems are genuinely hard to solve classically, but conversely it means that if one wanted to be certain that a particular instance of this problem is faster on a quantum computer than the classical equivalent, then more compute and larger-scale datasets would be needed.

On the one hand, generating random braids is reminiscent of generating random circuits to demonstrate quantum computational supremacy. On the other hand, the problem of approximating the Jones polynomial is not only a complete problem for quantum complexity classes but also potentially interesting in knot theory, a rich field of research in low-dimensional topology. See App.~\ref{app:simulations} for further details on the estimates, extrapolation, and benchmark results.

\section{End to End Pipeline}

We have introduced a comprehensive, configurable pipeline enabling the identification of problem sizes exhibiting quantum advantage, which serves as the main contribution of this work, beyond the introduction of the most competitive classical and quantum algorithms for the Jones polynomial.
Our benchmark can characterise the error scaling for both a quantum algorithm running on a noisy quantum computer and classical approximation algorithms running on state-of-the-art classical hardware.

Code for the end-to-end implementation of our quantum algorithm, the efficiently verifiable benchmark, as well as the competitive classical algorithms, can be found in an online repository~\cite{repo}. Given an input braid or link, this implementation can produce quantum circuits to be executed and perform the necessary classical pre- and post-processing. Additional tools for braid simplification and to perform the compilation of the braid generators for new architectures are also included.

Given the most efficient quantum and classical algorithms for solving a problem, the search for a problem instance for which it is advantageous to use a quantum processor over a classical computer is multidimensional.
These dimensions are the error on the estimated quantity, the wall-clock time taken to carry out the computation, the specifications of the quantum computer (2-qubit gate fidelity and depth-1 circuit time), and the analogous specifications of the classical computer (FLOPS per second). Note that achieving a quantum gate fidelity above a certain value may require error detection, error mitigation, or even error correction, which can significantly impact the clock speed of the quantum processor. The specifications of the quantum computer may vary widely over different platforms, but our approach is agnostic and applicable to all. 
To detect quantum advantage, one then fixes a desired error budget and searches for an instance for which time-to-solution on a quantum computer with given specifications is lower than on a state-of-the-art classical computer. Alternatively, given an instance of interest, one may search for the quantum system requirements for achieving a time-to-solution advantage. We have also demonstrated how energy efficiency advantage can be quantified and estimated. The reconfigurability of our pipeline enables the exploration of this multidimensional space. More generally, with our rigorously quantitative and practical approach, we aim to shift the quantum advantage discussion towards \emph{problem instances} rather than problem classes.

Since we have reduced the search for advantage instances, given hardware specifications, to a complex optimisation problem, we envision that the search for such quantum-easiest classically-hardest instances can readily be approached with state-of-the-art AI methods for scientific discovery~\cite{wagner2021constructionscombinatoricsneuralnetworks,Davies2021}, which we leave for future work. For instance, given a particular braid, finding a smaller but topologically equivalent braid that may be easier to solve classically is a difficult task that has received some attention~\cite{bangert2001search}, including by reinforcement learning~\cite{gukov2021learning}.

Finally, we stress that even though we focused on the APPROX-JONES-MARKOV for benchmarking and resource estimation, the problem APPROX-JONES-PLAT can also be approached similarly. The benchmark set of braids can be converted to Plat-closed braids (see Fig.~\ref{fig:markov-plat}(b)), and by construction, their Jones polynomials are invariant under this conversion. Similarly, \texttt{mpo-proj} can be adapted to use matrix product states rather than matrix product operators, and we would expect it to remain the most competitive classical baseline. 

\section{Discussion}

We have provided rigorous resource estimates for the DQC1-complete problem of approximating the Jones polynomial of a Markov-closed braid, up to an exponentially large scale-factor, at the fifth root of unity when solved on a noisy BQP machine, i.e. a digital quantum computer.
We optimally compiled the control-free version of the Hadamard test to reduce the 2-qubit gate count, and used echo verification to reduce the number of shots.
Our extrapolated results suggest that a quantum processor equipped with $\approx 100$ qubits, $\epsilon_\mathrm{2q} \leq 10^{-4}$, and depth-$1$ circuit-times on the order of tens of milliseconds is promising for \emph{demonstrating quantum advantage in terms of time-to-solution}.
The classical algorithms used as baselines include classical simulations of the quantum circuits, which are in general more challenging than they are for the BQP-complete formulation of the problem which regards Plat-closed braids.
Therefore, we have used a `less quantum' problem to gain `more advantage'.

Our problem-specific efficiently verifiable benchmark can be used to characterize the error scaling of a noisy digital quantum computer,
as well as the heuristic performance of approximate classical methods.
Thus, we provide practical tools for locating the advantage boundary given the current state-of-the-art classical methods.
Once a quantum computer with sufficient specifications is available, as the one assumed in this work, a quantum advantage for this problem may be identified.

We also showed that the BQP-complete version of the problem, where the braid is instead Plat-closed, can be solved using the same quantum protocol with a simple modification.
Therefore, our end-to-end pipeline can be used to identify the minimum requirements for quantum advantage for evaluating the Jones polynomial for Plat-closed braids as well.

Future work involves rigorous resource estimates in the fault-tolerant regime, compiling the algorithm to the appropriate gateset and assuming the relevant overheads and realistic methods. This analysis will inevitably be required for obtaining better precision, which depends on the fidelity of the quantum protocol. The current belief is that gate errors below $10^{-4}$ will require more error mitigation and perhaps error correction.
In the near term, memory errors may be addressed via decoherence-free subspace codes, and given that our quantum algorithm operates only in a small subspace of the full Hilbert space, it may be possible to exploit this to decrease the cost of error correction, for example by allowing smaller code distances.

Our algorithm and resource estimation pipeline can also be generalised to other non-lattice roots of unity $\e^{\im 2\pi/k}$, for $k\in\mathbb{N}$, where other unitary representations of the braid group would need to be compiled (see App.~\ref{app:non-fibo-sketch} for a brief sketch)~\cite{Kauffman_2010}.
With the ability to approximate the value of the Jones polynomial at roots of unity, one could distinguish knots and links with high probability, a computational tool that practitioners in the field of knot theory could utilise.
Consequently, our algorithm could be adapted to compute the coloured Jones polynomial \cite{Kauffman2001}.
Our general approach for benchmarking and diagnosing system requirements for quantum advantage in practice also applies to other more powerful invariants such as the Khovanov homology \cite{schmidhuber2025quantumalgorithmkhovanovhomology}.

\vspace{2mm}
\section*{Acknowledgements}
\vspace{-2mm}

We thank David Amaro and Etienne Granet for comments on the manuscript, and Adam Connolly, Natalie Brown, and Harry Buhrman for discussions.
We thank Julia Cline, Joan M. Dreiling, Alex Hall, Aaron Hankin, Azure Hansen, Nathan Hewitt, Chris Langer, Elliot Lehman, Dominic Lucchetti, and the entire operations team for technical support during the implementation on the Quantinuum H2-2 quantum computer.
This research used resources of the National Energy Research Scientific Computing Center, a DOE Office of Science User Facility supported by the Office of Science of the U.S. Department of Energy under Contract No.~{DE-AC02-05CH11231} using NERSC award NERSC~{DDR-ERCAP0029825}.

\onecolumngrid
\appendix

\vspace{20pt}
\begin{center}
{\bf\Large Appendix}    
\end{center}

\section{The control-free Hadamard test with echo-verification}
\label{app:qualgorithm}

In this work, we employ Hadamard tests as key algorithmic primitives.
Such protocols are designed to estimate $\bra{s}U\ket{s}$ for a given unitary $U$ and initial state $\ket{s}$. 
The execution of Hadamard test circuits on NISQ computers is affected by two main sources of errors: shot noise $\epsilon_\mathrm{shot}$, which contributes to the variance of any estimator, and gate noise $\epsilon_\mathrm{noise}$, which contributes to the bias of any estimator. It is therefore important to optimise the estimator and reduce the number of qubits and gates in the circuit. We now briefly summarise \texttt{ht}, and its variant \texttt{cfev} that forms the basis for our algorithm. The reader may refer to Ref.~\cite{Polla_2023} for a fully detailed exposition.

\texttt{ht} encodes the relevant information in the relative phase between the $\ket{0}$ and $\ket{1}$ states of the ancilla qubit by means of the following circuit: 
\begin{center}
    \input{figs/H-test.tikz}
\end{center}
The relevant information is then extracted by measuring the ancilla and averaging the binary outcomes. By setting $b=0$ one can estimate $\Re(\bra{s} U \ket{s})$, and setting $b=1$ allows one to estimate $\Im(\bra{s} U \ket{s})$.

Now, suppose that an eigenvalue/eigenstate pair $U\ket{\lambda} = \e^{\im \theta} \ket{\lambda}$ is known, and that the initial state is orthogonal to the eigenstate, $\braket{s}{\lambda} = 0$. Then one can \emph{remove the ancilla} qubit in \texttt{ht} and encode the relevant information in the relative phase between $\ket{\lambda}$ and $\ket{s}$. This is achieved using a cat-state preparation circuit, $V_\mathrm{cat}(s)(\ket{+} \otimes \ket{0}^{\otimes {n-1}}) = \frac{1}{\sqrt{2}}(\ket{\lambda} + \ket{s})$. A single qubit rotation gate $R_Z(\theta)$ is used to compensate the phase of the known eigenvalue.
Importantly, the desired unitary, $U$, can be compiled using fewer entangling gates than its controlled version, and so the gate error $\epsilon_\mathrm{noise}$ is smaller. 
Further, an error mitigation trick called \emph{echo verification} reduces the shot noise $\epsilon_\mathrm{shot}$.
Since the state preparation circuit for the initial state, $V_\mathrm{cat}(s)$, is known, at the end of the circuit its inverse, $V_\mathrm{cat}^\dagger(s)$, is applied.
The measurement outcomes are postprocessed by a Boolean function $g$ which checks whether the initial state is recovered or not. This information is then combined with measurements on the ancilla to form an unbiased estimator with reduced variance. This leads to the control-free Hadamard test with echo-verification (\texttt{cfev}) protocol that we use in this work.

Let us now specialise the discussion to our specific use case. Recall that for the estimation of both Markov and Plat closures of braid $B$, we have constructed a unitary $U_B$ with known eigenvector $\ket{0}^{\otimes n}$ and corresponding eigenvalue $\e^{\im \theta_B}$. Moreover, the initial state $\ket{s}$ always encodes a bitstring $s$, so there exist an efficient circuit $V_\mathrm{cat}(s)$ for cat-state preparation. This leads us to \texttt{cfev} circuit shown in Fig.~\ref{fig:H-test-CFEV-with-catstate-unitary}(a), main text. As we show next, this protocol has better scaling in both $\epsilon_\mathrm{shot}$ and $\epsilon_\mathrm{noise}$ errors compared to \texttt{ht}. Moreover, the final measurement reveals whether the main register is in the Fibonacci subspace. This is used for our error detection strategy explained in App.~\ref{app:mitigation}. 

For simplicity of notation we drop subscripts $B$ and $\mathrm{cat}$, and denote the initial state as $\ket{s} = \ket{01s_3s_{4..}}$. Without loss of generality we assume that $\theta_B = 0$, i.e. that the known eigenvalue is $1$. The action of the cat-state preparation circuit $V$ is such that $V\ket{0 x_2 x_{3..}} = \ket{x_2 0 x_{3..}}$ and $V\ket{1x_2x_{3..}} = \ket{x_21(x_{3..} \oplus s_{3..})}$. Note that $V\ket{0}^{\otimes n} = \ket{0}^{\otimes n}$ and $V\ket{10\cdots} = \ket{s}$.
Therefore 
\begin{align*}
    \ket{\psi} = UV (H \otimes I^{\otimes n -1}) \ket{0}^{\otimes n} = \frac{1}{\sqrt{2}}UV(\ket{00\cdots} + \ket{10\cdots}) = \frac{1}{\sqrt{2}}U(\ket{0}^{\otimes n} + \ket{s}) = \frac{1}{\sqrt{2}}(\ket{0}^{\otimes n} + U\ket{s}) 
\end{align*}

At the end of the \texttt{cfev} circuit, the computational basis has the following amplitudes
\begin{align*}
    \bra{\psi}V^\dag(H \otimes I^{\otimes n-1})\ket{x_1x_2x_{3..}} &= \frac{1}{\sqrt{2}}\bra{\psi}V^\dag(\ket{0x_2x_{3..}} + (-1)^{x_1} \ket{1x_2x_{3..}}) \\
    &= \frac{1}{\sqrt{2}}\bra{\psi} (\ket{x_20x_{3..}} + (-1)^{x_1} \ket{x_21(x_{3..} \oplus s_{3..})}) \\
    &= \frac{1}{2}(\braket{00\cdots}{x_2 0 x_{3..}} + \overbrace{(-1)^{x_1} \braket{00\cdots}{x_2 1 (x_{3..} \oplus s_{3..})}}^{=0} \\
    &\quad + \overbrace{\bra{s}U^\dag\ket{x_2 0 x_{3..}}}^{=0} + (-1)^{x_1}\bra{s}U^\dag\ket{x_2 1 (x_{3..} \oplus s_{3..})}) 
\end{align*}
Note that $U$ maps Fibonacci strings starting with $01$ to Fibonacci strings starting with $01$. Thus in the last step the third term is zero because $U\ket{s}$ is a Fibonacci string starting with $01$, while $\ket{x_2 0 x_{3\dots}}$ is not. Using this property of $U$ we can also write
\begin{align}
\label{eq:post_measurement_cases}
    \bra{\psi}V^\dag(H \otimes I^{\otimes n-1})\ket{x_1x_2x_{3..}} = \begin{cases}
    \frac{1}{2}(1 + (-1)^{x_1}\bra{s}U^\dag\ket{s}) & \text{if }x_2 = 0 \text{ and } x_{3..} = 0 \quad\text{(1)} \\
    \frac{1}{2}(-1)^{x_1}\bra{s}U^\dag\ket{x_2 1 (x_{3..} \oplus s_{3..})} & \text{if } x_2 = 0 \text{ and } x_{3..} \oplus s_{3..} \text{ is Fibonacci} \quad\text{(2)} \\
    0 & \text{otherwise} \quad\text{(3)}
\end{cases}
\end{align}

Let us now define two post-processing functions $g_1$ and $g_2$ for the measurement outcomes. If the measurement detects case (1), we set $g_1 = 1$, $g_2 = 1$. If the measurement detects case (2), we set $g_1 = 0$, $g_2 = 1$. Finally if the measurement detects case (3), we set $g_1 = 0$, $g_2 = 0$. The amplitude of case (3) is zero, meaning that this event should never happen (we will utilise this information for error mitigation in App.~\ref{app:mitigation}). Let us define the ternary random variable $r = (-1)^{x_1} \cdot g_1 \in \{-1, 0, 1\}$. From the Born rule, the probabilities of the three outcomes are:
\begin{align*}
    &P(r = +1) = P(x_1 = 0, x_2 = 0, x_{3..} = 0) = \frac{1}{4}|1 + \bra{s}U^\dagger\ket{s}|^2 = \frac{1}{4} + \frac{|\bra{s}U\ket{s}|^2}{4} + \frac{\mathrm{Re}(\bra{s}U\ket{s})}{2}  \\
    &P(r = -1) = P(x_1 = 1, x_2 = 0, x_{3..} = 0) = \frac{1}{4}|1 - \bra{s}U^\dagger\ket{s}|^2 = \frac{1}{4} + \frac{|\bra{s}U\ket{s}|^2}{4} - \frac{\mathrm{Re}(\bra{s}U\ket{s})}{2}  \\
    &P(r = 0) = 1 - P(r = +1) - P(r = -1) = \frac{1}{2} - \frac{|\bra{s}U\ket{s}|^2}{2} 
\end{align*}
The first two raw moments of $r$ are
\begin{align*} 
    &\mathbb{E}[r] = \left[\frac{1}{4} + \frac{|\bra{s}U\ket{s}|^2}{4} + \frac{\mathrm{Re}(\bra{s}U\ket{s})}{2}\right] - \left[\frac{1}{4} + \frac{|\bra{s}U\ket{s}|^2}{4} - \frac{\mathrm{Re}(\bra{s}U\ket{s})}{2}\right] = \mathrm{Re}(\bra{s}U\ket{s})  \\
    &\mathbb{E}[r^2] = \left[\frac{1}{4} + \frac{|\bra{s}U\ket{s}|^2}{4} + \frac{\mathrm{Re}(\bra{s}U\ket{s})}{2}\right] + \left[\frac{1}{4} + \frac{|\bra{s}U\ket{s}|^2}{4} - \frac{\mathrm{Re}(\bra{s}U\ket{s})}{2}\right] = \frac{1}{2} + \frac{|\bra{s}U\ket{s}|^2}{2} 
\end{align*}
Hence the variance is
\begin{align*}
    \mathrm{Var}(r) = \frac{1}{2} + \frac{|\bra{s}U\ket{s}|^2}{2} - \mathrm{Re}(\bra{s}U\ket{s})^2 
\end{align*}

Let us compare to the variance of the \texttt{ht} protocol. Recall that the \texttt{ht} protocol uses a controlled-$U$ operation denoted $CU$ such that $CU\ket{0}\ket{\psi} = \ket{0}\ket{\psi}$ and $CU\ket{1}\ket{\psi} = \ket{1}U\ket{\psi}$. The computational basis measurement for the ancilla has probabilities:
\begin{align*}
    P(x) &= \bra{+}\bra{s}CU^\dagger(H \otimes I^{\otimes n})(\dyad{x}\otimes I^{\otimes n})(H \otimes I^{\otimes n})CU\ket{+}\ket{s} \\
         &= \frac{1}{2}(\bra{0}\bra{s} + \bra{1}\bra{s}U^\dagger)(H \otimes I^{\otimes n})(\dyad{x}\otimes I^{\otimes n})(H \otimes I^{\otimes n})(\ket{0}\ket{s} + \ket{1}U\ket{s}) \\
         &= \frac{1}{4}(\bra{0}\bra{s} + \bra{1}\bra{s} + \bra{0}\bra{s}U^\dagger - \bra{1}\bra{s}U^\dagger)(\dyad{x}\otimes I^{\otimes n})(\ket{0}\ket{s} + \ket{1}\ket{s} + \ket{0}U\ket{s} - \ket{1}U\ket{s}) \\
         &= \frac{1}{4}(|\braket{0}{x}|^2(2 + \bra{s}U^\dagger\ket{s} + \bra{s}U\ket{s}) + |\braket{1}{x}|^2(2 - \bra{s}U^\dagger\ket{s} - \bra{s}U\ket{s})) \\
         &= \frac{1}{2}(1 + (-1)^{x}\mathrm{Re}(\bra{s}U\ket{s})) 
\end{align*}

We now define the binary random variable $r' = (-1)^x$ which is estimated by measuring the ancilla and discarding the rest. We have that:
\begin{align*}
    \mathbb{E}[r'] = \mathrm{Re}(\bra{s}U\ket{s})  \qquad \qquad 
    \mathbb{E}[(r')^2] = 1  \qquad \qquad
    \mathrm{Var}(r') = 1 - \mathrm{Re}(\bra{s}U\ket{s})^2     
\end{align*}
Therefore:
\begin{align*}
    \mathrm{Var}(r) = \mathrm{Var}(r') - \frac{1 - |\bra{s}U\ket{s}|^2}{2} \leq \mathrm{Var}(r') 
\end{align*}
with equality only when $|\bra{s}U\ket{s}|^2 = 1$. An identical result is obtained when estimating the imaginary part (this corresponds to inserting an extra $S^\dagger$ gate into the circuit). So \texttt{cfev} is usually better than \texttt{ht} in terms of variance.

Let's now consider the case of estimating $\bra{s}U\ket{s}$ by combining the estimator of the real part, $r$, with the one for the imaginary part, $r_c$. For \texttt{cfev} we have:
\begin{align*}
    &\mathbb{E}[r + \im r_c] = \bra{s}U\ket{s} \\
    &\mathrm{Var}(r + \im r_c) = \mathrm{Var}(r) + \mathrm{Var}(r_c) = 1 + |\bra{s}U\ket{s}|^2 - \mathrm{Re}(\bra{s}U\ket{s})^2 - \mathrm{Im}(\bra{s}U\ket{s})^2 = 1 
\end{align*}
Remarkably the variance of \texttt{cfev} does not depend on the state $\ket{s}$. This is not the case for \texttt{ht} since
\begin{align*}
    &\mathbb{E}[r' + \im r_c'] = \bra{s}U\ket{s} \\
    &\mathrm{Var}(r' + \im r_c') = \mathrm{Var}(r') + \mathrm{Var}(r_c') = 2 - \mathrm{Re}(\bra{s}U\ket{s})^2 - \mathrm{Im}(\bra{s}U\ket{s})^2 = 2 - |\bra{s}U\ket{s}|^2 
\end{align*} 
so we have
\begin{align*} 
    \mathrm{Var}(r + \im r_c) = \mathrm{Var}(r' + \im r_c') - (1 - |\bra{s}U\ket{s}|^2) \leq \mathrm{Var}(r' + \im r_c') 
\end{align*}
with equality only when $|\bra{s}U\ket{s}|^2 = 1$. \texttt{cfev} can have half the variance of \texttt{ht} when $\bra{s}U\ket{s} = 0$.

Finally, we comment on weighted trace estimation, which we use to estimate the Jones polynomial of Markov closed braids. In this case, the estimator comprises an additional average over different realizations of the initial state $\ket{s}$ each with probability $P(s)$. Let us denote the $s$-dependent random variables as $r_s +\im r_{c,s}$ and $r_s' +\im r_{c,s}'$ for \texttt{cfev} and \texttt{ht}, respectively. We have 
\begin{align*} 
    \mathbb{E}[ r_s + \im r_{c,s}] = \mathbb{E}[ \mathbb{E}[r_s + \im r_{c,s} | s] ] =  \mathbb{E}[\bra{s}U\ket{s}] = \tr[ \sum_s P(s) \dyad{s} U] 
\end{align*}
where we use the law of total expectation. This is the trace of $U$ weighted by $P(s)$. The result is identical for the \texttt{ht} protocol, $\mathbb{E}[ r_s' + \im r_{c,s}']$. Instead, for the variance we have
\begin{align*} 
    \mathrm{Var}(r_s + \im r_{c,s}) &= \mathbb{E}[ \mathrm{Var}(r_s + \im r_{c,s} | s)] + \mathrm{Var}( \mathbb{E}[r_s + \im r_{c,s} | s]) 
    = 1 + \mathbb{E}[|\bra{s}U\ket{s}|^2] - |\mathbb{E}[\bra{s}U\ket{s}]|^2  \\
    \mathrm{Var}(r_s' + \im r_{c,s}') &= \mathbb{E}[ \mathrm{Var}(r_s' + \im r_{c,s}' | s)] + \mathrm{Var}( \mathbb{E}[r_s' + \im r_{c,s}' | s]) 
    = 2 - |\mathbb{E}[\bra{s}U\ket{s}]|^2 
\end{align*}
where we use the law of total variance. Once again we observe that the variance of \texttt{cfev} is upper bounded by that of \texttt{ht}.

\section{Weighted sampling from \texorpdfstring{$\mathcal{F}$}{F}}
\label{app:zeckendorf-algo}

Consider $|s\rangle \in \mathcal{F}_n$ as defined in Eq.~\eqref{eq:fibonacci-basis}. Since $s_0 = 0$, we must have $s_1 = 1$, and that $s_2s_3\dots s_{n-1} \in \mathcal{F}'_{n-2}$, where we define the set $\mathcal{F}'_n$ of Fibonacci strings of length $n$ as
$$\mathcal{F}'_n = \{ s\in \mathbb{B}^n \mid s_i + s_{i+1} > 0\}$$
and indeed we have that $\{ s_2s_3\dots s_{n-1} \mid s \in \mathcal{F} \} = \mathcal{F}'_{n-2}$. Suppose we want to uniformly sample from $\mathcal{F}$ conditioned on the fact that $s_{n-1} = 0$. Then we must have $s_{n-2} = 1$, and the remaining substrings $s_2\dots s_{n-3}$ satisfy $\{ s_2s_3\dots s_{n-3} \mid |s\rangle \in \mathcal{F}, s_{n-1} = 0 \} = \mathcal{F}'_{n-4}$. Likewise if we condition on $s_{n - 1} = 1$, then we have $\{ s_2s_3\dots s_{n-2} \mid |s\rangle \in \mathcal{F}, s_{n-1} = 1 \} = \mathcal{F}'_{n-3}$. In either case, this is reduced to sampling uniformly from $\mathcal{F}'_m$ for some $m$. Zeckendorf's theorem states that there is a bijection between the integers $0 \leq k < f_{n+2}$ and the elements of $\mathcal{F}'_n$, so $|\mathcal{F}'_n| = f_{n+2}$. The bijection can be computed in linear time. Define a sequence $k_i$ with $k_{n - 1} = k$, then:
$$ k_{i - 1} = \begin{cases}
    k_{i} & \text{if } k_{i} < f_{i+2}  \\
    k_{i} - f_{i+2} & \text{if }  k_{i} \geq f_{i+2} 
\end{cases} \qquad \qquad s_{i} = \begin{cases}
    1 & \text{if } k_{i} < f_{i+2} \\ 
    0 & \text{if } k_{i} \geq f_{i+2}
\end{cases}$$
Now consider sampling from $\mathcal{F}$ weighted by $\kappa(s) = \phi^{s_{n-1}}$. Since $|\mathcal{F}'_{n-3}| = f_{n-1}$ and $|\mathcal{F}'_{n-4}| = f_{n-2}$, the probability of obtaining a given $s_{n-1}$ should be:
$$ P(s_{n-1} = 1) = \frac{\phi f_{n-1}}{\phi f_{n-1} + f_{n-2}} \qquad \qquad P(s_{n-1} = 0) = \frac{f_{n-2}}{\phi f_{n-1} + f_{n-2}}$$
Then given $s_{n-1}$, the problem reduces to sampling uniformly from some $\mathcal{F}'_m$, which can be done by picking an integer $0 \leq k < f_{n+2}$ and computing $s_i$ as described above.  The overall method (see Algorithm~\ref{alg:sampling}) is as follows:
\begin{enumerate}
    \item By definition $s_0 = 0$ and $s_1 = 1$. To obtain $s_{n-1}$ sample according to $P(s_{n-1})$.
    \item If $s_{n-1} = 0$ then $s_{n-2} = 1$ by definition. Pick $0 \leq k < f_{n-2}$ uniformly and use the bijection to fill in $s_2\dots s_{n-3}$.
    \item If $s_{n-1} = 1$ then pick $0 \leq k < f_{n-1}$ uniformly and use the bijection to  fill in $s_2\dots s_{n-2}$.
\end{enumerate}
We can see that sampling from $\mathcal{F}$ weighted by $\kappa(s)$ is the same as sampling from the distribution 
$$ P(s) = \begin{cases}
    0 & \text{if } |s\rangle \notin \mathcal{F}\\
    \frac{\phi^{s_{n-1}}}{\phi^{n-1}} & \text{otherwise}
    \end{cases} $$
since the total weight of $\kappa(s)$ over $\mathcal{F}$ is given by $\phi f_{n-1} + f_{n-2} = \phi^{n-1}$.

\begin{center}

\begin{minipage}{0.7\textwidth}
\begin{algorithm}[H]
\caption{A procedure based on Zeckendorf's theorem to efficiently classically sample a bit string $\ket{s}$ from $\mathcal{F}_n$ according to distribution $P(s)$.}\label{alg:sampling}
$s_0 \gets 0$\;
$s_1 \gets 1$\;
$r \gets \text{sample uniformly in } [0, 1]$\;
\eIf{$r \leq \frac{f_{n-1}}{\phi^{n-2}}$}{
    $s_{n - 1} \gets 1$\;
    $j \gets n - 2$\;
    $k \gets$ sample random integer in $[0, f_{n - 1})$\;
}{
    $s_{n - 1} \gets 0$\;
    $s_{n - 2} \gets 1$\;
    $j \gets n - 3$\;
    $k \gets$ sample random integer in $[0, f_{n - 2})$\;
}
\While{$j \geq 2$}{
  \eIf{$k \geq f_j$}{
    $s_j \gets 0$\;
    $k \gets k - f_j$\;
  }{
    $s_j \gets 1$\;
  }
  $j \gets j - 1$\;
}
\Return $\ket{s}$
\end{algorithm}
\end{minipage}

\end{center}

\section{Non-Fibonacci error detection and the conjugate trick}
\label{app:mitigation}

Suppose that $x_1, \dots, x_n$ is a sample from the circuit in Fig.~\ref{fig:H-test-CFEV-with-catstate-unitary}(a) in the main text. Without loss of generality let's assume that $\theta_B = 0$ and $b = 0$. According to Eq.~\eqref{eq:post_measurement_cases} case (3) has zero probability of occurring. That is, it should never be the case that $x_1 \neq 0$ or $(x_i \oplus s_i) + (x_{i+1} \oplus s_{i+1}) = 0$ for some $i \geq 3$. However, due to gate noise in the NISQ computer, we may actually observe case (3) implying that at least one hardware error has occurred. In this case the post-processing function defined after Eq.~\eqref{eq:post_measurement_cases} will take value $g_2=0$ on the erroneous samples. We now argue that under reasonable noise assumptions these samples can be discarded during the estimation of expectation values.

From the definitions of $r, g_1$ and $g_2$ it can be verified that 
$$\mathbb{E}[r] = \mathbb{E}[\mathbb{E}[r\mid g_2]] = P(g_2=1) \mathbb{E}[r | g_2=1]$$
Estimating $r$ from the retained samples ($g_2=1$) is then equivalent to rescaling the original estimator by a factor of $1/P(g_2 = 1)$, without affecting the phase. Certain noise models, such as the depolarizing noise, have the effect of shrinking the magnitude of any expectation value towards zero. Under these noise assumptions we can safely discard any sample for which $g_2 = 0$. This will reduce the bias at the cost of an increase in variance. We expect the same to be at least approximately true for any non-adversarial noise model. This effect can be observed in our experiments by comparing the green and orange markers in Fig.~\ref{fig:accuracy}. We call this technique the \emph{non-Fibonacci error detection}. 

Next we discuss the \emph{conjugate trick}. Suppose we estimated $\mathbb{E}_s[\e^{\im\theta}\bra{s}U_{B}\ket{s}]$ for braid $B$ from samples, obtaining the value $\e^{\im\theta}R$.
For the conjugate braid $B^*$ we have that $U_{B^*} = U_B^*$, since the braid generators satisfy $U_{\sigma^{-1}} = U_\sigma^\dag = U_{\sigma}^*$ due to symmetry. Let us then estimate $\mathbb{E}_s[\e^{\im\theta}\bra{s}U_{B^*}\ket{s}] = \mathbb{E}_s[\e^{\im\theta}\bra{s}U_{B}^*\ket{s}]$ from samples as well, obtaining $\e^{\im\theta}R^*$. We can now recover the value of $R$ up to a choice of signs as
\begin{align*}
    R_{1\cdots 4} = \pm\frac{1}{2}\left|\e^{\im\theta}R + \e^{\im\theta}R^*\right| \pm \frac{\im}{2} \left|\e^{\im\theta}R - \e^{\im\theta}R^*\right|
\end{align*}
where each $R_i$ correspond to a different choice of signs. To disambiguate, we proceed as follows. We can eliminate two of the four options by additionally calculating
\begin{align*}
    R'_{1,2} = \pm |\e^{\im\theta} R| \sqrt{\frac{\e^{\im\theta}R}{\e^{\im\theta}R^*}}
\end{align*}
and discarding the antipodal pair of $R_{1\cdots 4}$ that is furthest from $R'_{1,2}$. The reason for not using $R'_{1, 2}$ initially is that it is not obvious that the variance of $R'_i$ is the same as for $R_i$. The remaining values of $R_i$ are positioned at antipodal points. By selecting the value nearest to $\e^{\im\theta}R$, we ensure the signs are chosen correctly, provided that $|\theta| < \frac{\pi}{2}$.

Applying the conjugate trick relies on the assumption that the overall phase error $\e^{\im\theta}$ is the same for every measurement shot, for both $B$ and $B^*$. This may not be the case in practice -- for example, if each gate accumulated phase error at a fixed rate, then even slight changes to this rate (for instance due to hardware calibration) could cause large changes in the overall phase error $\theta$. Suppose that the phase error is consistent across shots, but for $B^*$ it is instead given by $\e^{\im\theta'}$ with $\theta' \neq \theta$. Then we can see that the value recovered by the conjugate trick is given by:
\begin{align*}
    \tilde{R}_{1\cdots4} &= \pm\frac{1}{2}\left|\e^{\im\theta}R + \e^{\im\theta'}R^*\right| \pm \frac{\im}{2}\left|\e^{\im\theta}R - \e^{\im\theta'}R^*\right| \\
    & = \pm\frac{1}{2}\left|\e^{\im\frac{\theta+\theta'}{2}}\e^{\im\frac{\theta-\theta'}{2}}R + \e^{\im\frac{\theta+\theta'}{2}}\e^{-\im\frac{\theta-\theta'}{2}}R^*\right| \pm \frac{\im}{2} \left|\e^{\im\frac{\theta+\theta'}{2}}\e^{\im\frac{\theta-\theta'}{2}}R - \e^{\im\frac{\theta+\theta'}{2}}\e^{-\im\frac{\theta-\theta'}{2}}R^* \right| \\
    & = \pm\frac{1}{2}\left|\e^{\im\frac{\theta+\theta'}{2}}(\e^{\im\frac{\theta-\theta'}{2}}R) + \e^{\im\frac{\theta+\theta'}{2}}(\e^{\im\frac{\theta-\theta'}{2}}R)^*\right| \pm \frac{\im}{2} \left|\e^{\im\frac{\theta+\theta'}{2}}(\e^{\im\frac{\theta-\theta'}{2}}R) - \e^{\im\frac{\theta+\theta'}{2}}(\e^{\im\frac{\theta-\theta'}{2}}R)^* \right|
\end{align*}
So we can see that this is the same as the case where $\theta = \theta'$, except with $R \to \e^{\im\frac{\theta-\theta'}{2}}R$ and $\theta \to \frac{\theta+\theta'}{2}$. This will be disambiguated to $\tilde{R} = \e^{\im\frac{\theta-\theta'}{2}}R$ if $\left|\frac{\theta + \theta'}{2}\right| < \frac{\pi}{2}$ by the same argument as above. Then, in the absence of finite sampling noise, the relative errors of the estimates with and without the conjugate trick are given by:
\begin{align*}
    \varepsilon_\textrm{with} &= \frac{|R - \tilde{R}|}{|R|} = \left|1 - \e^{\im\frac{\theta - \theta'}{2}}\right| = 2\left|\sin\left(\frac{\theta - \theta'}{4}\right)\right| \\ 
    \varepsilon_\textrm{without} &= \frac{|R -\e^{\im\theta}R|}{|R|} = \left|1 - \e^{\im\theta}\right| = 2\left|\sin\left(\frac{\theta}{2}\right)\right| 
\end{align*}
So applying the conjugate trick will be beneficial compared to the original estimate whenever $|\theta - \theta'| \leq 2|\theta|$, i.e. the relative error between $\theta$ and $\theta'$ must be at most 200\%.

Conversely, in the worst case, suppose that the phase error of each shot is uniformly random, then we have
$$\mathbb{E}[r] = \mathbb{E}[\mathbb{E}[r\mid \theta]] = \mathbb{E}[\e^{\im \theta}]\mathbb{E}[\bra{s}U_B\ket{s}] = 0$$
so the conjugate trick can't recover any information. Therefore, we will now consider the case where the change in phase error is roughly zero within a single measurement shot, but may be arbitrarily different between measurement shots. We will assume that we are estimating the amplitude of a single bitstring $\bra{s}U_B\ket{s}$, and extend this to the case of trace estimation later. We call this the \emph{shot-level conjugate trick}. Consider batching together four different measurements into a single circuit: the real and imaginary parts of $R$, and the real and imaginary parts of $R^*$. We assume that they all have the same phase error because they are within the same circuit (which is not too large). For measurement shot $i$, define the following random variables, corresponding to the measurement outcomes of each of the four parts respectively:
\begin{align*}
    X_i \sim \mathrm{EV}(x, y, \theta_i) \qquad
    Y_i \sim \mathrm{EV}(-y, x, \theta_i) \qquad
    X^*_i \sim \mathrm{EV}(x, -y, \theta_i) \qquad
    Y^*_i \sim \mathrm{EV}(y, x, \theta_i) 
\end{align*}
where $x + \im y = \bra{s}U_B\ket{s}$ and $\theta_i$ is the phase error for this measurement shot. $\mathrm{EV}(x, y, \theta)$ is the distribution of measurement outcomes given by \texttt{cfev} estimating a quantity with real part $x$, imaginary part $y$, and phase error $\theta$, so that
\begin{align*}
    P(X_i = 0) &= \frac{1 - x^2 - y^2}{2} \\
    P(X_i = +1) &= \frac{1 + x^2 + y^2}{4} +  \frac{x\cos(\theta_i) - y\sin(\theta_i)}{2}\\
    P(X_i = -1) &= \frac{1 + x^2 + y^2}{4} -  \frac{x\cos(\theta_i) - y\sin(\theta_i)}{2}
\end{align*}
as computed in App.~\ref{app:qualgorithm}. Then a straightforward but tedious calculation shows that
\begin{align*}
    \mathbb{E}[X_iX_i^* + Y_iY_i^*] = x^2 - y^2 \qquad \qquad\mathbb{E}\left[\frac{X_i^2 + Y_i^2 + (X_i^*)^2 + (Y_i^*)^2}{2} - 1\right] = x^2 + y^2
\end{align*}
which is independent of $\theta_i$, and hence we can recover $x^2$ and $y^2$ as:
\begin{align*}
    X^{EC}_i &= \frac{(X_i + X_i^*)^2 + (Y_i + Y_i^*)^2 - 2}{4} \qquad \qquad \mathbb{E}[X^{EC}_i]= x^2 \\
    Y^{EC}_i &= \frac{(X_i - X_i^*)^2 + (Y_i - Y_i^*)^2 - 2}{4} \qquad \qquad \mathbb{E}[Y^{EC}_i] = y^2
\end{align*}
Then taking $N$ shots, even if each $\theta_i$ is arbitrary, we have that
$$
    \pm\sqrt{\mathbb{E}\left[\frac{1}{N}\sum_{i=1}^NX^{EC}_i\right]} = \pm\sqrt{\frac{1}{N}\sum_{i=1}^N\mathbb{E}[X_i^{EC}]} = \pm\sqrt{\frac{1}{N}\sum_{i=1}^Nx^2} = \pm x
$$
and likewise for $Y^{EC}_i$. Let $M_N = \frac{1}{N}\sum_{i=1}^N X^{EC}_i$, then because $-\frac{1}{2} \leq X^{EC}_i \leq \frac{3}{2}$ we have
$$P(|M_N - x^2| \geq t) \leq 2 \e^{-N\frac{t^2}{2}}$$
from Hoeffding's inequality, and thus:
$$P\left(\left|\sqrt{M_N} - |x|\right| \geq t\right) = P\left(|M_N - x^2| \geq t\left|\sqrt{M_N} + |x|\right|\right) \leq P\left(|M_N - x^2| \geq t|x|\right) \leq 2\e^{-N\frac{t^2x^2}{2}}$$
Note that this inequality can be tightened further \cite{pinelis2023concentration}, but this is enough to imply that the following is a consistent estimator for $x + \im y$ up to the sign of its components:
$$
    R^{EC} = \pm\sqrt{\frac{1}{N}\sum_{i=1}^NX^{EC}_i} \pm \im\sqrt{\frac{1}{N}\sum_{i=1}^NY^{EC}_i}
$$
Since the square root function is not Lipschitz continuous, the variance will depend inversely on the value of $|x + \im y|$, and in any case will be substantially larger than the variance from the regular conjugate trick. 

To adapt this to trace estimation, if we let $x$ and $y$ vary with each shot such that $\mathbb{E}[x] = x'$ and $\mathbb{E}[y] = y'$ then this estimator will be biased, since $\mathbb{E}[x^2] \neq (x')^2$ in general. If the variance of $x$ is small, this may be a reasonable estimator, but to obtain a consistent estimator for the trace, one can instead estimate $R^{EC}$ for each bit string $\ket{s}$ individually and combine these. However, it is important to note that this can only be unbiased (up to a choice of signs) when the signs of the components of each amplitude on the diagonal are all positive or all negative, since the sign information cannot be recovered from the shot-level trick.

We will now prove that this works for estimating the absolute value of the real part of the trace. The imaginary part follows similarly. Let $\{s_j\}$ for $1 \leq j \leq m$ and $s$ be independent identically distributed Fibonacci strings according to $s \sim p(s)$. Then we will estimate $q(s_j) = |\mathrm{Re}(\bra{s_j}U_B\ket{s_j})|$ using the shot-level conjugate trick with $N$ shots, and average over these. Our estimator is thus defined as $Z = \frac{1}{m}\sum_{i=j}^m \sqrt{M_N(s_j)}$ where $M_N(s_j) = \frac{1}{N}\sum_{i = 1}^N X_i^{EC}(s_j)$ is defined above. Furthermore, we can define $B$ as follows, which is the expected bias due to the sign information of each amplitude being unavailable:
$$ B = \Big||\mathrm{Re}(\mathbb{E}[\bra{s}U_B\ket{s}])| - \mathbb{E}[q(s)]\Big|$$
Using the stronger concentration inequality of \cite{pinelis2023concentration}, we bound the conditional probability
$$ P\left(\left.\left|\sqrt{M_N(s_j)} - q(s_j)\right| \geq t~\right|s_j\right) \leq \begin{cases}
    2\e^{-\frac{N}{2}(q(s_j)^2 - (q(s_j) - t)^2)^2} & \text{if } t < q(s_j) \\
    2\e^{-\frac{N}{2}(q(s_j)^2 - (q(s_j) + t)^2)^2} & \text{if } t \geq q(s_j)
\end{cases} \leq 2\e^{-\frac{N}{2}t^4} .$$
Applying the tail-sum formula for expected values, we find:
$$ \mathbb{E}\left[\left.\left|\sqrt{M_N(s_j)} - q(s_j)\right|~\right|s_j\right] = \int_0^{\infty} P\left(\left.\left|\sqrt{M_N(s_j)} - q(s_j)\right| \geq t~\right|s_j\right) ~\mathrm{d}t \leq \int_0^{\infty} 2\e^{-\frac{N}{2}t^4} ~\mathrm{d}t = \frac{2^{\frac{5}{4}}\Gamma\left(\frac{5}{4}\right)}{N^{\frac{1}{4}}} \leq \frac{3}{N^{\frac{1}{4}}} $$
First we will show that $Z$ is a consistent estimator for $\mathbb{E}[q(s)]$:
\begin{align*}
    |\mathbb{E}[Z - \mathbb{E}[q(s)]]|
    &= \left|\mathbb{E}\left[\frac{1}{m}\sum_{j=1}^m \left(\sqrt{M_N(s_j)} - \mathbb{E}[q(s)]\right)\right]\right|\\
    &= \left|\frac{1}{m}\sum_{j=1}^m\mathbb{E}\left[\sqrt{M_N(s_1)} - \mathbb{E}[q(s)]\right]\right| = \left|\mathbb{E}\left[\sqrt{M_N(s_1)} - \mathbb{E}[q(s)]\right]\right|\\
    &=\left|\mathbb{E}\left[\sqrt{M_N(s_1)} - q(s)\right]\right| =\left|\mathbb{E}\left[\sqrt{M_N(s_1)} - q(s_1)\right]\right| \leq \mathbb{E}\left[\left|\sqrt{M_N(s_1)} - q(s_1)\right|\right] \\
    &= \mathbb{E}\left[\mathbb{E}\left[\left.\left|\sqrt{M_N(s_1)} - q(s_1)\right|~\right|s_1\right]\right] \leq \mathbb{E}\left[\frac{3}{N^\frac{1}{4}}\right] = \frac{3}{N^\frac{1}{4}}
\end{align*}
And so from the triangle inequality, we get a bound on the bias of $Z$ as an estimator for $|\mathrm{Re}(\mathbb{E}[\bra{s}U_B\ket{s}])|$:
\begin{align*}
    |\mathrm{Bias}| = |\mathbb{E}[Z - |\mathrm{Re}(\mathbb{E}[\bra{s}U_B\ket{s}])|]| &\leq \left|\mathbb{E}[Z - \mathbb{E}[q(s)]]\right| + B \leq \frac{3}{N^\frac{1}{4}} + B
\end{align*}
Thus, in the situation where $B = 0$, $Z$ is a consistent estimator, with bias that goes like $O(N^{-\frac{1}{4}})$. On the other hand, because each $M_N(s_j)$ is bounded, we would expect that the standard deviation of $Z$ is $O(m^{-\frac{1}{2}})$ from Hoeffding's inequality. This characterises the bias-variance trade-off. Now, imagine that we have $B =0$ and a fixed budget of $k = Nm$ shots, then we can pick $N = k^{\frac{2}{3}}$ and $m = k^{\frac{1}{3}}$ so that the bias and standard deviation both scale as $O(k^{-\frac{1}{6}})$. Note that this is much slower than standard Monte Carlo convergence, but we would expect that it performs much better than this in practice, since the bias of $Z$ will converge asymptotically faster than $O(N^{-\frac{1}{4}})$ whenever $\mathbb{E}_s[|\mathrm{Re}(\bra{s}U_B\ket{s})|]$ is bounded away from zero by a constant.

\section{Compilation Pipeline}

\subsection{Braid augmentation}
\label{app:braid-augmentation}

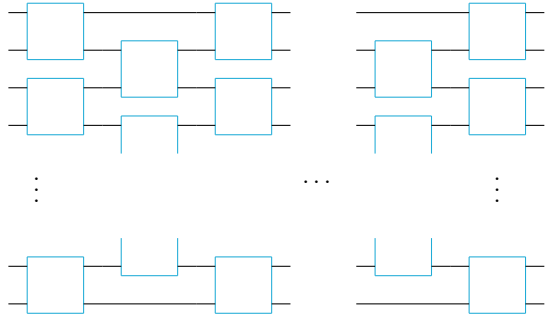
\begin{figure}[ht]
    \centering
    \input{figs/brick-wall.tikz}
    \caption{Random braid generated on a brick-wall pattern, where `bricks' are chosen from the set $\{\sigma, \sigma^{-1}, 1_2\}$ containing the braid generators and the 2-strand identity according to a strategy described in App.~\ref{app:braid-augmentation}.}
    \label{fig:random-braid}
\end{figure}

\begin{figure}[t]
    \centering
    \input{figs/markov-moves1.tikz}
    \caption{Braid rewrites that relate topologically equivalent braids upon Markov closure. From top to bottom we have the local moves: Reidemeister II or `poke' ($\sigma_{i}\sigma_{i}^{-1}=1_2$), Reidemeister III or `slide' ($\sigma_{i}\sigma_{i+1}\sigma_{i}\leftrightarrow \sigma_{i+1}\sigma_{i}\sigma_{i+1}$), `stabilisation' ($B\leftrightarrow B\sigma_{n+1}$), and the global move `cycle' or 'conjugation' ($A B \leftrightarrow B A $), where $A,B$ are arbitrary $n$-strand braids.}
    \label{fig:markov-moves}
\end{figure}
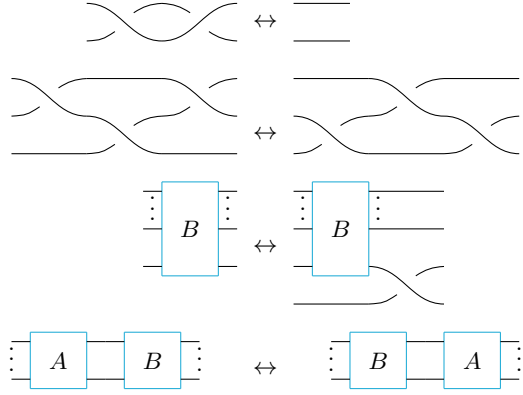

To generate braids for which the value of the Jones polynomial is known in advance, we take the following approach. Let $B = \otimes_{i=1}^k b_i$ with $3k$ strands be the tensor product of 3-strand braids $b_i$. The Jones polynomial for each $M(b_i)$ can be calculated trivially, and the Jones polynomial of $M(B)$ is given as
$$V_{M(B)}(\e^{\im\frac{2\pi}{5}}) = \phi^{k - 1}\prod_{i = 1}^k V_{M(b_i)}(\e^{\im\frac{2\pi}{5}})$$
so it can be computed efficiently classically. Then, we generate random braids $A$ and $A^{-1}$ such that $AA^{-1} = I$ is the identity braid on $3k$ strands. Let $B' = A^{-1}BA$, then since the Markov closure of a braid is unchanged by conjugation, we have $V_{M(B')}(\e^{\im\frac{2\pi}{5}}) = V_{M(B)}(\e^{\im\frac{2\pi}{5}})$ so this can also be computed efficiently. By varying the number of crossings in $A$ and the number of strands in $B$, the size of $B'$ can be controlled.

First, we describe how to generate $B$. We want the dominant error of \texttt{cfev} when applied to $B'$ to be $\epsilon_\mathrm{noise}$ (which is what we wish to characterize) rather than $\epsilon_\mathrm{shot}$, so we must ensure that the value of $|\mathbb{E}_s[\bra{s}U_B\ket{s}]|$ does not concentrate exponentially around zero. To this end, we first enumerate all 3-strand braids $b_j$ up to a given number of generators (we used four in practice), and calculated $T_j = \log |\mathbb{E}_s[\bra{s}U_{b_j}\ket{s}]|$ for each. Then, suppose that we want to construct braid $B$ on $3k$ strands. 
Each $b_i$ is sampled independently and identically from this set $\{b_j\}$ such that 
$$|\mathbb{E}_s[\bra{s}U_{B}\ket{s}]| = \prod_{i=1}^k |\mathbb{E}_s[\bra{s}U_{b_i}\ket{s}]| = \exp\left[\sum_{i=1}^k T_i\right]$$
is approximately uniformly distributed in $[0, 1]$. Note that it is not possible for this to be exactly uniformly distributed for any particular $k$ because the set $\{T_j\}$ is finite. We construct $D_k$ as 
$$P(D_k = b_i) = \frac{P(D'_k = T_i)}{|\{b_j \mid T_j = T_i\}|}$$
where $D'_k$ is a distribution over the set $\{T_j\}$. Clearly $T_i \sim D'_k$, and $b_i$ is selected uniformly at random from all $b_j$ such that $T_j = T_i$. Thus, the condition that $|\mathbb{E}_s[\bra{s}U_{B}\ket{s}]|$ be uniformly distributed depends only on $D'_k$. Let $t_1, t_2, \dots, t_n$ be the unique values of $T_j$, and define $P(D'_k = t_j) = p_j$. 

Next we use the idea of a moment generating function, which for a random variable $R$ is defined in terms of a free variable $x$ as $\mathbb{E}[\e^{xR}]$. It is defined this way because when expressed as a formal power series in $x$, the coefficient of $x^k$ is $m_k/k!$ where $m_k$ is the $k$-th raw moment of $R$. Let $Z = \log |\mathbb{E}_s[\bra{s}U_{B}\ket{s}]| = \sum_{i=1}^k T_i$ so that $|\mathbb{E}_s[\bra{s}U_{B}\ket{s}]| \sim \e^Z$. The moment generating function of $Z$ can be calculated as:
$$\mathbb{E}[\e^{xZ}] = \mathbb{E}[\e^{x\sum_{i=1}^k T_i}] = \prod_{i=1}^k\mathbb{E}[\e^{xT_i}] = \left(\sum_{j=1}^n p_j \e^{x t_j}\right)^k$$
To enforce that $|\mathbb{E}_s[\bra{s}U_{B}\ket{s}]| \sim \e^Z$ is approximately uniform, we can pick $p_j$ such that the moment generating function of $Z$ matches the desired distribution. In particular, let $U \sim \mathcal{U}([0, 1])$ then we want $\e^Z$ to match $U$, and hence $Z$ to match $\log U$. Since we have
$$\mathbb{E}[\e^{x \log(U)}] = \mathbb{E}[U^x] = \int_0^1 u^x ~\mathrm{du} = \frac{1}{x + 1}$$
then matching this for $Z$ yields:
$$\mathbb{E}[\e^{xZ}] = \left(\sum_{j=1}^n p_j \e^{x t_j}\right)^k \approx \frac{1}{x + 1} \quad \implies \quad \sum_{j=1}^n p_j \e^{x t_j} \approx \frac{1}{\sqrt[k]{x + 1}}$$
By picking a finite set of $x$ values where we want this to be satisfied (we used 100 equally spaced points in $0 \leq x \leq \frac{15}{2}$), this is transformed into a regression problem subject to the constraint that $p_j \geq 0$ and $\sum_j p_j = 1$, which can be solved by any standard constrained convex optimization method.

In theory, any method of generating random braids could be used to construct $A$, while $A^{-1}$ is given by reversing the order and sign of all the crossings in $A$. However, with this method, $A$ and $A^{-1}$ will be structurally very similar, and the overall circuit will be quite similar to mirror benchmarking techniques. Therefore, to generate a greater variety of circuit structures, and thus a more comprehensive benchmark, we use a method based on sorting networks.

First, we generate a random network $A$ of crossings on $n$ strands, not yet committing to whether each crossing is under or over. This is done in a brick-wall layout (see Fig.~\ref{fig:random-braid}) by picking randomly whether each `brick' is a crossing or the identity. The number of layers in the brick-wall can be varied to control the amount of crossings in the final braid. Regardless of the sign of each crossing, this induces the same permutation $\pi$ on the strands themselves, considering the group homomorphism $f$ from the braid group $B_n$ to the symmetric group $S_n$ that maps crossings to transpositions $f : \sigma_i^{\pm1} \to (i~~i+1)$. Now we can generate a similar network $A^{-1}$ which induces the permutation $\pi^{-1}$ via any deterministic method of decomposing a permutation into transpositions of consecutive pairs, for instance a sorting network. We used the odd-even sorting network to produce a brick-wall style network with depth at most $n$. For instance, see the following example with $n = 5$ where $A$ is generated using a $10$-layer brick-wall:
\begin{center}\input{figs/sorting-network-step1.tikz}\end{center}
Here, the small numbers on each strand indicate the induced permutation. Clearly, the concatenation of these two would induce the identity permutation as desired. Because of this structure, $A^{-1}$ will in general be much shallower than $A$ if the number of layers in the brick-wall is much larger than $n$. We need to pick a sign for each crossing so that $AA^{-1}$ is indeed the identity braid on $n$ strands. Note that not all sets of signs work, in general we might end up with some braid from the pure braid group $P_n \cong \ker f$ (that is, braids whose induced permutation is the identity). Thus, we use the following simple method which guarantees success, but could perhaps be improved in the future: we choose at random a total ordering on the strands given by $x \prec y$, and for every crossing $\sigma_i^{\pm 1}$ between strands $i$ and $i + 1$ we pick $\sigma_i^{+1}$ if $\pi'(i + 1) \prec \pi'(i)$ and $\sigma_i^{-1}$ if $\pi'(i) \prec \pi'(i + 1)$, where $\pi'$ is the permutation induced by all the crossings up to this point. Continuing the example above, using the ordering $4 \prec 3 \prec 5 \prec 1 \prec 2$, we obtain
\begin{center}\input{figs/sorting-network-step2.tikz}\end{center}
where the numbers after each layer indicate the induced permutation up to that point $\pi'$. To see that this construction will always generate $A$ and $A^{-1}$ such that $AA^{-1}$ is the identity braid, we can use induction on the strands. Starting from the top-most strand, since it lies either fully atop or beneath every other strand, a series of `slide' and `poke' moves (cf. Fig.~\ref{fig:markov-moves}) can reduce the braid so that no generators act on this strand. This can be repeated for all strands in sequence until the identity is reached.

Lastly, we comment briefly on why this method was chosen - we previously considered two options: either constructing $ABA^{-1}$ as above, but using a family of $n$-strand braids with a known Jones polynomial, in particular the torus knots, or taking a small braid for which the Jones polynomial of the closure is known or can be computed classically, and augmenting it with random Markov moves to form a large braid. In both cases, the value of the Jones polynomial does not scale exponentially with the number of strands, and hence the expectation we would need to measure in \texttt{cfev} would decrease as $O(\phi^{-n})$. Thus $\epsilon_\mathrm{shot}$ would quickly become the dominant error, which is unhelpful for this benchmark as it is intended to measure $\epsilon_\mathrm{noise}$.

\subsection{Braid simplification}
\label{app:braid-simplification}

Given a specific knot or link that you might wish to evaluate with $\texttt{cfev}$ it is useful to obtain a braid whose closure is the given link, called a \emph{braid representative}, which minimizes the number of crossings so as to minimize both the time required to run the algorithm and the $\epsilon_\mathrm{noise}$ error incurred. A braid representative can be efficiently obtained via Vogel's algorithm~\cite{Vogel1990}, but this may have many crossings. Note that identifying a small braid representative is hard in general. For example minimizing the number of strands is hard~\cite{Lee2010}, minimizing the number of crossings is NP-hard~\cite{paterson1991set}, and when minimizing both strands and crossings this is a complete knot invariant and hence even harder~\cite{gittings2004minimum}. Therefore, in our pipeline~\cite{repo} we include a heuristic solution to this problem based on applying peephole optimization to a given braid. Using Berger's algorithm for minimizing 3-strand braids~\cite{berger1994minimum}, we rewrite each 3-strand sub-braid and apply random `slide' and conjugation moves until no further improvement is obtained. This is fast in practice but may lead to larger numbers of crossings than other slower heuristic methods. Anecdotally, we observed that our method is $\sim 20\%$ worse than the physics-based heuristic proposed in Ref.~\cite{bangert2001search} for uniformly random $10$-strand braid words.

\subsection{Braid generator circuits}
\label{app:compilation-params}

For $U_{\sigma}^1$, the parameters for the native 1- and 2-qubit gates in Fig.~\ref{fig:braid-unitary-and-compiled-circuit}(b) are as follows
\begin{align*}
    \alpha_0 &= a_1 & \beta_0 &= 0 & \chi_0 &= -b_1 & \theta_0 &= -\frac{\pi}{10} \\
    \alpha_1 &= -a_1 & \beta_1 &= -b_1 & \chi_1 &= \frac{\pi}{2} & \theta_1 &= 2b_1 + \frac{3\pi}{10} \\
    \alpha_2 &= -a_1 & \beta_2 &= -b_1 - \frac{3\pi}{10} & \chi_2 &= -b_1 & \theta_2 &= \frac{\pi}{5} \\
    \alpha_3 &= a_1 & \beta_3 &= -2b_1 - \frac{3\pi}{10}
\end{align*}
where the angles $a_1$ and $b_1$ are defined by
$$ a_1 = \sin^{-1}\left(\sqrt{\frac{4}{3\phi}}\right) - \pi \qquad \qquad b_1 = \sin^{-1}\left(\phi\sqrt{\frac{3}{8}}\right)$$
and the eigenvalue corresponding to $\ket{000}$ is $\e^{\im \alpha} = \e^{\im \frac{2\pi}{5}}$. Similarly for $U_{\sigma}^2$ we have
\begin{align*}
    \alpha_0 &= \pi - a_2 & \beta_0 &= 0 & \chi_0 &= -b_2 & \theta_0 &= \frac{\pi}{5} \\
    \alpha_1 &= a_2 & \beta_1 &= -b_2 & \chi_1 &= \frac{3\pi}{5} & \theta_1 &= 2b_2 + \frac{3\pi}{5} \\
    \alpha_2 &= -a_2 & \beta_2 &= \frac{2\pi}{5} - b_2 & \chi_2 &= -b_2 & \theta_2 &= -\frac{2\pi}{5} \\
    \alpha_3 &= \pi - a_2 & \beta_3 &= -2b_2 - \frac{3\pi}{5}
\end{align*}
where the angles $a_2$ and $b_2$ are defined by
$$ a_2 = \sin^{-1}\left(\frac{2}{\sqrt{3 + 2\phi}}\right) \qquad \qquad b_2 = \sin^{-1}\left(\frac{\sqrt{3\phi - 1}}{2}\right) - \pi$$
and the eigenvalue corresponding to $\ket{000}$ is $\e^{\im \alpha} = \e^{-\im \frac{2\pi}{5}}$. The parameters for $U_\sigma^k$ with $k \geq 3$ have more complicated closed forms, so we calculated them numerically instead, and the values are listed in an online repository along with the pipeline code~\cite{repo}. Note that these values are not the only possible expression of these unitaries in the form given in Fig.~\ref{fig:braid-unitary-and-compiled-circuit}(b).

\section{Evaluating the Jones for Plat-closed braids}
\label{app:jones-plat}

\begin{theorem} For any braid $B$ with $n - 1$ strands, where $n - 1$ is even,
    $$V_{P(B)}(\e^{\im\frac{2\pi}{5}}) = (-\e^{-\im\frac{3\pi}{5}})^{3w_B}\phi^{\frac{n}{2} - \frac{3}{2}}\bra{\alpha}U_B\ket{\alpha}$$
    where $\ket{\alpha} = \ket{01010\cdots 10}$, $w_B$ is the writhe of $B$ and $P(B)$ is the link corresponding to the Plat closure of $B$.
\end{theorem}

\begin{proof}
    We adapt the proof of~\cite[Theorem 3.2]{aharonov2006polynomial} that uses the path model representation of $B_n$ to the Fibonacci string representation of~\cite{KAUFFMAN_2008} which we use in Sec.~\ref{sec:problem}. For any braid $B$, the following two link diagrams are ambient isotopic and hence have the same Jones polynomial:
    $$\input{figs/plat-proof-1.tikz}$$
    Note that the left diagram corresponds to the $P(B)$ while the right corresponds to $M(B \cdot E)$ where $B \cdot E$ is a tangle, and $E$ is an element of the Temperley-Lieb algebra:
    $$\input{figs/plat-proof-2.tikz}$$
    Therefore, we have that $V_{P(B)}(x) = V_{M(B \cdot E)}(x)$. In Sec.~\ref{sec:problem} we give an expression for $V_{M(B')}(\e^{\im\frac{2\pi}{5}})$ for any braid $B'$ in terms of a unitary representation $U_{B'}$, which was originally proven by Kauffman and Lomonaco~\cite[Theorem 2]{KAUFFMAN_2008} (this is equivalent to~\cite[Lemma 2.2]{aharonov2006polynomial} in the path model representation). They show that this also holds when $B'$ is any tangle (i.e. it contains Temperley-Lieb generators), although $U_{B'}$ may be non-unitary, and give the representation in this case. Specifically, if $B$ has $n-1$ strands, we have that
    $$U_{B\cdot E} = U_B U_E \qquad\qquad U_E = \prod_{i = 0}^{\frac{n-1}{2} - 1} I_{2i} \otimes M \otimes I_{n-1-2(i+1)} \qquad\qquad M = \begin{pmatrix} 
        1 &   &      &   &   &   &   &   \\ 
          & 1 &      &   &   &   &   &   \\ 
          &   & \phi &   &   &   &   &   \\ 
          &   &      & 0 &   &   &   &   \\ 
          &   &      &   & 1 &   &   &   \\ 
          &   &      &   &   & a &   & b \\ 
          &   &      &   &   &   & 0 &  \\ 
          &   &      &   &   & b &   & 1 \\ 
    \end{pmatrix}$$
    where $a = \frac{1}{\phi}$, $b = \sqrt{1 - \frac{1}{\phi^2}}$, and $I_{k}$ is the identity matrix of dimension $2^k \times 2^k$ (note that $M = \e^{\im\frac{3\pi}{5}}U_{\sigma_i} + \e^{\im\frac{\pi}{5}}I_3$).

    We will now show that $U_E\ket{s} = 0$ for all $\ket{s} \in \mathcal{F}$ except for $\ket{\alpha}$. Suppose that $U_E\ket{s} \neq 0$. Since the terms $I_{2i} \otimes M \otimes I_{n-1-2(i+1)}$ and $I_{2j} \otimes M \otimes I_{n-1-2(j+1)}$ commute for any $i$ and $j$, we must have that $(I_{2i} \otimes M \otimes I_{n-1-2(i+1)})\ket{s} \neq 0$ for all $0 \leq i < \frac{n-1}{2}$. If $s_{2i} = 0$, then $s_{2i+1} = 1$ by the definition of $\mathcal{F}$. In order for $(I_{2i} \otimes M \otimes I_{n-1-2(i+1)})\ket{s} \neq 0$, we must have $s_{2(i+1)} = 0$ since
    \begin{align*}
        (I_{2i} \otimes M \otimes I_{n-1-2(i+1)})\ket{s} &= \ket{s_0\dots s_{2i-1}} \otimes M\ket{s_{2i}s_{2i+1}s_{2(i+1)}}\otimes \ket{s_{2i+3}\dots s_{n-2}} \\
        &= \ket{s_0\dots s_{2i-1}} \otimes M\ket{01s_{2(i+1)}}\otimes \ket{s_{2i+3}\dots s_{n-2}}
    \end{align*}
    but $M\ket{011} = 0$ and $M\ket{010} = \phi\ket{010} \neq 0$. From the definition of $\mathcal{F}$, we have $s_0 = 0$ and $s_1 = 1$, so by induction we must have $s_{2i} = 0$ and $s_{2i+1} = 1$ for all $0 \leq i < \frac{n-1}{2}$, and thus $\ket{s} = \ket{\alpha}$. In this case, $(I_{2i} \otimes M \otimes I_{n-1-2(i+1)})\ket{\alpha} = \phi\ket{\alpha}$, so $U_E\ket{\alpha} = \phi^{\frac{n-1}{2}}\ket{\alpha}$. Therefore, we have that:
    \begin{align*}
        V_{P(B)}(\e^{\im\frac{2\pi}{5}}) &= V_{M(B\cdot E)}(\e^{\im\frac{2\pi}{5}}) = (-\e^{-\im\frac{3\pi}{5}})^{3w_B}\phi^{n-2}\mathbb{E}_s[\bra{s}U_{B\cdot E}\ket{s}] \\
        &= (-\e^{-\im\frac{3\pi}{5}})^{3w_B}\phi^{n-2}\sum_{s \in \mathcal{F}} \frac{\phi^{s_{n-1}}}{\phi^{n-1}}\bra{s}U_{B\cdot E}\ket{s} = (-\e^{-\im\frac{3\pi}{5}})^{3w_B}\phi^{-1}\sum_{s \in \mathcal{F}} \phi^{s_{n-1}}\bra{s}U_{B}U_E\ket{s} \\
        &= (-\e^{-\im\frac{3\pi}{5}})^{3w_B}\phi^{\frac{n}{2} - \frac{3}{2}}\bra{\alpha}U_{B}\ket{\alpha} \qedhere
    \end{align*}
\end{proof}

The quantum protocol presented in Sec.~\ref{sec:problem} can be used to solve the BQP-complete problem of additively approximating the Jones polynomial of Plat-closed braids at the fifth root of unity.
For every circuit execution, we fix $\ket{s}=\ket{01010\cdots 10}$ and estimate the desired quantity as $V_{P(B)}(\e^{\im\frac{2\pi}{5}}) = (-\e^{-\im\frac{3\pi}{5}})^{3w_B}\phi^{\frac{n}{2} - \frac{3}{2}}\bra{\alpha}U_B\ket{\alpha}$.

We will briefly comment on why this definition appears different to that in the Aharonov-Jones-Landau (AJL) algorithm. In the AJL algorithm, their computation of the Jones polynomial for the Plat closure has the following form \cite[Theorem 3.2]{aharonov2006polynomial} (where $n$ in their notation is $n - 1$ in our notation, since we count qubits rather than strands):
$$ V_{P(B)}(\e^{\im\frac{2\pi}{k}}) = (-A)^{3w_B}\frac{d^{\frac{3}{2}(n - 1) - 1}\sin\left(\frac{\pi}{k}\right)}{N}\bra{\alpha}U_B\ket{\alpha} $$
where $A$ is a phase depending on $k$ and $N$ is an exponentially large factor defined by
$$ N = \sum_{l = 1}^{k-1} \sin\left(\frac{l\pi}{k}\right) \dim \mathcal{H}_{n-1,k,l}$$
and $\dim \mathcal{H}_{n-1,k,l}$ is the number of walks of length $n - 1$ on the path graph of $k - 1$ vertices that start on the left-most vertex $l=1$ and end on vertex $l$. Note that their definition of $\ket{\alpha}$ and $U_B$ are slightly different than ours, although it is not important in this case. Substituting $k = 5$, we have
\begin{align*}
    d &= \phi & A &= \e^{-\im\frac{3\pi}{5}} &  \dim \mathcal{H}_{n-1,5,l} &= \begin{cases}
        f_{n - 2} & \text{if } l = 1 \text{ and } n \text{ is odd or } l = 4 \text{ and } n  \text{ is even}  \\
        f_{n - 1} & \text{if } l = 2 \text{ and } n \text{ is even or } l = 3 \text{ and } n  \text{ is odd} \\
        0 & \text{otherwise}
    \end{cases}
\end{align*}
so since $\sin(x)$ is symmetric around $\frac{\pi}{2}$ and $n - 1$ is either even or odd we can compute $N$ as
$$ \frac{N}{\sin\left(\frac{\pi}{5}\right)} = f_{n-2} + \frac{\sin\left(\frac{2\pi}{5}\right)}{\sin\left(\frac{\pi}{5}\right)}f_{n - 1} = f_{n-2} + \phi f_{n - 1} = \phi^{n - 1} $$
and thus we get
$$ V_{P(B)}(\e^{\im\frac{2\pi}{5}}) = (-A)^{3w_B}\phi^{\frac{n}{2} - \frac{3}{2}}\bra{\alpha}U_B\ket{\alpha} $$
which matches the theorem above, excepting the differences in $U_B$ and $\ket{\alpha}$.

\section{Numerical experiment details}
\label{app:simulations}

\subsection{Noisy circuit simulation}

\begin{figure}
    \includegraphics[width=0.5\textwidth]{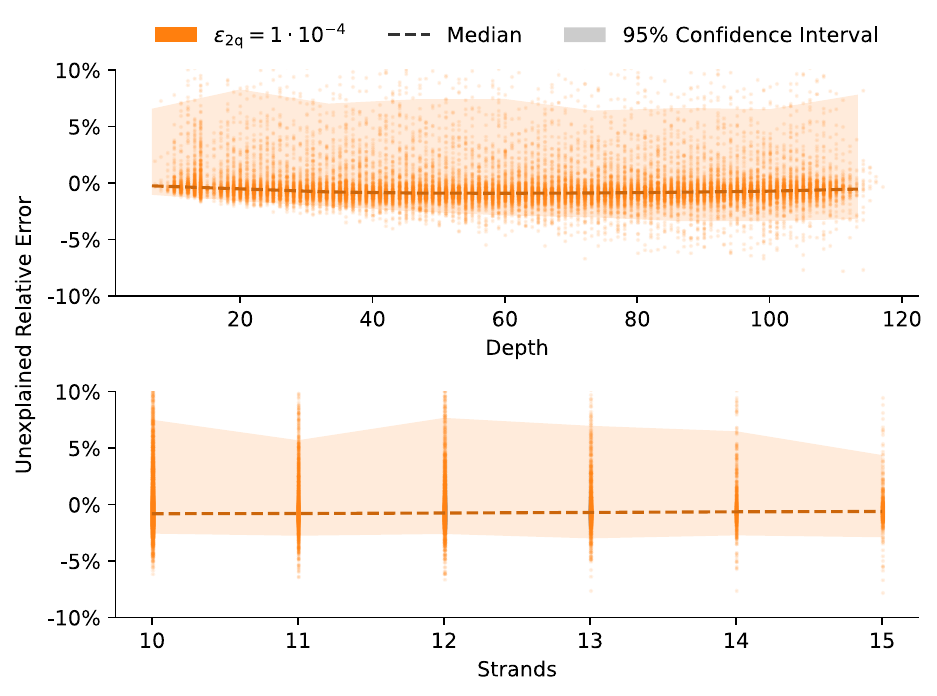}
    \caption{The dependence of the relative error in Fig.~\ref{fig:noise-characterisation} unexplained by the power-law fit for a set of 24k randomly genererated benchmark braids, when simulated with a depolarizing noise model and $\epsilon_\mathrm{2q} = 1\cdot 10^{-4}$. While the unexplained relative error can be relatively large ($5-10\%$), the median is low ($< 1\%$), and it is not correlated with either the depth of the circuits or the number of strands in each braid. This suggests that the total relative error of \texttt{cfev} depends mainly on the number of crossings.}
    \label{fig:unexplained-error}
\end{figure}

We simulated \texttt{cfev} as compiled to the native gateset of Quantinuum's H-series devices using a state-vector method. Random Pauli errors were inserted after every $U_\mathrm{1q}$ $1$- and $R_{ZZ}$ $2$-qubit gate to create a depolarizing channel. No errors were inserted for $R_Z$ gates, since in practice these are implemented entirely in software. Random bit-flip initialization errors were applied on each qubit at the beginning of each circuit, and asymmetric bit-flip errors were applied to the measurement results. See the repository \cite{repo} for the exact error rates and parameters. Non-Fibonacci error detection was applied to the output of these simulations, but the conjugate trick was not used (since no coherent errors were included, it is unnecessary). 

The simulations were run with a custom state-vector simulator built using the Triton GPU programming environment~\cite{tillet2019triton}. This was necessary to avoid the performance penalty incurred in existing state-vector simulators (for example, we saw up to 50x slow down relative to Qiskit-Aer~\cite{qiskit2024} on the same hardware) when simulating circuits which must vary between shots (this is due to the bitstring $\ket{s}$ being randomly sampled, and hence $V_\mathrm{cat}$ must be recompiled for each shot). A Qiskit-Aer-based simulation is also included in the repository, which we cross-checked against our custom implementation to ensure correctness. 

The benchmark set was generated using the technique from Sec.~\ref{sec:benchmark}, specifically according to the method in App.~\ref{app:braid-augmentation}. The number of layers in the initial brick-wall circuit was varied from $5$ to $100$ in steps of $5$. The number of strands varied from $10$ to $15$ with: $10$k braids on $10$ strands, $5$k each for $11$ and $12$ strands, $2.4$k for $13$ strands, $1.2$k for $14$ strands, and $640$ braids on $15$ strands. This was done to ensure that the time spent on each number of strands was roughly the same.

In Fig.~\ref{fig:noise-characterisation} we see how the relative error depends on the number of crossings. We use a power-law fit to extrapolate because it provided a better fit than a linear or exponential ($1 - \alpha^c$) regression. Fig.~\ref{fig:unexplained-error} shows that the relative error unexplained by the power-law fit does not depend strongly on the depth of the braid or the number of strands. Therefore, we did not believe it necessary to expand the benchmark set to include braids with more strands (which would have taken much more time to simulate).

\subsection{Tensor network methods}
\label{app:simulations-tns}

\begin{figure}
    \includegraphics[width=\textwidth]{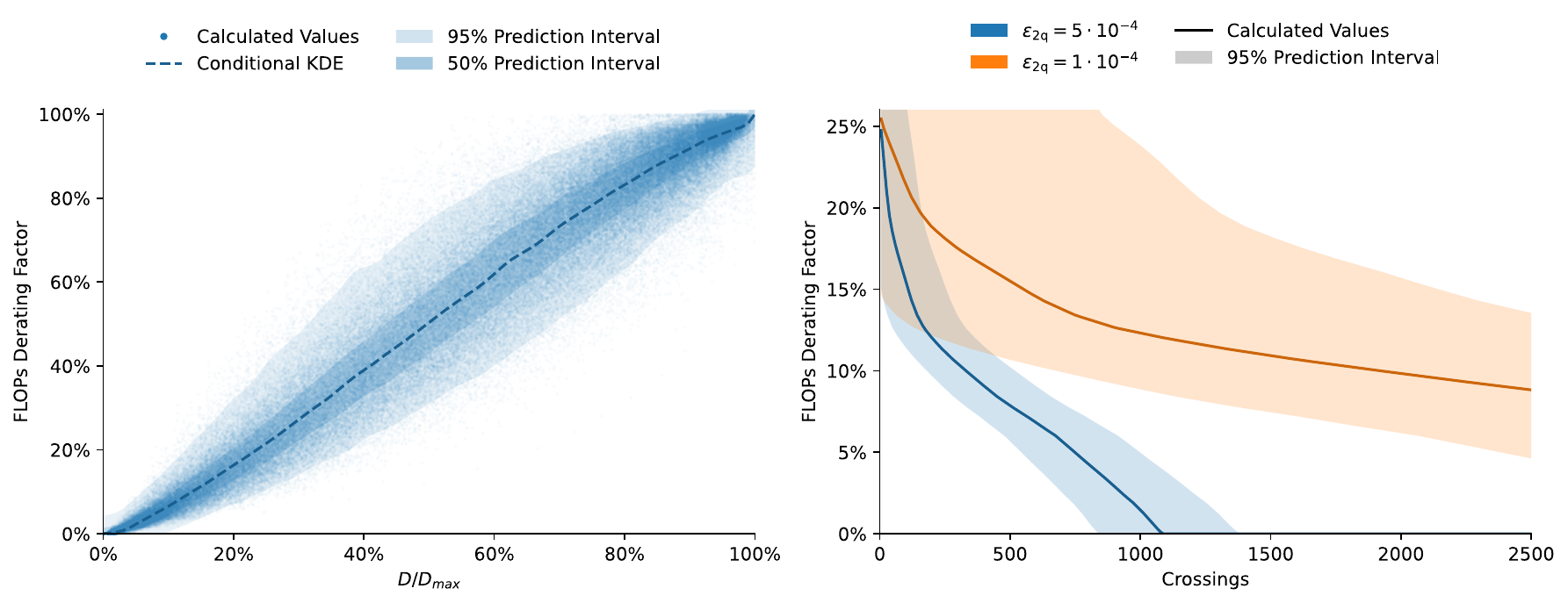}
    \caption{\emph{On the left:} The floating-point operations (FLOPS) required to evaluate \texttt{mpo-proj} relative to the total number of operations when the bond dimension is maximum $\chi = \chi_\mathrm{max}$ so that no truncation occurs, as it depends on the level of truncation. The level of truncation is expressed as the amount of intermediate memory used relative to the $\chi = \chi_\mathrm{max}$ case. It was determined for a set of 10k braids by measuring the total FLOPS and intermediate memory usage for $\chi = 2^{k}$ with $k \in [3, 10]$. \emph{On the right:} the amount of FLOPS needed for \texttt{mpo-proj} to achieve the same accuracy as \texttt{cfev}, as it depends on the number of crossings, relative to the number of FLOPS required for the $\chi = \chi_\mathrm{max}$ case. This is obtained by combining the left-hand plot with Fig.~\ref{fig:chi-and-mpo-crossings}(a).}
    \label{fig:flops-derating}
\end{figure}

\begin{figure}
    \includegraphics[width=0.5\textwidth]{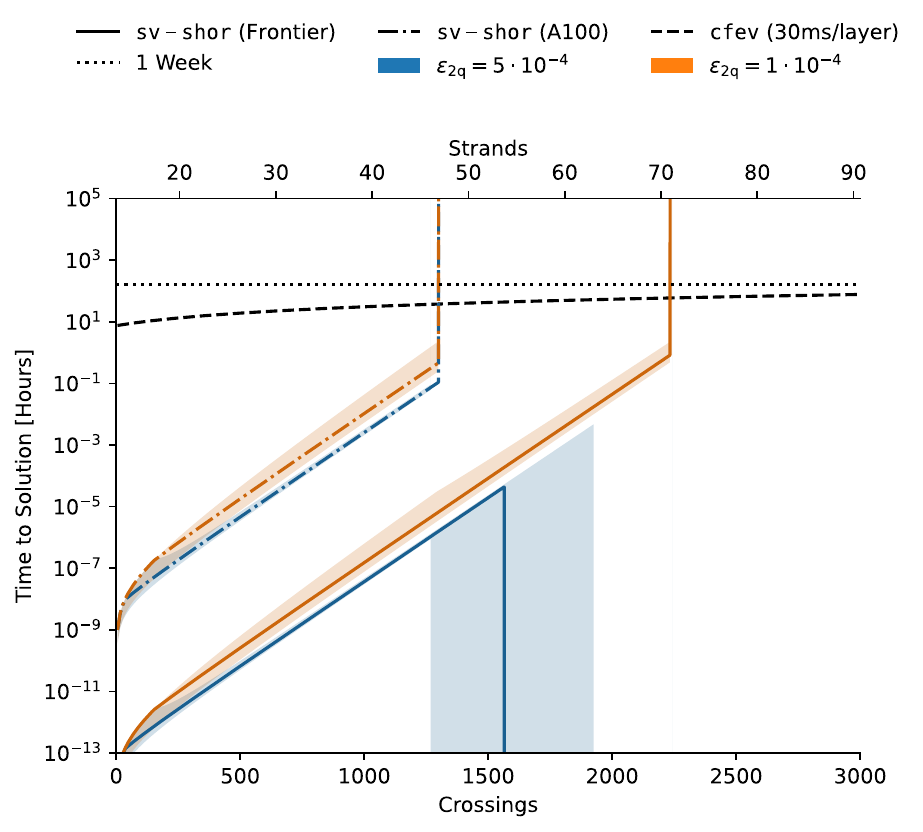}
    \caption{Predicted time needed to calculate $V_{M(B)}(\e^{\im\frac{2\pi}{5}})$ using the \texttt{sv-shor} method, on the same benchmark set as in Fig.~\ref{fig:chi-and-mpo-crossings}, to the same relative error achieved by \texttt{cfev}, as it depends on the number of crossings $c$ and the $2$-qubit error rate $\epsilon_\mathrm{2q}$. The solid lines represent computation on the Frontier cluster, while the dot-dashed lines represent the same computation on a single A100 GPU. The lines vertically downward represent the points where the relative error of \texttt{cfev} exceeds the 80\% threshold given by Fig.~\ref{fig:chi-and-mpo-crossings}, whereas the lines vertically upward represent the points where required memory exceeds the available memory of the classical hardware. For $\epsilon_\mathrm{2q} = 5 \cdot 10^{-4}$, the relative error of \texttt{cfev} grows large enough that \texttt{sv-shor} is always faster than \texttt{cfev}. For $\epsilon_\mathrm{2q} = 1 \cdot 10^{-4}$ the memory limit is reached before the point where \texttt{cfev} would become faster.}
    \label{fig:time-sv}
\end{figure}

\begin{figure}
    \includegraphics[width=\textwidth]{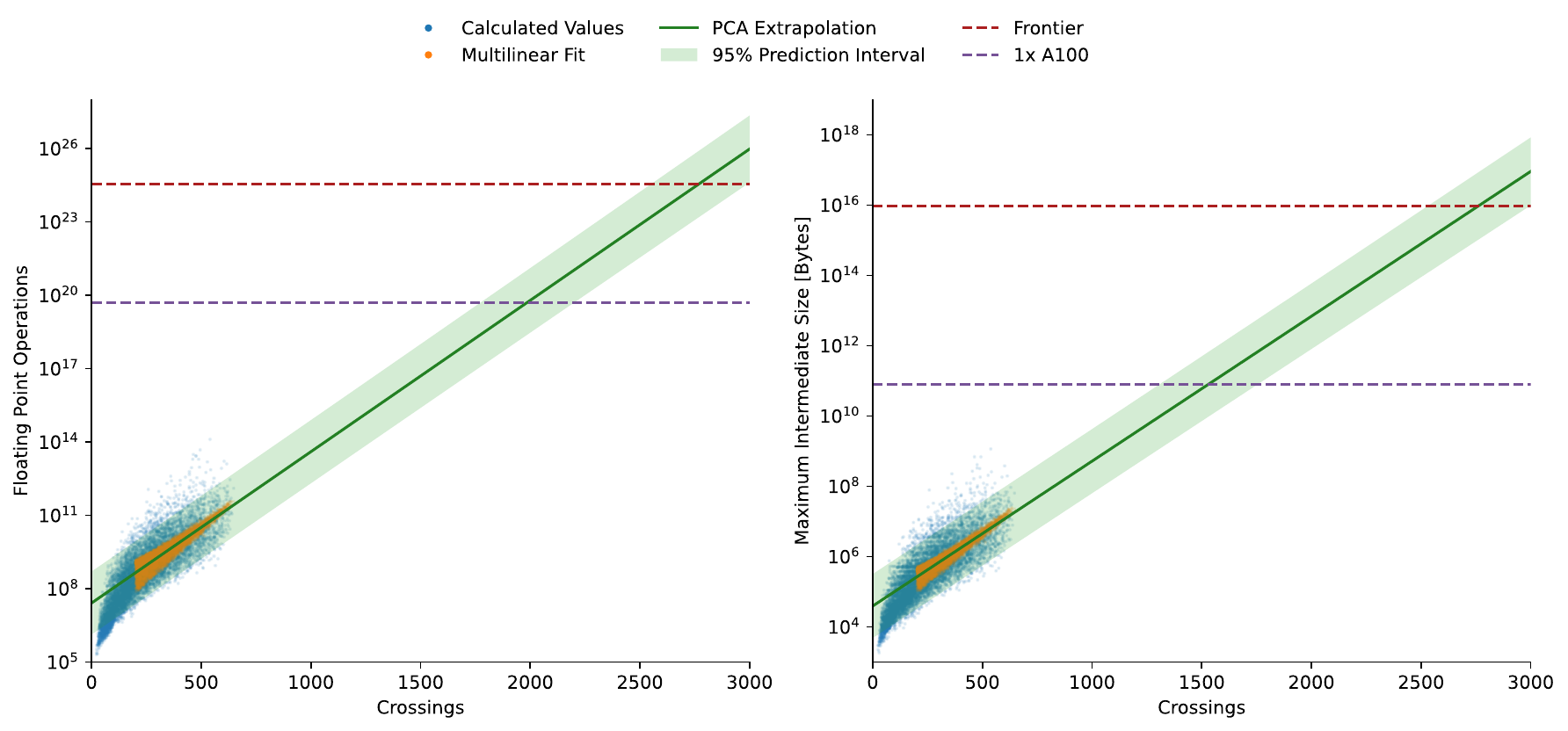}
    \caption{The resources required to evaluate $V_{M(B)}(\e^{\im\frac{2\pi}{5}})$ using the \texttt{mpo-proj} method for the benchmark set as in Fig.~\ref{fig:chi-and-mpo-crossings}. The bond dimension $\chi = \chi_\mathrm{max}$ was allowed to grow as large as needed to avoid any lossy compression. The extrapolation is a multi-linear fit against the number of strands $n - 1$, crossings $c$, and depth of the braid for $c \geq 200$. It is extrapolated along the principal component of the dataset to represent possible performance for larger datasets generated in the same way. \emph{On the left:} The number of floating point operations (FLOPS) required. The dashed lines labelled `Frontier' and `1x A100' represent one month of computation on the respective hardware, for scale. \emph{On the right:} The size of the largest intermediate tensor, in bytes. The dashed lines represent the memory limits of the respective hardware.}
    \label{fig:mem-mpo}
\end{figure}

\begin{figure}
    \includegraphics[width=\textwidth]{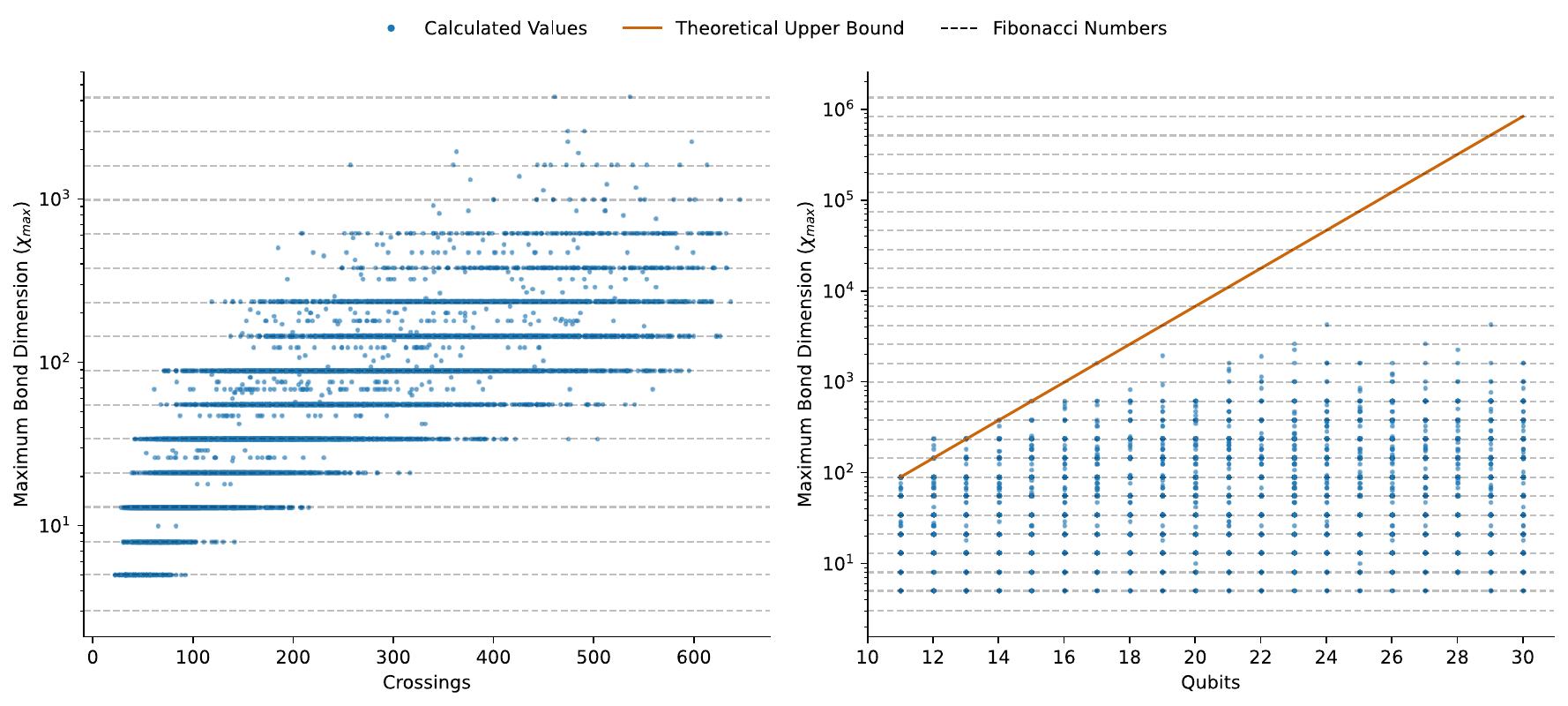}
    \caption{The maximum bond dimension $\chi_\mathrm{max}$ required to evaluate $V_{M(B)}(\e^{\im\frac{2\pi}{5}})$ without any compression using the \texttt{mpo-proj} method for the benchmark set as in Fig.~\ref{fig:chi-and-mpo-crossings}, as it depends on the number of qubits and the number of crossings. For almost all (96\%) of instances, $\chi_\mathrm{max}$ is equal to a Fibonacci number. This suggests that \texttt{mpo-proj} is successful at reducing the bond dimension by exploiting that $U_B$ acts non-trivially only on Fibonacci strings. Furthermore, $\chi_\mathrm{max}$ is usually much lower than the theoretical upper bound of $f_n$ for $n$ qubits, indicating that the complexity of braids in the benchmark set is likely limited by the number of crossings rather than the number of strands.}
    \label{fig:maxchi-fib}
\end{figure}

\begin{figure}
    \includegraphics[width=\textwidth]{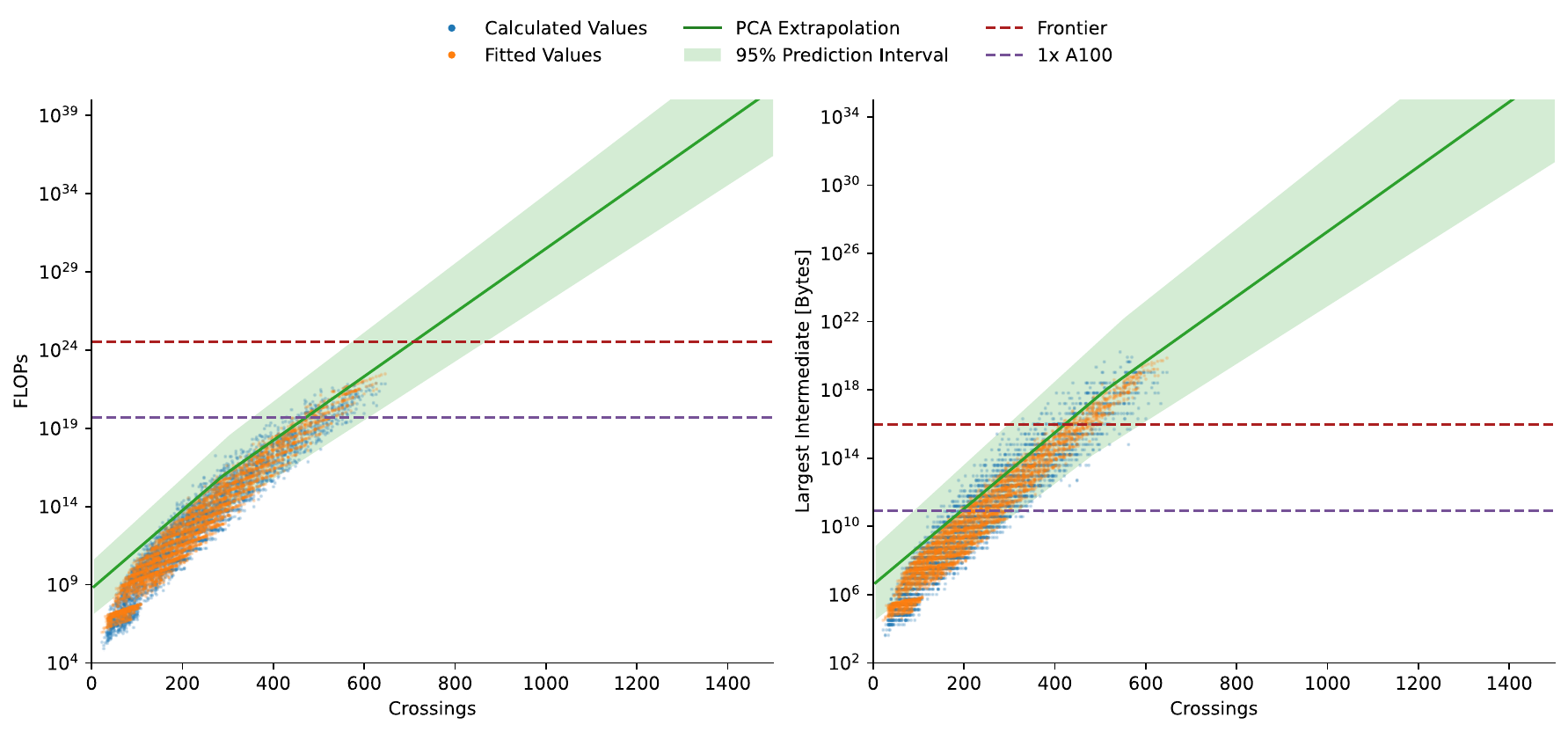}
    \caption{The resources required to evaluate $V_{M(B)}(\e^{\im\frac{2\pi}{5}})$ using the \texttt{tn-proj} method for the benchmark set as in Fig.~\ref{fig:chi-and-mpo-crossings}. The extrapolation is a multi-linear fit against the number of strands $n - 1$, crossings $c$, depth of the braid $d$, and $\min(n, d)$ (see App.~\ref{app:simulations-tns} for more details). It is extrapolated along the principal component of the dataset to represent possible performance for larger datasets generated in the same way. The abrupt change in slope around $c \sim 500$ is due to the $\min(n, d)$ term and results in a more conservative estimate. \emph{On the left:} The number of floating point operations (FLOPS) required. The dashed lines labelled `Frontier' and `1x A100' represent one month of computation on the respective hardware, for scale. \emph{On the right:} The size of the largest intermediate tensor, in bytes. The dashed lines represent the memory limits of the respective hardware.}
    \label{fig:mem-tn}
\end{figure}

The benchmark set for \texttt{tn-proj}, \texttt{mpo-proj}, and \texttt{bracket} was comprised of $10$k braids, generated as in App.~\ref{app:braid-augmentation} with initial brick-wall depths between $5$ and $50$ in increments of $5$ and between $10$ and $30$ strands. The braids were evenly distributed across all of these statistics. Since the time required to estimate the performance of \texttt{tn-proj}, \texttt{bracket},  and \texttt{mpo-proj} does not have a strong of a dependence on the number of strands (compared to the noisy simulations in the previous section), we elected to use a greater range of strands to have a more comprehensive benchmark.

We study the performance of \texttt{tn-proj} and \texttt{mpo-proj} by computing the weighted trace shown in Fig.~\ref{fig:tn-proj}. The description of the tensor network, including both the braid unitary and the projection matrices, is built up iteratively. The iteration is over the sequence of braid generators $U_{\sigma^+}$ and $U_{\sigma^-}$, here we do not use the compiled circuit representation of the braid generators, Fig.~\ref{fig:braid-unitary-and-compiled-circuit}(b), instead we use the exact 3-qubit matrix, Eq.~\eqref{eq:unitary}, where we set the action on the non-Fibonacci subspace to be a zero matrix. Once the braid is complete we similarly iterate over the projectors that apply the correct weightings in the trace. The inclusion of the projectors and the choice of the action of the braid generators on the non-Fibonacci subspace make the tensor network non-unitary. 

We can then evaluate this in two ways. For \texttt{tn-proj}, we use Cotengra~\cite{gray2021hyper} to find good contraction paths. We use the KaHyPar and random-greedy backends to Cotengra, configured to minimize total FLOPS. We used $128$ hyper-optimization samples for each network. We also tested $1024$ hyper-optimization samples on a subset of networks but did not observe any substantial improvement. We did not perform slicing of the network, but even if memory limitations are ignored, Fig.~\ref{fig:mpo-comparison-and-time-to-soln} indicates that this would not substantially change the boundaries for quantum advantage. We did not perform the actual contractions, but instead assumed that the calculated required FLOPS could be perfectly parallelized over our hardware targets. This is a conservative assumption that will bias the results in favour of classical algorithms. The results are shown in Fig.~\ref{fig:mem-tn}.

All matrix product operator (MPO) computations were carried out using the Quimb library \cite{gray2018quimb}. At each step of the tensor network construction, we apply compression to reduce the maximum bond dimension to a fixed value $\chi$. Once the compressed MPO is constructed in this way, we compute its trace and multiply by the prefactors to complete the estimate of the Jones polynomial for the Markov closure of the braid. From this, we can calculate the relative error by comparing against the known value of the Jones polynomial for the benchmark braids. We performed this for each braid for values of $\chi = 2^k$ with $k \in [3, 9]$. For each braid and for each $k$, we recorded the maximum intermediate size, and used this information to generate Fig.~\ref{fig:chi-and-mpo-crossings}(a). 

We found that $k = 13$ was sufficiently large that no compression needed to be applied for all braids in the benchmark set. The resources required in this case are given in Fig.~\ref{fig:mem-mpo}. Thus for $k = 13$ for each braid, we recorded the dimensions of each bond after the application of each unitary to the MPO. The largest achieved bond dimension defines $\chi_\mathrm{max} \leq 2^{13}$ for that braid, and the largest total intermediate defines $D_\mathrm{max}$. Then using the theoretical FLOP counts listed in~\cite{blackford1999lapack}, we could calculate the theoretical number of FLOPS needed to evaluate the MPO for any intermediate $\chi$ value by replaying the sequence of operations performed by Quimb in the MPO computation, but using truncated bond dimensions to calculate the FLOP count for each operation. With this data, Fig.~\ref{fig:noise-characterisation}, and Fig.~\ref{fig:chi-and-mpo-crossings}, we constructed Fig.~\ref{fig:flops-derating}, from which the time-to-solution of \texttt{mpo-proj} can be estimated. Again, we assumed that the calculated required FLOPS could be perfectly parallelized.

For \texttt{mpo-proj} we use a multilinear fit in Fig.~\ref{fig:mem-mpo} to extrapolate to larger numbers of crossings. Specifically, we take into account the number of crossings, strands, and depth of the braid. We fit to only the points with $c \geq 200$ because before this point the trend has not yet appeared to stabilize. We extrapolate along the principal axis of the dataset in terms of crossings, strands, and depth since this represents plausible larger instances generated in the same way as the original dataset.

For \texttt{tn-proj} in Fig.~\ref{fig:mem-tn} (and \texttt{bracket} discussed in the next section), we use a multilinear fit that includes crossings $c$, strands $n - 1$, depth $d$, and also the term $\min(n, d)$. The reasoning behind this is that rectangular planar tensor networks can be contracted in time that scales exponentially proportional to their shortest linear dimension (for instance, width or height). We would expect the networks generated by brick-wall style links to be roughly of this shape. More generally, planar tensor networks can be contracted in time that scales exponentially in the square root of the number of tensors \cite{kourtis2019fast}. This provides a more conservative scaling estimate than a multilinear fit of only $c$, $n$, $d$.

\subsection{Classical knot-theoretic algorithms}
\label{app:classicalalgos}

\begin{figure}
    \includegraphics[width=\textwidth]{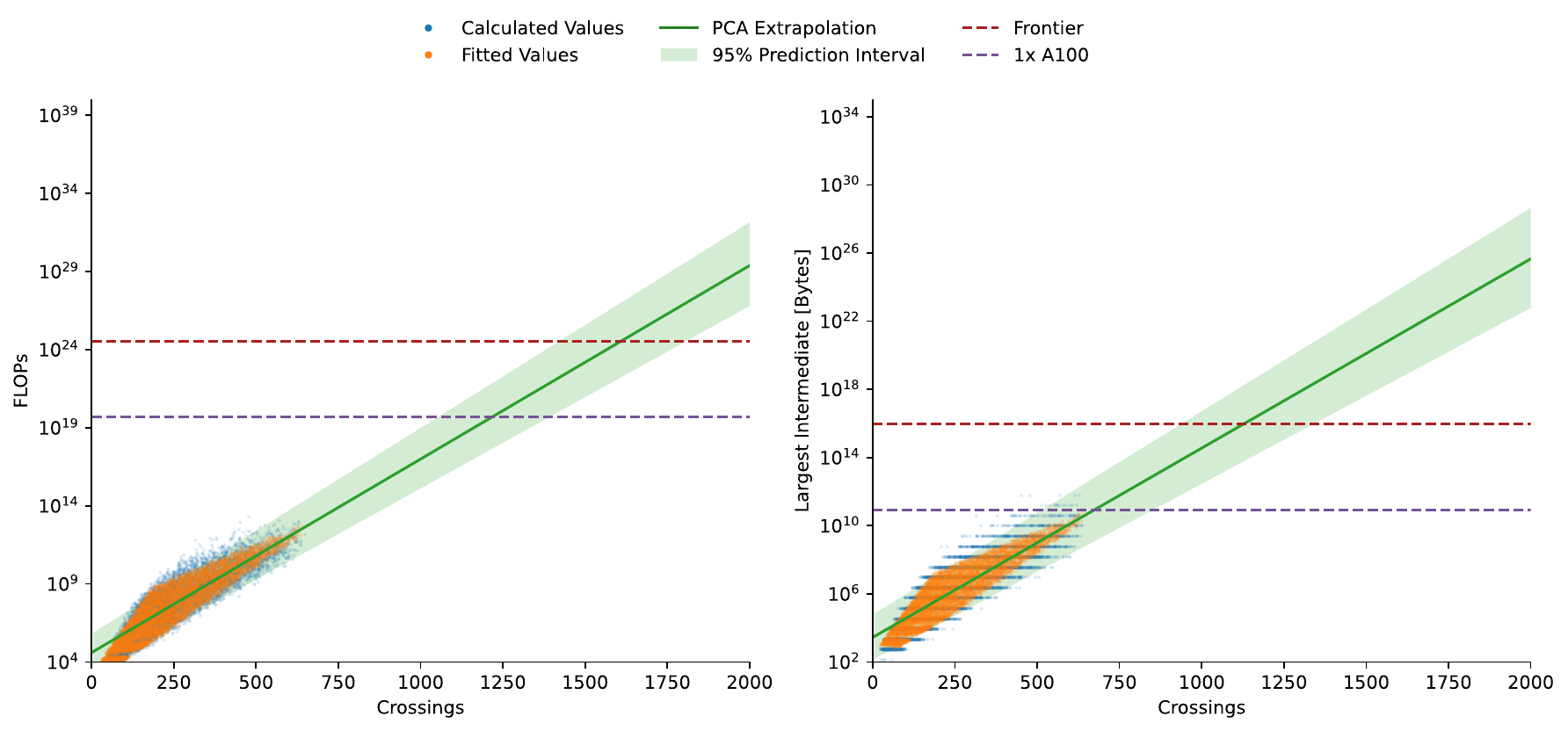}
    \caption{The resources required to evaluate $V_{M(B)}(\e^{\im\frac{2\pi}{5}})$ using the \texttt{bracket} method for the benchmark set as in Fig.~\ref{fig:chi-and-mpo-crossings}. This results from the optimization of the tensor network contraction path, making the conservative estimate that each tensor element takes 8 bytes of memory and each corresponding arithmetic operation costs only one floating-point operation (FLOP), see App.~\ref{app:classicalalgos} for more details. The extrapolation is a multi-linear fit against the number of strands $n - 1$, crossings $c$, depth of the braid $d$, and $\min(n, d)$ (see App.~\ref{app:simulations-tns} for more details). It is extrapolated along the principal component of the dataset to represent possible performance for larger datasets generated in the same way. \emph{On the left:} The number of floating point operations (FLOPS) required. The dashed lines labelled `Frontier' and `1x A100' represent one month of computation on the respective hardware, for scale. \emph{On the right:} The size of the largest intermediate tensor, in bytes. The dashed lines represent the memory limits of the respective hardware.}
    \label{fig:mem-bracket}
\end{figure}

There are various different classical algorithms for computing the Jones polynomial, or evaluating it at a certain point. Since doing this exactly is \#P-hard~\cite{jaeger1990computational}, we would expect these methods to scale exponentially. Indeed, they can be divided broadly into two paradigms based on whether this scaling is in terms of strands or crossings.

The first group of algorithms perform computations with a $2^{n - 1}$ dimensional vector (of Laurent polynomials or complex numbers, depending on the method). These can be derived for example through the connection to Hecke algebras~\cite{morton1990calculating} or quantum groups~\cite{tingley2015minus}, and has been implemented in practice~\cite{morton2014programs}. Such methods are essentially identical to quantum algorithms for evaluating the Jones polynomial, and in particular \texttt{sv-shor}, as it performs a statevector simulation of each amplitude on the diagonal of $U_B$ directly as a polynomial. Since they are not specialized to the $\e^{\im\frac{2\pi}{5}}$ evaluation point their cost scales as $O(2^n)$ in memory and $O(c4^n)$ in time.

The second group of algorithms expand the Jones polynomial as an exponential sum of terms using a skein relation (an equation that relates the original link to the links obtained by resolving a single crossing). For example, this includes direct expansion~\cite{kauffman1987state} (with complexity $O(2^c)$), methods based on `simple walks'~\cite{hajij2018efficient,boden2022braid}, or the connection to the Tutte polynomial of graphs~\cite{kyoko1995computing}. These computations can be optimized using dynamic programming, in a similar way to finding a good contraction path for a tensor network. Examples include \cite{ellenberg2013efficient} (with a stated complexity of at most $O(c2^{17\sqrt{c}})$), Bar-Natan's algorithm for the Khovanov homology~\cite{barnatan2006fast}, and an algorithm due to Kauffman~\cite[p. 104-124]{Kauffman2001} which we will call \texttt{bracket}. These methods are all abstractly similar to \texttt{tn-proj} in the same way that, taking another \#P-hard problem as an example, model counting relates to tensor contraction~\cite{kourtis2019fast}. Indeed, \texttt{tn-proj} scales exponentially in the width parameters (such as cut-width and tree-width) of the planar graph representing the quantum circuit~\cite{ogorman2019parameterization}, while~\cite{ellenberg2013efficient,barnatan2006fast} and \texttt{bracket} all scale exponentially in the width parameters of the planar graph representing the link (where each crossing is one vertex, and edges correspond to strands). This cannot be worse than the $O(2^n)$ scaling of the strand-based algorithms, so we choose to focus on these methods exclusively.

It is not a coincidence that both groups resemble different methods of simulating the corresponding quantum circuit; \cite[Section 5.5]{kuperberg2015how} gives a heuristic argument that this may be the optimal classical algorithm for large knots, by showing that the Jones polynomial provides a general model of planar quantum circuits -- the exact evaluation of the Jones polynomial corresponds to the exact simulation of post-selected quantum circuits, in the same way that the APPROX-JONES-PLAT problem corresponds to the simulation of quantum circuits up to bounded error. Moreover,~\cite{mann2017complexity} shows that for random braids, $U_B$ looks like a random quantum circuit. The theoretical scaling predicted for these algorithms is extremely poor, even making generous assumptions about the constant factors involved (e.g for $c = 2000$ and $n = 80$, $c2^{17\sqrt{c}} \approx 10^{232}$, $c4^{n} \approx 10^{51}$), so this would not be a suitable point of comparison for \texttt{cfev}. 

However, in the case of \texttt{bracket}, the method is given exactly by the contraction of a (non-unitary) tensor network over the ring of Laurent polynomials (that is, the elements of the tensors are Laurent polynomials). Thus, using Cotengra~\cite{gray2021hyper} to find good contraction paths for this network, and making the conservative assumption that each Laurent polynomial operation can be computed in only one floating-point operation, we can find a practical upper bound for the performance of \texttt{bracket}. We used the same contraction finding hyper-parameters as specified for \texttt{tn-proj} in the previous section. Since no large-scale implementations of any of these algorithms are available (e.g~\cite{hajij2018efficient} has $c \leq 12$, \cite{morton2014programs} has $n \leq 13$, ~\cite{barnatan2006fast} has $c \sim 50$), we will take this upper bound as representative of the best possible performance for all classical algorithms. Fig.~\ref{fig:mem-bracket} shows the upper bound calculated in this way for the benchmark braid set described in the previous section.

\subsection{Limitations}\label{app:limitations}

\begin{figure}
    \includegraphics[width=0.84\textwidth]{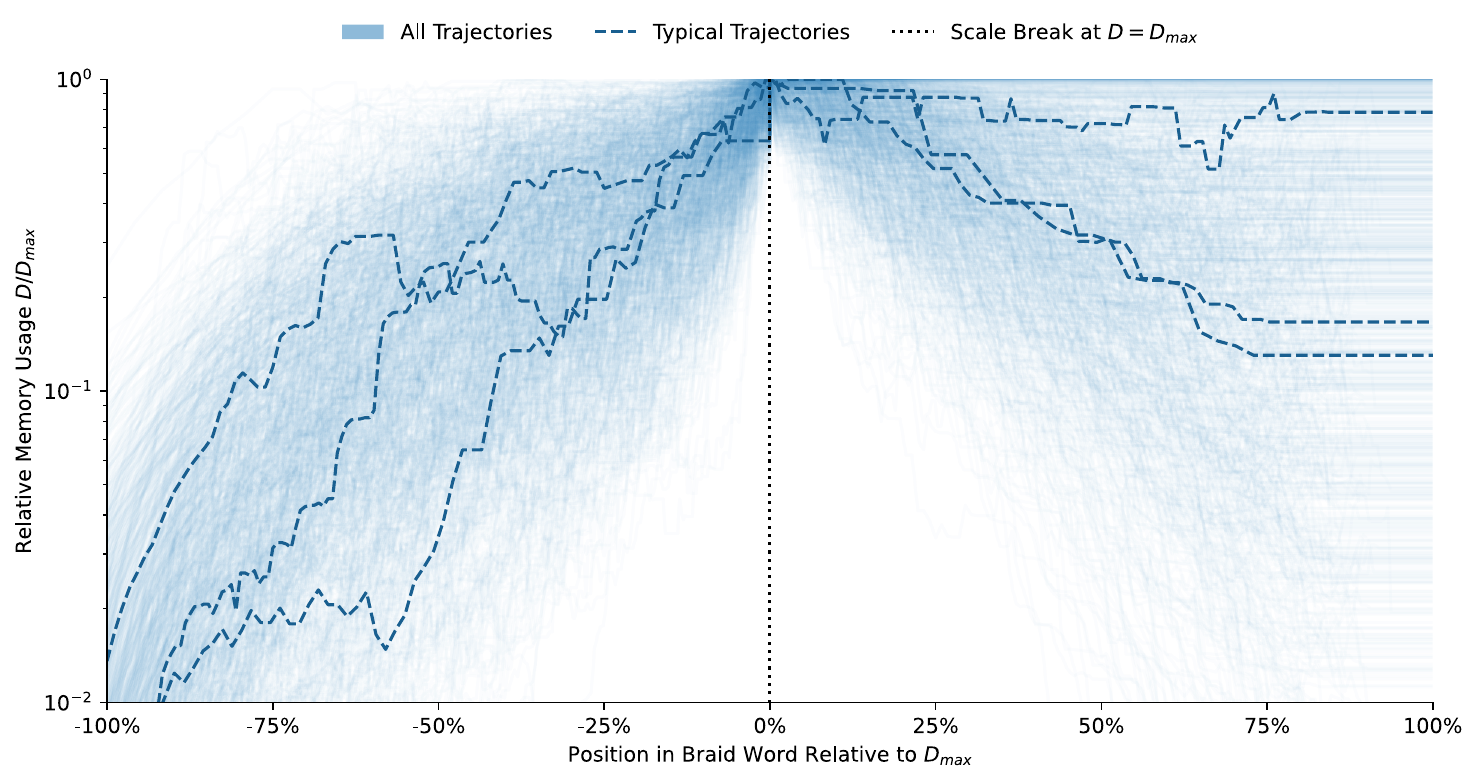}
    \caption{The memory usage required to evaluate $V_{M(B)}(\e^{\im\frac{2\pi}{5}})$ without any compression using the \texttt{mpo-proj} method for each braid in the benchmark set (as in Fig.~\ref{fig:chi-and-mpo-crossings}) as it evolves throughout the computation. We recorded the memory usage after each braid generator is applied, and normalize it relative to the maximum memory usage. The relative position in the braid word is measured in two different scales: -100\% to 0\% corresponds to all generators before the point $D = D_\mathrm{max}$ of maximum memory usage, and 0\% to 100\% is the same for all generators after that point. Three typical trajectories are shown with dashed lines. The memory usage increases roughly linearly (on a log-scale) up to the maximum before decreasing in the same way. This suggests that \texttt{mpo-proj} may be able to `see' the conjugate $ABA^{-1}$ structure of the braids in the benchmark set (see App.~\ref{app:braid-augmentation} for details), and hence it may perform differently for braids with a different structure.}
    \label{fig:maxchi-traj}
\end{figure}

There are several limitations of these numerical studies which are worth mentioning. Most importantly, the results in Fig.~\ref{fig:mem-bracket}, Fig.~\ref{fig:noise-characterisation} and especially Fig.~\ref{fig:chi-and-mpo-crossings} require significant extrapolation outside the range of data we were able to compute in order to estimate where advantage might occur. In the absence of a strong mathematical reason why linear and power-law fits ought to hold in general for these quantities, we cannot say with certainty where the threshold for time-to-solution or energy-efficiency advantage for \texttt{cfev} over classical algorithms will occur, or even if it necessarily must occur for the $\epsilon_\mathrm{2q}$ we are considering. The lack of further data in the case of \texttt{mpo-proj} is partially due to the fact that larger braids become difficult to benchmark on a single GPU due to VRAM constraints. This is promising as it suggests that this problem is indeed difficult to solve classically, but also means we must be less certain about our conclusions. Secondly, in using benchmark braids constructed according to the method in App.~\ref{app:braid-augmentation}, we cannot necessarily draw strong conclusions about the behaviour of braids that do not conform to this format (see Fig.~\ref{fig:maxchi-traj}). In the future, when a candidate braid for which advantage might be demonstrated is identified, it may be necessary to perform the same benchmarks with braids that more closely match this candidate.
 
\section{Generalization to other roots of unity}
\label{app:non-fibo-sketch}

Our compilation of Shor and Jordan's quantum algorithm \cite{shor2008estimating} has the appealing property that the unitaries corresponding to each crossing are \emph{local}, acting on only three qubits, and can thus be implemented in a hardware-efficient manner. On the other hand, the Aharonov-Jones-Landau (AJL) algorithm \cite{aharonov2006polynomial} can evaluate the Jones polynomial at any root of unity $\e^{\im\frac{2\pi}{k}}$, but the unitary for each crossing acts on \emph{all} qubits (in the worst case), so it is unlikely to be realizable on NISQ machines. Therefore, it would be advantageous to generalize Shor and Jordan's method so that roots of unity other than the fifth can be evaluated while preserving locality of crossings. Here we sketch such a generalization, with the understanding that all the error mitigation, error detection, and compilation tricks described in this paper can likely be generalized to this case also. Our generalization is based on the connection between the two algorithms for the $k = 5$ case described in Appendix C of \cite{shor2008estimating}.

The AJL algorithm utilizes the \emph{path representation} of the Temperley-Lieb algebra. In this model, given a braid of $n$ strands evaluated at the $\e^{\im\frac{2\pi}{k}}$ root of unity with $k \geq 3$, the Hilbert space that is acted on by the crossings consists of all walks of length $n$ on the path graph of $k - 1$ vertices that start at the left-most vertex. These paths are encoded as a list of transitions $\pm 1$, where $+1$ represents a step to the right, and $-1$ represents a step to the left. For example, in the $k = 5$ case, all the paths are generated as trajectories in the following finite state machine:
\begin{center}\input{figs/fsm-1.tikz}\end{center}
The inherent non-locality in the AJL algorithm comes from the fact that the action of each crossing depends on position of the path at a given point, and determining this position requires inspecting all transitions up to that point. To obtain Shor and Jordan's representation, first note that the action of the crossings on each path (given by Equation 4 and Definition 2.5 in \cite{aharonov2006polynomial}) is symmetric about the center of the graph; that is, paths that land on vertex $i$ or vertex $k - i - 2$ are treated exactly the same. Therefore, we can incorporate this symmetry into the finite state machine that generates these paths, by `folding' it over the line of symmetry:
\begin{center}\input{figs/fsm-2.tikz}\end{center}
This can also be interpreted as `reflecting' the path over the line of symmetry (cf. Appendix C of \cite{shor2008estimating}):
\begin{center}\input{figs/fsm-3.tikz}\end{center}
Then Shor and Jordan's representation is obtained by instead recording the \emph{positions} of the path at each step, rather than the transitions between steps, and labelling each vertex in the state machine rather than each edge. This corresponds to the trajectories of the following finite state machine, which clearly generates only Fibonacci strings:
\begin{center}\input{figs/fsm-4.tikz}\end{center}
In the example above, this maps the path $\{+1, +1, -1, +1, +1\}$ to $\{0, 1, 1, 1, 1, 0\}$. Recording the positions rather than the transitions enables crossings to be implemented by local gates, because the position at each step does not need to be calculated from the previous steps in the path. Therefore, to generalize this to $k > 5$, we can simply apply the same `folding' procedure to the corresponding path graph, and label each vertex. Note that for $k > 5$ the labels will require more than one qubit, since there are more than two vertices (a single qudit could also be used). For example for $k = 6, 7, 8$:
\begin{center}\input{figs/fsm-5.tikz}\end{center}
Then an $n$-strand braid will require $(n + 1)\lceil\log_2(\lfloor\frac{k}{2}\rfloor)\rceil$ qubits to encode. Indeed this gives the $n + 1$ we see in Shor and Jordan's algorithm when $k = 5$, and also shows that $k = 2, 3, 4$ (but not $k = 6$) must be trivial to compute since the folded state machines only admit a single trajectory. This is asymptotically more qubits than the $n + O(\log_2(k))$ used in the AJL algorithm, but the savings due to the locality (in terms of gates) may be worth the trade-off. 

Substituting these encodings into \cite[Equation 4]{aharonov2006polynomial} then allows one to derive the equivalent of our Eq.~\eqref{eq:unitary} for $k > 5$. Concretely, consider all positions on the left half of the original path graph (by symmetry, this is sufficient), and all combinations of the following two $\pm 1$ moves. This corresponds to a sequence of three positions in the reflected state machine, and the action of the Temperley-Lieb generators on this bit string can be read off from \cite[Equation 4]{aharonov2006polynomial}. Substituting this into \cite[Definition 2.5]{aharonov2006polynomial} then gives the action of a crossing on this bit string. For example for $k = 5$, we have:
$$\arraycolsep=5pt
    \begin{array}{cccc}
    \toprule
    &             & \multicolumn{2}{c}{\text{Positions}} \\ 
    \multirow{4}{*}[-3em]{\rotatebox{90}{Moves}} &             & \multicolumn{1}{c}{1}     & \multicolumn{1}{c}{2}    \\ \midrule
    & \{-1,-1\} &  - &   - \\ \cmidrule{3-4}
    & \{-1,+1\} &  - & {\scriptstyle  \begin{aligned} &U_\sigma\ket{101} = \phi^{-1}\e^{\im\frac{4\pi}{5}}\ket{101} +\phi^{-\frac{1}{2}}\e^{-\im\frac{3\pi}{5}}\ket{111}\end{aligned}} \\ \cmidrule{3-4}
    & \{+1,-1\} & {\scriptstyle  \begin{aligned} &U_\sigma\ket{010} = \e^{-\im\frac{4\pi}{5}}\ket{010} \end{aligned}} & {\scriptstyle  \begin{aligned} &U_\sigma \ket{111} = -\phi^{-1}\ket{111} + \phi^{-\frac{1}{2}}\e^{-\im\frac{3\pi}{5}}\ket{101}\end{aligned}} \\ \cmidrule{3-4}
    & \{+1,+1\} & {\scriptstyle  \begin{aligned} &U_\sigma\ket{011} = \e^{\im \frac{3\pi}{5}}\ket{011}\end{aligned}} & {\scriptstyle  \begin{aligned} &U_\sigma\ket{110} = \e^{\im \frac{3\pi}{5}}\ket{110}\end{aligned}}\\ \bottomrule
    \end{array}
$$
Which recovers Eq.~\eqref{eq:unitary} exactly. We also include the case of $k = 6$:
$$\arraycolsep=0pt
    \begin{array}{c@{\hskip 10pt}c@{\hskip 5pt}c@{\hskip 5pt}c@{\hskip 5pt}c}
    \toprule
    &             & \multicolumn{3}{c}{\text{Positions}} \\ 
    \multirow{4}{*}[-6.5em]{\rotatebox{90}{Moves}} &             & \multicolumn{1}{c}{1}     & \multicolumn{1}{c}{2}     & \multicolumn{1}{c}{3}   \\ \midrule
    & \{-1,-1\} &  - &  - & {\scriptstyle  \begin{aligned} &U_\sigma\ket{100100} = \e^{\im \frac{19\pi}{12}}\ket{100100}\end{aligned}}  \\ \cmidrule{3-5}
    & \{-1,+1\} & -  & {\scriptstyle  \begin{aligned} U_\sigma\ket{010001} &= \frac{\e^{\im\frac{\pi}{4}}}{\sqrt{3}}\ket{010001} \\&+ \sqrt{\frac{2}{3}}\e^{\im\frac{5\pi}{12}}\ket{011001}\end{aligned}} & {\scriptstyle  \begin{aligned} &U_\sigma\ket{100110} = \e^{\im\frac{\pi}{4}}\ket{100110} \end{aligned}} \\ \cmidrule{3-5}
    & \{+1,-1\} & {\scriptstyle  \begin{aligned} &U_\sigma\ket{000100} = \e^{\im\frac{\pi}{4}}\ket{000100} \end{aligned}} & {\scriptstyle  \begin{aligned} U_\sigma \ket{011001} &= \frac{\e^{\im\frac{\pi}{12}}}{\sqrt{3}}\ket{011001} \\&+ \sqrt{\frac{2}{3}}\e^{\im\frac{5\pi}{12}}\ket{010001}\end{aligned}} & {\scriptstyle  \begin{aligned} &U_\sigma\ket{100110} = \e^{\im\frac{\pi}{4}}\ket{100110} \end{aligned}}  \\ \cmidrule{3-5}
    & \{+1,+1\} & {\scriptstyle  \begin{aligned} &U_\sigma\ket{000110} = \e^{\im \frac{19\pi}{12}}\ket{000110}\end{aligned}} & {\scriptstyle  \begin{aligned} &U_\sigma\ket{011001} = \e^{\im \frac{19\pi}{12}}\ket{011001}\end{aligned}} & {\scriptstyle  \begin{aligned} &U_\sigma\ket{100100} = \e^{\im \frac{19\pi}{12}}\ket{100100}\end{aligned}} \\ \bottomrule
    \end{array}
$$
For higher $k$ values, the corresponding matrix elements of the unitary representations can be derived similarly.

Critically, whenever $k \geq 4$, the string $00\dots0$ is not a valid trajectory of the state machine (in the same way that it is not a Fibonacci basis state), so it is not constrained by the above equations. Thus we can arrange for it to be an eigenstate of the crossing unitary, allowing \texttt{cfev} to be implemented essentially unchanged. Additionally, in the same way that we can efficiently sample from $\mathcal{F}_n$ classically to generate $\ket{s}$ bitstrings for \texttt{cfev}, AJL provides an algorithm \cite[Claim 3.4]{aharonov2006polynomial} to sample paths uniformly for $k > 5$, which can then be `reflected' as above to be compatible with our representation.

\bibliographystyle{eptcs}
\bibliography{refs}

\end{document}

%% file: quantum.tikzstyles
\tikzstyle{gate}=[shape=rectangle, text height=1.5ex, text depth=0.25ex, yshift=0.5mm, fill=white, draw=black, minimum height=3mm, yshift=-0.5mm, minimum width=3mm, font={\small}, tikzit category=circuit, inner sep=2pt]
\tikzstyle{big gate}=[shape=rectangle, text height=1.5ex, text depth=0.25ex, yshift=0.5mm, fill=white, draw=black, minimum height=10mm, yshift=-0.5mm, minimum width=5mm, font={\small}, tikzit category=circuit]
\tikzstyle{Z dot}=[inner sep=0mm, minimum size=2mm, shape=circle, draw=black, fill={rgb,255: red,221; green,255; blue,221}, tikzit category=zx]
\tikzstyle{Z phase dot}=[minimum size=5mm, font={\footnotesize\boldmath}, shape=rectangle, rounded corners=2mm, inner sep=0.2mm, outer sep=-2mm, scale=0.8, tikzit shape=circle, draw=black, fill={rgb,255: red,221; green,255; blue,221}, tikzit draw=blue, tikzit category=zx]
\tikzstyle{X dot}=[Z dot, shape=circle, draw=black, fill={rgb,255: red,255; green,136; blue,136}, tikzit category=zx]
\tikzstyle{X phase dot}=[Z phase dot, tikzit shape=circle, tikzit draw=blue, fill={rgb,255: red,255; green,136; blue,136}, font={\footnotesize\boldmath}, tikzit category=zx]
\tikzstyle{hadamard}=[fill=yellow, draw=black, shape=rectangle, inner sep=0.6mm, minimum height=1.5mm, minimum width=1.5mm, tikzit category=zx]
\tikzstyle{paulibox}=[fill={rgb,255: red,221; green,221; blue,255}, draw=black, shape=rectangle, inner sep=0.6mm, minimum height=5mm, minimum width=5mm, font={\footnotesize}, text height=1.5ex, text depth=0.25ex, tikzit category=zx]
\tikzstyle{vertex}=[inner sep=0mm, minimum size=1mm, shape=circle, draw=black, fill=black, tikzit category=misc]
\tikzstyle{vertex set}=[inner sep=0mm, minimum size=1mm, shape=circle, draw=black, fill=white, font={\footnotesize\boldmath}, tikzit category=misc]
\tikzstyle{small black dot}=[fill=black, draw=black, shape=circle, inner sep=0pt, minimum width=1.2mm, tikzit category=circuit]
\tikzstyle{cnot ctrl}=[fill=black, draw=black, shape=circle, inner sep=0pt, minimum width=1.2mm, tikzit category=circuit]
\tikzstyle{cnot targ}=[fill=white, draw=white, shape=circle, tikzit category=circuit, label={center:$\oplus$}, inner sep=0pt, minimum width=2.1mm, tikzit fill={rgb,255: red,102; green,204; blue,255}, tikzit draw=black]
\tikzstyle{ket}=[fill=white, draw=black, shape=regular polygon, regular polygon sides=3, regular polygon rotate=-30, scale=0.7, inner sep=1pt, tikzit category=circuit, tikzit shape=rectangle, tikzit fill=green]
\tikzstyle{bra}=[fill=white, draw=black, shape=regular polygon, regular polygon sides=3, regular polygon rotate=30, scale=0.7, inner sep=1pt, tikzit category=circuit, tikzit shape=rectangle, tikzit fill=red]
\tikzstyle{scalar}=[shape=rectangle, text height=1.5ex, text depth=0.25ex, yshift=0.5mm, fill=white, draw=black, minimum height=5mm, yshift=-0.5mm, minimum width=5mm, font={\small}]
\tikzstyle{clabel}=[fill=white, draw=none, shape=rectangle, tikzit fill={rgb,255: red,56; green,255; blue,242}, font={\footnotesize}, inner sep=1pt, tikzit category=labels]
\tikzstyle{empty diagram}=[draw={gray!40!white}, dashed, shape=rectangle, minimum width=1cm, minimum height=1cm, tikzit category=misc]
\tikzstyle{amap}=[fill=white, draw=black, shape=NEbox, tikzit category=asymmetric, tikzit fill=yellow, tikzit shape=rectangle]
\tikzstyle{amap conj}=[fill=white, draw=black, shape=NWbox, tikzit category=asymmetric, tikzit fill=green, tikzit shape=rectangle]
\tikzstyle{amap adj}=[fill=white, draw=black, shape=SEbox, tikzit category=asymmetric, tikzit fill=red, tikzit shape=rectangle]
\tikzstyle{amap trans}=[fill=white, draw=black, shape=SWbox, tikzit category=asymmetric, tikzit fill=orange, tikzit shape=rectangle]
\tikzstyle{astate}=[fill=white, draw=black, shape=NEtriangle, tikzit category=asymmetric, tikzit shape=circle, tikzit fill=yellow]
\tikzstyle{astate conj}=[fill=white, draw=black, shape=NWtriangle, tikzit category=asymmetric, tikzit shape=circle, tikzit fill=green]
\tikzstyle{astate adj}=[fill=white, draw=black, shape=SEtriangle, tikzit category=asymmetric, tikzit shape=circle, tikzit fill=red]
\tikzstyle{astate trans}=[fill=white, draw=black, shape=SWtriangle, tikzit category=asymmetric, tikzit shape=circle, tikzit fill=orange]
\tikzstyle{white dot}=[inner sep=0mm, minimum size=2mm, shape=circle, draw=black, fill={rgb,255: red,250; green,250; blue,250}]
\tikzstyle{white phase dot}=[minimum size=5mm, font={\footnotesize\boldmath}, shape=rectangle, rounded corners=2mm, inner sep=0.2mm, outer sep=-2mm, scale=0.8, tikzit shape=circle, draw=black, fill={rgb,255: red,250; green,250; blue,250}, tikzit draw=blue]
\tikzstyle{hbox}=[shape=rectangle, text height=2mm, fill={rgb,255: red,255; green,235; blue,61}, draw=black, minimum height=2mm, minimum width=2mm, font={\small}, tikzit category=zh, inner sep=0pt, rounded corners=0.5mm]
\tikzstyle{Z dot (zh)}=[inner sep=0mm, minimum size=2mm, shape=circle, draw=black, fill={rgb,255: red,250; green,250; blue,250}, tikzit category=zh]
\tikzstyle{X dot (zh)}=[Z dot, shape=circle, draw=black, fill={rgb,255: red,193; green,193; blue,193}, tikzit category=zh]
\tikzstyle{triangle}=[fill={rgb,255: red,255; green,136; blue,136}, draw=black, shape=isosceles triangle, isosceles triangle apex angle=60, minimum size=2.5mm, inner sep=0mm]
\tikzstyle{labelled hbox}=[shape=rectangle, text height=1.75ex, text depth=0.5ex, fill={rgb,255: red,255; green,235; blue,61}, draw=black, minimum height=3mm, minimum width=4mm, font={\small}, tikzit category=zh, inner sep=1.3pt, rounded corners=0.5mm]
\tikzstyle{Z phase dot (zh)}=[Z phase dot, tikzit shape=circle, tikzit draw=blue, fill={rgb,255: red,250; green,250; blue,250}, font={\footnotesize\boldmath}, tikzit category=zh]
\tikzstyle{X phase dot (zh)}=[Z phase dot, tikzit shape=circle, tikzit draw=blue, fill={rgb,255: red,193; green,193; blue,193}, font={\footnotesize\boldmath}, tikzit category=zh]
\tikzstyle{W node}=[fill=black, draw=black, shape=regular polygon, regular polygon sides=3, minimum size=2mm]
\tikzstyle{Z dot (zw)}=[fill=white, draw=black, shape=circle, minimum width=1.2mm, inner sep=0pt]
\tikzstyle{Z phase dot XL}=[Z phase dot, fill={rgb,255: red,250; green,250; blue,250}, draw=black, shape=circle, tikzit draw={rgb,255: red,191; green,0; blue,64}, tikzit shape=circle, font={\large\boldmath}, inner sep=0.0mm]

\tikzstyle{hadamard edge}=[-, dashed, dash pattern=on 2pt off 0.5pt, thick, draw={rgb,255: red,68; green,136; blue,255}]
\tikzstyle{box edge}=[-, dashed, dash pattern=on 2pt off 0.5pt, thick, draw={rgb,255: red,203; green,192; blue,225}]
\tikzstyle{brace edge}=[-, tikzit draw=blue, decorate, decoration={brace,amplitude=1mm,raise=-1mm}]
\tikzstyle{diredge}=[->, thick]
\tikzstyle{double edge}=[-, double, shorten <=-1mm, shorten >=-1mm, double distance=2pt]
\tikzstyle{gray edge}=[-, {gray!60!white}]
\tikzstyle{pointer edge}=[->, very thick, gray]
\tikzstyle{boldedge}=[-, line width=1.0pt, shorten <=-0.17mm, shorten >=-0.17mm]
\tikzstyle{bidir edge}=[<->, very thick, draw={rgb,255: red,191; green,191; blue,191}]
\tikzstyle{purple edge}=[->, thick, draw={rgb,255: red,225; green,117; blue,216}]
\tikzstyle{green edge}=[->, thick, draw={rgb,255: red,167; green,231; blue,137}]
\tikzstyle{orange edge}=[->, thick, draw={rgb,255: red,245; green,170; blue,63}]
\tikzstyle{any edge}=[->, thick, draw=cyan]
\tikzstyle{red edge}=[->, thick, draw={rgb,255: red,255; green,136; blue,136}]
\tikzstyle{bidiredge}=[<->, thick]
\tikzstyle{dashed diredge}=[->, dashed, dash pattern=on 1pt off 0.5pt]
\tikzstyle{bidashed diredge}=[<->, dashed, dash pattern=on 1pt off 0.5pt]
\tikzstyle{blue edge}=[-, draw={rgb,255: red,28; green,115; blue,237}]

%% file: figs/markov-plat.tikz
\begin{tikzpicture}
	\begin{pgfonlayer}{nodelayer}
		\node [style=none] (0) at (-7.5, -0.75) {$M(B)=$};
		\node [style=none] (3) at (-5, 0.5) {};
		\node [style=none] (6) at (-5, -0.5) {};
		\node [style=none] (9) at (-5, -1.5) {};
		\node [style=none] (10) at (5, -0.5) {};
		\node [style=none] (11) at (5, -1.5) {};
		\node [style=none] (12) at (5, -0.5) {};
		\node [style=none] (13) at (5, 0.5) {};
		\node [style=none] (14) at (-5, -3) {};
		\node [style=none] (15) at (-5, -2.75) {};
		\node [style=none] (16) at (-5, -2.5) {};
		\node [style=none] (17) at (5, -2.5) {};
		\node [style=none] (18) at (5, -2.75) {};
		\node [style=none] (19) at (5, -3) {};
		\node [style=none] (49) at (5, -0.5) {};
		\node [style=none] (50) at (5, -1.5) {};
		\node [style=none] (54) at (-1, -8) {};
		\node [style=none] (55) at (-3, -7) {};
		\node [style=none] (56) at (-1, -7) {};
		\node [style=none] (57) at (-1, -8) {};
		\node [style=none] (58) at (-3, -8) {};
		\node [style=none] (59) at (-2.25, -7.25) {};
		\node [style=none] (60) at (-1.75, -7.75) {};
		\node [style=none] (62) at (-1, -7) {};
		\node [style=none] (63) at (-1, -8) {};
		\node [style=none] (64) at (-5, -6) {};
		\node [style=none] (65) at (-5, -6) {};
		\node [style=none] (66) at (-3, -6) {};
		\node [style=none] (67) at (-3, -7) {};
		\node [style=none] (68) at (-5, -7) {};
		\node [style=none] (69) at (-4.25, -6.75) {};
		\node [style=none] (70) at (-3.75, -6.25) {};
		\node [style=none] (71) at (-3, -6) {};
		\node [style=none] (72) at (-5, -9.5) {};
		\node [style=none] (73) at (-5, -5.5) {};
		\node [style=none] (74) at (5, -5.5) {};
		\node [style=none] (75) at (5, -9.5) {};
		\node [style=none] (76) at (-1, -6) {};
		\node [style=none] (77) at (-1, -6) {};
		\node [style=none] (78) at (1, -6) {};
		\node [style=none] (79) at (1, -7) {};
		\node [style=none] (80) at (-1, -7) {};
		\node [style=none] (81) at (-0.25, -6.75) {};
		\node [style=none] (82) at (0.25, -6.25) {};
		\node [style=none] (83) at (1, -6) {};
		\node [style=none] (84) at (1, -8) {};
		\node [style=none] (85) at (1, -8) {};
		\node [style=none] (86) at (3, -8) {};
		\node [style=none] (87) at (3, -9) {};
		\node [style=none] (88) at (1, -9) {};
		\node [style=none] (89) at (1.75, -8.75) {};
		\node [style=none] (90) at (2.25, -8.25) {};
		\node [style=none] (91) at (3, -8) {};
		\node [style=none] (92) at (-5, -8) {};
		\node [style=none] (93) at (-5, -8) {};
		\node [style=none] (94) at (-3, -8) {};
		\node [style=none] (95) at (-3, -9) {};
		\node [style=none] (96) at (-5, -9) {};
		\node [style=none] (97) at (-4.25, -8.75) {};
		\node [style=none] (98) at (-3.75, -8.25) {};
		\node [style=none] (99) at (-3, -8) {};
		\node [style=none] (100) at (1, -6) {};
		\node [style=none] (101) at (1, -6) {};
		\node [style=none] (102) at (3, -6) {};
		\node [style=none] (103) at (3, -7) {};
		\node [style=none] (104) at (1, -7) {};
		\node [style=none] (105) at (1.75, -6.75) {};
		\node [style=none] (106) at (2.25, -6.25) {};
		\node [style=none] (107) at (3, -6) {};
		\node [style=none] (108) at (5, -8) {};
		\node [style=none] (109) at (3, -7) {};
		\node [style=none] (110) at (5, -7) {};
		\node [style=none] (111) at (5, -8) {};
		\node [style=none] (112) at (3, -8) {};
		\node [style=none] (113) at (3.75, -7.25) {};
		\node [style=none] (114) at (4.25, -7.75) {};
		\node [style=none] (115) at (5, -7) {};
		\node [style=none] (116) at (5, -8) {};
		\node [style=none] (117) at (3, -7) {};
		\node [style=none] (118) at (5, -7) {};
		\node [style=none] (119) at (3, -8) {};
		\node [style=none] (120) at (3, -8) {};
		\node [style=none] (121) at (5, -9) {};
		\node [style=none] (122) at (5, -6) {};
		\node [style=none] (123) at (-5, -8) {};
		\node [style=none] (124) at (-5, -9) {};
		\node [style=none] (125) at (5, -6) {};
		\node [style=none] (126) at (5, -7) {};
		\node [style=none] (127) at (5, -8) {};
		\node [style=none] (128) at (5, -9) {};
		\node [style=none] (129) at (-7.5, -7.5) {$P(B)=$};
		\node [style=none] (130) at (-5, 1.5) {};
		\node [style=none] (131) at (5, 1.5) {};
		\node [style=none] (132) at (-5, -3.25) {};
		\node [style=none] (133) at (5, -3.25) {};
		\node [style=none] (134) at (-1, -0.5) {};
		\node [style=none] (135) at (-3, 0.5) {};
		\node [style=none] (136) at (-1, 0.5) {};
		\node [style=none] (137) at (-1, -0.5) {};
		\node [style=none] (138) at (-3, -0.5) {};
		\node [style=none] (139) at (-2.25, 0.25) {};
		\node [style=none] (140) at (-1.75, -0.25) {};
		\node [style=none] (141) at (-1, 0.5) {};
		\node [style=none] (142) at (-1, -0.5) {};
		\node [style=none] (143) at (-5, 1.5) {};
		\node [style=none] (144) at (-5, 1.5) {};
		\node [style=none] (145) at (-3, 1.5) {};
		\node [style=none] (146) at (-3, 0.5) {};
		\node [style=none] (147) at (-5, 0.5) {};
		\node [style=none] (148) at (-4.25, 0.75) {};
		\node [style=none] (149) at (-3.75, 1.25) {};
		\node [style=none] (150) at (-3, 1.5) {};
		\node [style=none] (151) at (-5, -2) {};
		\node [style=none] (152) at (-5, 2) {};
		\node [style=none] (153) at (5, 2) {};
		\node [style=none] (154) at (5, -2) {};
		\node [style=none] (155) at (-1, 1.5) {};
		\node [style=none] (156) at (-1, 1.5) {};
		\node [style=none] (157) at (1, 1.5) {};
		\node [style=none] (158) at (1, 0.5) {};
		\node [style=none] (159) at (-1, 0.5) {};
		\node [style=none] (160) at (-0.25, 0.75) {};
		\node [style=none] (161) at (0.25, 1.25) {};
		\node [style=none] (162) at (1, 1.5) {};
		\node [style=none] (163) at (1, -0.5) {};
		\node [style=none] (164) at (1, -0.5) {};
		\node [style=none] (165) at (3, -0.5) {};
		\node [style=none] (166) at (3, -1.5) {};
		\node [style=none] (167) at (1, -1.5) {};
		\node [style=none] (168) at (1.75, -1.25) {};
		\node [style=none] (169) at (2.25, -0.75) {};
		\node [style=none] (170) at (3, -0.5) {};
		\node [style=none] (171) at (-5, -0.5) {};
		\node [style=none] (172) at (-5, -0.5) {};
		\node [style=none] (173) at (-3, -0.5) {};
		\node [style=none] (174) at (-3, -1.5) {};
		\node [style=none] (175) at (-5, -1.5) {};
		\node [style=none] (176) at (-4.25, -1.25) {};
		\node [style=none] (177) at (-3.75, -0.75) {};
		\node [style=none] (178) at (-3, -0.5) {};
		\node [style=none] (179) at (1, 1.5) {};
		\node [style=none] (180) at (1, 1.5) {};
		\node [style=none] (181) at (3, 1.5) {};
		\node [style=none] (182) at (3, 0.5) {};
		\node [style=none] (183) at (1, 0.5) {};
		\node [style=none] (184) at (1.75, 0.75) {};
		\node [style=none] (185) at (2.25, 1.25) {};
		\node [style=none] (186) at (3, 1.5) {};
		\node [style=none] (187) at (5, -0.5) {};
		\node [style=none] (188) at (3, 0.5) {};
		\node [style=none] (189) at (5, 0.5) {};
		\node [style=none] (190) at (5, -0.5) {};
		\node [style=none] (191) at (3, -0.5) {};
		\node [style=none] (192) at (3.75, 0.25) {};
		\node [style=none] (193) at (4.25, -0.25) {};
		\node [style=none] (195) at (5, -0.5) {};
		\node [style=none] (196) at (3, 0.5) {};
		\node [style=none] (198) at (3, -0.5) {};
		\node [style=none] (199) at (3, -0.5) {};
		\node [style=none] (200) at (5, -1.5) {};
		\node [style=none] (201) at (5, 1.5) {};
		\node [style=none] (202) at (-5, -0.5) {};
		\node [style=none] (203) at (-5, -1.5) {};
		\node [style=none] (206) at (5, -0.5) {};
		\node [style=none] (207) at (5, -1.5) {};
		\node [style=none] (329) at (0, 2.75) {DQC1};
		\node [style=none] (330) at (0, -4.75) {BQP};
	\end{pgfonlayer}
	\begin{pgfonlayer}{edgelayer}
		\draw [bend left=90, looseness=0.75] (13.center) to (19.center);
		\draw [bend left=90, looseness=0.75] (12.center) to (18.center);
		\draw [bend left=90, looseness=0.75] (11.center) to (17.center);
		\draw [bend left=90, looseness=0.75] (14.center) to (3.center);
		\draw [bend left=90, looseness=0.75] (15.center) to (6.center);
		\draw [bend left=270, looseness=0.75] (9.center) to (16.center);
		\draw (14.center) to (19.center);
		\draw (18.center) to (15.center);
		\draw (16.center) to (17.center);
		\draw [in=180, out=0] (58.center) to (56.center);
		\draw [bend right=15] (59.center) to (55.center);
		\draw [bend right=15] (60.center) to (57.center);
		\draw [in=-180, out=0] (65.center) to (67.center);
		\draw [bend left=15, looseness=1.25] (69.center) to (68.center);
		\draw [bend right=15] (66.center) to (70.center);
		\draw [style=new edge style 0] (73.center) to (74.center);
		\draw [style=new edge style 0] (74.center) to (75.center);
		\draw [style=new edge style 0] (75.center) to (72.center);
		\draw [style=new edge style 0] (72.center) to (73.center);
		\draw [in=-180, out=0] (77.center) to (79.center);
		\draw [bend left=15, looseness=1.25] (81.center) to (80.center);
		\draw [bend right=15] (78.center) to (82.center);
		\draw [in=-180, out=0] (85.center) to (87.center);
		\draw [bend left=15, looseness=1.25] (89.center) to (88.center);
		\draw [bend right=15] (86.center) to (90.center);
		\draw [in=-180, out=0] (93.center) to (95.center);
		\draw [bend left=15, looseness=1.25] (97.center) to (96.center);
		\draw [bend right=15] (94.center) to (98.center);
		\draw [in=-180, out=0] (101.center) to (103.center);
		\draw [bend left=15, looseness=1.25] (105.center) to (104.center);
		\draw [bend right=15] (102.center) to (106.center);
		\draw [in=180, out=0] (112.center) to (110.center);
		\draw [bend right=15] (113.center) to (109.center);
		\draw [bend right=15] (114.center) to (111.center);
		\draw (71.center) to (77.center);
		\draw (95.center) to (88.center);
		\draw (63.center) to (85.center);
		\draw (87.center) to (121.center);
		\draw (107.center) to (122.center);
		\draw [bend right=90, looseness=2.00] (65.center) to (68.center);
		\draw [bend right=90, looseness=2.00] (123.center) to (124.center);
		\draw [bend left=90, looseness=2.25] (125.center) to (126.center);
		\draw [bend left=90, looseness=2.25] (127.center) to (128.center);
		\draw [bend left=90, looseness=0.75] (131.center) to (133.center);
		\draw [bend left=90, looseness=0.75] (132.center) to (130.center);
		\draw (132.center) to (133.center);
		\draw [in=180, out=0] (138.center) to (136.center);
		\draw [bend right=15] (139.center) to (135.center);
		\draw [bend right=15] (140.center) to (137.center);
		\draw [in=-180, out=0] (144.center) to (146.center);
		\draw [bend left=15, looseness=1.25] (148.center) to (147.center);
		\draw [bend right=15] (145.center) to (149.center);
		\draw [style=new edge style 0] (152.center) to (153.center);
		\draw [style=new edge style 0] (153.center) to (154.center);
		\draw [style=new edge style 0] (154.center) to (151.center);
		\draw [style=new edge style 0] (151.center) to (152.center);
		\draw [in=-180, out=0] (156.center) to (158.center);
		\draw [bend left=15, looseness=1.25] (160.center) to (159.center);
		\draw [bend right=15] (157.center) to (161.center);
		\draw [in=-180, out=0] (164.center) to (166.center);
		\draw [bend left=15, looseness=1.25] (168.center) to (167.center);
		\draw [bend right=15] (165.center) to (169.center);
		\draw [in=-180, out=0] (172.center) to (174.center);
		\draw [bend left=15, looseness=1.25] (176.center) to (175.center);
		\draw [bend right=15] (173.center) to (177.center);
		\draw [in=-180, out=0] (180.center) to (182.center);
		\draw [bend left=15, looseness=1.25] (184.center) to (183.center);
		\draw [bend right=15] (181.center) to (185.center);
		\draw [in=180, out=0] (191.center) to (189.center);
		\draw [bend right=15] (192.center) to (188.center);
		\draw [bend right=15] (193.center) to (190.center);
		\draw (150.center) to (156.center);
		\draw (174.center) to (167.center);
		\draw (142.center) to (164.center);
		\draw (166.center) to (200.center);
		\draw (186.center) to (201.center);
	\end{pgfonlayer}
\end{tikzpicture}

%% file: figs/markov-plat-slide.tikz
\begin{tikzpicture}
	\begin{pgfonlayer}{nodelayer}
		\node [style=none] (0) at (-7.5, -0.25) {$M(B)=$};
		\node [style=none] (3) at (-5, 1) {};
		\node [style=none] (6) at (-5, 0) {};
		\node [style=none] (9) at (-5, -1) {};
		\node [style=none] (10) at (5, 0) {};
		\node [style=none] (11) at (5, -1) {};
		\node [style=none] (12) at (5, 0) {};
		\node [style=none] (13) at (5, 1) {};
		\node [style=none] (14) at (-5, -2.5) {};
		\node [style=none] (15) at (-5, -2.25) {};
		\node [style=none] (16) at (-5, -2) {};
		\node [style=none] (17) at (5, -2) {};
		\node [style=none] (18) at (5, -2.25) {};
		\node [style=none] (19) at (5, -2.5) {};
		\node [style=none] (49) at (5, 0) {};
		\node [style=none] (50) at (5, -1) {};
		\node [style=none] (130) at (-5, 2) {};
		\node [style=none] (131) at (5, 2) {};
		\node [style=none] (132) at (-5, -2.75) {};
		\node [style=none] (133) at (5, -2.75) {};
		\node [style=none] (134) at (-1, 0) {};
		\node [style=none] (135) at (-3, 1) {};
		\node [style=none] (136) at (-1, 1) {};
		\node [style=none] (137) at (-1, 0) {};
		\node [style=none] (138) at (-3, 0) {};
		\node [style=none] (139) at (-2.25, 0.75) {};
		\node [style=none] (140) at (-1.75, 0.25) {};
		\node [style=none] (141) at (-1, 1) {};
		\node [style=none] (142) at (-1, 0) {};
		\node [style=none] (143) at (-5, 2) {};
		\node [style=none] (144) at (-5, 2) {};
		\node [style=none] (145) at (-3, 2) {};
		\node [style=none] (146) at (-3, 1) {};
		\node [style=none] (147) at (-5, 1) {};
		\node [style=none] (148) at (-4.25, 1.25) {};
		\node [style=none] (149) at (-3.75, 1.75) {};
		\node [style=none] (150) at (-3, 2) {};
		\node [style=none] (151) at (-5, -1.5) {};
		\node [style=none] (152) at (-5, 2.5) {};
		\node [style=none] (153) at (5, 2.5) {};
		\node [style=none] (154) at (5, -1.5) {};
		\node [style=none] (155) at (-1, 2) {};
		\node [style=none] (156) at (-1, 2) {};
		\node [style=none] (157) at (1, 2) {};
		\node [style=none] (158) at (1, 1) {};
		\node [style=none] (159) at (-1, 1) {};
		\node [style=none] (160) at (-0.25, 1.25) {};
		\node [style=none] (161) at (0.25, 1.75) {};
		\node [style=none] (162) at (1, 2) {};
		\node [style=none] (163) at (1, 0) {};
		\node [style=none] (164) at (1, 0) {};
		\node [style=none] (165) at (3, 0) {};
		\node [style=none] (166) at (3, -1) {};
		\node [style=none] (167) at (1, -1) {};
		\node [style=none] (168) at (1.75, -0.75) {};
		\node [style=none] (169) at (2.25, -0.25) {};
		\node [style=none] (170) at (3, 0) {};
		\node [style=none] (171) at (-5, 0) {};
		\node [style=none] (172) at (-5, 0) {};
		\node [style=none] (173) at (-3, 0) {};
		\node [style=none] (174) at (-3, -1) {};
		\node [style=none] (175) at (-5, -1) {};
		\node [style=none] (176) at (-4.25, -0.75) {};
		\node [style=none] (177) at (-3.75, -0.25) {};
		\node [style=none] (178) at (-3, 0) {};
		\node [style=none] (179) at (1, 2) {};
		\node [style=none] (180) at (1, 2) {};
		\node [style=none] (181) at (3, 2) {};
		\node [style=none] (182) at (3, 1) {};
		\node [style=none] (183) at (1, 1) {};
		\node [style=none] (184) at (1.75, 1.25) {};
		\node [style=none] (185) at (2.25, 1.75) {};
		\node [style=none] (186) at (3, 2) {};
		\node [style=none] (187) at (5, 0) {};
		\node [style=none] (188) at (3, 1) {};
		\node [style=none] (189) at (5, 1) {};
		\node [style=none] (190) at (5, 0) {};
		\node [style=none] (191) at (3, 0) {};
		\node [style=none] (192) at (3.75, 0.75) {};
		\node [style=none] (193) at (4.25, 0.25) {};
		\node [style=none] (195) at (5, 0) {};
		\node [style=none] (196) at (3, 1) {};
		\node [style=none] (198) at (3, 0) {};
		\node [style=none] (199) at (3, 0) {};
		\node [style=none] (200) at (5, -1) {};
		\node [style=none] (201) at (5, 2) {};
		\node [style=none] (202) at (-5, 0) {};
		\node [style=none] (203) at (-5, -1) {};
		\node [style=none] (206) at (5, 0) {};
		\node [style=none] (207) at (5, -1) {};
		\node [style=none] (208) at (-5, -5.25) {};
		\node [style=none] (209) at (-5, -6.5) {};
		\node [style=none] (210) at (-5, -7.75) {};
		\node [style=none] (211) at (5, -6.5) {};
		\node [style=none] (212) at (5, -7.75) {};
		\node [style=none] (213) at (5, -6.5) {};
		\node [style=none] (214) at (5, -5.25) {};
		\node [style=none] (215) at (-5, -6.25) {};
		\node [style=none] (216) at (-5, -7.5) {};
		\node [style=none] (217) at (-5, -8.75) {};
		\node [style=none] (218) at (5, -8.75) {};
		\node [style=none] (219) at (5, -7.5) {};
		\node [style=none] (220) at (5, -6.25) {};
		\node [style=none] (221) at (5, -6.5) {};
		\node [style=none] (222) at (5, -7.75) {};
		\node [style=none] (223) at (-5, -4) {};
		\node [style=none] (224) at (5, -4) {};
		\node [style=none] (225) at (-5, -5) {};
		\node [style=none] (226) at (5, -5) {};
		\node [style=none] (227) at (-1, -6.5) {};
		\node [style=none] (228) at (-3, -5.25) {};
		\node [style=none] (229) at (-1, -5.25) {};
		\node [style=none] (230) at (-1, -6.5) {};
		\node [style=none] (231) at (-3, -6.5) {};
		\node [style=none] (232) at (-2.25, -5.5) {};
		\node [style=none] (233) at (-1.75, -6) {};
		\node [style=none] (234) at (-1, -5.25) {};
		\node [style=none] (235) at (-1, -6.5) {};
		\node [style=none] (236) at (-5, -4) {};
		\node [style=none] (237) at (-5, -4) {};
		\node [style=none] (238) at (-3, -4) {};
		\node [style=none] (239) at (-3, -5.25) {};
		\node [style=none] (240) at (-5, -5.25) {};
		\node [style=none] (241) at (-4.25, -4.75) {};
		\node [style=none] (242) at (-3.75, -4.25) {};
		\node [style=none] (243) at (-3, -4) {};
		\node [style=none] (244) at (-5, -9.25) {};
		\node [style=none] (245) at (-5, -3.5) {};
		\node [style=none] (246) at (5, -3.5) {};
		\node [style=none] (247) at (5, -9.25) {};
		\node [style=none] (248) at (-1, -4) {};
		\node [style=none] (249) at (-1, -4) {};
		\node [style=none] (250) at (1, -4) {};
		\node [style=none] (251) at (1, -5.25) {};
		\node [style=none] (252) at (-1, -5.25) {};
		\node [style=none] (253) at (-0.25, -4.75) {};
		\node [style=none] (254) at (0.25, -4.25) {};
		\node [style=none] (255) at (1, -4) {};
		\node [style=none] (256) at (1, -6.5) {};
		\node [style=none] (257) at (1, -6.5) {};
		\node [style=none] (258) at (3, -6.5) {};
		\node [style=none] (259) at (3, -7.75) {};
		\node [style=none] (260) at (1, -7.75) {};
		\node [style=none] (261) at (1.75, -7.25) {};
		\node [style=none] (262) at (2.25, -6.75) {};
		\node [style=none] (263) at (3, -6.5) {};
		\node [style=none] (264) at (-5, -6.5) {};
		\node [style=none] (265) at (-5, -6.5) {};
		\node [style=none] (266) at (-3, -6.5) {};
		\node [style=none] (267) at (-3, -7.75) {};
		\node [style=none] (268) at (-5, -7.75) {};
		\node [style=none] (269) at (-4.25, -7.25) {};
		\node [style=none] (270) at (-3.75, -6.75) {};
		\node [style=none] (271) at (-3, -6.5) {};
		\node [style=none] (272) at (1, -4) {};
		\node [style=none] (273) at (1, -4) {};
		\node [style=none] (274) at (3, -4) {};
		\node [style=none] (275) at (3, -5.25) {};
		\node [style=none] (276) at (1, -5.25) {};
		\node [style=none] (277) at (1.75, -4.75) {};
		\node [style=none] (278) at (2.25, -4.25) {};
		\node [style=none] (279) at (3, -4) {};
		\node [style=none] (280) at (5, -6.5) {};
		\node [style=none] (281) at (3, -5.25) {};
		\node [style=none] (282) at (5, -5.25) {};
		\node [style=none] (283) at (5, -6.5) {};
		\node [style=none] (284) at (3, -6.5) {};
		\node [style=none] (285) at (3.75, -5.5) {};
		\node [style=none] (286) at (4.25, -6) {};
		\node [style=none] (287) at (5, -6.5) {};
		\node [style=none] (288) at (3, -5.25) {};
		\node [style=none] (289) at (3, -6.5) {};
		\node [style=none] (290) at (3, -6.5) {};
		\node [style=none] (291) at (5, -7.75) {};
		\node [style=none] (292) at (5, -4) {};
		\node [style=none] (293) at (-5, -6.5) {};
		\node [style=none] (294) at (-5, -7.75) {};
		\node [style=none] (295) at (5, -6.5) {};
		\node [style=none] (296) at (5, -7.75) {};
		\node [style=none] (297) at (-4.75, -5) {};
		\node [style=none] (298) at (-4.25, -5) {};
		\node [style=none] (299) at (-4, -5) {};
		\node [style=none] (300) at (-3.25, -5) {};
		\node [style=none] (301) at (-0.75, -5) {};
		\node [style=none] (302) at (-0.25, -5) {};
		\node [style=none] (303) at (0, -5) {};
		\node [style=none] (304) at (0.75, -5) {};
		\node [style=none] (305) at (1.25, -5) {};
		\node [style=none] (306) at (1.75, -5) {};
		\node [style=none] (307) at (2, -5) {};
		\node [style=none] (308) at (2.75, -5) {};
		\node [style=none] (309) at (-2.5, -6.25) {};
		\node [style=none] (310) at (-2, -6.25) {};
		\node [style=none] (311) at (-1.75, -6.25) {};
		\node [style=none] (312) at (-1.25, -6.25) {};
		\node [style=none] (313) at (3.25, -6.25) {};
		\node [style=none] (314) at (4, -6.25) {};
		\node [style=none] (315) at (4.25, -6.25) {};
		\node [style=none] (316) at (-4.75, -7.5) {};
		\node [style=none] (317) at (-4.25, -7.5) {};
		\node [style=none] (318) at (-4, -7.5) {};
		\node [style=none] (319) at (-3.25, -7.5) {};
		\node [style=none] (320) at (1, -7.5) {};
		\node [style=none] (321) at (1.75, -7.5) {};
		\node [style=none] (322) at (2, -7.5) {};
		\node [style=none] (323) at (2.75, -7.5) {};
		\node [style=none] (324) at (-9.5, -6.5) {$\sim$};
		\node [style=none] (325) at (4.75, -6.25) {};
		\node [style=none] (326) at (-7.5, -6.5) {$P(B')=$};
		\node [style=none] (327) at (-9.5, -5.75) {slide};
	\end{pgfonlayer}
	\begin{pgfonlayer}{edgelayer}
		\draw [bend left=90, looseness=0.75] (13.center) to (19.center);
		\draw [bend left=90, looseness=0.75] (12.center) to (18.center);
		\draw [bend left=90, looseness=0.75] (11.center) to (17.center);
		\draw [bend left=90, looseness=0.75] (14.center) to (3.center);
		\draw [bend left=90, looseness=0.75] (15.center) to (6.center);
		\draw [bend left=270, looseness=0.75] (9.center) to (16.center);
		\draw (14.center) to (19.center);
		\draw (18.center) to (15.center);
		\draw (16.center) to (17.center);
		\draw [bend left=90, looseness=0.75] (131.center) to (133.center);
		\draw [bend left=90, looseness=0.75] (132.center) to (130.center);
		\draw (132.center) to (133.center);
		\draw [in=180, out=0] (138.center) to (136.center);
		\draw [bend right=15] (139.center) to (135.center);
		\draw [bend right=15] (140.center) to (137.center);
		\draw [in=-180, out=0] (144.center) to (146.center);
		\draw [bend left=15, looseness=1.25] (148.center) to (147.center);
		\draw [bend right=15] (145.center) to (149.center);
		\draw [style=new edge style 0] (152.center) to (153.center);
		\draw [style=new edge style 0] (153.center) to (154.center);
		\draw [style=new edge style 0] (154.center) to (151.center);
		\draw [style=new edge style 0] (151.center) to (152.center);
		\draw [in=-180, out=0] (156.center) to (158.center);
		\draw [bend left=15, looseness=1.25] (160.center) to (159.center);
		\draw [bend right=15] (157.center) to (161.center);
		\draw [in=-180, out=0] (164.center) to (166.center);
		\draw [bend left=15, looseness=1.25] (168.center) to (167.center);
		\draw [bend right=15] (165.center) to (169.center);
		\draw [in=-180, out=0] (172.center) to (174.center);
		\draw [bend left=15, looseness=1.25] (176.center) to (175.center);
		\draw [bend right=15] (173.center) to (177.center);
		\draw [in=-180, out=0] (180.center) to (182.center);
		\draw [bend left=15, looseness=1.25] (184.center) to (183.center);
		\draw [bend right=15] (181.center) to (185.center);
		\draw [in=180, out=0] (191.center) to (189.center);
		\draw [bend right=15] (192.center) to (188.center);
		\draw [bend right=15] (193.center) to (190.center);
		\draw (150.center) to (156.center);
		\draw (174.center) to (167.center);
		\draw (142.center) to (164.center);
		\draw (166.center) to (200.center);
		\draw (186.center) to (201.center);
		\draw [bend left=90, looseness=1.50] (214.center) to (220.center);
		\draw [bend left=90, looseness=1.50] (213.center) to (219.center);
		\draw [bend left=90, looseness=1.50] (212.center) to (218.center);
		\draw [bend left=90, looseness=1.50] (215.center) to (208.center);
		\draw [bend left=90, looseness=1.50] (216.center) to (209.center);
		\draw [bend right=90, looseness=1.50] (210.center) to (217.center);
		\draw (217.center) to (218.center);
		\draw [bend left=90, looseness=1.50] (224.center) to (226.center);
		\draw [bend left=90, looseness=1.50] (225.center) to (223.center);
		\draw [in=180, out=0] (231.center) to (229.center);
		\draw [bend right=15] (232.center) to (228.center);
		\draw [bend right=15] (233.center) to (230.center);
		\draw [in=-180, out=0] (237.center) to (239.center);
		\draw [bend left=15, looseness=1.25] (241.center) to (240.center);
		\draw [bend right=15] (238.center) to (242.center);
		\draw [style=new edge style 0] (245.center) to (246.center);
		\draw [style=new edge style 0] (246.center) to (247.center);
		\draw [style=new edge style 0] (247.center) to (244.center);
		\draw [style=new edge style 0] (244.center) to (245.center);
		\draw [in=-180, out=0] (249.center) to (251.center);
		\draw [bend left=15, looseness=1.25] (253.center) to (252.center);
		\draw [bend right=15] (250.center) to (254.center);
		\draw [in=-180, out=0] (257.center) to (259.center);
		\draw [bend left=15, looseness=1.25] (261.center) to (260.center);
		\draw [bend right=15] (258.center) to (262.center);
		\draw [in=-180, out=0] (265.center) to (267.center);
		\draw [bend left=15, looseness=1.25] (269.center) to (268.center);
		\draw [bend right=15] (266.center) to (270.center);
		\draw [in=-180, out=0] (273.center) to (275.center);
		\draw [bend left=15, looseness=1.25] (277.center) to (276.center);
		\draw [bend right=15] (274.center) to (278.center);
		\draw [in=180, out=0] (284.center) to (282.center);
		\draw [bend right=15] (285.center) to (281.center);
		\draw [bend right=15] (286.center) to (283.center);
		\draw (243.center) to (249.center);
		\draw (267.center) to (260.center);
		\draw (235.center) to (257.center);
		\draw (259.center) to (291.center);
		\draw (279.center) to (292.center);
		\draw (225.center) to (297.center);
		\draw (298.center) to (299.center);
		\draw (300.center) to (301.center);
		\draw (302.center) to (303.center);
		\draw (304.center) to (305.center);
		\draw (306.center) to (307.center);
		\draw (308.center) to (226.center);
		\draw (215.center) to (309.center);
		\draw (310.center) to (311.center);
		\draw (312.center) to (313.center);
		\draw (314.center) to (315.center);
		\draw (216.center) to (316.center);
		\draw (317.center) to (318.center);
		\draw (319.center) to (320.center);
		\draw (321.center) to (322.center);
		\draw (323.center) to (219.center);
		\draw (325.center) to (220.center);
	\end{pgfonlayer}
\end{tikzpicture}

%% file: figs/braid-gens.tikz
\begin{tikzpicture}
	\begin{pgfonlayer}{nodelayer}
		\node [style=none] (0) at (-4.5, 0.5) {};
		\node [style=none] (2) at (-2.5, 0.5) {};
		\node [style=none] (3) at (-2.5, -0.5) {};
		\node [style=none] (5) at (-4.5, -0.5) {};
		\node [style=none] (6) at (-3.75, -0.25) {};
		\node [style=none] (7) at (-3.25, 0.25) {};
		\node [style=none] (8) at (4.75, 0.5) {};
		\node [style=none] (9) at (6.75, 0.5) {};
		\node [style=none] (10) at (6.75, -0.5) {};
		\node [style=none] (11) at (4.75, -0.5) {};
		\node [style=none] (12) at (5.5, 0.25) {};
		\node [style=none] (13) at (6, -0.25) {};
		\node [style=none] (90) at (-7.25, 0) {$\sigma_i=$};
		\node [style=none] (91) at (2, 0) {$\sigma_i^{-1}=$};
		\node [style=none] (92) at (-5.5, 0.5) {$i$};
		\node [style=none] (93) at (-5.5, -0.5) {$i+1$};
		\node [style=none] (94) at (4, 0.5) {$i$};
		\node [style=none] (95) at (4, -0.5) {$i+1$};
	\end{pgfonlayer}
	\begin{pgfonlayer}{edgelayer}
		\draw [in=-180, out=0] (0.center) to (3.center);
		\draw [bend left=15, looseness=1.25] (6.center) to (5.center);
		\draw [bend right=15] (2.center) to (7.center);
		\draw [in=180, out=0] (11.center) to (9.center);
		\draw [bend right=15] (12.center) to (8.center);
		\draw [bend right=15] (13.center) to (10.center);
	\end{pgfonlayer}
\end{tikzpicture}

%% file: figs/example-braids.tikz
\begin{tikzpicture}
	\begin{pgfonlayer}{nodelayer}
		\node [style=none] (14) at (-4.75, 0) {};
		\node [style=none] (15) at (-2.75, 0) {};
		\node [style=none] (16) at (-2.75, -1) {};
		\node [style=none] (17) at (-4.75, -1) {};
		\node [style=none] (18) at (-4, -0.75) {};
		\node [style=none] (19) at (-3.5, -0.25) {};
		\node [style=none] (20) at (-2.75, 0) {};
		\node [style=none] (21) at (-0.75, 0) {};
		\node [style=none] (22) at (-0.75, -1) {};
		\node [style=none] (23) at (-2.75, -1) {};
		\node [style=none] (24) at (-2, -0.75) {};
		\node [style=none] (25) at (-1.5, -0.25) {};
		\node [style=none] (26) at (-0.75, 0) {};
		\node [style=none] (27) at (1.25, 0) {};
		\node [style=none] (28) at (1.25, -1) {};
		\node [style=none] (29) at (-0.75, -1) {};
		\node [style=none] (30) at (0, -0.75) {};
		\node [style=none] (31) at (0.5, -0.25) {};
		\node [style=none] (32) at (-4.75, -2) {};
		\node [style=none] (33) at (-4.75, -2.25) {};
		\node [style=none] (34) at (1.25, -2) {};
		\node [style=none] (35) at (1.25, -2.25) {};
		\node [style=none] (36) at (-4.75, -1.25) {};
		\node [style=none] (37) at (-4.75, 0.25) {};
		\node [style=none] (38) at (1.25, 0.25) {};
		\node [style=none] (39) at (1.25, -1.25) {};
		\node [style=none] (46) at (-11, -10.25) {};
		\node [style=none] (49) at (-11, -11.25) {};
		\node [style=none] (59) at (-11, -12.25) {};
		\node [style=none] (61) at (1, -11.25) {};
		\node [style=none] (62) at (1, -12.25) {};
		\node [style=none] (67) at (1, -11.25) {};
		\node [style=none] (68) at (1, -10.25) {};
		\node [style=none] (69) at (-11, -13.75) {};
		\node [style=none] (70) at (-11, -13.5) {};
		\node [style=none] (71) at (-11, -13.25) {};
		\node [style=none] (72) at (1, -13.25) {};
		\node [style=none] (73) at (1, -13.5) {};
		\node [style=none] (74) at (1, -13.75) {};
		\node [style=none] (75) at (-11, -12.5) {};
		\node [style=none] (76) at (-11, -10) {};
		\node [style=none] (77) at (1, -10) {};
		\node [style=none] (78) at (1, -12.5) {};
		\node [style=none] (80) at (-5, -3.5) {$B_{6_2} = (\sigma^{-1}_1)^3 \sigma_2 \sigma^{-1}_1 \sigma_2$};
		\node [style=none] (82) at (-11.5, -0.5) {};
		\node [style=none] (83) at (-10.5, -0.5) {};
		\node [style=none] (85) at (-11.5, -0.75) {};
		\node [style=none] (86) at (-10.5, -0.75) {};
		\node [style=none] (87) at (-11.5, -1.75) {};
		\node [style=none] (88) at (-10.5, -1.75) {};
		\node [style=none] (89) at (-13.5, -1.25) {$M(B_{0_1})=$};
		\node [style=none] (96) at (-7, -1.25) {$M(B_{3_1})=$};
		\node [style=none] (97) at (-13.5, -6.25) {$M(B_{6_2})=$};
		\node [style=none] (123) at (1, -11.25) {};
		\node [style=none] (124) at (1, -12.25) {};
		\node [style=none] (128) at (-1.5, 1) {$B_{3_1} = (\sigma_1)^3$};
		\node [style=none] (129) at (-11, 0.25) {$B_{0_1}=1_1$};
		\node [style=none] (130) at (-11.5, -1) {};
		\node [style=none] (131) at (-10.5, -1) {};
		\node [style=none] (132) at (-9, -5.5) {};
		\node [style=none] (133) at (-9, -6.5) {};
		\node [style=none] (134) at (-11, -4.5) {};
		\node [style=none] (135) at (-9, -4.5) {};
		\node [style=none] (136) at (-9, -5.5) {};
		\node [style=none] (137) at (-11, -5.5) {};
		\node [style=none] (138) at (-10.25, -4.75) {};
		\node [style=none] (139) at (-9.75, -5.25) {};
		\node [style=none] (140) at (-11, -6.5) {};
		\node [style=none] (141) at (1, -5.5) {};
		\node [style=none] (142) at (1, -6.5) {};
		\node [style=none] (143) at (1, -5.5) {};
		\node [style=none] (144) at (1, -4.5) {};
		\node [style=none] (145) at (-11, -8) {};
		\node [style=none] (146) at (-11, -7.75) {};
		\node [style=none] (147) at (-11, -7.5) {};
		\node [style=none] (148) at (1, -7.5) {};
		\node [style=none] (149) at (1, -7.75) {};
		\node [style=none] (150) at (1, -8) {};
		\node [style=none] (151) at (-11, -6.75) {};
		\node [style=none] (152) at (-11, -4.25) {};
		\node [style=none] (153) at (1, -4.25) {};
		\node [style=none] (154) at (1, -6.75) {};
		\node [style=none] (155) at (-5.25, -9.25) {$B_{6^3_2} =( \sigma_1 \sigma_2^{-1})^3$};
		\node [style=none] (156) at (-13.5, -12) {$M(B_{6^3_2})=$};
		\node [style=none] (157) at (-9, -4.5) {};
		\node [style=none] (158) at (-7, -4.5) {};
		\node [style=none] (159) at (-7, -5.5) {};
		\node [style=none] (160) at (-9, -5.5) {};
		\node [style=none] (161) at (-8.25, -4.75) {};
		\node [style=none] (162) at (-7.75, -5.25) {};
		\node [style=none] (163) at (-7, -4.5) {};
		\node [style=none] (164) at (-5, -4.5) {};
		\node [style=none] (165) at (-5, -5.5) {};
		\node [style=none] (166) at (-7, -5.5) {};
		\node [style=none] (167) at (-6.25, -4.75) {};
		\node [style=none] (168) at (-5.75, -5.25) {};
		\node [style=none] (169) at (-5, -5.5) {};
		\node [style=none] (170) at (-3, -5.5) {};
		\node [style=none] (171) at (-3, -6.5) {};
		\node [style=none] (172) at (-5, -6.5) {};
		\node [style=none] (173) at (-4.25, -6.25) {};
		\node [style=none] (174) at (-3.75, -5.75) {};
		\node [style=none] (175) at (-3, -4.5) {};
		\node [style=none] (176) at (-1, -4.5) {};
		\node [style=none] (177) at (-1, -5.5) {};
		\node [style=none] (178) at (-3, -5.5) {};
		\node [style=none] (179) at (-2.25, -4.75) {};
		\node [style=none] (180) at (-1.75, -5.25) {};
		\node [style=none] (181) at (-1, -5.5) {};
		\node [style=none] (182) at (1, -5.5) {};
		\node [style=none] (183) at (1, -6.5) {};
		\node [style=none] (184) at (-1, -6.5) {};
		\node [style=none] (185) at (-0.25, -6.25) {};
		\node [style=none] (186) at (0.25, -5.75) {};
		\node [style=none] (187) at (-11, -10.25) {};
		\node [style=none] (188) at (-9, -10.25) {};
		\node [style=none] (189) at (-9, -11.25) {};
		\node [style=none] (190) at (-11, -11.25) {};
		\node [style=none] (191) at (-10.25, -11) {};
		\node [style=none] (192) at (-9.75, -10.5) {};
		\node [style=none] (193) at (-9, -11.25) {};
		\node [style=none] (194) at (-7, -11.25) {};
		\node [style=none] (195) at (-7, -12.25) {};
		\node [style=none] (196) at (-9, -12.25) {};
		\node [style=none] (197) at (-8.25, -11.5) {};
		\node [style=none] (198) at (-7.75, -12) {};
		\node [style=none] (199) at (-7, -10.25) {};
		\node [style=none] (200) at (-5, -10.25) {};
		\node [style=none] (201) at (-5, -11.25) {};
		\node [style=none] (202) at (-7, -11.25) {};
		\node [style=none] (203) at (-6.25, -11) {};
		\node [style=none] (204) at (-5.75, -10.5) {};
		\node [style=none] (205) at (-5, -11.25) {};
		\node [style=none] (206) at (-3, -11.25) {};
		\node [style=none] (207) at (-3, -12.25) {};
		\node [style=none] (208) at (-5, -12.25) {};
		\node [style=none] (209) at (-4.25, -11.5) {};
		\node [style=none] (210) at (-3.75, -12) {};
		\node [style=none] (211) at (-3, -10.25) {};
		\node [style=none] (212) at (-1, -10.25) {};
		\node [style=none] (213) at (-1, -11.25) {};
		\node [style=none] (214) at (-3, -11.25) {};
		\node [style=none] (215) at (-2.25, -11) {};
		\node [style=none] (216) at (-1.75, -10.5) {};
		\node [style=none] (217) at (-1, -11.25) {};
		\node [style=none] (218) at (1, -11.25) {};
		\node [style=none] (219) at (1, -12.25) {};
		\node [style=none] (220) at (-1, -12.25) {};
		\node [style=none] (221) at (-0.25, -11.5) {};
		\node [style=none] (222) at (0.25, -12) {};
	\end{pgfonlayer}
	\begin{pgfonlayer}{edgelayer}
		\draw [in=-180, out=0] (14.center) to (16.center);
		\draw [bend left=15, looseness=1.25] (18.center) to (17.center);
		\draw [bend right=15] (15.center) to (19.center);
		\draw [in=-180, out=0] (20.center) to (22.center);
		\draw [bend left=15, looseness=1.25] (24.center) to (23.center);
		\draw [bend right=15] (21.center) to (25.center);
		\draw [in=-180, out=0] (26.center) to (28.center);
		\draw [bend left=15, looseness=1.25] (30.center) to (29.center);
		\draw [bend right=15] (27.center) to (31.center);
		\draw [bend right=90, looseness=0.75] (14.center) to (33.center);
		\draw [bend left=270, looseness=0.75] (17.center) to (32.center);
		\draw (32.center) to (34.center);
		\draw (33.center) to (35.center);
		\draw [bend left=90, looseness=0.75] (27.center) to (35.center);
		\draw [bend left=90, looseness=0.75] (28.center) to (34.center);
		\draw [bend left=90, looseness=0.75] (68.center) to (74.center);
		\draw [bend left=90, looseness=0.75] (67.center) to (73.center);
		\draw [bend left=90, looseness=0.75] (62.center) to (72.center);
		\draw [bend left=90, looseness=0.75] (69.center) to (46.center);
		\draw [bend left=90, looseness=0.75] (70.center) to (49.center);
		\draw [bend left=270, looseness=0.75] (59.center) to (71.center);
		\draw (69.center) to (74.center);
		\draw (73.center) to (70.center);
		\draw (71.center) to (72.center);
		\draw (85.center) to (86.center);
		\draw (87.center) to (88.center);
		\draw [bend left=90] (86.center) to (88.center);
		\draw [bend left=90] (87.center) to (85.center);
		\draw [in=180, out=0] (137.center) to (135.center);
		\draw [bend right=15] (138.center) to (134.center);
		\draw [bend right=15] (139.center) to (136.center);
		\draw (140.center) to (133.center);
		\draw [bend left=90, looseness=0.75] (144.center) to (150.center);
		\draw [bend left=90, looseness=0.75] (143.center) to (149.center);
		\draw [bend left=90, looseness=0.75] (142.center) to (148.center);
		\draw [bend left=90, looseness=0.75] (145.center) to (134.center);
		\draw [bend left=90, looseness=0.75] (146.center) to (137.center);
		\draw [bend left=270, looseness=0.75] (140.center) to (147.center);
		\draw (145.center) to (150.center);
		\draw (149.center) to (146.center);
		\draw (147.center) to (148.center);
		\draw [in=180, out=0] (160.center) to (158.center);
		\draw [bend right=15] (161.center) to (157.center);
		\draw [bend right=15] (162.center) to (159.center);
		\draw [in=180, out=0] (166.center) to (164.center);
		\draw [bend right=15] (167.center) to (163.center);
		\draw [bend right=15] (168.center) to (165.center);
		\draw [in=-180, out=0] (169.center) to (171.center);
		\draw [bend left=15, looseness=1.25] (173.center) to (172.center);
		\draw [bend right=15] (170.center) to (174.center);
		\draw [in=180, out=0] (178.center) to (176.center);
		\draw [bend right=15] (179.center) to (175.center);
		\draw [bend right=15] (180.center) to (177.center);
		\draw [in=-180, out=0] (181.center) to (183.center);
		\draw [bend left=15, looseness=1.25] (185.center) to (184.center);
		\draw [bend right=15] (182.center) to (186.center);
		\draw (164.center) to (175.center);
		\draw (176.center) to (144.center);
		\draw (184.center) to (171.center);
		\draw (172.center) to (133.center);
		\draw [in=-180, out=0] (187.center) to (189.center);
		\draw [bend left=15, looseness=1.25] (191.center) to (190.center);
		\draw [bend right=15] (188.center) to (192.center);
		\draw [in=180, out=0] (196.center) to (194.center);
		\draw [bend right=15] (197.center) to (193.center);
		\draw [bend right=15] (198.center) to (195.center);
		\draw [in=-180, out=0] (199.center) to (201.center);
		\draw [bend left=15, looseness=1.25] (203.center) to (202.center);
		\draw [bend right=15] (200.center) to (204.center);
		\draw [in=180, out=0] (208.center) to (206.center);
		\draw [bend right=15] (209.center) to (205.center);
		\draw [bend right=15] (210.center) to (207.center);
		\draw [in=-180, out=0] (211.center) to (213.center);
		\draw [bend left=15, looseness=1.25] (215.center) to (214.center);
		\draw [bend right=15] (212.center) to (216.center);
		\draw [in=180, out=0] (220.center) to (218.center);
		\draw [bend right=15] (221.center) to (217.center);
		\draw [bend right=15] (222.center) to (219.center);
		\draw (188.center) to (199.center);
		\draw (59.center) to (196.center);
		\draw (195.center) to (208.center);
		\draw (200.center) to (211.center);
		\draw (207.center) to (220.center);
		\draw (212.center) to (68.center);
		\draw [style=new edge style 0] (82.center) to (83.center);
		\draw [style=new edge style 0] (82.center) to (130.center);
		\draw [style=new edge style 0] (130.center) to (131.center);
		\draw [style=new edge style 0] (131.center) to (83.center);
		\draw [style=new edge style 0] (152.center) to (153.center);
		\draw [style=new edge style 0] (153.center) to (154.center);
		\draw [style=new edge style 0] (154.center) to (151.center);
		\draw [style=new edge style 0] (151.center) to (152.center);
		\draw [style=new edge style 0] (76.center) to (77.center);
		\draw [style=new edge style 0] (77.center) to (78.center);
		\draw [style=new edge style 0] (78.center) to (75.center);
		\draw [style=new edge style 0] (75.center) to (76.center);
		\draw [style=new edge style 0] (37.center) to (38.center);
		\draw [style=new edge style 0] (38.center) to (39.center);
		\draw [style=new edge style 0] (39.center) to (36.center);
		\draw [style=new edge style 0] (36.center) to (37.center);
	\end{pgfonlayer}
\end{tikzpicture}

%% file: figs/example-braid-unitary.tikz
\begin{tikzpicture}
	\begin{pgfonlayer}{nodelayer}
		\node [style=none] (0) at (-1.25, 2.75) {};
		\node [style=none] (1) at (-1.5, 1.75) {};
		\node [style=none] (2) at (-1.25, 0.75) {};
		\node [style=none] (3) at (-1.5, -0.25) {};
		\node [style=none] (4) at (-1.5, 3.75) {};
		\node [style=none] (5) at (-1.25, 4.75) {};
		\node [style=none] (6) at (-1.5, 5.75) {};
		\node [style=none] (10) at (4.75, -1.25) {};
		\node [style=none] (11) at (5.25, 3.75) {};
		\node [style=none] (12) at (9.25, 4.75) {};
		\node [style=none] (13) at (7.5, 5.75) {};
		\node [style=none] (14) at (-0.75, 6) {};
		\node [style=none] (15) at (0.75, 6) {};
		\node [style=none] (16) at (-0.75, 1.5) {};
		\node [style=none] (17) at (0.75, 1.5) {};
		\node [style=none] (18) at (-0.75, 5.75) {};
		\node [style=none] (19) at (-0.75, 3.75) {};
		\node [style=none] (20) at (-0.75, 1.75) {};
		\node [style=none] (21) at (0.75, 5.75) {};
		\node [style=none] (22) at (0.75, 3.75) {};
		\node [style=none] (23) at (0.75, 1.75) {};
		\node [style=none] (30) at (1.25, 4) {};
		\node [style=none] (31) at (2.75, 4) {};
		\node [style=none] (32) at (1.25, -0.5) {};
		\node [style=none] (33) at (2.75, -0.5) {};
		\node [style=none] (34) at (1.25, 3.75) {};
		\node [style=none] (35) at (1.25, 1.75) {};
		\node [style=none] (36) at (1.25, -0.25) {};
		\node [style=none] (37) at (2.75, 3.75) {};
		\node [style=none] (38) at (2.75, 1.75) {};
		\node [style=none] (39) at (2.75, -0.25) {};
		\node [style=none] (40) at (1.25, 0.75) {};
		\node [style=none] (41) at (1.25, 2.75) {};
		\node [style=none] (42) at (2.75, 2.75) {};
		\node [style=none] (43) at (2.75, 0.75) {};
		\node [style=none] (46) at (3.25, 6) {};
		\node [style=none] (47) at (4.75, 6) {};
		\node [style=none] (48) at (3.25, 1.5) {};
		\node [style=none] (49) at (4.75, 1.5) {};
		\node [style=none] (50) at (3.25, 5.75) {};
		\node [style=none] (51) at (3.25, 3.75) {};
		\node [style=none] (52) at (3.25, 1.75) {};
		\node [style=none] (53) at (4.75, 5.75) {};
		\node [style=none] (54) at (4.75, 3.75) {};
		\node [style=none] (55) at (4.75, 1.75) {};
		\node [style=none] (56) at (3.25, 2.75) {};
		\node [style=none] (57) at (3.25, 4.75) {};
		\node [style=none] (58) at (4.75, 4.75) {};
		\node [style=none] (59) at (4.75, 2.75) {};
		\node [style=none] (60) at (0, 5.5) {$U_{\sigma_1}$};
		\node [style=none] (61) at (4, 5.5) {$U_{\sigma_1}^\dagger$};
		\node [style=none] (62) at (2, 3.5) {$U_{\sigma_2}$};
		\node [style=none] (63) at (4.75, 2.75) {};
		\node [style=none] (64) at (7.25, 2.75) {};
		\node [style=none] (65) at (2.75, 0.75) {};
		\node [style=none] (66) at (5.25, 0.75) {};
		\node [style=none] (67) at (4.75, 1.75) {};
		\node [style=none] (68) at (5.25, 1.75) {};
		\node [style=none] (69) at (4.75, 3.75) {};
		\node [style=none] (70) at (7.25, 3.75) {};
		\node [style=none] (71) at (2.75, -0.25) {};
		\node [style=none] (72) at (4.75, -0.25) {};
		\node [style=none] (73) at (3.25, 4.75) {};
		\node [style=none] (74) at (3.25, 4.75) {};
		\node [style=none] (75) at (4.75, 2.75) {};
		\node [style=none] (76) at (3.25, 2.75) {};
		\node [style=none] (79) at (4.75, 4.75) {};
		\node [style=none] (88) at (-3.75, -0.25) {};
		\node [style=none] (91) at (-3.75, 3.75) {};
		\node [style=none] (93) at (-3.75, 5.75) {};
		\node [style=none] (101) at (-6.25, 5.75) {};
		\node [style=none] (117) at (-6.25, 3.75) {};
		\node [style=none] (118) at (-6.25, -0.25) {};
		\node [style=none] (126) at (-5.75, 6) {};
		\node [style=none] (127) at (-4.25, 6) {};
		\node [style=none] (128) at (-5.75, -2.5) {};
		\node [style=none] (129) at (-4.25, -2.5) {};
		\node [style=none] (130) at (-5.75, 5.75) {};
		\node [style=none] (131) at (-5.75, 3.75) {};
		\node [style=none] (132) at (-5.75, -0.25) {};
		\node [style=none] (133) at (-4.25, 5.75) {};
		\node [style=none] (134) at (-4.25, 3.75) {};
		\node [style=none] (135) at (-4.25, -0.25) {};
		\node [style=none] (141) at (-5, 1.75) {$U_B$};
		\node [style=none] (147) at (-4.25, -0.25) {};
		\node [style=none] (148) at (-3.75, -0.25) {};
		\node [style=none] (149) at (-4.25, 3.75) {};
		\node [style=none] (150) at (-3.75, 3.75) {};
		\node [style=none] (154) at (-3.75, 1.75) {};
		\node [style=none] (155) at (-6.25, 1.75) {};
		\node [style=none] (156) at (-5.75, 1.75) {};
		\node [style=none] (157) at (-4.25, 1.75) {};
		\node [style=none] (158) at (-4.25, 1.75) {};
		\node [style=none] (159) at (-3.75, 1.75) {};
		\node [style=none] (160) at (-2.75, 1.75) {$=$};
		\node [style=none] (161) at (-0.75, 4.75) {};
		\node [style=none] (162) at (-0.75, 2.75) {};
		\node [style=none] (163) at (0.75, 2.75) {};
		\node [style=none] (164) at (0.75, 4.75) {};
		\node [style=none] (165) at (-0.125, 3.35) {};
		\node [style=none] (166) at (0.125, 3.925) {};
		\node [style=none] (167) at (0.75, 2.75) {};
		\node [style=none] (168) at (0.75, 2.75) {};
		\node [style=none] (169) at (0.75, 4.75) {};
		\node [style=none] (170) at (0.75, 4.75) {};
		\node [style=none] (173) at (3.25, 2.75) {};
		\node [style=none] (174) at (3.25, 4.75) {};
		\node [style=none] (175) at (4.75, 4.75) {};
		\node [style=none] (176) at (4.75, 2.75) {};
		\node [style=none] (177) at (3.875, 4.1) {};
		\node [style=none] (178) at (4.125, 3.425) {};
		\node [style=none] (179) at (4.75, 2.75) {};
		\node [style=none] (180) at (-0.75, 4.75) {};
		\node [style=none] (181) at (-1.25, 4.75) {};
		\node [style=none] (182) at (-1.25, 4.75) {};
		\node [style=none] (183) at (-0.75, 2.75) {};
		\node [style=none] (184) at (-1.25, 2.75) {};
		\node [style=none] (185) at (-1.25, 2.75) {};
		\node [style=none] (186) at (3.25, 4.75) {};
		\node [style=none] (187) at (0.75, 4.75) {};
		\node [style=none] (188) at (0.75, 4.75) {};
		\node [style=none] (189) at (1.25, 2.75) {};
		\node [style=none] (190) at (0.75, 2.75) {};
		\node [style=none] (191) at (0.75, 2.75) {};
		\node [style=none] (192) at (1.25, 1.75) {};
		\node [style=none] (193) at (2.75, 1.75) {};
		\node [style=none] (194) at (1.25, 2.75) {};
		\node [style=none] (195) at (1.25, 0.75) {};
		\node [style=none] (196) at (2.75, 0.75) {};
		\node [style=none] (197) at (2.75, 2.75) {};
		\node [style=none] (198) at (1.875, 1.35) {};
		\node [style=none] (199) at (2.125, 1.925) {};
		\node [style=none] (200) at (2.75, 0.75) {};
		\node [style=none] (201) at (2.75, 0.75) {};
		\node [style=none] (202) at (2.75, 2.75) {};
		\node [style=none] (203) at (2.75, 2.75) {};
		\node [style=none] (204) at (1.25, 2.75) {};
		\node [style=none] (205) at (1.25, 0.75) {};
		\node [style=none] (206) at (2.75, 2.75) {};
		\node [style=none] (207) at (2.75, 2.75) {};
		\node [style=none] (208) at (2.75, 0.75) {};
		\node [style=none] (209) at (2.75, 0.75) {};
		\node [style=none] (210) at (7.25, -1.25) {};
		\node [style=none] (211) at (7.5, -2.25) {};
		\node [style=none] (212) at (7.25, -0.25) {};
		\node [style=none] (213) at (7.25, 0.75) {};
		\node [style=none] (214) at (7.25, 1.75) {};
		\node [style=none] (215) at (4.75, -0.25) {};
		\node [style=none] (216) at (-1.5, -2.25) {};
		\node [style=none] (217) at (-1.25, -1.25) {};
		\node [style=none] (218) at (5.25, 2) {};
		\node [style=none] (219) at (6.75, 2) {};
		\node [style=none] (220) at (5.25, -2.5) {};
		\node [style=none] (221) at (6.75, -2.5) {};
		\node [style=none] (222) at (5.25, 1.75) {};
		\node [style=none] (223) at (5.25, -0.25) {};
		\node [style=none] (224) at (5.25, -2.25) {};
		\node [style=none] (225) at (6.75, 1.75) {};
		\node [style=none] (226) at (6.75, -0.25) {};
		\node [style=none] (227) at (6.75, -2.25) {};
		\node [style=none] (228) at (5.25, -1.25) {};
		\node [style=none] (229) at (5.25, 0.75) {};
		\node [style=none] (230) at (6.75, 0.75) {};
		\node [style=none] (231) at (6.75, -1.25) {};
		\node [style=none] (232) at (6, 1.5) {$U_{\sigma_3}$};
		\node [style=none] (233) at (6.75, -1.25) {};
		\node [style=none] (234) at (11.25, -1.25) {};
		\node [style=none] (235) at (6.75, -2.25) {};
		\node [style=none] (236) at (11.5, -2.25) {};
		\node [style=none] (237) at (6.75, -0.25) {};
		\node [style=none] (238) at (7.25, -0.25) {};
		\node [style=none] (239) at (5.25, 0.75) {};
		\node [style=none] (240) at (5.25, 0.75) {};
		\node [style=none] (241) at (6.75, -1.25) {};
		\node [style=none] (242) at (5.25, -1.25) {};
		\node [style=none] (243) at (6.75, 0.75) {};
		\node [style=none] (244) at (5.25, -1.25) {};
		\node [style=none] (245) at (5.25, 0.75) {};
		\node [style=none] (246) at (6.75, 0.75) {};
		\node [style=none] (247) at (6.75, -1.25) {};
		\node [style=none] (250) at (6.75, -1.25) {};
		\node [style=none] (251) at (5.25, 0.75) {};
		\node [style=none] (252) at (-3.75, -2.25) {};
		\node [style=none] (253) at (-6.25, -2.25) {};
		\node [style=none] (254) at (-5.75, -2.25) {};
		\node [style=none] (255) at (-4.25, -2.25) {};
		\node [style=none] (256) at (-4.25, -2.25) {};
		\node [style=none] (257) at (-3.75, -2.25) {};
		\node [style=none] (258) at (5.25, -0.25) {};
		\node [style=none] (259) at (6.75, -0.25) {};
		\node [style=none] (260) at (5.25, 0.75) {};
		\node [style=none] (261) at (5.25, -1.25) {};
		\node [style=none] (262) at (6.75, -1.25) {};
		\node [style=none] (263) at (6.75, 0.75) {};
		\node [style=none] (264) at (5.875, -0.65) {};
		\node [style=none] (265) at (6.125, -0.075) {};
		\node [style=none] (266) at (6.75, -1.25) {};
		\node [style=none] (267) at (6.75, -1.25) {};
		\node [style=none] (268) at (6.75, 0.75) {};
		\node [style=none] (269) at (6.75, 0.75) {};
		\node [style=none] (270) at (5.25, 0.75) {};
		\node [style=none] (271) at (5.25, -1.25) {};
		\node [style=none] (272) at (6.75, 0.75) {};
		\node [style=none] (273) at (6.75, 0.75) {};
		\node [style=none] (274) at (6.75, -1.25) {};
		\node [style=none] (275) at (6.75, -1.25) {};
		\node [style=none] (276) at (7.25, 4) {};
		\node [style=none] (277) at (8.75, 4) {};
		\node [style=none] (278) at (7.25, -0.5) {};
		\node [style=none] (279) at (8.75, -0.5) {};
		\node [style=none] (280) at (7.25, 3.75) {};
		\node [style=none] (281) at (7.25, 1.75) {};
		\node [style=none] (282) at (7.25, -0.25) {};
		\node [style=none] (283) at (8.75, 3.75) {};
		\node [style=none] (284) at (8.75, 1.75) {};
		\node [style=none] (285) at (8.75, -0.25) {};
		\node [style=none] (286) at (7.25, 0.75) {};
		\node [style=none] (287) at (7.25, 2.75) {};
		\node [style=none] (288) at (8.75, 2.75) {};
		\node [style=none] (289) at (8.75, 0.75) {};
		\node [style=none] (290) at (8, 3.5) {$U_{\sigma_2}$};
		\node [style=none] (291) at (8.75, 0.75) {};
		\node [style=none] (292) at (8.75, -0.25) {};
		\node [style=none] (293) at (7.25, 2.75) {};
		\node [style=none] (294) at (7.25, 1.75) {};
		\node [style=none] (295) at (8.75, 1.75) {};
		\node [style=none] (296) at (7.25, 2.75) {};
		\node [style=none] (297) at (7.25, 0.75) {};
		\node [style=none] (298) at (8.75, 0.75) {};
		\node [style=none] (299) at (8.75, 2.75) {};
		\node [style=none] (300) at (7.875, 1.35) {};
		\node [style=none] (301) at (8.125, 1.925) {};
		\node [style=none] (302) at (8.75, 0.75) {};
		\node [style=none] (303) at (8.75, 0.75) {};
		\node [style=none] (304) at (8.75, 2.75) {};
		\node [style=none] (305) at (8.75, 2.75) {};
		\node [style=none] (306) at (7.25, 2.75) {};
		\node [style=none] (307) at (7.25, 0.75) {};
		\node [style=none] (308) at (8.75, 2.75) {};
		\node [style=none] (309) at (8.75, 2.75) {};
		\node [style=none] (310) at (8.75, 0.75) {};
		\node [style=none] (311) at (11.25, 0.75) {};
		\node [style=none] (312) at (9.25, 6) {};
		\node [style=none] (313) at (10.75, 6) {};
		\node [style=none] (314) at (9.25, 1.5) {};
		\node [style=none] (315) at (10.75, 1.5) {};
		\node [style=none] (316) at (9.25, 5.75) {};
		\node [style=none] (317) at (9.25, 3.75) {};
		\node [style=none] (318) at (9.25, 1.75) {};
		\node [style=none] (319) at (10.75, 5.75) {};
		\node [style=none] (320) at (10.75, 3.75) {};
		\node [style=none] (321) at (10.75, 1.75) {};
		\node [style=none] (322) at (9.25, 2.75) {};
		\node [style=none] (323) at (9.25, 4.75) {};
		\node [style=none] (324) at (10.75, 4.75) {};
		\node [style=none] (325) at (10.75, 2.75) {};
		\node [style=none] (326) at (10, 5.5) {$U_{\sigma_1}^\dagger$};
		\node [style=none] (327) at (10.75, 2.75) {};
		\node [style=none] (328) at (10.75, 1.75) {};
		\node [style=none] (329) at (10.75, 3.75) {};
		\node [style=none] (330) at (9.25, 4.75) {};
		\node [style=none] (331) at (9.25, 4.75) {};
		\node [style=none] (332) at (10.75, 2.75) {};
		\node [style=none] (333) at (9.25, 2.75) {};
		\node [style=none] (334) at (10.75, 4.75) {};
		\node [style=none] (335) at (9.25, 2.75) {};
		\node [style=none] (336) at (9.25, 4.75) {};
		\node [style=none] (337) at (10.75, 4.75) {};
		\node [style=none] (338) at (10.75, 2.75) {};
		\node [style=none] (339) at (9.875, 4.1) {};
		\node [style=none] (340) at (10.125, 3.425) {};
		\node [style=none] (341) at (10.75, 2.75) {};
		\node [style=none] (342) at (9.25, 4.75) {};
		\node [style=none] (343) at (11.5, 5.75) {};
		\node [style=none] (344) at (11.5, 3.75) {};
		\node [style=none] (345) at (11.5, 1.75) {};
		\node [style=none] (346) at (11.5, -0.25) {};
		\node [style=none] (347) at (11.25, 4.75) {};
		\node [style=none] (348) at (11.25, 2.75) {};
	\end{pgfonlayer}
	\begin{pgfonlayer}{edgelayer}
		\draw (14.center) to (15.center);
		\draw (15.center) to (17.center);
		\draw (17.center) to (16.center);
		\draw (16.center) to (14.center);
		\draw (20.center) to (1.center);
		\draw (4.center) to (19.center);
		\draw (18.center) to (6.center);
		\draw (30.center) to (31.center);
		\draw (31.center) to (33.center);
		\draw (33.center) to (32.center);
		\draw (32.center) to (30.center);
		\draw (46.center) to (47.center);
		\draw (47.center) to (49.center);
		\draw (49.center) to (48.center);
		\draw (48.center) to (46.center);
		\draw (36.center) to (3.center);
		\draw (23.center) to (35.center);
		\draw (38.center) to (52.center);
		\draw (51.center) to (37.center);
		\draw (34.center) to (22.center);
		\draw (21.center) to (50.center);
		\draw (53.center) to (13.center);
		\draw [style=blue edge] (40.center) to (2.center);
		\draw [style=blue edge] (42.center) to (56.center);
		\draw [style=blue edge] (12.center) to (58.center);
		\draw [style=blue edge] (63.center) to (64.center);
		\draw [style=blue edge] (65.center) to (66.center);
		\draw (67.center) to (68.center);
		\draw (69.center) to (70.center);
		\draw (72.center) to (71.center);
		\draw (126.center) to (127.center);
		\draw (127.center) to (129.center);
		\draw (129.center) to (128.center);
		\draw (128.center) to (126.center);
		\draw (118.center) to (132.center);
		\draw (131.center) to (117.center);
		\draw (101.center) to (130.center);
		\draw (133.center) to (93.center);
		\draw (147.center) to (148.center);
		\draw (149.center) to (150.center);
		\draw (156.center) to (155.center);
		\draw (158.center) to (159.center);
		\draw [style=blue edge, in=-180, out=0, looseness=0.75] (161.center) to (163.center);
		\draw [style=blue edge, in=-15, out=-120, looseness=0.75] (165.center) to (162.center);
		\draw [style=blue edge, in=-180, out=75] (166.center) to (164.center);
		\draw [style=blue edge, in=-180, out=0, looseness=0.75] (173.center) to (175.center);
		\draw [style=blue edge, in=-15, out=120] (177.center) to (174.center);
		\draw [style=blue edge, in=180, out=-60] (178.center) to (176.center);
		\draw [style=blue edge] (180.center) to (181.center);
		\draw [style=blue edge] (183.center) to (184.center);
		\draw [style=blue edge] (186.center) to (187.center);
		\draw [style=blue edge] (189.center) to (190.center);
		\draw [style=blue edge, in=-180, out=0, looseness=0.75] (194.center) to (196.center);
		\draw [style=blue edge, in=-15, out=-120, looseness=0.75] (198.center) to (195.center);
		\draw [style=blue edge, in=-180, out=75] (199.center) to (197.center);
		\draw (218.center) to (219.center);
		\draw (219.center) to (221.center);
		\draw (221.center) to (220.center);
		\draw (220.center) to (218.center);
		\draw (216.center) to (224.center);
		\draw (223.center) to (215.center);
		\draw (225.center) to (214.center);
		\draw [style=blue edge] (217.center) to (228.center);
		\draw [style=blue edge] (213.center) to (230.center);
		\draw [style=blue edge] (233.center) to (234.center);
		\draw (235.center) to (236.center);
		\draw (237.center) to (238.center);
		\draw (253.center) to (254.center);
		\draw (256.center) to (257.center);
		\draw [style=blue edge, in=-180, out=0, looseness=0.75] (260.center) to (262.center);
		\draw [style=blue edge, in=-15, out=-120, looseness=0.75] (264.center) to (261.center);
		\draw [style=blue edge, in=-180, out=75] (265.center) to (263.center);
		\draw (276.center) to (277.center);
		\draw (277.center) to (279.center);
		\draw (279.center) to (278.center);
		\draw (278.center) to (276.center);
		\draw [style=blue edge, in=-180, out=0, looseness=0.75] (296.center) to (298.center);
		\draw [style=blue edge, in=-15, out=-120, looseness=0.75] (300.center) to (297.center);
		\draw [style=blue edge, in=-180, out=75] (301.center) to (299.center);
		\draw (312.center) to (313.center);
		\draw (313.center) to (315.center);
		\draw (315.center) to (314.center);
		\draw (314.center) to (312.center);
		\draw [style=blue edge, in=-180, out=0, looseness=0.75] (335.center) to (337.center);
		\draw [style=blue edge, in=-15, out=120] (339.center) to (336.center);
		\draw [style=blue edge, in=180, out=-60] (340.center) to (338.center);
		\draw (317.center) to (283.center);
		\draw (318.center) to (295.center);
		\draw (316.center) to (13.center);
		\draw [style=blue edge] (309.center) to (335.center);
		\draw (346.center) to (292.center);
		\draw (345.center) to (328.center);
		\draw (344.center) to (329.center);
		\draw (343.center) to (319.center);
		\draw [style=blue edge] (310.center) to (311.center);
		\draw [style=blue edge] (347.center) to (337.center);
		\draw [style=blue edge] (348.center) to (341.center);
	\end{pgfonlayer}
\end{tikzpicture}

%% file: figs/braid-generator-circuit1.tikz
\begin{tikzpicture}
	\begin{pgfonlayer}{nodelayer}
		\node [style=none] (120) at (-13.5, 0) {};
		\node [style=gate] (121) at (-3.25, 0) {$R_Z$};
		\node [style=none] (128) at (-13.5, -2) {};
		\node [style=none] (131) at (-13.5, -4) {};
		\node [style=gate] (132) at (-3.25, -4) {$R_Z$};
		\node [style=none] (155) at (-2.25, -4) {};
		\node [style=none] (156) at (-2.25, -2) {};
		\node [style=none] (157) at (-2.25, 0) {};
		\node [style=none] (162) at (-17, 0) {};
		\node [style=none] (163) at (-17, 0.25) {};
		\node [style=none] (164) at (-15.75, 0.25) {};
		\node [style=none] (165) at (-17, -4.25) {};
		\node [style=none] (166) at (-15.75, -4.25) {};
		\node [style=none] (167) at (-16.25, -2) {$U_{\sigma_i}$};
		\node [style=none] (168) at (-17, -4) {};
		\node [style=none] (169) at (-15.75, 0) {};
		\node [style=none] (170) at (-15.75, -4) {};
		\node [style=none] (171) at (-15.5, 0) {};
		\node [style=none] (172) at (-17.25, 0) {};
		\node [style=none] (173) at (-17, -4) {};
		\node [style=none] (174) at (-15.75, -4) {};
		\node [style=none] (175) at (-15.5, -4) {};
		\node [style=none] (176) at (-17.25, -4) {};
		\node [style=none] (177) at (-17, -2) {};
		\node [style=none] (178) at (-15.75, -2) {};
		\node [style=none] (179) at (-15.5, -2) {};
		\node [style=none] (180) at (-17.25, -2) {};
		\node [style=none] (181) at (-14.5, -2) {$=$};
		\node [style=cnot ctrl] (182) at (-11.25, -2) {};
		\node [style=cnot ctrl] (187) at (-11.25, -4) {};
		\node [style=cnot ctrl] (190) at (-8.75, 0) {};
		\node [style=cnot ctrl] (191) at (-8.75, -2) {};
		\node [style=cnot ctrl] (196) at (-6.25, -2) {};
		\node [style=cnot ctrl] (197) at (-6.25, -4) {};
		\node [style=gate] (201) at (-12.5, -2) {$U_\mathrm{1q}$};
		\node [style=gate] (206) at (-10, -2) {$U_\mathrm{1q}$};
		\node [style=gate] (212) at (-7.5, -2) {$U_\mathrm{1q}$};
		\node [style=gate] (216) at (-5, -2) {$U_\mathrm{1q}$};
		\node [style=gate] (218) at (-3.25, -2) {$R_Z$};
		\node [style=none] (223) at (-3.25, -1.25) {\tiny $\theta_1$};
		\node [style=none] (224) at (-3.25, 0.75) {\tiny $\theta_0$};
		\node [style=none] (225) at (-3.25, -3.25) {\tiny $\theta_2$};
		\node [style=none] (226) at (-12.5, -1.25) {\tiny $\alpha_0,\beta_0$};
		\node [style=none] (227) at (-10, -1.25) {\tiny $\alpha_1,\beta_1$};
		\node [style=none] (228) at (-7.5, -1.25) {\tiny $\alpha_2,\beta_2$};
		\node [style=none] (229) at (-5, -1.25) {\tiny $\alpha_3,\beta_3$};
		\node [style=none] (230) at (-5.75, -3) {\tiny $\chi_2$};
		\node [style=none] (231) at (-10.75, -3) {\tiny $\chi_0$};
		\node [style=none] (232) at (-8.75, 0.5) {\tiny $\chi_1$};
	\end{pgfonlayer}
	\begin{pgfonlayer}{edgelayer}
		\draw (163.center) to (164.center);
		\draw (164.center) to (166.center);
		\draw (166.center) to (165.center);
		\draw (165.center) to (163.center);
		\draw (172.center) to (162.center);
		\draw (169.center) to (171.center);
		\draw (176.center) to (173.center);
		\draw (174.center) to (175.center);
		\draw (180.center) to (177.center);
		\draw (178.center) to (179.center);
		\draw (182) to (187);
		\draw (190) to (191);
		\draw (196) to (197);
		\draw (201) to (182);
		\draw (206) to (191);
		\draw (212) to (196);
		\draw (216) to (218);
		\draw (218) to (156.center);
		\draw (155.center) to (132);
		\draw (132) to (197);
		\draw (197) to (187);
		\draw (187) to (131.center);
		\draw (157.center) to (121);
		\draw (121) to (190);
		\draw (190) to (120.center);
		\draw (216) to (196);
		\draw (212) to (191);
		\draw (206) to (182);
		\draw (201) to (128.center);
	\end{pgfonlayer}
\end{tikzpicture}

%% file: figs/H-test-CFEV.tikz
\begin{tikzpicture}
	\begin{pgfonlayer}{nodelayer}
		\node [style=none] (0) at (-1.5, 2.75) {};
		\node [style=none] (1) at (-1.75, 0.75) {};
		\node [style=none] (2) at (-1.5, 0.75) {};
		\node [style=none] (3) at (-1.75, -0.25) {};
		\node [style=none] (4) at (-1.75, 3) {};
		\node [style=none] (5) at (-1.5, 3) {};
		\node [style=none] (6) at (-1.75, 4) {};
		\node [style=none] (7) at (0, 2.75) {};
		\node [style=none] (8) at (0.25, 0.75) {};
		\node [style=none] (9) at (0, 0.75) {};
		\node [style=none] (10) at (0.25, -0.25) {};
		\node [style=none] (11) at (0.25, 3) {};
		\node [style=none] (12) at (0, 3) {};
		\node [style=none] (13) at (0.25, 4) {};
		\node [style=none] (190) at (-3.75, 4.25) {};
		\node [style=none] (191) at (-1.75, 4.25) {};
		\node [style=none] (192) at (-3.75, -0.5) {};
		\node [style=none] (193) at (-1.75, -0.5) {};
		\node [style=none] (194) at (-2.75, 2) {$V_\mathrm{cat}(s)$};
		\node [style=none] (226) at (2.25, 2.75) {};
		\node [style=none] (227) at (2.25, 0.75) {};
		\node [style=none] (228) at (2.25, 0.75) {};
		\node [style=none] (229) at (2.25, -0.25) {};
		\node [style=none] (230) at (2.25, 3) {};
		\node [style=none] (231) at (2.25, 3) {};
		\node [style=none] (232) at (2.25, 4) {};
		\node [style=none] (233) at (0.25, 4.25) {};
		\node [style=none] (234) at (2.25, 4.25) {};
		\node [style=none] (235) at (0.25, -0.5) {};
		\node [style=none] (236) at (2.25, -0.5) {};
		\node [style=none] (237) at (1.25, 2) {$V_\mathrm{cat}^\dagger(s)$};
		\node [style=none] (238) at (3.25, 0.75) {};
		\node [style=none] (239) at (3.25, -0.25) {};
		\node [style=none] (240) at (3.25, 3) {};
		\node [style=none] (242) at (2.25, 4) {};
		\node [style=none] (243) at (2.25, 3) {};
		\node [style=none] (244) at (2.25, 0.75) {};
		\node [style=none] (245) at (2.25, -0.25) {};
		\node [style=gate] (246) at (4.5, 4) {$M_Z$};
		\node [style=none] (247) at (5.75, 4) {};
		\node [style=gate] (248) at (3.25, 3) {$M_Z$};
		\node [style=none] (249) at (4, 3) {};
		\node [style=gate] (250) at (3.25, 0.75) {$M_Z$};
		\node [style=none] (251) at (4, 0.75) {};
		\node [style=gate] (252) at (3.25, -0.25) {$M_Z$};
		\node [style=none] (253) at (4, -0.25) {};
		\node [style=none] (254) at (-3.75, 4) {};
		\node [style=none] (256) at (-3.75, 3) {};
		\node [style=none] (257) at (-8.75, 3) {};
		\node [style=none] (258) at (-3.75, 0.75) {};
		\node [style=none] (259) at (-8.75, 0.75) {};
		\node [style=none] (260) at (-3.75, -0.25) {};
		\node [style=none] (261) at (-8.75, -0.25) {};
		\node [style=none] (262) at (4.25, 3.25) {};
		\node [style=none] (263) at (5.25, 3.25) {};
		\node [style=none] (264) at (4.25, -0.5) {};
		\node [style=none] (265) at (5.25, -0.5) {};
		\node [style=none] (266) at (5.5, 2) {};
		\node [style=none] (267) at (6, 2) {};
		\node [style=none] (268) at (4.75, 1.5) {$g$};
		\node [style=none] (269) at (6.25, 2.25) {};
		\node [style=none] (270) at (6.25, 3.5) {};
		\node [style=none] (271) at (6.75, 4) {};
		\node [style=none] (272) at (7, 4) {};
		\node [style=none] (273) at (6.25, 4) {$\otimes$};
		\node [style=none] (274) at (-9.25, 4) {\small $|+\rangle$};
		\node [style=none] (275) at (-9.25, 3) {\small $|0\rangle$};
		\node [style=none] (276) at (-9.25, 0.75) {\small $|0\rangle$};
		\node [style=none] (277) at (-9.25, -0.25) {\small $|0\rangle$};
		\node [style=gate] (278) at (-5, 4) {$(S^\dagger)^b$};
		\node [style=none] (281) at (0, 2.75) {};
		\node [style=none] (283) at (0, 0.75) {};
		\node [style=none] (285) at (0.25, 0.75) {};
		\node [style=none] (287) at (0.25, 3) {};
		\node [style=none] (289) at (0.25, -0.25) {};
		\node [style=none] (290) at (0.25, -0.25) {};
		\node [style=none] (291) at (0.25, 3) {};
		\node [style=none] (292) at (0.25, 4) {};
		\node [style=none] (293) at (-1.75, 4) {};
		\node [style=none] (294) at (-1.75, 3) {};
		\node [style=none] (295) at (-1.75, -0.25) {};
		\node [style=none] (296) at (-1.25, 4.25) {};
		\node [style=none] (297) at (-0.25, 4.25) {};
		\node [style=none] (298) at (-1.25, -0.5) {};
		\node [style=none] (299) at (-0.25, -0.5) {};
		\node [style=none] (300) at (-1.25, 4) {};
		\node [style=none] (301) at (-1.25, 3) {};
		\node [style=none] (302) at (-1.25, -0.25) {};
		\node [style=none] (303) at (-0.25, 4) {};
		\node [style=none] (304) at (-0.25, 3) {};
		\node [style=none] (305) at (-0.25, -0.25) {};
		\node [style=none] (306) at (-0.75, 2) {$U_B$};
		\node [style=none] (307) at (-0.25, -0.25) {};
		\node [style=none] (308) at (0.25, -0.25) {};
		\node [style=none] (309) at (-0.25, 3) {};
		\node [style=none] (310) at (0.25, 3) {};
		\node [style=none] (311) at (0.25, 0.75) {};
		\node [style=none] (312) at (-1.75, 0.75) {};
		\node [style=none] (313) at (-1.25, 0.75) {};
		\node [style=none] (314) at (-0.25, 0.75) {};
		\node [style=none] (315) at (-0.25, 0.75) {};
		\node [style=none] (316) at (0.25, 0.75) {};
		\node [style=none] (317) at (-5, 2) {$\vdots$};
		\node [style=none] (318) at (3.25, 2) {$\vdots$};
		\node [style=none] (319) at (-8.75, 4) {};
		\node [style=gate] (320) at (-7.25, 4) {$R_Z(\theta_B)$};
		\node [style=gate] (321) at (3, 4) {$H$};
		\node [style=none] (322) at (5.5, 1) {};
		\node [style=none] (323) at (6.75, 1) {};
		\node [style=none] (324) at (7, 0.25) {$g_2$};
		\node [style=none] (325) at (7, 2.5) {$g_1$};
		\node [style=none] (326) at (5.5, 4.5) {$z$};
		\node [style=none] (327) at (7.5, 4) {$r$};
	\end{pgfonlayer}
	\begin{pgfonlayer}{edgelayer}
		\draw (190.center) to (191.center);
		\draw (191.center) to (193.center);
		\draw (193.center) to (192.center);
		\draw (192.center) to (190.center);
		\draw (233.center) to (234.center);
		\draw (234.center) to (236.center);
		\draw (236.center) to (235.center);
		\draw (235.center) to (233.center);
		\draw (239.center) to (245.center);
		\draw (244.center) to (238.center);
		\draw (240.center) to (243.center);
		\draw [style=double edge] (246) to (247.center);
		\draw [style=double edge] (248) to (249.center);
		\draw [style=double edge] (250) to (251.center);
		\draw [style=double edge] (252) to (253.center);
		\draw (256.center) to (257.center);
		\draw (258.center) to (259.center);
		\draw (260.center) to (261.center);
		\draw (262.center) to (263.center);
		\draw (263.center) to (265.center);
		\draw (265.center) to (264.center);
		\draw (264.center) to (262.center);
		\draw [style=double edge] (266.center) to (267.center);
		\draw [style=double edge] (269.center) to (270.center);
		\draw [style=double edge] (271.center) to (272.center);
		\draw (254.center) to (278);
		\draw (296.center) to (297.center);
		\draw (297.center) to (299.center);
		\draw (299.center) to (298.center);
		\draw (298.center) to (296.center);
		\draw (295.center) to (302.center);
		\draw (301.center) to (294.center);
		\draw (293.center) to (300.center);
		\draw (303.center) to (292.center);
		\draw (307.center) to (308.center);
		\draw (309.center) to (310.center);
		\draw (313.center) to (312.center);
		\draw (315.center) to (316.center);
		\draw (319.center) to (320);
		\draw (320) to (278);
		\draw (242.center) to (321);
		\draw (321) to (246);
		\draw [style=double edge] (322.center) to (323.center);
	\end{pgfonlayer}
\end{tikzpicture}

%% file: figs/catstate-unitary.tikz
\begin{tikzpicture}
	\begin{pgfonlayer}{nodelayer}
		\node [style=none] (254) at (-2.5, 4) {};
		\node [style=none] (256) at (-2.5, 3) {};
		\node [style=none] (257) at (-5.75, 3) {};
		\node [style=none] (258) at (-2.5, 2) {};
		\node [style=none] (259) at (-5.75, 2) {};
		\node [style=none] (260) at (-2.5, 0) {};
		\node [style=none] (261) at (-5.75, 0) {};
		\node [style=none] (279) at (-3.5, 4) {};
		\node [style=none] (317) at (-4.25, 1) {$\vdots$};
		\node [style=cnot ctrl] (319) at (-2.5, 4) {};
		\node [style=gate] (320) at (-2.5, 3) {$X^{s_2}$};
		\node [style=none] (321) at (-2.5, 5) {};
		\node [style=none] (322) at (-3.5, 5) {};
		\node [style=none] (323) at (-3.5, 4) {};
		\node [style=none] (324) at (-5, 5) {};
		\node [style=none] (325) at (-3.5, 5) {};
		\node [style=none] (326) at (-5, 4) {};
		\node [style=none] (327) at (-5, 4) {};
		\node [style=none] (328) at (-5.75, 4) {};
		\node [style=none] (329) at (-5, 5) {};
		\node [style=none] (330) at (-5.75, 5) {};
		\node [style=none] (332) at (-2.5, 2) {};
		\node [style=cnot ctrl] (333) at (-1.25, 4) {};
		\node [style=gate] (334) at (-1.25, 2) {$X^{s_3}$};
		\node [style=none] (336) at (-2.5, 0) {};
		\node [style=cnot ctrl] (337) at (0.25, 4) {};
		\node [style=gate] (338) at (0.25, 0) {$X^{s_{n-1}}$};
		\node [style=none] (339) at (1.75, 4) {};
		\node [style=none] (340) at (1.75, 3) {};
		\node [style=none] (341) at (1.75, 2) {};
		\node [style=none] (342) at (1.75, 0) {};
		\node [style=none] (343) at (1.75, 5) {};
		\node [style=none] (345) at (-8.5, 2) {};
		\node [style=none] (346) at (-8.25, 2) {};
		\node [style=none] (347) at (-8.5, 0) {};
		\node [style=none] (348) at (-8.5, 4) {};
		\node [style=none] (349) at (-8.25, 4) {};
		\node [style=none] (350) at (-8.5, 5) {};
		\node [style=none] (351) at (-10.5, 5.25) {};
		\node [style=none] (352) at (-8.5, 5.25) {};
		\node [style=none] (353) at (-10.5, -0.25) {};
		\node [style=none] (354) at (-8.5, -0.25) {};
		\node [style=none] (355) at (-9.5, 2.5) {$V_\mathrm{cat}(s)$};
		\node [style=none] (356) at (-10.5, 5) {};
		\node [style=none] (357) at (-10.5, 4) {};
		\node [style=none] (358) at (-11.25, 4) {};
		\node [style=none] (359) at (-10.5, 2) {};
		\node [style=none] (360) at (-11.25, 2) {};
		\node [style=none] (361) at (-10.5, 0) {};
		\node [style=none] (362) at (-11.25, 0) {};
		\node [style=none] (363) at (-11.25, 5) {};
		\node [style=none] (364) at (-8.5, 5) {};
		\node [style=none] (365) at (-8.5, 4) {};
		\node [style=none] (366) at (-8.5, 0) {};
		\node [style=none] (367) at (-7.75, 5) {};
		\node [style=none] (368) at (-7.75, 4) {};
		\node [style=none] (369) at (-7.75, 0) {};
		\node [style=none] (370) at (-8.5, 2) {};
		\node [style=none] (371) at (-7.75, 2) {};
		\node [style=none] (372) at (-8.5, 3) {};
		\node [style=none] (373) at (-8.25, 3) {};
		\node [style=none] (374) at (-10.5, 3) {};
		\node [style=none] (375) at (-11.25, 3) {};
		\node [style=none] (376) at (-8.5, 3) {};
		\node [style=none] (377) at (-7.75, 3) {};
		\node [style=none] (378) at (-6.75, 2.75) {$=$};
	\end{pgfonlayer}
	\begin{pgfonlayer}{edgelayer}
		\draw (256.center) to (257.center);
		\draw (258.center) to (259.center);
		\draw (260.center) to (261.center);
		\draw (279.center) to (254.center);
		\draw (322.center) to (321.center);
		\draw [in=-180, out=0] (324.center) to (323.center);
		\draw [in=-180, out=0] (326.center) to (325.center);
		\draw (328.center) to (327.center);
		\draw (330.center) to (329.center);
		\draw (254.center) to (319.center);
		\draw (256.center) to (320.center);
		\draw (319) to (320);
		\draw (332.center) to (334);
		\draw (333) to (334);
		\draw (336.center) to (338);
		\draw (337) to (338);
		\draw (319) to (333);
		\draw (333) to (337);
		\draw (320) to (340.center);
		\draw (334) to (341.center);
		\draw (338) to (342.center);
		\draw (337) to (339.center);
		\draw (321.center) to (343.center);
		\draw (351.center) to (352.center);
		\draw (352.center) to (354.center);
		\draw (354.center) to (353.center);
		\draw (353.center) to (351.center);
		\draw (357.center) to (358.center);
		\draw (359.center) to (360.center);
		\draw (361.center) to (362.center);
		\draw (366.center) to (369.center);
		\draw (368.center) to (365.center);
		\draw (364.center) to (367.center);
		\draw (371.center) to (370.center);
		\draw (363.center) to (356.center);
		\draw (374.center) to (375.center);
		\draw (377.center) to (376.center);
	\end{pgfonlayer}
\end{tikzpicture}

%% file: figs/conjugation.tikz
\begin{tikzpicture}
	\begin{pgfonlayer}{nodelayer}
		\node [style=none] (3) at (3, 2) {};
		\node [style=none] (4) at (3, 0.5) {};
		\node [style=none] (5) at (7.25, 1.25) {$\leftrightarrow$};
		\node [style=none] (17) at (3, 2) {};
		\node [style=none] (18) at (3, 0.5) {};
		\node [style=none] (19) at (3, 0.5) {};
		\node [style=none] (20) at (5.5, 2) {};
		\node [style=none] (21) at (5.5, 0.5) {};
		\node [style=none] (22) at (4.25, 1.25) {$B'$};
		\node [style=none] (23) at (3.5, 2.25) {};
		\node [style=none] (24) at (3.5, 0.25) {};
		\node [style=none] (25) at (5, 0.25) {};
		\node [style=none] (26) at (5, 2.25) {};
		\node [style=none] (27) at (5, 2) {};
		\node [style=none] (28) at (5, 0.5) {};
		\node [style=none] (29) at (5, 0.5) {};
		\node [style=none] (30) at (3.5, 2) {};
		\node [style=none] (31) at (3.5, 0.5) {};
		\node [style=none] (32) at (3.5, 0.5) {};
		\node [style=none] (34) at (5.5, 1.5) {$\vdots$};
		\node [style=none] (35) at (9, 2) {};
		\node [style=none] (36) at (9, 0.5) {};
		\node [style=none] (37) at (9, 0.5) {};
		\node [style=none] (38) at (11.5, 2) {};
		\node [style=none] (39) at (11.5, 0.5) {};
		\node [style=none] (40) at (10.25, 1.25) {$A^{-1}$};
		\node [style=none] (41) at (9.5, 2.25) {};
		\node [style=none] (42) at (9.5, 0.25) {};
		\node [style=none] (43) at (11, 0.25) {};
		\node [style=none] (44) at (11, 2.25) {};
		\node [style=none] (45) at (11, 2) {};
		\node [style=none] (46) at (11, 0.5) {};
		\node [style=none] (47) at (11, 0.5) {};
		\node [style=none] (48) at (9.5, 2) {};
		\node [style=none] (49) at (9.5, 0.5) {};
		\node [style=none] (50) at (9.5, 0.5) {};
		\node [style=none] (51) at (11.5, 2) {};
		\node [style=none] (52) at (11.5, 0.5) {};
		\node [style=none] (53) at (11.5, 0.5) {};
		\node [style=none] (54) at (14, 2) {};
		\node [style=none] (55) at (14, 0.5) {};
		\node [style=none] (56) at (12.75, 1.25) {$B$};
		\node [style=none] (57) at (12, 2.25) {};
		\node [style=none] (58) at (12, 0.25) {};
		\node [style=none] (59) at (13.5, 0.25) {};
		\node [style=none] (60) at (13.5, 2.25) {};
		\node [style=none] (61) at (13.5, 2) {};
		\node [style=none] (62) at (13.5, 0.5) {};
		\node [style=none] (63) at (13.5, 0.5) {};
		\node [style=none] (64) at (12, 2) {};
		\node [style=none] (65) at (12, 0.5) {};
		\node [style=none] (66) at (12, 0.5) {};
		\node [style=none] (67) at (9, 1.5) {$\vdots$};
		\node [style=none] (69) at (14, 2) {};
		\node [style=none] (70) at (14, 0.5) {};
		\node [style=none] (71) at (14, 0.5) {};
		\node [style=none] (72) at (16.5, 2) {};
		\node [style=none] (73) at (16.5, 0.5) {};
		\node [style=none] (74) at (15.25, 1.25) {$A$};
		\node [style=none] (75) at (14.5, 2.25) {};
		\node [style=none] (76) at (14.5, 0.25) {};
		\node [style=none] (77) at (16, 0.25) {};
		\node [style=none] (78) at (16, 2.25) {};
		\node [style=none] (79) at (16, 2) {};
		\node [style=none] (80) at (16, 0.5) {};
		\node [style=none] (81) at (16, 0.5) {};
		\node [style=none] (82) at (14.5, 2) {};
		\node [style=none] (83) at (14.5, 0.5) {};
		\node [style=none] (84) at (14.5, 0.5) {};
		\node [style=none] (85) at (16.5, 2) {};
		\node [style=none] (86) at (16.5, 0.5) {};
		\node [style=none] (87) at (16.5, 0.5) {};
		\node [style=none] (88) at (16.5, 1.5) {$\vdots$};
		\node [style=none] (89) at (11.5, -1) {};
		\node [style=none] (90) at (11.5, -2) {};
		\node [style=none] (91) at (11.5, -1) {};
		\node [style=none] (92) at (11.5, -2) {};
		\node [style=none] (93) at (11.5, -2) {};
		\node [style=none] (94) at (14, -1) {};
		\node [style=none] (95) at (14, -2) {};
		\node [style=none] (96) at (12.75, -1.5) {$b_1$};
		\node [style=none] (97) at (12, -0.75) {};
		\node [style=none] (98) at (12, -2.25) {};
		\node [style=none] (99) at (13.5, -2.25) {};
		\node [style=none] (100) at (13.5, -0.75) {};
		\node [style=none] (101) at (13.5, -1) {};
		\node [style=none] (102) at (13.5, -2) {};
		\node [style=none] (103) at (13.5, -2) {};
		\node [style=none] (104) at (12, -1) {};
		\node [style=none] (105) at (12, -2) {};
		\node [style=none] (106) at (12, -2) {};
		\node [style=none] (107) at (11.5, -1.5) {};
		\node [style=none] (108) at (14, -1.5) {};
		\node [style=none] (109) at (13.5, -1.5) {};
		\node [style=none] (110) at (12, -1.5) {};
		\node [style=none] (112) at (11.5, -3) {};
		\node [style=none] (113) at (11.5, -4) {};
		\node [style=none] (114) at (11.5, -3) {};
		\node [style=none] (115) at (11.5, -4) {};
		\node [style=none] (116) at (11.5, -4) {};
		\node [style=none] (117) at (14, -3) {};
		\node [style=none] (118) at (14, -4) {};
		\node [style=none] (119) at (12.75, -3.5) {$b_2$};
		\node [style=none] (120) at (12, -2.75) {};
		\node [style=none] (121) at (12, -4.25) {};
		\node [style=none] (122) at (13.5, -4.25) {};
		\node [style=none] (123) at (13.5, -2.75) {};
		\node [style=none] (124) at (13.5, -3) {};
		\node [style=none] (125) at (13.5, -4) {};
		\node [style=none] (126) at (13.5, -4) {};
		\node [style=none] (127) at (12, -3) {};
		\node [style=none] (128) at (12, -4) {};
		\node [style=none] (129) at (12, -4) {};
		\node [style=none] (130) at (11.5, -3.5) {};
		\node [style=none] (131) at (14, -3.5) {};
		\node [style=none] (132) at (13.5, -3.5) {};
		\node [style=none] (133) at (12, -3.5) {};
		\node [style=none] (134) at (11.5, -5.75) {};
		\node [style=none] (135) at (11.5, -6.75) {};
		\node [style=none] (136) at (11.5, -5.75) {};
		\node [style=none] (137) at (11.5, -6.75) {};
		\node [style=none] (138) at (11.5, -6.75) {};
		\node [style=none] (139) at (14, -5.75) {};
		\node [style=none] (140) at (14, -6.75) {};
		\node [style=none] (141) at (12.75, -6.25) {$b_k$};
		\node [style=none] (142) at (12, -5.5) {};
		\node [style=none] (143) at (12, -7) {};
		\node [style=none] (144) at (13.5, -7) {};
		\node [style=none] (145) at (13.5, -5.5) {};
		\node [style=none] (146) at (13.5, -5.75) {};
		\node [style=none] (147) at (13.5, -6.75) {};
		\node [style=none] (148) at (13.5, -6.75) {};
		\node [style=none] (149) at (12, -5.75) {};
		\node [style=none] (150) at (12, -6.75) {};
		\node [style=none] (151) at (12, -6.75) {};
		\node [style=none] (152) at (11.5, -6.25) {};
		\node [style=none] (153) at (14, -6.25) {};
		\node [style=none] (154) at (13.5, -6.25) {};
		\node [style=none] (155) at (12, -6.25) {};
		\node [style=none] (156) at (12.75, -4.75) {$\vdots$};
		\node [style=none] (157) at (6.5, -3) {};
		\node [style=none] (158) at (6.5, -4.5) {};
		\node [style=none] (159) at (6.5, -3) {};
		\node [style=none] (160) at (6.5, -4.5) {};
		\node [style=none] (161) at (6.5, -4.5) {};
		\node [style=none] (162) at (9, -3) {};
		\node [style=none] (163) at (9, -4.5) {};
		\node [style=none] (164) at (7.75, -3.75) {$B$};
		\node [style=none] (165) at (7, -2.75) {};
		\node [style=none] (166) at (7, -4.75) {};
		\node [style=none] (167) at (8.5, -4.75) {};
		\node [style=none] (168) at (8.5, -2.75) {};
		\node [style=none] (169) at (8.5, -3) {};
		\node [style=none] (170) at (8.5, -4.5) {};
		\node [style=none] (171) at (8.5, -4.5) {};
		\node [style=none] (172) at (7, -3) {};
		\node [style=none] (173) at (7, -4.5) {};
		\node [style=none] (174) at (7, -4.5) {};
		\node [style=none] (175) at (9, -3.5) {$\vdots$};
		\node [style=none] (176) at (10.25, -3.75) {$=$};
		\node [style=none] (177) at (4, -3.5) {where};
		\node [style=none] (179) at (6.5, -3.5) {$\vdots$};
		\node [style=none] (180) at (3, 1.5) {$\vdots$};
	\end{pgfonlayer}
	\begin{pgfonlayer}{edgelayer}
		\draw [style=new edge style 0] (23.center) to (26.center);
		\draw [style=new edge style 0] (26.center) to (25.center);
		\draw [style=new edge style 0] (25.center) to (24.center);
		\draw [style=new edge style 0] (24.center) to (23.center);
		\draw (17.center) to (30.center);
		\draw (32.center) to (19.center);
		\draw (27.center) to (20.center);
		\draw (21.center) to (29.center);
		\draw [style=new edge style 0] (41.center) to (44.center);
		\draw [style=new edge style 0] (44.center) to (43.center);
		\draw [style=new edge style 0] (43.center) to (42.center);
		\draw [style=new edge style 0] (42.center) to (41.center);
		\draw (35.center) to (48.center);
		\draw (50.center) to (37.center);
		\draw (45.center) to (38.center);
		\draw (39.center) to (47.center);
		\draw [style=new edge style 0] (57.center) to (60.center);
		\draw [style=new edge style 0] (60.center) to (59.center);
		\draw [style=new edge style 0] (59.center) to (58.center);
		\draw [style=new edge style 0] (58.center) to (57.center);
		\draw (51.center) to (64.center);
		\draw (66.center) to (53.center);
		\draw (61.center) to (54.center);
		\draw (55.center) to (63.center);
		\draw [style=new edge style 0] (75.center) to (78.center);
		\draw [style=new edge style 0] (78.center) to (77.center);
		\draw [style=new edge style 0] (77.center) to (76.center);
		\draw [style=new edge style 0] (76.center) to (75.center);
		\draw (69.center) to (82.center);
		\draw (84.center) to (71.center);
		\draw (79.center) to (72.center);
		\draw (73.center) to (81.center);
		\draw [style=new edge style 0] (97.center) to (100.center);
		\draw [style=new edge style 0] (100.center) to (99.center);
		\draw [style=new edge style 0] (99.center) to (98.center);
		\draw [style=new edge style 0] (98.center) to (97.center);
		\draw (91.center) to (104.center);
		\draw (106.center) to (93.center);
		\draw (101.center) to (94.center);
		\draw (95.center) to (103.center);
		\draw (107.center) to (110.center);
		\draw (109.center) to (108.center);
		\draw [style=new edge style 0] (120.center) to (123.center);
		\draw [style=new edge style 0] (123.center) to (122.center);
		\draw [style=new edge style 0] (122.center) to (121.center);
		\draw [style=new edge style 0] (121.center) to (120.center);
		\draw (114.center) to (127.center);
		\draw (129.center) to (116.center);
		\draw (124.center) to (117.center);
		\draw (118.center) to (126.center);
		\draw (130.center) to (133.center);
		\draw (132.center) to (131.center);
		\draw [style=new edge style 0] (142.center) to (145.center);
		\draw [style=new edge style 0] (145.center) to (144.center);
		\draw [style=new edge style 0] (144.center) to (143.center);
		\draw [style=new edge style 0] (143.center) to (142.center);
		\draw (136.center) to (149.center);
		\draw (151.center) to (138.center);
		\draw (146.center) to (139.center);
		\draw (140.center) to (148.center);
		\draw (152.center) to (155.center);
		\draw (154.center) to (153.center);
		\draw [style=new edge style 0] (165.center) to (168.center);
		\draw [style=new edge style 0] (168.center) to (167.center);
		\draw [style=new edge style 0] (167.center) to (166.center);
		\draw [style=new edge style 0] (166.center) to (165.center);
		\draw (159.center) to (172.center);
		\draw (174.center) to (161.center);
		\draw (169.center) to (162.center);
		\draw (163.center) to (171.center);
	\end{pgfonlayer}
\end{tikzpicture}

%% file: figs/tn-proj.tikz
\begin{tikzpicture}
	\begin{pgfonlayer}{nodelayer}
		\node [style=none] (0) at (3.75, -0.25) {$M(B)=$};
		\node [style=none] (1) at (6.25, 1) {};
		\node [style=none] (2) at (6.25, 0) {};
		\node [style=none] (3) at (6.25, -1) {};
		\node [style=none] (4) at (16.25, 0) {};
		\node [style=none] (5) at (16.25, -1) {};
		\node [style=none] (6) at (16.25, 0) {};
		\node [style=none] (7) at (16.25, 1) {};
		\node [style=none] (8) at (6.25, -2.5) {};
		\node [style=none] (9) at (6.25, -2.25) {};
		\node [style=none] (10) at (6.25, -2) {};
		\node [style=none] (11) at (16.25, -2) {};
		\node [style=none] (12) at (16.25, -2.25) {};
		\node [style=none] (13) at (16.25, -2.5) {};
		\node [style=none] (14) at (16.25, 0) {};
		\node [style=none] (15) at (16.25, -1) {};
		\node [style=none] (16) at (6.25, 2) {};
		\node [style=none] (17) at (16.25, 2) {};
		\node [style=none] (18) at (6.25, -2.75) {};
		\node [style=none] (19) at (16.25, -2.75) {};
		\node [style=none] (20) at (10.25, 0) {};
		\node [style=none] (21) at (8.25, 1) {};
		\node [style=none] (22) at (10.25, 1) {};
		\node [style=none] (23) at (10.25, 0) {};
		\node [style=none] (24) at (8.25, 0) {};
		\node [style=none] (25) at (9, 0.75) {};
		\node [style=none] (26) at (9.5, 0.25) {};
		\node [style=none] (27) at (10.25, 1) {};
		\node [style=none] (28) at (10.25, 0) {};
		\node [style=none] (29) at (6.25, 2) {};
		\node [style=none] (30) at (6.25, 2) {};
		\node [style=none] (31) at (8.25, 2) {};
		\node [style=none] (32) at (8.25, 1) {};
		\node [style=none] (33) at (6.25, 1) {};
		\node [style=none] (34) at (7, 1.25) {};
		\node [style=none] (35) at (7.5, 1.75) {};
		\node [style=none] (36) at (8.25, 2) {};
		\node [style=none] (37) at (6.25, -1.5) {};
		\node [style=none] (38) at (6.25, 2.5) {};
		\node [style=none] (39) at (16.25, 2.5) {};
		\node [style=none] (40) at (16.25, -1.5) {};
		\node [style=none] (41) at (10.25, 2) {};
		\node [style=none] (42) at (10.25, 2) {};
		\node [style=none] (43) at (12.25, 2) {};
		\node [style=none] (44) at (12.25, 1) {};
		\node [style=none] (45) at (10.25, 1) {};
		\node [style=none] (46) at (11, 1.25) {};
		\node [style=none] (47) at (11.5, 1.75) {};
		\node [style=none] (48) at (12.25, 2) {};
		\node [style=none] (49) at (12.25, 0) {};
		\node [style=none] (50) at (12.25, 0) {};
		\node [style=none] (51) at (14.25, 0) {};
		\node [style=none] (52) at (14.25, -1) {};
		\node [style=none] (53) at (12.25, -1) {};
		\node [style=none] (54) at (13, -0.75) {};
		\node [style=none] (55) at (13.5, -0.25) {};
		\node [style=none] (56) at (14.25, 0) {};
		\node [style=none] (57) at (6.25, 0) {};
		\node [style=none] (58) at (6.25, 0) {};
		\node [style=none] (59) at (8.25, 0) {};
		\node [style=none] (60) at (8.25, -1) {};
		\node [style=none] (61) at (6.25, -1) {};
		\node [style=none] (62) at (7, -0.75) {};
		\node [style=none] (63) at (7.5, -0.25) {};
		\node [style=none] (64) at (8.25, 0) {};
		\node [style=none] (65) at (12.25, 2) {};
		\node [style=none] (66) at (12.25, 2) {};
		\node [style=none] (67) at (14.25, 2) {};
		\node [style=none] (68) at (14.25, 1) {};
		\node [style=none] (69) at (12.25, 1) {};
		\node [style=none] (70) at (13, 1.25) {};
		\node [style=none] (71) at (13.5, 1.75) {};
		\node [style=none] (72) at (14.25, 2) {};
		\node [style=none] (73) at (16.25, 0) {};
		\node [style=none] (74) at (14.25, 1) {};
		\node [style=none] (75) at (16.25, 1) {};
		\node [style=none] (76) at (16.25, 0) {};
		\node [style=none] (77) at (14.25, 0) {};
		\node [style=none] (78) at (15, 0.75) {};
		\node [style=none] (79) at (15.5, 0.25) {};
		\node [style=none] (80) at (16.25, 0) {};
		\node [style=none] (81) at (14.25, 1) {};
		\node [style=none] (82) at (14.25, 0) {};
		\node [style=none] (83) at (14.25, 0) {};
		\node [style=none] (84) at (16.25, -1) {};
		\node [style=none] (85) at (16.25, 2) {};
		\node [style=none] (86) at (6.25, 0) {};
		\node [style=none] (87) at (6.25, -1) {};
		\node [style=none] (88) at (16.25, 0) {};
		\node [style=none] (89) at (16.25, -1) {};
		\node [style=none] (206) at (3.25, -6.5) {$\mapsto$};
		\node [style=none] (207) at (6.25, -6) {};
		\node [style=none] (208) at (6.25, -7) {};
		\node [style=none] (209) at (6.25, -8) {};
		\node [style=none] (210) at (12.25, -7) {};
		\node [style=none] (211) at (12.25, -8) {};
		\node [style=none] (212) at (12.25, -7) {};
		\node [style=none] (213) at (12.25, -6) {};
		\node [style=none] (214) at (6.25, -9.5) {};
		\node [style=none] (215) at (6.25, -9.25) {};
		\node [style=none] (216) at (6.25, -9) {};
		\node [style=none] (220) at (12.25, -7) {};
		\node [style=none] (221) at (12.25, -8) {};
		\node [style=none] (222) at (6.25, -5) {};
		\node [style=none] (223) at (12.25, -5) {};
		\node [style=none] (224) at (6.25, -9.75) {};
		\node [style=none] (235) at (6.25, -5) {};
		\node [style=none] (236) at (6.25, -5) {};
		\node [style=none] (239) at (6.25, -6) {};
		\node [style=none] (243) at (6.25, -8.5) {};
		\node [style=none] (244) at (6.25, -3.5) {};
		\node [style=none] (245) at (12.25, -3.5) {};
		\node [style=none] (246) at (12.25, -8.5) {};
		\node [style=none] (263) at (6.25, -7) {};
		\node [style=none] (264) at (6.25, -7) {};
		\node [style=none] (267) at (6.25, -8) {};
		\node [style=none] (279) at (12.25, -7) {};
		\node [style=none] (281) at (12.25, -6) {};
		\node [style=none] (282) at (12.25, -7) {};
		\node [style=none] (286) at (12.25, -7) {};
		\node [style=none] (290) at (12.25, -8) {};
		\node [style=none] (291) at (12.25, -5) {};
		\node [style=none] (292) at (6.25, -7) {};
		\node [style=none] (293) at (6.25, -8) {};
		\node [style=none] (294) at (12.25, -7) {};
		\node [style=none] (295) at (12.25, -8) {};
		\node [style=none] (296) at (6.25, -4) {};
		\node [style=none] (297) at (12.25, -4) {};
		\node [style=none] (298) at (6.25, -10) {};
		\node [style=none] (300) at (6.25, -4) {};
		\node [style=none] (301) at (6.25, -4) {};
		\node [style=none] (302) at (12.25, -4) {};
		\node [style=none] (305) at (12.25, -5) {};
		\node [style=none] (308) at (12.25, -6) {};
		\node [style=none] (311) at (12.25, -7) {};
		\node [style=none] (314) at (12.25, -8) {};
		\node [style=none] (316) at (9.25, -6) {$U_B$};
		\node [style=none] (317) at (16.25, -7) {};
		\node [style=none] (318) at (16.25, -8) {};
		\node [style=none] (319) at (16.25, -7) {};
		\node [style=none] (320) at (16.25, -6) {};
		\node [style=none] (321) at (16.25, -9) {};
		\node [style=none] (322) at (16.25, -9.25) {};
		\node [style=none] (323) at (16.25, -9.5) {};
		\node [style=none] (324) at (16.25, -7) {};
		\node [style=none] (325) at (16.25, -8) {};
		\node [style=none] (326) at (16.25, -5) {};
		\node [style=none] (327) at (16.25, -9.75) {};
		\node [style=none] (328) at (16.25, -7) {};
		\node [style=none] (329) at (16.25, -6) {};
		\node [style=none] (330) at (16.25, -7) {};
		\node [style=none] (331) at (16.25, -7) {};
		\node [style=none] (332) at (16.25, -8) {};
		\node [style=none] (333) at (16.25, -5) {};
		\node [style=none] (334) at (16.25, -7) {};
		\node [style=none] (335) at (16.25, -8) {};
		\node [style=none] (336) at (16.25, -4) {};
		\node [style=none] (337) at (16.25, -10) {};
		\node [style=none] (339) at (12.75, -3.75) {};
		\node [style=none] (340) at (14.25, -3.75) {};
		\node [style=none] (341) at (12.75, -5.25) {};
		\node [style=none] (342) at (14.25, -5.25) {};
		\node [style=none] (343) at (13.5, -4.5) {$P_\mathcal{F}^{\partial_0}$};
		\node [style=none] (344) at (14.75, -4.75) {};
		\node [style=none] (345) at (16.25, -4.75) {};
		\node [style=none] (346) at (14.75, -6.25) {};
		\node [style=none] (347) at (16.25, -6.25) {};
		\node [style=none] (348) at (15.5, -5.5) {$P_\mathcal{F}$};
		\node [style=none] (349) at (12.75, -5.75) {};
		\node [style=none] (350) at (14.25, -5.75) {};
		\node [style=none] (351) at (12.75, -7.25) {};
		\node [style=none] (352) at (14.25, -7.25) {};
		\node [style=none] (353) at (13.5, -6.5) {$P_\mathcal{F}$};
		\node [style=none] (354) at (14.75, -6.75) {};
		\node [style=none] (355) at (16.25, -6.75) {};
		\node [style=none] (356) at (14.75, -8.25) {};
		\node [style=none] (357) at (16.25, -8.25) {};
		\node [style=none] (358) at (15.5, -7.5) {$P_\mathcal{F}^{\partial_1}$};
		\node [style=none] (359) at (14.75, -8) {};
		\node [style=none] (360) at (14.75, -7) {};
		\node [style=none] (361) at (14.25, -7) {};
		\node [style=none] (362) at (14.25, -6) {};
		\node [style=none] (363) at (14.75, -6) {};
		\node [style=none] (364) at (16.25, -6) {};
		\node [style=none] (365) at (14.25, -5) {};
		\node [style=none] (366) at (12.75, -5) {};
		\node [style=none] (367) at (14.75, -5) {};
		\node [style=none] (368) at (16.25, -5) {};
		\node [style=none] (369) at (14.25, -4) {};
		\node [style=none] (370) at (12.75, -4) {};
		\node [style=none] (371) at (12.75, -7) {};
		\node [style=none] (372) at (12.75, -6) {};
		\node [style=none] (373) at (16.25, -7) {};
		\node [style=none] (374) at (16.25, -8) {};
	\end{pgfonlayer}
	\begin{pgfonlayer}{edgelayer}
		\draw [bend left=90, looseness=0.75] (7.center) to (13.center);
		\draw [bend left=90, looseness=0.75] (6.center) to (12.center);
		\draw [bend left=90, looseness=0.75] (5.center) to (11.center);
		\draw [bend left=90, looseness=0.75] (8.center) to (1.center);
		\draw [bend left=90, looseness=0.75] (9.center) to (2.center);
		\draw [bend left=270, looseness=0.75] (3.center) to (10.center);
		\draw (8.center) to (13.center);
		\draw (12.center) to (9.center);
		\draw (10.center) to (11.center);
		\draw [bend left=90, looseness=0.75] (17.center) to (19.center);
		\draw [bend left=90, looseness=0.75] (18.center) to (16.center);
		\draw (18.center) to (19.center);
		\draw [in=180, out=0] (24.center) to (22.center);
		\draw [bend right=15] (25.center) to (21.center);
		\draw [bend right=15] (26.center) to (23.center);
		\draw [in=-180, out=0] (30.center) to (32.center);
		\draw [bend left=15, looseness=1.25] (34.center) to (33.center);
		\draw [bend right=15] (31.center) to (35.center);
		\draw [style=new edge style 0] (38.center) to (39.center);
		\draw [style=new edge style 0] (39.center) to (40.center);
		\draw [style=new edge style 0] (40.center) to (37.center);
		\draw [style=new edge style 0] (37.center) to (38.center);
		\draw [in=-180, out=0] (42.center) to (44.center);
		\draw [bend left=15, looseness=1.25] (46.center) to (45.center);
		\draw [bend right=15] (43.center) to (47.center);
		\draw [in=-180, out=0] (50.center) to (52.center);
		\draw [bend left=15, looseness=1.25] (54.center) to (53.center);
		\draw [bend right=15] (51.center) to (55.center);
		\draw [in=-180, out=0] (58.center) to (60.center);
		\draw [bend left=15, looseness=1.25] (62.center) to (61.center);
		\draw [bend right=15] (59.center) to (63.center);
		\draw [in=-180, out=0] (66.center) to (68.center);
		\draw [bend left=15, looseness=1.25] (70.center) to (69.center);
		\draw [bend right=15] (67.center) to (71.center);
		\draw [in=180, out=0] (77.center) to (75.center);
		\draw [bend right=15] (78.center) to (74.center);
		\draw [bend right=15] (79.center) to (76.center);
		\draw (36.center) to (42.center);
		\draw (60.center) to (53.center);
		\draw (28.center) to (50.center);
		\draw (52.center) to (84.center);
		\draw (72.center) to (85.center);
		\draw [bend left=90, looseness=0.75] (214.center) to (207.center);
		\draw [bend left=90, looseness=0.75] (215.center) to (208.center);
		\draw [bend left=270, looseness=0.75] (209.center) to (216.center);
		\draw [bend left=90, looseness=0.75] (224.center) to (222.center);
		\draw (244.center) to (243.center);
		\draw (243.center) to (246.center);
		\draw (246.center) to (245.center);
		\draw (245.center) to (244.center);
		\draw [bend left=90, looseness=0.75] (298.center) to (296.center);
		\draw [bend left=90, looseness=0.75] (320.center) to (323.center);
		\draw [bend left=90, looseness=0.75] (319.center) to (322.center);
		\draw [bend left=90, looseness=0.75] (318.center) to (321.center);
		\draw [bend left=90, looseness=0.75] (326.center) to (327.center);
		\draw [bend left=90, looseness=0.75] (336.center) to (337.center);
		\draw (339.center) to (340.center);
		\draw (340.center) to (342.center);
		\draw (342.center) to (341.center);
		\draw (341.center) to (339.center);
		\draw (344.center) to (345.center);
		\draw (345.center) to (347.center);
		\draw (347.center) to (346.center);
		\draw (346.center) to (344.center);
		\draw (349.center) to (350.center);
		\draw (350.center) to (352.center);
		\draw (352.center) to (351.center);
		\draw (351.center) to (349.center);
		\draw (354.center) to (355.center);
		\draw (355.center) to (357.center);
		\draw (357.center) to (356.center);
		\draw (356.center) to (354.center);
		\draw (302.center) to (370.center);
		\draw (305.center) to (366.center);
		\draw (369.center) to (336.center);
		\draw (368.center) to (333.center);
		\draw (365.center) to (367.center);
		\draw (308.center) to (372.center);
		\draw (311.center) to (371.center);
		\draw (361.center) to (360.center);
		\draw (362.center) to (363.center);
		\draw (364.center) to (329.center);
		\draw (373.center) to (334.center);
		\draw (374.center) to (335.center);
		\draw (314.center) to (359.center);
		\draw (321.center) to (216.center);
		\draw (215.center) to (322.center);
		\draw (323.center) to (214.center);
		\draw (224.center) to (327.center);
		\draw (337.center) to (298.center);
	\end{pgfonlayer}
\end{tikzpicture}

%% file: figs/H-test.tikz
\begin{tikzpicture}
	\begin{pgfonlayer}{nodelayer}
		\node [style=none] (164) at (-3, -1.75) {};
		\node [style=none] (165) at (2.75, -1.75) {};
		\node [style=none] (166) at (-4, -1.75) {\small$| s \rangle$};
		\node [style=none] (167) at (2.75, -1) {};
		\node [style=none] (168) at (2.75, -2.5) {};
		\node [style=none] (169) at (3, -1.25) {};
		\node [style=none] (170) at (3, -2.25) {};
		\node [style=none] (171) at (3.25, -1.5) {};
		\node [style=none] (172) at (3.25, -2) {};
		\node [style=none] (175) at (-4, -0.25) {\small$|+\rangle$};
		\node [style=gate] (176) at (-1.35, -0.25) {$(S^\dagger)^b$};
		\node [style=none] (177) at (-3, -0.25) {};
		\node [style=gate] (178) at (0.25, -1.75) {$U$};
		\node [style=gate] (179) at (1.25, -0.25) {$H$};
		\node [style=gate] (180) at (2.75, -0.25) {$M_Z$};
		\node [style=none] (181) at (3.75, -0.25) {};
		\node [style=cnot ctrl] (182) at (0.25, -0.25) {};
	\end{pgfonlayer}
	\begin{pgfonlayer}{edgelayer}
		\draw (167.center) to (168.center);
		\draw (169.center) to (170.center);
		\draw (171.center) to (172.center);
		\draw (177.center) to (176);
		\draw (164.center) to (178);
		\draw (178) to (165.center);
		\draw [style=double edge] (180) to (181.center);
		\draw (182) to (178);
		\draw (176) to (182);
		\draw (182) to (179);
		\draw (179) to (180);
	\end{pgfonlayer}
\end{tikzpicture}

%% file: figs/brick-wall.tikz
\begin{tikzpicture}
	\begin{pgfonlayer}{nodelayer}
		\node [style=none] (89) at (9.5, -1) {};
		\node [style=none] (90) at (9.5, -2) {};
		\node [style=none] (91) at (9.5, -1) {};
		\node [style=none] (92) at (9.5, -2) {};
		\node [style=none] (93) at (9.5, -2) {};
		\node [style=none] (94) at (14.5, -1) {};
		\node [style=none] (95) at (12, -2) {};
		\node [style=none] (97) at (10, -0.75) {};
		\node [style=none] (98) at (10, -2.25) {};
		\node [style=none] (99) at (11.5, -2.25) {};
		\node [style=none] (100) at (11.5, -0.75) {};
		\node [style=none] (101) at (11.5, -1) {};
		\node [style=none] (102) at (11.5, -2) {};
		\node [style=none] (103) at (11.5, -2) {};
		\node [style=none] (104) at (10, -1) {};
		\node [style=none] (105) at (10, -2) {};
		\node [style=none] (106) at (10, -2) {};
		\node [style=none] (112) at (9.5, -3) {};
		\node [style=none] (113) at (9.5, -4) {};
		\node [style=none] (114) at (9.5, -3) {};
		\node [style=none] (115) at (9.5, -4) {};
		\node [style=none] (116) at (9.5, -4) {};
		\node [style=none] (117) at (12, -3) {};
		\node [style=none] (118) at (12, -4) {};
		\node [style=none] (120) at (10, -2.75) {};
		\node [style=none] (121) at (10, -4.25) {};
		\node [style=none] (122) at (11.5, -4.25) {};
		\node [style=none] (123) at (11.5, -2.75) {};
		\node [style=none] (124) at (11.5, -3) {};
		\node [style=none] (125) at (11.5, -4) {};
		\node [style=none] (126) at (11.5, -4) {};
		\node [style=none] (127) at (10, -3) {};
		\node [style=none] (128) at (10, -4) {};
		\node [style=none] (129) at (10, -4) {};
		\node [style=none] (134) at (9.5, -7.75) {};
		\node [style=none] (135) at (9.5, -8.75) {};
		\node [style=none] (136) at (9.5, -7.75) {};
		\node [style=none] (137) at (9.5, -8.75) {};
		\node [style=none] (138) at (9.5, -8.75) {};
		\node [style=none] (139) at (12, -7.75) {};
		\node [style=none] (140) at (14.5, -8.75) {};
		\node [style=none] (142) at (10, -7.5) {};
		\node [style=none] (143) at (10, -9) {};
		\node [style=none] (144) at (11.5, -9) {};
		\node [style=none] (145) at (11.5, -7.5) {};
		\node [style=none] (146) at (11.5, -7.75) {};
		\node [style=none] (147) at (11.5, -8.75) {};
		\node [style=none] (148) at (11.5, -8.75) {};
		\node [style=none] (149) at (10, -7.75) {};
		\node [style=none] (150) at (10, -8.75) {};
		\node [style=none] (151) at (10, -8.75) {};
		\node [style=none] (156) at (10.25, -5.5) {$\vdots$};
		\node [style=none] (157) at (12, -2) {};
		\node [style=none] (158) at (12, -3) {};
		\node [style=none] (159) at (12, -2) {};
		\node [style=none] (160) at (12, -3) {};
		\node [style=none] (161) at (12, -3) {};
		\node [style=none] (162) at (14.5, -2) {};
		\node [style=none] (163) at (14.5, -3) {};
		\node [style=none] (164) at (12.5, -1.75) {};
		\node [style=none] (165) at (12.5, -3.25) {};
		\node [style=none] (166) at (14, -3.25) {};
		\node [style=none] (167) at (14, -1.75) {};
		\node [style=none] (168) at (14, -2) {};
		\node [style=none] (169) at (14, -3) {};
		\node [style=none] (170) at (14, -3) {};
		\node [style=none] (171) at (12.5, -2) {};
		\node [style=none] (172) at (12.5, -3) {};
		\node [style=none] (173) at (12.5, -3) {};
		\node [style=none] (174) at (12, -4) {};
		\node [style=none] (176) at (12, -4) {};
		\node [style=none] (179) at (14.5, -4) {};
		\node [style=none] (181) at (12.5, -3.75) {};
		\node [style=none] (182) at (12.5, -4.75) {};
		\node [style=none] (183) at (14, -4.75) {};
		\node [style=none] (184) at (14, -3.75) {};
		\node [style=none] (185) at (14, -4) {};
		\node [style=none] (188) at (12.5, -4) {};
		\node [style=none] (195) at (12, -7.75) {};
		\node [style=none] (197) at (14.5, -7.75) {};
		\node [style=none] (198) at (12.5, -7) {};
		\node [style=none] (199) at (12.5, -8) {};
		\node [style=none] (200) at (14, -8) {};
		\node [style=none] (201) at (14, -7) {};
		\node [style=none] (203) at (14, -7.75) {};
		\node [style=none] (204) at (14, -7.75) {};
		\node [style=none] (206) at (12.5, -7.75) {};
		\node [style=none] (207) at (12.5, -7.75) {};
		\node [style=none] (208) at (17.75, -5.5) {$\dots$};
		\node [style=none] (209) at (21.25, -1) {};
		\node [style=none] (210) at (21.25, -2) {};
		\node [style=none] (211) at (18.75, -1) {};
		\node [style=none] (212) at (21.25, -2) {};
		\node [style=none] (213) at (21.25, -2) {};
		\node [style=none] (214) at (23.75, -2) {};
		\node [style=none] (215) at (21.75, -0.75) {};
		\node [style=none] (216) at (21.75, -2.25) {};
		\node [style=none] (217) at (23.25, -2.25) {};
		\node [style=none] (218) at (23.25, -0.75) {};
		\node [style=none] (219) at (23.25, -1) {};
		\node [style=none] (220) at (23.25, -2) {};
		\node [style=none] (221) at (23.25, -2) {};
		\node [style=none] (222) at (21.75, -1) {};
		\node [style=none] (223) at (21.75, -2) {};
		\node [style=none] (224) at (21.75, -2) {};
		\node [style=none] (225) at (21.25, -3) {};
		\node [style=none] (226) at (21.25, -4) {};
		\node [style=none] (227) at (21.25, -3) {};
		\node [style=none] (228) at (21.25, -4) {};
		\node [style=none] (229) at (21.25, -4) {};
		\node [style=none] (230) at (23.75, -3) {};
		\node [style=none] (231) at (23.75, -4) {};
		\node [style=none] (232) at (21.75, -2.75) {};
		\node [style=none] (233) at (21.75, -4.25) {};
		\node [style=none] (234) at (23.25, -4.25) {};
		\node [style=none] (235) at (23.25, -2.75) {};
		\node [style=none] (236) at (23.25, -3) {};
		\node [style=none] (237) at (23.25, -4) {};
		\node [style=none] (238) at (23.25, -4) {};
		\node [style=none] (239) at (21.75, -3) {};
		\node [style=none] (240) at (21.75, -4) {};
		\node [style=none] (241) at (21.75, -4) {};
		\node [style=none] (242) at (21.25, -7.75) {};
		\node [style=none] (243) at (21.25, -8.75) {};
		\node [style=none] (244) at (21.25, -7.75) {};
		\node [style=none] (245) at (21.25, -8.75) {};
		\node [style=none] (246) at (18.75, -8.75) {};
		\node [style=none] (247) at (23.75, -7.75) {};
		\node [style=none] (248) at (21.75, -7.5) {};
		\node [style=none] (249) at (21.75, -9) {};
		\node [style=none] (250) at (23.25, -9) {};
		\node [style=none] (251) at (23.25, -7.5) {};
		\node [style=none] (252) at (23.25, -7.75) {};
		\node [style=none] (253) at (23.25, -8.75) {};
		\node [style=none] (254) at (23.25, -8.75) {};
		\node [style=none] (255) at (21.75, -7.75) {};
		\node [style=none] (256) at (21.75, -8.75) {};
		\node [style=none] (257) at (21.75, -8.75) {};
		\node [style=none] (258) at (22.5, -5.5) {$\vdots$};
		\node [style=none] (259) at (23.75, -2) {};
		\node [style=none] (260) at (23.75, -3) {};
		\node [style=none] (261) at (23.75, -2) {};
		\node [style=none] (262) at (23.75, -3) {};
		\node [style=none] (263) at (23.75, -3) {};
		\node [style=none] (264) at (23.75, -4) {};
		\node [style=none] (265) at (23.75, -5) {};
		\node [style=none] (266) at (23.75, -4) {};
		\node [style=none] (267) at (23.75, -5) {};
		\node [style=none] (268) at (23.75, -5) {};
		\node [style=none] (269) at (23.75, -6.75) {};
		\node [style=none] (270) at (23.75, -6.75) {};
		\node [style=none] (271) at (23.75, -7.75) {};
		\node [style=none] (272) at (23.75, -8.75) {};
		\node [style=none] (273) at (23.25, -8.75) {};
		\node [style=none] (274) at (23.75, -8.75) {};
		\node [style=none] (275) at (23.75, -1) {};
		\node [style=none] (276) at (23.25, -1) {};
		\node [style=none] (277) at (23.25, -1) {};
		\node [style=none] (278) at (23.75, -1) {};
		\node [style=none] (279) at (23.75, -1) {};
		\node [style=none] (280) at (14.5, -1) {};
		\node [style=none] (281) at (14.5, -2) {};
		\node [style=none] (282) at (14.5, -1) {};
		\node [style=none] (283) at (14.5, -2) {};
		\node [style=none] (284) at (14.5, -2) {};
		\node [style=none] (285) at (17, -2) {};
		\node [style=none] (286) at (15, -0.75) {};
		\node [style=none] (287) at (15, -2.25) {};
		\node [style=none] (288) at (16.5, -2.25) {};
		\node [style=none] (289) at (16.5, -0.75) {};
		\node [style=none] (290) at (16.5, -1) {};
		\node [style=none] (291) at (16.5, -2) {};
		\node [style=none] (292) at (16.5, -2) {};
		\node [style=none] (293) at (15, -1) {};
		\node [style=none] (294) at (15, -2) {};
		\node [style=none] (295) at (15, -2) {};
		\node [style=none] (296) at (14.5, -3) {};
		\node [style=none] (297) at (14.5, -4) {};
		\node [style=none] (298) at (14.5, -3) {};
		\node [style=none] (299) at (14.5, -4) {};
		\node [style=none] (300) at (14.5, -4) {};
		\node [style=none] (301) at (17, -3) {};
		\node [style=none] (302) at (17, -4) {};
		\node [style=none] (303) at (15, -2.75) {};
		\node [style=none] (304) at (15, -4.25) {};
		\node [style=none] (305) at (16.5, -4.25) {};
		\node [style=none] (306) at (16.5, -2.75) {};
		\node [style=none] (307) at (16.5, -3) {};
		\node [style=none] (308) at (16.5, -4) {};
		\node [style=none] (309) at (16.5, -4) {};
		\node [style=none] (310) at (15, -3) {};
		\node [style=none] (311) at (15, -4) {};
		\node [style=none] (312) at (15, -4) {};
		\node [style=none] (313) at (14.5, -7.75) {};
		\node [style=none] (314) at (14.5, -8.75) {};
		\node [style=none] (315) at (14.5, -7.75) {};
		\node [style=none] (316) at (14.5, -8.75) {};
		\node [style=none] (317) at (14.5, -8.75) {};
		\node [style=none] (318) at (17, -7.75) {};
		\node [style=none] (319) at (15, -7.5) {};
		\node [style=none] (320) at (15, -9) {};
		\node [style=none] (321) at (16.5, -9) {};
		\node [style=none] (322) at (16.5, -7.5) {};
		\node [style=none] (323) at (16.5, -7.75) {};
		\node [style=none] (324) at (16.5, -8.75) {};
		\node [style=none] (325) at (16.5, -8.75) {};
		\node [style=none] (326) at (15, -7.75) {};
		\node [style=none] (327) at (15, -8.75) {};
		\node [style=none] (328) at (15, -8.75) {};
		\node [style=none] (330) at (17, -2) {};
		\node [style=none] (331) at (17, -3) {};
		\node [style=none] (332) at (17, -2) {};
		\node [style=none] (333) at (17, -3) {};
		\node [style=none] (334) at (17, -3) {};
		\node [style=none] (335) at (17, -4) {};
		\node [style=none] (337) at (17, -4) {};
		\node [style=none] (342) at (17, -7.75) {};
		\node [style=none] (343) at (17, -8.75) {};
		\node [style=none] (344) at (16.5, -8.75) {};
		\node [style=none] (345) at (17, -8.75) {};
		\node [style=none] (346) at (17, -1) {};
		\node [style=none] (347) at (16.5, -1) {};
		\node [style=none] (348) at (16.5, -1) {};
		\node [style=none] (349) at (17, -1) {};
		\node [style=none] (350) at (17, -1) {};
		\node [style=none] (352) at (18.75, -2) {};
		\node [style=none] (353) at (18.75, -3) {};
		\node [style=none] (354) at (18.75, -4) {};
		\node [style=none] (355) at (18.75, -7.75) {};
		\node [style=none] (357) at (18.75, -2) {};
		\node [style=none] (358) at (18.75, -3) {};
		\node [style=none] (359) at (18.75, -2) {};
		\node [style=none] (360) at (18.75, -3) {};
		\node [style=none] (361) at (18.75, -3) {};
		\node [style=none] (362) at (21.25, -2) {};
		\node [style=none] (363) at (21.25, -3) {};
		\node [style=none] (364) at (19.25, -1.75) {};
		\node [style=none] (365) at (19.25, -3.25) {};
		\node [style=none] (366) at (20.75, -3.25) {};
		\node [style=none] (367) at (20.75, -1.75) {};
		\node [style=none] (368) at (20.75, -2) {};
		\node [style=none] (369) at (20.75, -3) {};
		\node [style=none] (370) at (20.75, -3) {};
		\node [style=none] (371) at (19.25, -2) {};
		\node [style=none] (372) at (19.25, -3) {};
		\node [style=none] (373) at (19.25, -3) {};
		\node [style=none] (374) at (18.75, -4) {};
		\node [style=none] (376) at (18.75, -4) {};
		\node [style=none] (378) at (21.25, -4) {};
		\node [style=none] (379) at (19.25, -3.75) {};
		\node [style=none] (380) at (19.25, -4.75) {};
		\node [style=none] (381) at (20.75, -4.75) {};
		\node [style=none] (382) at (20.75, -3.75) {};
		\node [style=none] (383) at (20.75, -4) {};
		\node [style=none] (385) at (19.25, -4) {};
		\node [style=none] (387) at (18.75, -6.75) {};
		\node [style=none] (388) at (18.75, -7.75) {};
		\node [style=none] (389) at (21.25, -7.75) {};
		\node [style=none] (390) at (19.25, -7) {};
		\node [style=none] (391) at (19.25, -8) {};
		\node [style=none] (392) at (20.75, -8) {};
		\node [style=none] (393) at (20.75, -7) {};
		\node [style=none] (394) at (20.75, -7.75) {};
		\node [style=none] (395) at (20.75, -7.75) {};
		\node [style=none] (396) at (19.25, -7.75) {};
		\node [style=none] (397) at (19.25, -7.75) {};
		\node [style=none] (399) at (21.25, -2) {};
		\node [style=none] (401) at (21.25, -2) {};
		\node [style=none] (402) at (21.25, -2) {};
		\node [style=none] (403) at (21.25, -3) {};
		\node [style=none] (404) at (21.25, -4) {};
		\node [style=none] (405) at (21.25, -3) {};
		\node [style=none] (406) at (21.25, -4) {};
		\node [style=none] (407) at (21.25, -4) {};
		\node [style=none] (408) at (21.25, -7.75) {};
		\node [style=none] (410) at (21.25, -7.75) {};
	\end{pgfonlayer}
	\begin{pgfonlayer}{edgelayer}
		\draw [style=new edge style 0] (97.center) to (100.center);
		\draw [style=new edge style 0] (100.center) to (99.center);
		\draw [style=new edge style 0] (99.center) to (98.center);
		\draw [style=new edge style 0] (98.center) to (97.center);
		\draw (91.center) to (104.center);
		\draw (106.center) to (93.center);
		\draw (101.center) to (94.center);
		\draw (95.center) to (103.center);
		\draw [style=new edge style 0] (120.center) to (123.center);
		\draw [style=new edge style 0] (123.center) to (122.center);
		\draw [style=new edge style 0] (122.center) to (121.center);
		\draw [style=new edge style 0] (121.center) to (120.center);
		\draw (114.center) to (127.center);
		\draw (129.center) to (116.center);
		\draw (124.center) to (117.center);
		\draw (118.center) to (126.center);
		\draw [style=new edge style 0] (142.center) to (145.center);
		\draw [style=new edge style 0] (145.center) to (144.center);
		\draw [style=new edge style 0] (144.center) to (143.center);
		\draw [style=new edge style 0] (143.center) to (142.center);
		\draw (136.center) to (149.center);
		\draw (151.center) to (138.center);
		\draw (146.center) to (139.center);
		\draw (140.center) to (148.center);
		\draw [style=new edge style 0] (164.center) to (167.center);
		\draw [style=new edge style 0] (167.center) to (166.center);
		\draw [style=new edge style 0] (166.center) to (165.center);
		\draw [style=new edge style 0] (165.center) to (164.center);
		\draw (159.center) to (171.center);
		\draw (173.center) to (161.center);
		\draw (168.center) to (162.center);
		\draw (163.center) to (170.center);
		\draw [style=new edge style 0] (181.center) to (184.center);
		\draw [style=new edge style 0] (184.center) to (183.center);
		\draw [style=new edge style 0] (182.center) to (181.center);
		\draw (176.center) to (188.center);
		\draw (185.center) to (179.center);
		\draw [style=new edge style 0] (201.center) to (200.center);
		\draw [style=new edge style 0] (200.center) to (199.center);
		\draw [style=new edge style 0] (199.center) to (198.center);
		\draw (207.center) to (195.center);
		\draw (197.center) to (204.center);
		\draw [style=new edge style 0] (215.center) to (218.center);
		\draw [style=new edge style 0] (218.center) to (217.center);
		\draw [style=new edge style 0] (217.center) to (216.center);
		\draw [style=new edge style 0] (216.center) to (215.center);
		\draw (211.center) to (222.center);
		\draw (224.center) to (213.center);
		\draw (214.center) to (221.center);
		\draw [style=new edge style 0] (232.center) to (235.center);
		\draw [style=new edge style 0] (235.center) to (234.center);
		\draw [style=new edge style 0] (234.center) to (233.center);
		\draw [style=new edge style 0] (233.center) to (232.center);
		\draw (227.center) to (239.center);
		\draw (241.center) to (229.center);
		\draw (236.center) to (230.center);
		\draw (231.center) to (238.center);
		\draw [style=new edge style 0] (248.center) to (251.center);
		\draw [style=new edge style 0] (251.center) to (250.center);
		\draw [style=new edge style 0] (250.center) to (249.center);
		\draw [style=new edge style 0] (249.center) to (248.center);
		\draw (244.center) to (255.center);
		\draw (257.center) to (246.center);
		\draw (252.center) to (247.center);
		\draw (273.center) to (272.center);
		\draw (275.center) to (277.center);
		\draw [style=new edge style 0] (286.center) to (289.center);
		\draw [style=new edge style 0] (289.center) to (288.center);
		\draw [style=new edge style 0] (288.center) to (287.center);
		\draw [style=new edge style 0] (287.center) to (286.center);
		\draw (282.center) to (293.center);
		\draw (295.center) to (284.center);
		\draw (285.center) to (292.center);
		\draw [style=new edge style 0] (303.center) to (306.center);
		\draw [style=new edge style 0] (306.center) to (305.center);
		\draw [style=new edge style 0] (305.center) to (304.center);
		\draw [style=new edge style 0] (304.center) to (303.center);
		\draw (298.center) to (310.center);
		\draw (312.center) to (300.center);
		\draw (307.center) to (301.center);
		\draw (302.center) to (309.center);
		\draw [style=new edge style 0] (319.center) to (322.center);
		\draw [style=new edge style 0] (322.center) to (321.center);
		\draw [style=new edge style 0] (321.center) to (320.center);
		\draw [style=new edge style 0] (320.center) to (319.center);
		\draw (315.center) to (326.center);
		\draw (328.center) to (317.center);
		\draw (323.center) to (318.center);
		\draw (344.center) to (343.center);
		\draw (346.center) to (348.center);
		\draw [style=new edge style 0] (364.center) to (367.center);
		\draw [style=new edge style 0] (367.center) to (366.center);
		\draw [style=new edge style 0] (366.center) to (365.center);
		\draw [style=new edge style 0] (365.center) to (364.center);
		\draw (359.center) to (371.center);
		\draw (373.center) to (361.center);
		\draw (368.center) to (362.center);
		\draw (363.center) to (370.center);
		\draw [style=new edge style 0] (379.center) to (382.center);
		\draw [style=new edge style 0] (382.center) to (381.center);
		\draw [style=new edge style 0] (380.center) to (379.center);
		\draw (376.center) to (385.center);
		\draw (383.center) to (378.center);
		\draw [style=new edge style 0] (393.center) to (392.center);
		\draw [style=new edge style 0] (392.center) to (391.center);
		\draw [style=new edge style 0] (391.center) to (390.center);
		\draw (397.center) to (388.center);
		\draw (389.center) to (395.center);
	\end{pgfonlayer}
\end{tikzpicture}

%% file: figs/markov-moves1.tikz
\begin{tikzpicture}
	\begin{pgfonlayer}{nodelayer}
		\node [style=none] (62) at (-3.25, 6) {};
		\node [style=none] (65) at (-3.25, 5) {};
		\node [style=none] (68) at (-3.25, 5) {};
		\node [style=none] (77) at (-0.75, 6) {};
		\node [style=none] (78) at (-0.75, 5) {};
		\node [style=none] (106) at (2.75, 4) {};
		\node [style=none] (108) at (0.75, 3) {};
		\node [style=none] (109) at (2.75, 4) {};
		\node [style=none] (110) at (4.75, 4) {};
		\node [style=none] (111) at (4.75, 3) {};
		\node [style=none] (112) at (2.75, 3) {};
		\node [style=none] (113) at (3.5, 3.25) {};
		\node [style=none] (114) at (4, 3.75) {};
		\node [style=none] (115) at (2.75, 3) {};
		\node [style=none] (118) at (0, 4.5) {$\leftrightarrow$};
		\node [style=none] (120) at (-2, 5) {$B$};
		\node [style=none] (122) at (-4.75, 11) {};
		\node [style=none] (123) at (-2.75, 11) {};
		\node [style=none] (124) at (-2.75, 10) {};
		\node [style=none] (125) at (-4.75, 10) {};
		\node [style=none] (126) at (-4, 10.25) {};
		\node [style=none] (127) at (-3.5, 10.75) {};
		\node [style=none] (128) at (-2.75, 10) {};
		\node [style=none] (129) at (-2.75, 11) {};
		\node [style=none] (130) at (-2.75, 11) {};
		\node [style=none] (131) at (-0.75, 11) {};
		\node [style=none] (132) at (-0.75, 10) {};
		\node [style=none] (133) at (-2.75, 10) {};
		\node [style=none] (134) at (-2, 10.75) {};
		\node [style=none] (135) at (-1.5, 10.25) {};
		\node [style=none] (136) at (-0.75, 11) {};
		\node [style=none] (137) at (0, 10.5) {$\leftrightarrow$};
		\node [style=none] (138) at (0.75, 10) {};
		\node [style=none] (139) at (0.75, 11) {};
		\node [style=none] (140) at (2.25, 11) {};
		\node [style=none] (141) at (2.25, 10) {};
		\node [style=none] (142) at (4.75, 8) {};
		\node [style=none] (143) at (6.75, 8) {};
		\node [style=none] (144) at (6.75, 7) {};
		\node [style=none] (145) at (4.75, 7) {};
		\node [style=none] (146) at (5.5, 7.25) {};
		\node [style=none] (147) at (6, 7.75) {};
		\node [style=none] (148) at (6.75, 7) {};
		\node [style=none] (149) at (6.75, 8) {};
		\node [style=none] (150) at (6.75, 8) {};
		\node [style=none] (151) at (6.75, 7) {};
		\node [style=none] (163) at (-6.75, 9) {};
		\node [style=none] (164) at (-4.75, 9) {};
		\node [style=none] (165) at (-4.75, 8) {};
		\node [style=none] (166) at (-6.75, 8) {};
		\node [style=none] (167) at (-6, 8.25) {};
		\node [style=none] (168) at (-5.5, 8.75) {};
		\node [style=none] (169) at (-6.75, 8) {};
		\node [style=none] (170) at (-4.75, 8) {};
		\node [style=none] (171) at (-2.75, 8) {};
		\node [style=none] (172) at (-2.75, 7) {};
		\node [style=none] (173) at (-4.75, 7) {};
		\node [style=none] (174) at (-4, 7.25) {};
		\node [style=none] (175) at (-3.5, 7.75) {};
		\node [style=none] (176) at (-4.75, 7) {};
		\node [style=none] (177) at (-2.75, 9) {};
		\node [style=none] (178) at (-0.75, 9) {};
		\node [style=none] (179) at (-0.75, 8) {};
		\node [style=none] (180) at (-2.75, 8) {};
		\node [style=none] (181) at (-2, 8.25) {};
		\node [style=none] (182) at (-1.5, 8.75) {};
		\node [style=none] (183) at (-2.75, 8) {};
		\node [style=none] (184) at (-0.75, 7) {};
		\node [style=none] (185) at (-6.75, 7) {};
		\node [style=none] (186) at (0, 7.5) {$\leftrightarrow$};
		\node [style=none] (199) at (2.75, 8) {};
		\node [style=none] (200) at (2.75, 7) {};
		\node [style=none] (201) at (2.75, 9) {};
		\node [style=none] (204) at (2.75, 9) {};
		\node [style=none] (205) at (2.75, 9) {};
		\node [style=none] (206) at (4.75, 9) {};
		\node [style=none] (207) at (4.75, 8) {};
		\node [style=none] (208) at (2.75, 8) {};
		\node [style=none] (209) at (3.5, 8.25) {};
		\node [style=none] (210) at (4, 8.75) {};
		\node [style=none] (211) at (2.75, 8) {};
		\node [style=none] (212) at (4.75, 7) {};
		\node [style=none] (213) at (0.75, 9) {};
		\node [style=none] (214) at (0.75, 8) {};
		\node [style=none] (215) at (2.75, 8) {};
		\node [style=none] (216) at (2.75, 7) {};
		\node [style=none] (217) at (0.75, 7) {};
		\node [style=none] (218) at (1.5, 7.25) {};
		\node [style=none] (219) at (2, 7.75) {};
		\node [style=none] (220) at (0.75, 7) {};
		\node [style=none] (221) at (2.75, 7) {};
		\node [style=none] (230) at (4.75, 9) {};
		\node [style=none] (231) at (6.75, 9) {};
		\node [style=none] (232) at (-2.75, 6.25) {};
		\node [style=none] (233) at (-2.75, 3.75) {};
		\node [style=none] (234) at (-1.25, 3.75) {};
		\node [style=none] (235) at (-1.25, 6.25) {};
		\node [style=none] (236) at (-1.25, 6) {};
		\node [style=none] (237) at (-1.25, 5) {};
		\node [style=none] (238) at (-1.25, 5) {};
		\node [style=none] (239) at (-2.75, 6) {};
		\node [style=none] (240) at (-2.75, 5) {};
		\node [style=none] (241) at (-2.75, 5) {};
		\node [style=none] (242) at (0.75, 6) {};
		\node [style=none] (244) at (0.75, 5) {};
		\node [style=none] (245) at (4.75, 6) {};
		\node [style=none] (246) at (4.75, 5) {};
		\node [style=none] (247) at (1.25, 6.25) {};
		\node [style=none] (248) at (1.25, 3.75) {};
		\node [style=none] (249) at (2.75, 3.75) {};
		\node [style=none] (250) at (2.75, 6.25) {};
		\node [style=none] (251) at (2.75, 6) {};
		\node [style=none] (252) at (2.75, 5) {};
		\node [style=none] (253) at (2.75, 5) {};
		\node [style=none] (254) at (1.25, 6) {};
		\node [style=none] (255) at (1.25, 5) {};
		\node [style=none] (256) at (1.25, 5) {};
		\node [style=none] (257) at (-3.25, 4) {};
		\node [style=none] (258) at (-3.25, 4) {};
		\node [style=none] (259) at (-0.75, 4) {};
		\node [style=none] (260) at (-1.25, 4) {};
		\node [style=none] (261) at (-1.25, 4) {};
		\node [style=none] (262) at (-2.75, 4) {};
		\node [style=none] (263) at (-2.75, 4) {};
		\node [style=none] (264) at (0.75, 4) {};
		\node [style=none] (265) at (1.25, 4) {};
		\node [style=none] (266) at (1.25, 4) {};
		\node [style=none] (267) at (2, 5) {$B$};
		\node [style=none] (268) at (-6.75, 2) {};
		\node [style=none] (269) at (-6.75, 1) {};
		\node [style=none] (270) at (-6.75, 1) {};
		\node [style=none] (271) at (-4.25, 2) {};
		\node [style=none] (272) at (-4.25, 1) {};
		\node [style=none] (275) at (0, 1.25) {$\leftrightarrow$};
		\node [style=none] (276) at (-5.5, 1.5) {$A$};
		\node [style=none] (277) at (-6.25, 2.25) {};
		\node [style=none] (278) at (-6.25, 0.75) {};
		\node [style=none] (279) at (-4.75, 0.75) {};
		\node [style=none] (280) at (-4.75, 2.25) {};
		\node [style=none] (281) at (-4.75, 2) {};
		\node [style=none] (282) at (-4.75, 1) {};
		\node [style=none] (283) at (-4.75, 1) {};
		\node [style=none] (284) at (-6.25, 2) {};
		\node [style=none] (285) at (-6.25, 1) {};
		\node [style=none] (286) at (-6.25, 1) {};
		\node [style=none] (310) at (-4.25, 2) {};
		\node [style=none] (311) at (-4.25, 1) {};
		\node [style=none] (312) at (-4.25, 1) {};
		\node [style=none] (313) at (-1.75, 2) {};
		\node [style=none] (314) at (-1.75, 1) {};
		\node [style=none] (315) at (-3, 1.5) {$B$};
		\node [style=none] (316) at (-3.75, 2.25) {};
		\node [style=none] (317) at (-3.75, 0.75) {};
		\node [style=none] (318) at (-2.25, 0.75) {};
		\node [style=none] (319) at (-2.25, 2.25) {};
		\node [style=none] (320) at (-2.25, 2) {};
		\node [style=none] (321) at (-2.25, 1) {};
		\node [style=none] (322) at (-2.25, 1) {};
		\node [style=none] (323) at (-3.75, 2) {};
		\node [style=none] (324) at (-3.75, 1) {};
		\node [style=none] (325) at (-3.75, 1) {};
		\node [style=none] (330) at (-6.5, 1.75) {$\vdots$};
		\node [style=none] (331) at (-2, 1.75) {$\vdots$};
		\node [style=none] (332) at (-3, 5.75) {$\vdots$};
		\node [style=none] (333) at (-1, 5.75) {$\vdots$};
		\node [style=none] (334) at (1, 5.75) {$\vdots$};
		\node [style=none] (335) at (3, 5.75) {$\vdots$};
		\node [style=none] (336) at (1.75, 2) {};
		\node [style=none] (337) at (1.75, 1) {};
		\node [style=none] (338) at (1.75, 1) {};
		\node [style=none] (339) at (4.25, 2) {};
		\node [style=none] (340) at (4.25, 1) {};
		\node [style=none] (341) at (3, 1.5) {$B$};
		\node [style=none] (342) at (2.25, 2.25) {};
		\node [style=none] (343) at (2.25, 0.75) {};
		\node [style=none] (344) at (3.75, 0.75) {};
		\node [style=none] (345) at (3.75, 2.25) {};
		\node [style=none] (346) at (3.75, 2) {};
		\node [style=none] (347) at (3.75, 1) {};
		\node [style=none] (348) at (3.75, 1) {};
		\node [style=none] (349) at (2.25, 2) {};
		\node [style=none] (350) at (2.25, 1) {};
		\node [style=none] (351) at (2.25, 1) {};
		\node [style=none] (352) at (4.25, 2) {};
		\node [style=none] (353) at (4.25, 1) {};
		\node [style=none] (354) at (4.25, 1) {};
		\node [style=none] (355) at (6.75, 2) {};
		\node [style=none] (356) at (6.75, 1) {};
		\node [style=none] (357) at (5.5, 1.5) {$A$};
		\node [style=none] (358) at (4.75, 2.25) {};
		\node [style=none] (359) at (4.75, 0.75) {};
		\node [style=none] (360) at (6.25, 0.75) {};
		\node [style=none] (361) at (6.25, 2.25) {};
		\node [style=none] (362) at (6.25, 2) {};
		\node [style=none] (363) at (6.25, 1) {};
		\node [style=none] (364) at (6.25, 1) {};
		\node [style=none] (365) at (4.75, 2) {};
		\node [style=none] (366) at (4.75, 1) {};
		\node [style=none] (367) at (4.75, 1) {};
		\node [style=none] (372) at (2.0, 1.75) {$\vdots$};
		\node [style=none] (373) at (6.5, 1.75) {$\vdots$};
	\end{pgfonlayer}
	\begin{pgfonlayer}{edgelayer}
		\draw [in=-180, out=0] (109.center) to (111.center);
		\draw [bend left=15, looseness=1.25] (113.center) to (112.center);
		\draw [bend right=15] (110.center) to (114.center);
		\draw (108.center) to (115.center);
		\draw [in=-180, out=0] (122.center) to (124.center);
		\draw [bend left=15, looseness=1.25] (126.center) to (125.center);
		\draw [bend right=15] (123.center) to (127.center);
		\draw [in=180, out=0] (133.center) to (131.center);
		\draw [bend right=15] (134.center) to (130.center);
		\draw [bend right=15] (135.center) to (132.center);
		\draw (139.center) to (140.center);
		\draw (138.center) to (141.center);
		\draw [in=-180, out=0] (142.center) to (144.center);
		\draw [bend left=15, looseness=1.25] (146.center) to (145.center);
		\draw [bend right=15] (143.center) to (147.center);
		\draw [in=-180, out=0] (163.center) to (165.center);
		\draw [bend left=15, looseness=1.25] (167.center) to (166.center);
		\draw [bend right=15] (164.center) to (168.center);
		\draw [in=-180, out=0] (170.center) to (172.center);
		\draw [bend left=15, looseness=1.25] (174.center) to (173.center);
		\draw [bend right=15] (171.center) to (175.center);
		\draw [in=-180, out=0] (177.center) to (179.center);
		\draw [bend left=15, looseness=1.25] (181.center) to (180.center);
		\draw [bend right=15] (178.center) to (182.center);
		\draw (164.center) to (177.center);
		\draw (172.center) to (184.center);
		\draw (185.center) to (176.center);
		\draw [in=-180, out=0] (205.center) to (207.center);
		\draw [bend left=15, looseness=1.25] (209.center) to (208.center);
		\draw [bend right=15] (206.center) to (210.center);
		\draw (200.center) to (212.center);
		\draw (213.center) to (204.center);
		\draw [in=-180, out=0] (214.center) to (216.center);
		\draw [bend left=15, looseness=1.25] (218.center) to (217.center);
		\draw [bend right=15] (215.center) to (219.center);
		\draw (230.center) to (231.center);
		\draw [style=new edge style 0] (232.center) to (235.center);
		\draw [style=new edge style 0] (235.center) to (234.center);
		\draw [style=new edge style 0] (234.center) to (233.center);
		\draw [style=new edge style 0] (233.center) to (232.center);
		\draw (62.center) to (239.center);
		\draw (241.center) to (68.center);
		\draw (236.center) to (77.center);
		\draw (78.center) to (238.center);
		\draw [style=new edge style 0] (247.center) to (250.center);
		\draw [style=new edge style 0] (250.center) to (249.center);
		\draw [style=new edge style 0] (249.center) to (248.center);
		\draw [style=new edge style 0] (248.center) to (247.center);
		\draw (242.center) to (254.center);
		\draw (256.center) to (244.center);
		\draw (251.center) to (245.center);
		\draw (246.center) to (253.center);
		\draw (263.center) to (258.center);
		\draw (259.center) to (261.center);
		\draw (266.center) to (264.center);
		\draw [style=new edge style 0] (277.center) to (280.center);
		\draw [style=new edge style 0] (280.center) to (279.center);
		\draw [style=new edge style 0] (279.center) to (278.center);
		\draw [style=new edge style 0] (278.center) to (277.center);
		\draw (268.center) to (284.center);
		\draw (286.center) to (270.center);
		\draw (281.center) to (271.center);
		\draw (272.center) to (283.center);
		\draw [style=new edge style 0] (316.center) to (319.center);
		\draw [style=new edge style 0] (319.center) to (318.center);
		\draw [style=new edge style 0] (318.center) to (317.center);
		\draw [style=new edge style 0] (317.center) to (316.center);
		\draw (310.center) to (323.center);
		\draw (325.center) to (312.center);
		\draw (320.center) to (313.center);
		\draw (314.center) to (322.center);
		\draw [style=new edge style 0] (342.center) to (345.center);
		\draw [style=new edge style 0] (345.center) to (344.center);
		\draw [style=new edge style 0] (344.center) to (343.center);
		\draw [style=new edge style 0] (343.center) to (342.center);
		\draw (336.center) to (349.center);
		\draw (351.center) to (338.center);
		\draw (346.center) to (339.center);
		\draw (340.center) to (348.center);
		\draw [style=new edge style 0] (358.center) to (361.center);
		\draw [style=new edge style 0] (361.center) to (360.center);
		\draw [style=new edge style 0] (360.center) to (359.center);
		\draw [style=new edge style 0] (359.center) to (358.center);
		\draw (352.center) to (365.center);
		\draw (367.center) to (354.center);
		\draw (362.center) to (355.center);
		\draw (356.center) to (364.center);
	\end{pgfonlayer}
\end{tikzpicture}

%% file: figs/sorting-network-step1.tikz
\begin{tikzpicture}
	\begin{pgfonlayer}{nodelayer}
		\node [style=none] (0) at (0, 1) {};
		\node [style=none] (1) at (0, 0) {};
		\node [style=none] (2) at (1, 1) {};
		\node [style=none] (3) at (1, 0) {};
		\node [style=none] (9) at (0, -1) {};
		\node [style=none] (10) at (0, -2) {};
		\node [style=none] (11) at (1, -1) {};
		\node [style=none] (12) at (1, -2) {};
		\node [style=none] (13) at (1, 2) {};
		\node [style=none] (14) at (1, 1) {};
		\node [style=none] (15) at (2, 2) {};
		\node [style=none] (16) at (2, 1) {};
		\node [style=none] (17) at (2, 1) {};
		\node [style=none] (18) at (2, 0) {};
		\node [style=none] (19) at (3, 1) {};
		\node [style=none] (20) at (3, 0) {};
		\node [style=none] (21) at (2, -1) {};
		\node [style=none] (22) at (2, -2) {};
		\node [style=none] (23) at (3, -1) {};
		\node [style=none] (24) at (3, -2) {};
		\node [style=none] (25) at (3, 2) {};
		\node [style=none] (26) at (3, 1) {};
		\node [style=none] (27) at (4, 2) {};
		\node [style=none] (28) at (4, 1) {};
		\node [style=none] (29) at (3, 0) {};
		\node [style=none] (30) at (3, -1) {};
		\node [style=none] (31) at (4, 0) {};
		\node [style=none] (32) at (4, -1) {};
		\node [style=none] (33) at (4, 1) {};
		\node [style=none] (34) at (4, 0) {};
		\node [style=none] (35) at (5, 1) {};
		\node [style=none] (36) at (5, 0) {};
		\node [style=none] (37) at (4, -1) {};
		\node [style=none] (38) at (4, -2) {};
		\node [style=none] (39) at (5, -1) {};
		\node [style=none] (40) at (5, -2) {};
		\node [style=none] (41) at (5, 2) {};
		\node [style=none] (42) at (5, 1) {};
		\node [style=none] (43) at (6, 2) {};
		\node [style=none] (44) at (6, 1) {};
		\node [style=none] (45) at (6, 1) {};
		\node [style=none] (46) at (6, 0) {};
		\node [style=none] (47) at (7, 1) {};
		\node [style=none] (48) at (7, 0) {};
		\node [style=none] (49) at (7, 2) {};
		\node [style=none] (50) at (7, 1) {};
		\node [style=none] (51) at (8, 2) {};
		\node [style=none] (52) at (8, 1) {};
		\node [style=none] (53) at (7, 0) {};
		\node [style=none] (54) at (7, -1) {};
		\node [style=none] (55) at (8, 0) {};
		\node [style=none] (56) at (8, -1) {};
		\node [style=none] (57) at (8, -1) {};
		\node [style=none] (58) at (8, -2) {};
		\node [style=none] (59) at (9, -1) {};
		\node [style=none] (60) at (9, -2) {};
		\node [style=none] (66) at (-1, 2) {};
		\node [style=none] (67) at (9, 0) {};
		\node [style=none] (68) at (9, -1) {};
		\node [style=none] (69) at (9, -2) {};
		\node [style=none] (70) at (-3, 0) {$A$};
		\node [style=none] (71) at (-2, 0) {$=$};
		\node [style=none] (73) at (14, 0) {};
		\node [style=none] (74) at (14, -1) {};
		\node [style=none] (75) at (15, 0) {};
		\node [style=none] (76) at (15, -1) {};
		\node [style=none] (77) at (15, 1) {};
		\node [style=none] (78) at (15, 0) {};
		\node [style=none] (79) at (16, 1) {};
		\node [style=none] (80) at (16, 0) {};
		\node [style=none] (81) at (15, -1) {};
		\node [style=none] (82) at (15, -2) {};
		\node [style=none] (83) at (16, -1) {};
		\node [style=none] (84) at (16, -2) {};
		\node [style=none] (85) at (16, 2) {};
		\node [style=none] (86) at (16, 1) {};
		\node [style=none] (87) at (17, 2) {};
		\node [style=none] (88) at (17, 1) {};
		\node [style=none] (89) at (16, 0) {};
		\node [style=none] (90) at (16, -1) {};
		\node [style=none] (91) at (17, 0) {};
		\node [style=none] (92) at (17, -1) {};
		\node [style=none] (93) at (17, 1) {};
		\node [style=none] (94) at (17, 0) {};
		\node [style=none] (95) at (18, 1) {};
		\node [style=none] (96) at (18, 0) {};
		\node [style=none] (97) at (9, 1) {};
		\node [style=none] (98) at (9, 2) {};
		\node [style=none] (101) at (14, 2) {};
		\node [style=none] (102) at (14, 1) {};
		\node [style=none] (103) at (15, 2) {};
		\node [style=none] (104) at (15, 1) {};
		\node [style=none] (105) at (-1, 0) {};
		\node [style=none] (106) at (-1, -1) {};
		\node [style=none] (107) at (0, 0) {};
		\node [style=none] (108) at (0, -1) {};
		\node [style=none] (109) at (-1, 1) {};
		\node [style=none] (110) at (-1, -2) {};
		\node [style=none] (111) at (18, 2) {};
		\node [style=none] (112) at (18, -1) {};
		\node [style=none] (113) at (18, -2) {};
		\node [style=none] (114) at (14, -2) {};
		\node [style=none] (115) at (11.5, 0) {$A^{-1}$};
		\node [style=none] (116) at (13, 0) {$=$};
		\node [style=none] (117) at (9, -0.75) {\tiny $1$};
		\node [style=none] (118) at (9, -1.75) {\tiny $2$};
		\node [style=none] (119) at (9, 1.25) {\tiny $3$};
		\node [style=none] (120) at (9, 2.25) {\tiny $4$};
		\node [style=none] (121) at (9, 0.25) {\tiny $5$};
		\node [style=none] (122) at (14, -0.75) {\tiny $1$};
		\node [style=none] (123) at (14, -1.75) {\tiny $2$};
		\node [style=none] (124) at (14, 1.25) {\tiny $3$};
		\node [style=none] (125) at (14, 2.25) {\tiny $4$};
		\node [style=none] (126) at (14, 0.25) {\tiny $5$};
		\node [style=none] (127) at (18, 2.25) {\tiny $1$};
		\node [style=none] (128) at (18, 1.25) {\tiny $2$};
		\node [style=none] (129) at (18, 0.25) {\tiny $3$};
		\node [style=none] (130) at (18, -0.75) {\tiny $4$};
		\node [style=none] (131) at (18, -1.75) {\tiny $5$};
		\node [style=none] (132) at (-1, 2.25) {\tiny $1$};
		\node [style=none] (133) at (-1, 1.25) {\tiny $2$};
		\node [style=none] (134) at (-1, 0.25) {\tiny $3$};
		\node [style=none] (135) at (-1, -0.75) {\tiny $4$};
		\node [style=none] (136) at (-1, -1.75) {\tiny $5$};
	\end{pgfonlayer}
	\begin{pgfonlayer}{edgelayer}
		\draw [in=-180, out=0, looseness=0.75] (1.center) to (2.center);
		\draw [in=180, out=0, looseness=0.75] (0.center) to (3.center);
		\draw [in=-180, out=0, looseness=0.75] (10.center) to (11.center);
		\draw [in=180, out=0, looseness=0.75] (9.center) to (12.center);
		\draw [in=-180, out=0, looseness=0.75] (14.center) to (15.center);
		\draw [in=180, out=0, looseness=0.75] (13.center) to (16.center);
		\draw [in=-180, out=0, looseness=0.75] (18.center) to (19.center);
		\draw [in=180, out=0, looseness=0.75] (17.center) to (20.center);
		\draw [in=-180, out=0, looseness=0.75] (22.center) to (23.center);
		\draw [in=180, out=0, looseness=0.75] (21.center) to (24.center);
		\draw [in=-180, out=0, looseness=0.75] (26.center) to (27.center);
		\draw [in=180, out=0, looseness=0.75] (25.center) to (28.center);
		\draw [in=-180, out=0, looseness=0.75] (30.center) to (31.center);
		\draw [in=180, out=0, looseness=0.75] (29.center) to (32.center);
		\draw [in=-180, out=0, looseness=0.75] (34.center) to (35.center);
		\draw [in=180, out=0, looseness=0.75] (33.center) to (36.center);
		\draw [in=-180, out=0, looseness=0.75] (38.center) to (39.center);
		\draw [in=180, out=0, looseness=0.75] (37.center) to (40.center);
		\draw [in=-180, out=0, looseness=0.75] (42.center) to (43.center);
		\draw [in=180, out=0, looseness=0.75] (41.center) to (44.center);
		\draw [in=-180, out=0, looseness=0.75] (46.center) to (47.center);
		\draw [in=180, out=0, looseness=0.75] (45.center) to (48.center);
		\draw [in=-180, out=0, looseness=0.75] (50.center) to (51.center);
		\draw [in=180, out=0, looseness=0.75] (49.center) to (52.center);
		\draw [in=-180, out=0, looseness=0.75] (54.center) to (55.center);
		\draw [in=180, out=0, looseness=0.75] (53.center) to (56.center);
		\draw [in=-180, out=0, looseness=0.75] (58.center) to (59.center);
		\draw [in=180, out=0, looseness=0.75] (57.center) to (60.center);
		\draw (49.center) to (43.center);
		\draw (41.center) to (27.center);
		\draw (25.center) to (15.center);
		\draw (13.center) to (66.center);
		\draw (3.center) to (18.center);
		\draw (11.center) to (21.center);
		\draw (12.center) to (22.center);
		\draw (24.center) to (38.center);
		\draw (39.center) to (54.center);
		\draw (40.center) to (58.center);
		\draw (36.center) to (46.center);
		\draw (55.center) to (67.center);
		\draw (59.center) to (68.center);
		\draw (60.center) to (69.center);
		\draw [in=-180, out=0, looseness=0.75] (74.center) to (75.center);
		\draw [in=180, out=0, looseness=0.75] (73.center) to (76.center);
		\draw [in=-180, out=0, looseness=0.75] (78.center) to (79.center);
		\draw [in=180, out=0, looseness=0.75] (77.center) to (80.center);
		\draw [in=-180, out=0, looseness=0.75] (82.center) to (83.center);
		\draw [in=180, out=0, looseness=0.75] (81.center) to (84.center);
		\draw [in=-180, out=0, looseness=0.75] (86.center) to (87.center);
		\draw [in=180, out=0, looseness=0.75] (85.center) to (88.center);
		\draw [in=-180, out=0, looseness=0.75] (90.center) to (91.center);
		\draw [in=180, out=0, looseness=0.75] (89.center) to (92.center);
		\draw [in=-180, out=0, looseness=0.75] (94.center) to (95.center);
		\draw [in=180, out=0, looseness=0.75] (93.center) to (96.center);
		\draw (52.center) to (97.center);
		\draw (98.center) to (51.center);
		\draw [in=-180, out=0, looseness=0.75] (102.center) to (103.center);
		\draw [in=180, out=0, looseness=0.75] (101.center) to (104.center);
		\draw [in=-180, out=0, looseness=0.75] (106.center) to (107.center);
		\draw [in=180, out=0, looseness=0.75] (105.center) to (108.center);
		\draw (109.center) to (0.center);
		\draw (110.center) to (10.center);
		\draw (103.center) to (85.center);
		\draw (111.center) to (87.center);
		\draw (112.center) to (92.center);
		\draw (113.center) to (84.center);
		\draw (114.center) to (82.center);
	\end{pgfonlayer}
\end{tikzpicture}

%% file: figs/sorting-network-step2.tikz
\begin{tikzpicture}
	\begin{pgfonlayer}{nodelayer}
		\node [style=none] (12) at (0, 0) {};
		\node [style=none] (13) at (0, -1) {};
		\node [style=none] (14) at (1, -1) {};
		\node [style=none] (15) at (1, 0) {};
		\node [style=none] (16) at (0.4, -0.65) {};
		\node [style=none] (17) at (0.6, -0.35) {};
		\node [style=none] (18) at (1, 1) {};
		\node [style=none] (19) at (1, 0) {};
		\node [style=none] (20) at (2, 0) {};
		\node [style=none] (21) at (2, 1) {};
		\node [style=none] (22) at (1.4, 0.35) {};
		\node [style=none] (23) at (1.6, 0.65) {};
		\node [style=none] (24) at (1, -2) {};
		\node [style=none] (25) at (1, -1) {};
		\node [style=none] (26) at (2, -1) {};
		\node [style=none] (27) at (2, -2) {};
		\node [style=none] (28) at (1.4, -1.35) {};
		\node [style=none] (29) at (1.6, -1.65) {};
		\node [style=none] (30) at (2, 2) {};
		\node [style=none] (31) at (2, 1) {};
		\node [style=none] (32) at (3, 1) {};
		\node [style=none] (33) at (3, 2) {};
		\node [style=none] (34) at (2.4, 1.35) {};
		\node [style=none] (35) at (2.6, 1.65) {};
		\node [style=none] (36) at (3, 0) {};
		\node [style=none] (37) at (3, 1) {};
		\node [style=none] (38) at (4, 1) {};
		\node [style=none] (39) at (4, 0) {};
		\node [style=none] (40) at (3.4, 0.65) {};
		\node [style=none] (41) at (3.6, 0.35) {};
		\node [style=none] (42) at (3, -1) {};
		\node [style=none] (43) at (3, -2) {};
		\node [style=none] (44) at (4, -2) {};
		\node [style=none] (45) at (4, -1) {};
		\node [style=none] (46) at (3.4, -1.65) {};
		\node [style=none] (47) at (3.6, -1.35) {};
		\node [style=none] (48) at (4, 1) {};
		\node [style=none] (49) at (4, 2) {};
		\node [style=none] (50) at (5, 2) {};
		\node [style=none] (51) at (5, 1) {};
		\node [style=none] (52) at (4.4, 1.65) {};
		\node [style=none] (53) at (4.6, 1.35) {};
		\node [style=none] (54) at (4, 0) {};
		\node [style=none] (55) at (4, -1) {};
		\node [style=none] (56) at (5, -1) {};
		\node [style=none] (57) at (5, 0) {};
		\node [style=none] (58) at (4.4, -0.65) {};
		\node [style=none] (59) at (4.6, -0.35) {};
		\node [style=none] (60) at (5, 0) {};
		\node [style=none] (61) at (5, 1) {};
		\node [style=none] (62) at (6, 1) {};
		\node [style=none] (63) at (6, 0) {};
		\node [style=none] (64) at (5.4, 0.65) {};
		\node [style=none] (65) at (5.6, 0.35) {};
		\node [style=none] (66) at (5, -1) {};
		\node [style=none] (67) at (5, -2) {};
		\node [style=none] (68) at (6, -2) {};
		\node [style=none] (69) at (6, -1) {};
		\node [style=none] (70) at (5.4, -1.65) {};
		\node [style=none] (71) at (5.6, -1.35) {};
		\node [style=none] (72) at (6, 2) {};
		\node [style=none] (73) at (6, 1) {};
		\node [style=none] (74) at (7, 1) {};
		\node [style=none] (75) at (7, 2) {};
		\node [style=none] (76) at (6.4, 1.35) {};
		\node [style=none] (77) at (6.6, 1.65) {};
		\node [style=none] (78) at (7, 1) {};
		\node [style=none] (79) at (7, 0) {};
		\node [style=none] (80) at (8, 0) {};
		\node [style=none] (81) at (8, 1) {};
		\node [style=none] (82) at (7.4, 0.35) {};
		\node [style=none] (83) at (7.6, 0.65) {};
		\node [style=none] (84) at (8, 2) {};
		\node [style=none] (85) at (8, 1) {};
		\node [style=none] (86) at (9, 1) {};
		\node [style=none] (87) at (9, 2) {};
		\node [style=none] (88) at (8.4, 1.35) {};
		\node [style=none] (89) at (8.6, 1.65) {};
		\node [style=none] (90) at (8, 0) {};
		\node [style=none] (91) at (8, -1) {};
		\node [style=none] (92) at (9, -1) {};
		\node [style=none] (93) at (9, 0) {};
		\node [style=none] (94) at (8.4, -0.65) {};
		\node [style=none] (95) at (8.6, -0.35) {};
		\node [style=none] (96) at (9, -1) {};
		\node [style=none] (97) at (9, -2) {};
		\node [style=none] (98) at (10, -2) {};
		\node [style=none] (99) at (10, -1) {};
		\node [style=none] (100) at (9.4, -1.65) {};
		\node [style=none] (101) at (9.6, -1.35) {};
		\node [style=none] (102) at (0, -2) {};
		\node [style=none] (103) at (0, 1) {};
		\node [style=none] (104) at (0, 2) {};
		\node [style=none] (105) at (10, 0) {};
		\node [style=none] (106) at (10, 1) {};
		\node [style=none] (107) at (10, 2) {};
		\node [style=none] (108) at (15, 1) {};
		\node [style=none] (109) at (15, 2) {};
		\node [style=none] (110) at (16, 2) {};
		\node [style=none] (111) at (16, 1) {};
		\node [style=none] (112) at (15.4, 1.65) {};
		\node [style=none] (113) at (15.6, 1.35) {};
		\node [style=none] (114) at (15, -1) {};
		\node [style=none] (115) at (15, 0) {};
		\node [style=none] (116) at (16, 0) {};
		\node [style=none] (117) at (16, -1) {};
		\node [style=none] (118) at (15.4, -0.35) {};
		\node [style=none] (119) at (15.6, -0.65) {};
		\node [style=none] (120) at (16, 0) {};
		\node [style=none] (121) at (16, 1) {};
		\node [style=none] (122) at (17, 1) {};
		\node [style=none] (123) at (17, 0) {};
		\node [style=none] (124) at (16.4, 0.65) {};
		\node [style=none] (125) at (16.6, 0.35) {};
		\node [style=none] (126) at (16, -2) {};
		\node [style=none] (127) at (16, -1) {};
		\node [style=none] (128) at (17, -1) {};
		\node [style=none] (129) at (17, -2) {};
		\node [style=none] (130) at (16.4, -1.35) {};
		\node [style=none] (131) at (16.6, -1.65) {};
		\node [style=none] (132) at (17, 1) {};
		\node [style=none] (133) at (17, 2) {};
		\node [style=none] (134) at (18, 2) {};
		\node [style=none] (135) at (18, 1) {};
		\node [style=none] (136) at (17.4, 1.65) {};
		\node [style=none] (137) at (17.6, 1.35) {};
		\node [style=none] (138) at (17, -1) {};
		\node [style=none] (139) at (17, 0) {};
		\node [style=none] (140) at (18, 0) {};
		\node [style=none] (141) at (18, -1) {};
		\node [style=none] (142) at (17.4, -0.35) {};
		\node [style=none] (143) at (17.6, -0.65) {};
		\node [style=none] (144) at (18, 0) {};
		\node [style=none] (145) at (18, 1) {};
		\node [style=none] (146) at (19, 1) {};
		\node [style=none] (147) at (19, 0) {};
		\node [style=none] (148) at (18.4, 0.65) {};
		\node [style=none] (149) at (18.6, 0.35) {};
		\node [style=none] (150) at (15, -2) {};
		\node [style=none] (151) at (19, 2) {};
		\node [style=none] (152) at (19, -1) {};
		\node [style=none] (153) at (19, -2) {};
		\node [style=none] (154) at (-2, 0) {$A$};
		\node [style=none] (155) at (-1, 0) {$=$};
		\node [style=none] (156) at (14, 0) {$=$};
		\node [style=none] (157) at (12.5, 0) {$A^{-1}$};
		\node [style=none] (158) at (0, 2.25) {\tiny $1$};
		\node [style=none] (159) at (0, 1.25) {\tiny $2$};
		\node [style=none] (160) at (0, 0.25) {\tiny $3$};
		\node [style=none] (161) at (0, -0.75) {\tiny $4$};
		\node [style=none] (162) at (0, -1.75) {\tiny $5$};
		\node [style=none] (163) at (1, 2.25) {\tiny $1$};
		\node [style=none] (164) at (1, 1.25) {\tiny $2$};
		\node [style=none] (165) at (1, -0.75) {\tiny $3$};
		\node [style=none] (166) at (1, 0.25) {\tiny $4$};
		\node [style=none] (167) at (1, -1.75) {\tiny $5$};
		\node [style=none] (168) at (2, 2.25) {\tiny $1$};
		\node [style=none] (169) at (2, 0.25) {\tiny $2$};
		\node [style=none] (170) at (2, -1.75) {\tiny $3$};
		\node [style=none] (171) at (2, 1.25) {\tiny $4$};
		\node [style=none] (172) at (2, -0.75) {\tiny $5$};
		\node [style=none] (173) at (3, 1.25) {\tiny $1$};
		\node [style=none] (174) at (3, 0.25) {\tiny $2$};
		\node [style=none] (175) at (3, -1.75) {\tiny $3$};
		\node [style=none] (176) at (3, 2.25) {\tiny $4$};
		\node [style=none] (177) at (3, -0.75) {\tiny $5$};
		\node [style=none] (178) at (4, 0.25) {\tiny $1$};
		\node [style=none] (179) at (4, 1.25) {\tiny $2$};
		\node [style=none] (180) at (4, -0.75) {\tiny $3$};
		\node [style=none] (181) at (4, 2.25) {\tiny $4$};
		\node [style=none] (182) at (4, -1.75) {\tiny $5$};
		\node [style=none] (183) at (5, -0.75) {\tiny $1$};
		\node [style=none] (184) at (5, 2.25) {\tiny $2$};
		\node [style=none] (185) at (5, 0.25) {\tiny $3$};
		\node [style=none] (186) at (5, 1.25) {\tiny $4$};
		\node [style=none] (187) at (5, -1.75) {\tiny $5$};
		\node [style=none] (188) at (6, -1.75) {\tiny $1$};
		\node [style=none] (189) at (6, 2.25) {\tiny $2$};
		\node [style=none] (190) at (6, 1.25) {\tiny $3$};
		\node [style=none] (191) at (6, 0.25) {\tiny $4$};
		\node [style=none] (192) at (6, -0.75) {\tiny $5$};
		\node [style=none] (193) at (7, -1.75) {\tiny $1$};
		\node [style=none] (194) at (7, 1.25) {\tiny $2$};
		\node [style=none] (195) at (7, 2.25) {\tiny $3$};
		\node [style=none] (196) at (7, 0.25) {\tiny $4$};
		\node [style=none] (197) at (7, -0.75) {\tiny $5$};
		\node [style=none] (198) at (8, -1.75) {\tiny $1$};
		\node [style=none] (199) at (8, 0.25) {\tiny $2$};
		\node [style=none] (200) at (8, 2.25) {\tiny $3$};
		\node [style=none] (201) at (8, 1.25) {\tiny $4$};
		\node [style=none] (202) at (8, -0.75) {\tiny $5$};
		\node [style=none] (203) at (9, -1.75) {\tiny $1$};
		\node [style=none] (204) at (9, -0.75) {\tiny $2$};
		\node [style=none] (205) at (9, 1.25) {\tiny $3$};
		\node [style=none] (206) at (9, 2.25) {\tiny $4$};
		\node [style=none] (207) at (9, 0.25) {\tiny $5$};
		\node [style=none] (208) at (10, -0.75) {\tiny $1$};
		\node [style=none] (209) at (10, -1.75) {\tiny $2$};
		\node [style=none] (210) at (10, 1.25) {\tiny $3$};
		\node [style=none] (211) at (10, 2.25) {\tiny $4$};
		\node [style=none] (212) at (10, 0.25) {\tiny $5$};
		\node [style=none] (213) at (15, -0.75) {\tiny $1$};
		\node [style=none] (214) at (15, -1.75) {\tiny $2$};
		\node [style=none] (215) at (15, 1.25) {\tiny $3$};
		\node [style=none] (216) at (15, 2.25) {\tiny $4$};
		\node [style=none] (217) at (15, 0.25) {\tiny $5$};
		\node [style=none] (218) at (16, 0.25) {\tiny $1$};
		\node [style=none] (219) at (16, -1.75) {\tiny $2$};
		\node [style=none] (220) at (16, 2.25) {\tiny $3$};
		\node [style=none] (221) at (16, 1.25) {\tiny $4$};
		\node [style=none] (222) at (16, -0.75) {\tiny $5$};
		\node [style=none] (223) at (17, 1.25) {\tiny $1$};
		\node [style=none] (224) at (17, -0.75) {\tiny $2$};
		\node [style=none] (225) at (17, 2.25) {\tiny $3$};
		\node [style=none] (226) at (17, 0.25) {\tiny $4$};
		\node [style=none] (227) at (17, -1.75) {\tiny $5$};
		\node [style=none] (228) at (18, 2.25) {\tiny $1$};
		\node [style=none] (229) at (18, 0.25) {\tiny $2$};
		\node [style=none] (230) at (18, 1.25) {\tiny $3$};
		\node [style=none] (231) at (18, -0.75) {\tiny $4$};
		\node [style=none] (232) at (18, -1.75) {\tiny $5$};
		\node [style=none] (233) at (19, 2.25) {\tiny $1$};
		\node [style=none] (234) at (19, 1.25) {\tiny $2$};
		\node [style=none] (235) at (19, 0.25) {\tiny $3$};
		\node [style=none] (236) at (19, -0.75) {\tiny $4$};
		\node [style=none] (237) at (19, -1.75) {\tiny $5$};
	\end{pgfonlayer}
	\begin{pgfonlayer}{edgelayer}
		\draw [in=180, out=0, looseness=0.75] (12.center) to (14.center);
		\draw [in=-180, out=60] (17.center) to (15.center);
		\draw [in=0, out=-120] (16.center) to (13.center);
		\draw [in=180, out=0, looseness=0.75] (18.center) to (20.center);
		\draw [in=-180, out=60] (23.center) to (21.center);
		\draw [in=0, out=-120] (22.center) to (19.center);
		\draw [in=-180, out=0, looseness=0.75] (24.center) to (26.center);
		\draw [in=180, out=-60] (29.center) to (27.center);
		\draw [in=0, out=120] (28.center) to (25.center);
		\draw [in=180, out=0, looseness=0.75] (30.center) to (32.center);
		\draw [in=-180, out=60] (35.center) to (33.center);
		\draw [in=0, out=-120] (34.center) to (31.center);
		\draw [in=-180, out=0, looseness=0.75] (36.center) to (38.center);
		\draw [in=180, out=-60] (41.center) to (39.center);
		\draw [in=0, out=120] (40.center) to (37.center);
		\draw [in=180, out=0, looseness=0.75] (42.center) to (44.center);
		\draw [in=-180, out=60] (47.center) to (45.center);
		\draw [in=0, out=-120] (46.center) to (43.center);
		\draw [in=-180, out=0, looseness=0.75] (48.center) to (50.center);
		\draw [in=180, out=-60] (53.center) to (51.center);
		\draw [in=0, out=120] (52.center) to (49.center);
		\draw [in=180, out=0, looseness=0.75] (54.center) to (56.center);
		\draw [in=-180, out=60] (59.center) to (57.center);
		\draw [in=0, out=-120] (58.center) to (55.center);
		\draw [in=-180, out=0, looseness=0.75] (60.center) to (62.center);
		\draw [in=180, out=-60] (65.center) to (63.center);
		\draw [in=0, out=120] (64.center) to (61.center);
		\draw [in=180, out=0, looseness=0.75] (66.center) to (68.center);
		\draw [in=-180, out=60] (71.center) to (69.center);
		\draw [in=0, out=-120] (70.center) to (67.center);
		\draw [in=180, out=0, looseness=0.75] (72.center) to (74.center);
		\draw [in=-180, out=60] (77.center) to (75.center);
		\draw [in=0, out=-120] (76.center) to (73.center);
		\draw [in=180, out=0, looseness=0.75] (78.center) to (80.center);
		\draw [in=-180, out=60] (83.center) to (81.center);
		\draw [in=0, out=-120] (82.center) to (79.center);
		\draw [in=180, out=0, looseness=0.75] (84.center) to (86.center);
		\draw [in=-180, out=60] (89.center) to (87.center);
		\draw [in=0, out=-120] (88.center) to (85.center);
		\draw [in=180, out=0, looseness=0.75] (90.center) to (92.center);
		\draw [in=-180, out=60] (95.center) to (93.center);
		\draw [in=0, out=-120] (94.center) to (91.center);
		\draw [in=180, out=0, looseness=0.75] (96.center) to (98.center);
		\draw [in=-180, out=60] (101.center) to (99.center);
		\draw [in=0, out=-120] (100.center) to (97.center);
		\draw (104.center) to (30.center);
		\draw (33.center) to (49.center);
		\draw (50.center) to (72.center);
		\draw (75.center) to (84.center);
		\draw (18.center) to (103.center);
		\draw (102.center) to (24.center);
		\draw (27.center) to (43.center);
		\draw (42.center) to (26.center);
		\draw (44.center) to (67.center);
		\draw (68.center) to (97.center);
		\draw (91.center) to (69.center);
		\draw (63.center) to (79.center);
		\draw (36.center) to (20.center);
		\draw (107.center) to (87.center);
		\draw (86.center) to (106.center);
		\draw (105.center) to (93.center);
		\draw [in=-180, out=0, looseness=0.75] (108.center) to (110.center);
		\draw [in=180, out=-60] (113.center) to (111.center);
		\draw [in=0, out=120] (112.center) to (109.center);
		\draw [in=-180, out=0, looseness=0.75] (114.center) to (116.center);
		\draw [in=180, out=-60] (119.center) to (117.center);
		\draw [in=0, out=120] (118.center) to (115.center);
		\draw [in=-180, out=0, looseness=0.75] (120.center) to (122.center);
		\draw [in=180, out=-60] (125.center) to (123.center);
		\draw [in=0, out=120] (124.center) to (121.center);
		\draw [in=-180, out=0, looseness=0.75] (126.center) to (128.center);
		\draw [in=180, out=-60] (131.center) to (129.center);
		\draw [in=0, out=120] (130.center) to (127.center);
		\draw [in=-180, out=0, looseness=0.75] (132.center) to (134.center);
		\draw [in=180, out=-60] (137.center) to (135.center);
		\draw [in=0, out=120] (136.center) to (133.center);
		\draw [in=-180, out=0, looseness=0.75] (138.center) to (140.center);
		\draw [in=180, out=-60] (143.center) to (141.center);
		\draw [in=0, out=120] (142.center) to (139.center);
		\draw [in=-180, out=0, looseness=0.75] (144.center) to (146.center);
		\draw [in=180, out=-60] (149.center) to (147.center);
		\draw [in=0, out=120] (148.center) to (145.center);
		\draw (153.center) to (129.center);
		\draw (126.center) to (150.center);
		\draw (151.center) to (134.center);
		\draw (133.center) to (110.center);
		\draw (152.center) to (141.center);
	\end{pgfonlayer}
\end{tikzpicture}

%% file: figs/plat-proof-1.tikz
\begin{tikzpicture}
	\begin{pgfonlayer}{nodelayer}
		\node [style=none] (0) at (0, 0.5) {};
		\node [style=none] (1) at (0, -3.5) {};
		\node [style=none] (2) at (3, -3.5) {};
		\node [style=none] (3) at (3, 0.5) {};
		\node [style=none] (4) at (1.5, -1.5) {$B$};
		\node [style=none] (5) at (3, 0.25) {};
		\node [style=none] (6) at (3, -0.75) {};
		\node [style=none] (7) at (3, -2.25) {};
		\node [style=none] (8) at (3, -3.25) {};
		\node [style=none] (9) at (3.5, -0.75) {};
		\node [style=none] (10) at (3.5, 0.25) {};
		\node [style=none] (11) at (3.5, -3.25) {};
		\node [style=none] (12) at (3.5, -2.25) {};
		\node [style=none] (13) at (3.5, -1.25) {$\vdots$};
		\node [style=none] (14) at (0, 0.25) {};
		\node [style=none] (15) at (0, -0.75) {};
		\node [style=none] (16) at (0, -2.25) {};
		\node [style=none] (17) at (0, -3.25) {};
		\node [style=none] (18) at (-0.5, -0.75) {};
		\node [style=none] (19) at (-0.5, 0.25) {};
		\node [style=none] (20) at (-0.5, -3.25) {};
		\node [style=none] (21) at (-0.5, -2.25) {};
		\node [style=none] (22) at (-0.5, -1.25) {$\vdots$};
		\node [style=none] (23) at (10.75, 1.75) {};
		\node [style=none] (24) at (10.75, -2.25) {};
		\node [style=none] (25) at (13.75, -2.25) {};
		\node [style=none] (26) at (13.75, 1.75) {};
		\node [style=none] (27) at (12.25, -0.25) {$B$};
		\node [style=none] (28) at (13.75, 1.5) {};
		\node [style=none] (29) at (13.75, 0.5) {};
		\node [style=none] (30) at (13.75, -1) {};
		\node [style=none] (31) at (13.75, -2) {};
		\node [style=none] (32) at (9, 0.5) {};
		\node [style=none] (33) at (9, 1.5) {};
		\node [style=none] (34) at (9, -2) {};
		\node [style=none] (35) at (9, -1) {};
		\node [style=none] (36) at (14, -0.15) {$\vdots$};
		\node [style=none] (37) at (10.75, 1.5) {};
		\node [style=none] (38) at (10.75, 0.5) {};
		\node [style=none] (39) at (10.75, -1) {};
		\node [style=none] (40) at (10.75, -2) {};
		\node [style=none] (41) at (10.25, 0.5) {};
		\node [style=none] (42) at (10.25, 1.5) {};
		\node [style=none] (43) at (10.25, -2) {};
		\node [style=none] (44) at (10.25, -1) {};
		\node [style=none] (45) at (10.25, 0) {$\vdots$};
		\node [style=none] (46) at (9, -2.75) {};
		\node [style=none] (47) at (9, -3.25) {};
		\node [style=none] (48) at (9, -4.25) {};
		\node [style=none] (49) at (9, -4.75) {};
		\node [style=none] (50) at (9, -3.625) {$\vdots$};
		\node [style=none] (51) at (13.75, -2.75) {};
		\node [style=none] (52) at (13.75, -3.25) {};
		\node [style=none] (53) at (13.75, -4.25) {};
		\node [style=none] (54) at (13.75, -4.75) {};
		\node [style=none] (55) at (13.75, -3.625) {$\vdots$};
		\node [style=none] (56) at (9, 0) {$\vdots$};
		\node [style=none] (59) at (5.75, -1.5) {$\leftrightsquigarrow$};
	\end{pgfonlayer}
	\begin{pgfonlayer}{edgelayer}
		\draw (3.center) to (0.center);
		\draw (0.center) to (1.center);
		\draw (1.center) to (2.center);
		\draw (2.center) to (3.center);
		\draw [bend left=90, looseness=1.50] (10.center) to (9.center);
		\draw (9.center) to (6.center);
		\draw (5.center) to (10.center);
		\draw (12.center) to (7.center);
		\draw (8.center) to (11.center);
		\draw [bend right=90, looseness=1.50] (11.center) to (12.center);
		\draw [bend right=90, looseness=1.50] (19.center) to (18.center);
		\draw (18.center) to (15.center);
		\draw (14.center) to (19.center);
		\draw (21.center) to (16.center);
		\draw (17.center) to (20.center);
		\draw [bend left=90, looseness=1.50] (20.center) to (21.center);
		\draw (26.center) to (23.center);
		\draw (23.center) to (24.center);
		\draw (24.center) to (25.center);
		\draw (25.center) to (26.center);
		\draw [bend left=90, looseness=1.50] (33.center) to (32.center);
		\draw [bend right=90, looseness=1.50] (34.center) to (35.center);
		\draw [bend right=90, looseness=1.50] (42.center) to (41.center);
		\draw (41.center) to (38.center);
		\draw (37.center) to (42.center);
		\draw (44.center) to (39.center);
		\draw (40.center) to (43.center);
		\draw [bend left=90, looseness=1.50] (43.center) to (44.center);
		\draw [bend right=90, looseness=1.25] (34.center) to (46.center);
		\draw [bend right=90] (35.center) to (47.center);
		\draw [bend right=90, looseness=0.75] (32.center) to (48.center);
		\draw [bend right=90, looseness=0.75] (33.center) to (49.center);
		\draw (46.center) to (51.center);
		\draw (52.center) to (47.center);
		\draw (48.center) to (53.center);
		\draw (54.center) to (49.center);
		\draw [bend right=90, looseness=1.25] (51.center) to (31.center);
		\draw [bend right=90] (52.center) to (30.center);
		\draw [bend left=90, looseness=0.75] (29.center) to (53.center);
		\draw [bend right=90, looseness=0.75] (54.center) to (28.center);
	\end{pgfonlayer}
\end{tikzpicture}

%% file: figs/plat-proof-2.tikz
\begin{tikzpicture}
	\begin{pgfonlayer}{nodelayer}
		\node [style=none] (1) at (1, 1.75) {};
		\node [style=none] (2) at (1, 0.75) {};
		\node [style=none] (3) at (1, -0.75) {};
		\node [style=none] (4) at (1, -1.75) {};
		\node [style=none] (5) at (0.5, 0.75) {};
		\node [style=none] (6) at (0.5, 1.75) {};
		\node [style=none] (7) at (0.5, -1.75) {};
		\node [style=none] (8) at (0.5, -0.75) {};
		\node [style=none] (9) at (0.5, 0.25) {$\vdots$};
		\node [style=none] (10) at (-1.25, 1.75) {};
		\node [style=none] (11) at (-1.25, 0.75) {};
		\node [style=none] (12) at (-1.25, -0.75) {};
		\node [style=none] (13) at (-1.25, -1.75) {};
		\node [style=none] (14) at (-0.75, 0.75) {};
		\node [style=none] (15) at (-0.75, 1.75) {};
		\node [style=none] (16) at (-0.75, -1.75) {};
		\node [style=none] (17) at (-0.75, -0.75) {};
		\node [style=none] (18) at (-0.75, 0.25) {$\vdots$};
		\node [style=none] (19) at (-3.25, 0) {$E$};
		\node [style=none] (20) at (-2.25, 0) {$=$};
		\node [style=none] (21) at (4, 0) {$B \cdot E$};
		\node [style=none] (22) at (5.75, 0) {$=$};
		\node [style=none] (23) at (9, 2) {};
		\node [style=none] (24) at (9, -2) {};
		\node [style=none] (25) at (12, -2) {};
		\node [style=none] (26) at (12, 2) {};
		\node [style=none] (27) at (10.5, 0) {$B$};
		\node [style=none] (28) at (12, 1.75) {};
		\node [style=none] (29) at (12, 0.75) {};
		\node [style=none] (30) at (12, -0.75) {};
		\node [style=none] (31) at (12, -1.75) {};
		\node [style=none] (32) at (7.25, 0.75) {};
		\node [style=none] (33) at (7.25, 1.75) {};
		\node [style=none] (34) at (7.25, -1.75) {};
		\node [style=none] (35) at (7.25, -0.75) {};
		\node [style=none] (36) at (9, 1.75) {};
		\node [style=none] (37) at (9, 0.75) {};
		\node [style=none] (38) at (9, -0.75) {};
		\node [style=none] (39) at (9, -1.75) {};
		\node [style=none] (40) at (8.5, 0.75) {};
		\node [style=none] (41) at (8.5, 1.75) {};
		\node [style=none] (42) at (8.5, -1.75) {};
		\node [style=none] (43) at (8.5, -0.75) {};
		\node [style=none] (44) at (8.5, 0.25) {$\vdots$};
		\node [style=none] (45) at (7.25, 0.25) {$\vdots$};
		\node [style=none] (46) at (12.5, 1.75) {};
		\node [style=none] (47) at (12.5, 0.75) {};
		\node [style=none] (48) at (12.5, -0.75) {};
		\node [style=none] (49) at (12.5, -1.75) {};
		\node [style=none] (50) at (6.75, 1.75) {};
		\node [style=none] (51) at (6.75, 0.75) {};
		\node [style=none] (52) at (6.75, -0.75) {};
		\node [style=none] (53) at (6.75, -1.75) {};
	\end{pgfonlayer}
	\begin{pgfonlayer}{edgelayer}
		\draw [bend right=90, looseness=1.50] (6.center) to (5.center);
		\draw (5.center) to (2.center);
		\draw (1.center) to (6.center);
		\draw (8.center) to (3.center);
		\draw (4.center) to (7.center);
		\draw [bend left=90, looseness=1.50] (7.center) to (8.center);
		\draw [bend left=90, looseness=1.50] (15.center) to (14.center);
		\draw (14.center) to (11.center);
		\draw (10.center) to (15.center);
		\draw (17.center) to (12.center);
		\draw (13.center) to (16.center);
		\draw [bend right=90, looseness=1.50] (16.center) to (17.center);
		\draw (26.center) to (23.center);
		\draw (23.center) to (24.center);
		\draw (24.center) to (25.center);
		\draw (25.center) to (26.center);
		\draw [bend left=90, looseness=1.50] (33.center) to (32.center);
		\draw [bend right=90, looseness=1.50] (34.center) to (35.center);
		\draw [bend right=90, looseness=1.50] (41.center) to (40.center);
		\draw (40.center) to (37.center);
		\draw (36.center) to (41.center);
		\draw (43.center) to (38.center);
		\draw (39.center) to (42.center);
		\draw [bend left=90, looseness=1.50] (42.center) to (43.center);
		\draw (51.center) to (32.center);
		\draw (33.center) to (50.center);
		\draw (52.center) to (35.center);
		\draw (34.center) to (53.center);
		\draw (31.center) to (49.center);
		\draw (48.center) to (30.center);
		\draw (29.center) to (47.center);
		\draw (46.center) to (28.center);
	\end{pgfonlayer}
\end{tikzpicture}

%% file: figs/fsm-1.tikz
\begin{tikzpicture}
	\begin{pgfonlayer}{nodelayer}
		\node [style={Z phase dot (zh)}] (0) at (0, 0) {};
		\node [style={Z phase dot (zh)}] (1) at (2, 0) {};
		\node [style={Z phase dot (zh)}] (2) at (4, 0) {};
		\node [style={Z phase dot (zh)}] (3) at (6, 0) {};
		\node [style=none] (4) at (1.75, -0.5) {};
		\node [style=none] (5) at (0.25, -0.5) {};
		\node [style=none] (6) at (0.25, 0.5) {};
		\node [style=none] (7) at (1.75, 0.5) {};
		\node [style=none] (8) at (3.75, -0.5) {};
		\node [style=none] (9) at (2.25, -0.5) {};
		\node [style=none] (10) at (2.25, 0.5) {};
		\node [style=none] (11) at (3.75, 0.5) {};
		\node [style=none] (12) at (5.75, -0.5) {};
		\node [style=none] (13) at (4.25, -0.5) {};
		\node [style=none] (14) at (4.25, 0.5) {};
		\node [style=none] (15) at (5.75, 0.5) {};
		\node [style=none] (16) at (1, 1.25) {$+1$};
		\node [style=none] (17) at (3, 1.25) {$+1$};
		\node [style=none] (18) at (5, 1.25) {$+1$};
		\node [style=none] (19) at (1, -1.25) {$-1$};
		\node [style=none] (20) at (3, -1.25) {$-1$};
		\node [style=none] (21) at (5, -1.25) {$-1$};
	\end{pgfonlayer}
	\begin{pgfonlayer}{edgelayer}
		\draw [style=diredge, bend left] (4.center) to (5.center);
		\draw [style=diredge, bend left] (6.center) to (7.center);
		\draw [style=diredge, bend left] (8.center) to (9.center);
		\draw [style=diredge, bend left] (10.center) to (11.center);
		\draw [style=diredge, bend left] (12.center) to (13.center);
		\draw [style=diredge, bend left] (14.center) to (15.center);
	\end{pgfonlayer}
\end{tikzpicture}

%% file: figs/fsm-2.tikz
\begin{tikzpicture}
	\begin{pgfonlayer}{nodelayer}
		\node [style={Z phase dot (zh)}] (0) at (-3, 0) {};
		\node [style={Z phase dot (zh)}] (1) at (-1, 0) {};
		\node [style={Z phase dot (zh)}] (2) at (1, 0) {};
		\node [style={Z phase dot (zh)}] (3) at (3, 0) {};
		\node [style=none] (4) at (-1.25, -0.5) {};
		\node [style=none] (5) at (-2.75, -0.5) {};
		\node [style=none] (6) at (-2.75, 0.5) {};
		\node [style=none] (7) at (-1.25, 0.5) {};
		\node [style=none] (8) at (0.75, -0.5) {};
		\node [style=none] (9) at (-0.75, -0.5) {};
		\node [style=none] (10) at (-0.75, 0.5) {};
		\node [style=none] (11) at (0.75, 0.5) {};
		\node [style=none] (12) at (2.75, -0.5) {};
		\node [style=none] (13) at (1.25, -0.5) {};
		\node [style=none] (14) at (1.25, 0.5) {};
		\node [style=none] (15) at (2.75, 0.5) {};
		\node [style=none] (16) at (0, 1.5) {};
		\node [style=none] (17) at (0, -1.5) {};
		\node [style=none] (18) at (5, 0) {$\to$};
		\node [style={Z phase dot (zh)}] (19) at (7, 0) {};
		\node [style={Z phase dot (zh)}] (20) at (9, 0) {};
		\node [style=none] (21) at (8.75, -0.5) {};
		\node [style=none] (22) at (7.25, -0.5) {};
		\node [style=none] (23) at (7.25, 0.5) {};
		\node [style=none] (24) at (8.75, 0.5) {};
		\node [style=none] (25) at (9.25, 0.5) {};
		\node [style=none] (26) at (9.25, -0.5) {};
	\end{pgfonlayer}
	\begin{pgfonlayer}{edgelayer}
		\draw [style=diredge, bend left] (4.center) to (5.center);
		\draw [style=diredge, bend left] (6.center) to (7.center);
		\draw [style=diredge, bend left] (8.center) to (9.center);
		\draw [style=diredge, bend left] (10.center) to (11.center);
		\draw [style=diredge, bend left] (12.center) to (13.center);
		\draw [style=diredge, bend left] (14.center) to (15.center);
		\draw [style=box edge] (17.center) to (16.center);
		\draw [style=diredge, bend left] (21.center) to (22.center);
		\draw [style=diredge, bend left] (23.center) to (24.center);
		\draw [style=diredge, bend left=120, looseness=3.50] (25.center) to (26.center);
	\end{pgfonlayer}
\end{tikzpicture}

%% file: figs/fsm-3.tikz
\begin{tikzpicture}
	\begin{pgfonlayer}{nodelayer}
		\node [style=none] (0) at (0, 0) {};
		\node [style=none] (1) at (0, -1) {};
		\node [style=none] (2) at (0, -2) {};
		\node [style=none] (3) at (0, -3) {};
		\node [style=none] (4) at (5, 0) {};
		\node [style=none] (5) at (5, -1) {};
		\node [style=none] (6) at (5, -2) {};
		\node [style=none] (7) at (5, -3) {};
		\node [style=small black dot] (8) at (0, -3) {};
		\node [style=small black dot] (9) at (1, -2) {};
		\node [style=small black dot] (10) at (2, -1) {};
		\node [style=small black dot] (11) at (3, -2) {};
		\node [style=small black dot] (12) at (4, -1) {};
		\node [style=small black dot] (13) at (5, 0) {};
		\node [style=none] (14) at (7, -1.5) {$\to$};
		\node [style=none] (15) at (9, 0) {};
		\node [style=none] (16) at (9, -1) {};
		\node [style=none] (17) at (9, -2) {};
		\node [style=none] (18) at (9, -3) {};
		\node [style=none] (19) at (14, 0) {};
		\node [style=none] (20) at (14, -1) {};
		\node [style=none] (21) at (14, -2) {};
		\node [style=none] (22) at (14, -3) {};
		\node [style=small black dot] (23) at (9, -3) {};
		\node [style=small black dot] (24) at (10, -2) {};
		\node [style=small black dot] (25) at (11, -2) {};
		\node [style=small black dot] (26) at (12, -2) {};
		\node [style=small black dot] (27) at (13, -2) {};
		\node [style=small black dot] (28) at (14, -3) {};
		\node [style=none] (29) at (5.5, -1.5) {};
		\node [style=none] (30) at (-0.5, -1.5) {};
		\node [style=none] (31) at (14.5, -1.5) {};
		\node [style=none] (32) at (8.5, -1.5) {};
	\end{pgfonlayer}
	\begin{pgfonlayer}{edgelayer}
		\draw (7.center) to (3.center);
		\draw (2.center) to (6.center);
		\draw (5.center) to (1.center);
		\draw (0.center) to (4.center);
		\draw [style=diredge] (8) to (9);
		\draw [style=diredge] (9) to (10);
		\draw [style=diredge] (10) to (11);
		\draw [style=diredge] (11) to (12);
		\draw [style=diredge] (12) to (13);
		\draw (22.center) to (18.center);
		\draw (17.center) to (21.center);
		\draw (20.center) to (16.center);
		\draw (15.center) to (19.center);
		\draw [style=diredge] (23) to (24);
		\draw [style=diredge] (24) to (25);
		\draw [style=diredge] (25) to (26);
		\draw [style=diredge] (26) to (27);
		\draw [style=diredge] (27) to (28);
		\draw [style=box edge] (29.center) to (30.center);
		\draw [style=box edge] (31.center) to (32.center);
	\end{pgfonlayer}
\end{tikzpicture}

%% file: figs/fsm-4.tikz
\begin{tikzpicture}
	\begin{pgfonlayer}{nodelayer}
		\node [style={Z phase dot (zh)}] (0) at (0, 0) {$0$};
		\node [style={Z phase dot (zh)}] (1) at (2, 0) {$1$};
		\node [style=none] (2) at (1.75, -0.5) {};
		\node [style=none] (3) at (0.25, -0.5) {};
		\node [style=none] (4) at (0.25, 0.5) {};
		\node [style=none] (5) at (1.75, 0.5) {};
		\node [style=none] (6) at (2.25, 0.5) {};
		\node [style=none] (7) at (2.25, -0.5) {};
	\end{pgfonlayer}
	\begin{pgfonlayer}{edgelayer}
		\draw [style=diredge, bend left] (2.center) to (3.center);
		\draw [style=diredge, bend left] (4.center) to (5.center);
		\draw [style=diredge, bend left=120, looseness=3.50] (6.center) to (7.center);
	\end{pgfonlayer}
\end{tikzpicture}

%% file: figs/fsm-5.tikz
\begin{tikzpicture}
	\begin{pgfonlayer}{nodelayer}
		\node [style={Z phase dot (zh)}] (0) at (-9, 0) {};
		\node [style={Z phase dot (zh)}] (1) at (-7, 0) {};
		\node [style={Z phase dot (zh)}] (2) at (-5, 0) {};
		\node [style={Z phase dot (zh)}] (3) at (-3, 0) {};
		\node [style=none] (4) at (-7.25, -0.5) {};
		\node [style=none] (5) at (-8.75, -0.5) {};
		\node [style=none] (6) at (-8.75, 0.5) {};
		\node [style=none] (7) at (-7.25, 0.5) {};
		\node [style=none] (8) at (-5.25, -0.5) {};
		\node [style=none] (9) at (-6.75, -0.5) {};
		\node [style=none] (10) at (-6.75, 0.5) {};
		\node [style=none] (11) at (-5.25, 0.5) {};
		\node [style=none] (12) at (-3.25, -0.5) {};
		\node [style=none] (13) at (-4.75, -0.5) {};
		\node [style=none] (14) at (-4.75, 0.5) {};
		\node [style=none] (15) at (-3.25, 0.5) {};
		\node [style=none] (16) at (-5, 1.5) {};
		\node [style=none] (17) at (-5, -1.5) {};
		\node [style=none] (18) at (1, 0) {$\to$};
		\node [style={Z phase dot (zh)}] (19) at (3, 0) {$00$};
		\node [style={Z phase dot (zh)}] (20) at (5, 0) {$01$};
		\node [style=none] (21) at (4.75, -0.5) {};
		\node [style=none] (22) at (3.25, -0.5) {};
		\node [style=none] (23) at (3.25, 0.5) {};
		\node [style=none] (24) at (4.75, 0.5) {};
		\node [style={Z phase dot (zh)}] (27) at (-1, 0) {};
		\node [style=none] (28) at (-1.25, -0.5) {};
		\node [style=none] (29) at (-2.75, -0.5) {};
		\node [style=none] (30) at (-2.75, 0.5) {};
		\node [style=none] (31) at (-1.25, 0.5) {};
		\node [style={Z phase dot (zh)}] (32) at (7, 0) {$10$};
		\node [style=none] (33) at (6.75, -0.5) {};
		\node [style=none] (34) at (5.25, -0.5) {};
		\node [style=none] (35) at (5.25, 0.5) {};
		\node [style=none] (36) at (6.75, 0.5) {};
		\node [style=none] (37) at (-8, 1.25) {$+1$};
		\node [style=none] (38) at (-6, 1.25) {$+1$};
		\node [style=none] (39) at (-4, 1.25) {$+1$};
		\node [style=none] (40) at (-2, 1.25) {$+1$};
		\node [style=none] (41) at (-8, -1.25) {$-1$};
		\node [style=none] (42) at (-6, -1.25) {$-1$};
		\node [style=none] (43) at (-4, -1.25) {$-1$};
		\node [style=none] (44) at (-2, -1.25) {$-1$};
		\node [style={Z phase dot (zh)}] (45) at (-11, -8) {};
		\node [style={Z phase dot (zh)}] (46) at (-9, -8) {};
		\node [style={Z phase dot (zh)}] (47) at (-7, -8) {};
		\node [style={Z phase dot (zh)}] (48) at (-5, -8) {};
		\node [style=none] (49) at (-9.25, -8.5) {};
		\node [style=none] (50) at (-10.75, -8.5) {};
		\node [style=none] (51) at (-10.75, -7.5) {};
		\node [style=none] (52) at (-9.25, -7.5) {};
		\node [style=none] (53) at (-7.25, -8.5) {};
		\node [style=none] (54) at (-8.75, -8.5) {};
		\node [style=none] (55) at (-8.75, -7.5) {};
		\node [style=none] (56) at (-7.25, -7.5) {};
		\node [style=none] (57) at (-5.25, -8.5) {};
		\node [style=none] (58) at (-6.75, -8.5) {};
		\node [style=none] (59) at (-6.75, -7.5) {};
		\node [style=none] (60) at (-5.25, -7.5) {};
		\node [style=none] (61) at (-7, -6.5) {};
		\node [style=none] (62) at (-7, -9.5) {};
		\node [style=none] (63) at (1, -8) {$\to$};
		\node [style={Z phase dot (zh)}] (64) at (3, -8) {$00$};
		\node [style={Z phase dot (zh)}] (65) at (5, -8) {$01$};
		\node [style=none] (66) at (4.75, -8.5) {};
		\node [style=none] (67) at (3.25, -8.5) {};
		\node [style=none] (68) at (3.25, -7.5) {};
		\node [style=none] (69) at (4.75, -7.5) {};
		\node [style={Z phase dot (zh)}] (70) at (-3, -8) {};
		\node [style=none] (71) at (-3.25, -8.5) {};
		\node [style=none] (72) at (-4.75, -8.5) {};
		\node [style=none] (73) at (-4.75, -7.5) {};
		\node [style=none] (74) at (-3.25, -7.5) {};
		\node [style={Z phase dot (zh)}] (75) at (7, -8) {$10$};
		\node [style=none] (76) at (6.75, -8.5) {};
		\node [style=none] (77) at (5.25, -8.5) {};
		\node [style=none] (78) at (5.25, -7.5) {};
		\node [style=none] (79) at (6.75, -7.5) {};
		\node [style=none] (80) at (-10, -6.75) {$+1$};
		\node [style=none] (81) at (-8, -6.75) {$+1$};
		\node [style=none] (82) at (-6, -6.75) {$+1$};
		\node [style=none] (83) at (-4, -6.75) {$+1$};
		\node [style=none] (84) at (-10, -9.25) {$-1$};
		\node [style=none] (85) at (-8, -9.25) {$-1$};
		\node [style=none] (86) at (-6, -9.25) {$-1$};
		\node [style=none] (87) at (-4, -9.25) {$-1$};
		\node [style={Z phase dot (zh)}] (88) at (9, -8) {$11$};
		\node [style=none] (89) at (8.75, -8.5) {};
		\node [style=none] (90) at (7.25, -8.5) {};
		\node [style=none] (91) at (7.25, -7.5) {};
		\node [style=none] (92) at (8.75, -7.5) {};
		\node [style={Z phase dot (zh)}] (93) at (-1, -8) {};
		\node [style=none] (94) at (-1.25, -8.5) {};
		\node [style=none] (95) at (-2.75, -8.5) {};
		\node [style=none] (96) at (-2.75, -7.5) {};
		\node [style=none] (97) at (-1.25, -7.5) {};
		\node [style=none] (98) at (-2, -6.75) {$+1$};
		\node [style=none] (99) at (-2, -9.25) {$-1$};
		\node [style={Z phase dot (zh)}] (100) at (-13, -8) {};
		\node [style=none] (101) at (-11.25, -8.5) {};
		\node [style=none] (102) at (-12.75, -8.5) {};
		\node [style=none] (103) at (-12.75, -7.5) {};
		\node [style=none] (104) at (-11.25, -7.5) {};
		\node [style=none] (105) at (-12, -6.75) {$+1$};
		\node [style=none] (106) at (-12, -9.25) {$-1$};
		\node [style={Z phase dot (zh)}] (107) at (-11, -4) {};
		\node [style={Z phase dot (zh)}] (108) at (-9, -4) {};
		\node [style={Z phase dot (zh)}] (109) at (-5, -4) {};
		\node [style={Z phase dot (zh)}] (110) at (-3, -4) {};
		\node [style=none] (111) at (-9.25, -4.5) {};
		\node [style=none] (112) at (-10.75, -4.5) {};
		\node [style=none] (113) at (-10.75, -3.5) {};
		\node [style=none] (114) at (-9.25, -3.5) {};
		\node [style=none] (115) at (-5.25, -4.5) {};
		\node [style=none] (116) at (-6.75, -4.5) {};
		\node [style=none] (117) at (-6.75, -3.5) {};
		\node [style=none] (118) at (-5.25, -3.5) {};
		\node [style=none] (119) at (-3.25, -4.5) {};
		\node [style=none] (120) at (-4.75, -4.5) {};
		\node [style=none] (121) at (-4.75, -3.5) {};
		\node [style=none] (122) at (-3.25, -3.5) {};
		\node [style=none] (123) at (-6, -2.5) {};
		\node [style=none] (124) at (-6, -5.5) {};
		\node [style=none] (125) at (1, -4) {$\to$};
		\node [style={Z phase dot (zh)}] (126) at (5, -4) {$01$};
		\node [style={Z phase dot (zh)}] (127) at (7, -4) {$10$};
		\node [style=none] (128) at (6.75, -4.5) {};
		\node [style=none] (129) at (5.25, -4.5) {};
		\node [style=none] (130) at (5.25, -3.5) {};
		\node [style=none] (131) at (6.75, -3.5) {};
		\node [style=none] (132) at (7.25, -3.5) {};
		\node [style=none] (133) at (7.25, -4.5) {};
		\node [style={Z phase dot (zh)}] (134) at (-1, -4) {};
		\node [style=none] (135) at (-1.25, -4.5) {};
		\node [style=none] (136) at (-2.75, -4.5) {};
		\node [style=none] (137) at (-2.75, -3.5) {};
		\node [style=none] (138) at (-1.25, -3.5) {};
		\node [style={Z phase dot (zh)}] (139) at (-7, -4) {};
		\node [style=none] (140) at (-7.25, -4.5) {};
		\node [style=none] (141) at (-8.75, -4.5) {};
		\node [style=none] (142) at (-8.75, -3.5) {};
		\node [style=none] (143) at (-7.25, -3.5) {};
		\node [style={Z phase dot (zh)}] (144) at (3, -4) {$00$};
		\node [style=none] (145) at (4.75, -4.5) {};
		\node [style=none] (146) at (3.25, -4.5) {};
		\node [style=none] (147) at (3.25, -3.5) {};
		\node [style=none] (148) at (4.75, -3.5) {};
		\node [style=none] (149) at (-8, -2.75) {$+1$};
		\node [style=none] (150) at (-6, -2.75) {$+1$};
		\node [style=none] (151) at (-4, -2.75) {$+1$};
		\node [style=none] (152) at (-2, -2.75) {$+1$};
		\node [style=none] (154) at (-10, -2.75) {$+1$};
		\node [style=none] (155) at (-8, -5.25) {$-1$};
		\node [style=none] (156) at (-6, -5.25) {$-1$};
		\node [style=none] (157) at (-4, -5.25) {$-1$};
		\node [style=none] (158) at (-2, -5.25) {$-1$};
		\node [style=none] (160) at (-10, -5.25) {$-1$};
	\end{pgfonlayer}
	\begin{pgfonlayer}{edgelayer}
		\draw [style=diredge, bend left] (4.center) to (5.center);
		\draw [style=diredge, bend left] (6.center) to (7.center);
		\draw [style=diredge, bend left] (8.center) to (9.center);
		\draw [style=diredge, bend left] (10.center) to (11.center);
		\draw [style=diredge, bend left] (12.center) to (13.center);
		\draw [style=diredge, bend left] (14.center) to (15.center);
		\draw [style=box edge] (17.center) to (16.center);
		\draw [style=diredge, bend left] (21.center) to (22.center);
		\draw [style=diredge, bend left] (23.center) to (24.center);
		\draw [style=diredge, bend left] (28.center) to (29.center);
		\draw [style=diredge, bend left] (30.center) to (31.center);
		\draw [style=diredge, bend left] (33.center) to (34.center);
		\draw [style=diredge, bend left] (35.center) to (36.center);
		\draw [style=diredge, bend left] (49.center) to (50.center);
		\draw [style=diredge, bend left] (51.center) to (52.center);
		\draw [style=diredge, bend left] (53.center) to (54.center);
		\draw [style=diredge, bend left] (55.center) to (56.center);
		\draw [style=diredge, bend left] (57.center) to (58.center);
		\draw [style=diredge, bend left] (59.center) to (60.center);
		\draw [style=box edge] (62.center) to (61.center);
		\draw [style=diredge, bend left] (66.center) to (67.center);
		\draw [style=diredge, bend left] (68.center) to (69.center);
		\draw [style=diredge, bend left] (71.center) to (72.center);
		\draw [style=diredge, bend left] (73.center) to (74.center);
		\draw [style=diredge, bend left] (76.center) to (77.center);
		\draw [style=diredge, bend left] (78.center) to (79.center);
		\draw [style=diredge, bend left] (89.center) to (90.center);
		\draw [style=diredge, bend left] (91.center) to (92.center);
		\draw [style=diredge, bend left] (94.center) to (95.center);
		\draw [style=diredge, bend left] (96.center) to (97.center);
		\draw [style=diredge, bend left] (101.center) to (102.center);
		\draw [style=diredge, bend left] (103.center) to (104.center);
		\draw [style=diredge, bend left] (111.center) to (112.center);
		\draw [style=diredge, bend left] (113.center) to (114.center);
		\draw [style=diredge, bend left] (115.center) to (116.center);
		\draw [style=diredge, bend left] (117.center) to (118.center);
		\draw [style=diredge, bend left] (119.center) to (120.center);
		\draw [style=diredge, bend left] (121.center) to (122.center);
		\draw [style=box edge] (124.center) to (123.center);
		\draw [style=diredge, bend left] (128.center) to (129.center);
		\draw [style=diredge, bend left] (130.center) to (131.center);
		\draw [style=diredge, bend left=120, looseness=3.50] (132.center) to (133.center);
		\draw [style=diredge, bend left] (135.center) to (136.center);
		\draw [style=diredge, bend left] (137.center) to (138.center);
		\draw [style=diredge, bend left] (140.center) to (141.center);
		\draw [style=diredge, bend left] (142.center) to (143.center);
		\draw [style=diredge, bend left] (145.center) to (146.center);
		\draw [style=diredge, bend left] (147.center) to (148.center);
	\end{pgfonlayer}
\end{tikzpicture}